\documentclass[draftmode]{clases/UCMmatTesis}

\newtheorem{thm}{Theorem}[chapter]
\newtheorem{lem}{Lemma}[chapter]
\newtheorem{cor}{Corollary}[chapter]
\newtheorem{prop}{Proposition}[chapter]

\theoremstyle{definition}
\newtheorem{defn}{Definition}[chapter]

\usepackage{physics}
\usepackage{cite}
\usepackage{bbm}
\usepackage{enumitem}
\usepackage{amsfonts}
\usepackage{mathtools}
\allowdisplaybreaks

\providecommand{\ketbra}[1]{\ket{#1}\bra{#1}}
\providecommand{\trace}{\textnormal{tr}}

\def\conv{\textrm{conv}}

\newcommand{\ui}{\mathrm{i}}
\newcommand{\ue}{\mathrm{e}}

\providecommand{\casename}{Case}
\newlist{casenv}{enumerate}{4}
\setlist[casenv]{leftmargin=*,align=left,widest={iiii}}
\setlist[casenv,1]{label={{\itshape\ \casename} \arabic*.},ref=\arabic*}
\setlist[casenv,2]{label={{\itshape\ \casename} \roman*.},ref=\roman*}
\setlist[casenv,3]{label={{\itshape\ \casename\ \alph*.}},ref=\alph*}
\setlist[casenv,4]{label={{\itshape\ \casename} \arabic*.},ref=\arabic*}

\begin{document}

%% ****************** Informacion del trabajo ********************
% Thesis title and author information, refernce file for biblatex
% ************************ Informacion del Tesis  **********************
%% Titulo del Tesis en ingles
\title{Resource characterisation of \\ quantum entanglement and \\ nonlocality \\ in multipartite settings.}
%% Titulo en español
\titulo{Caracterización del entrelazamiento \\ y no localidad cuánticos como recursos \\ en sistemas multipartitos}

%% Nombre y Apellidos
\author{Patricia Contreras Tejada}

%% Director
\supervisor{Carlos Palazuelos Cabezón}
%% Co-director (descomentar si procede)
\cosupervisor{Julio Íñigo de Vicente Majúa}

%% Full title of the Degree
\degreetitle{Doctorado en Investigación Matemática}

%% Fecha de defensa (mes y año)
\degreedate{2021}

%% ************************* Portada ***************************
%	\makeTesistitle
\pagestyle{empty}
\newgeometry{left=0cm,right=0cm,top=0cm,bottom=0cm}
\begin{figure}[t]
   % \makebox[\linewidth]{
        \includegraphics[width=1\linewidth]{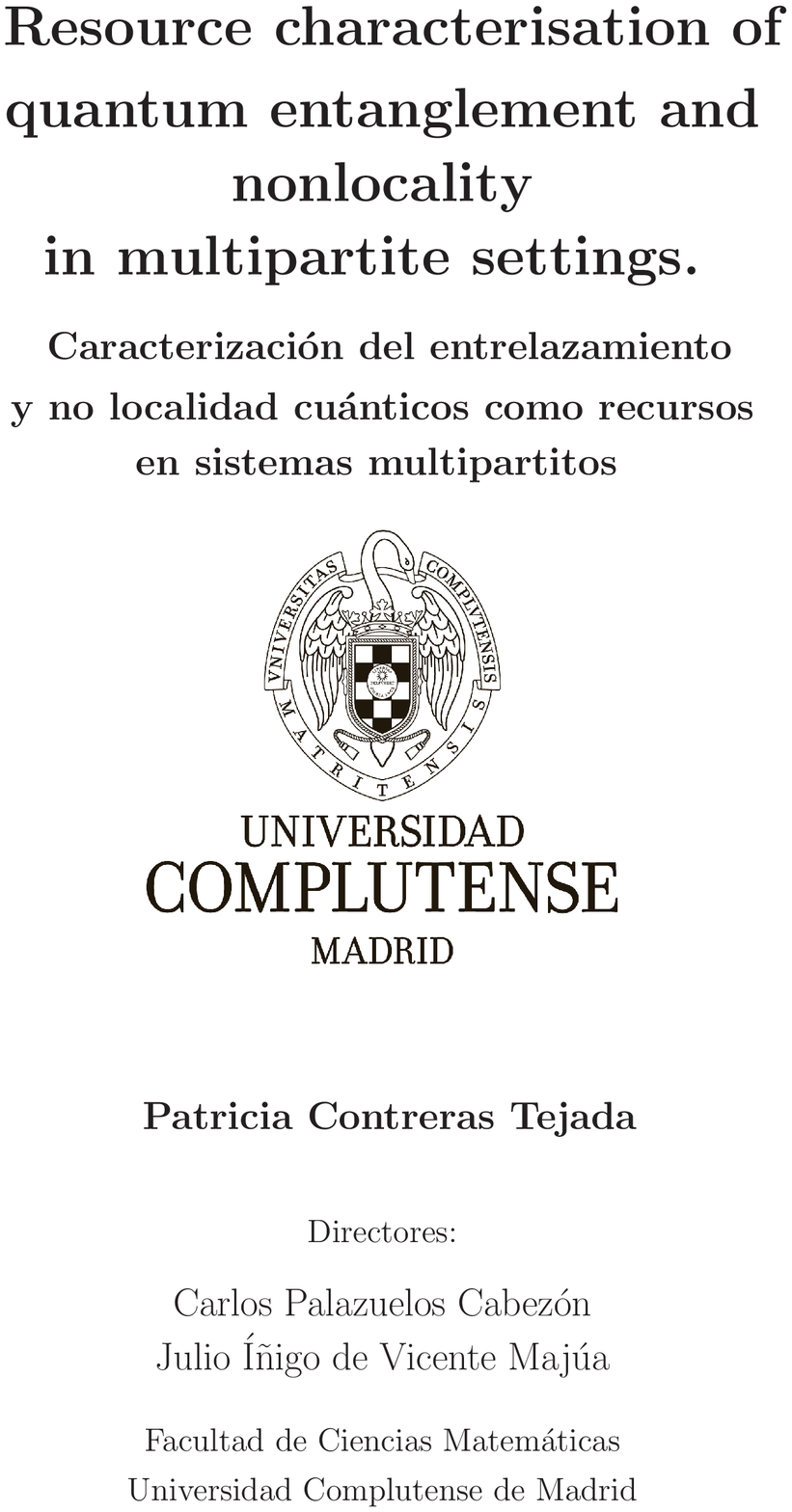}
    %}
\end{figure}
\restoregeometry
%\includepdf{Portada.png}
%	\pagenumbering{roman}

%	 \includepdf[page=1]{blankpage}
%	 \includepdf[page=1]{1declaracion_fda}
%	 \includepdf[page=1]{blankpage}
	
%%	*********************** Dedicatoria ************************

%\input{agradecimientos/dedicatoria}
	
%% ********************* Agradecimientos **********************
%\input{caps/1Acknowledgments}

\cleardoublepage

%% ************************* Indice *****************************
\TesisTableOfContents
\pagestyle{plain}
\listoffigures
\begingroup
\let\clearpage\relax
\listoftables
\endgroup

%%************************ Resumen ***************************
%%% LyX 2.3.0rc1 created this file.  For more info, see http://www.lyx.org/.
%%% Do not edit unless you really know what you are doing.
%\documentclass[11pt,oneside,english]{book}
%\usepackage[T1]{fontenc}
%\usepackage[latin9]{inputenc}
%\usepackage{geometry}
%\geometry{verbose,tmargin=1in,bmargin=1in,lmargin=1in,rmargin=1in}
%\setcounter{secnumdepth}{3}
%\setcounter{tocdepth}{3}
%\usepackage{amsmath}
%\usepackage{amsthm}
%\usepackage{amssymb}
%\usepackage{setspace}
%\onehalfspacing
%
%\makeatletter
%%%%%%%%%%%%%%%%%%%%%%%%%%%%%%% User specified LaTeX commands.
%\usepackage{physics}
%\usepackage{cite}
%\usepackage{bbm}
%\usepackage{enumitem}
%\usepackage{amsfonts}
%
%
%\renewcommand{\theenumi}{(\roman{enumi})}
%\renewcommand{\labelenumi}{\theenumi}
%\providecommand{\ketbra}[1]{\ket{#1}\bra{#1}}
%\providecommand{\trace}{\textnormal{tr}}
%\providecommand{\FSP}{\mathcal{FSP}}
%
%\newcommand{\ui}{\mathrm{i}}
%\newcommand{\ue}{\mathrm{e}}
%
%
%
%\usepackage{babel}
%\providecommand{\definitionname}{Definition}
%\providecommand{\theoremname}{Theorem}
%
%\makeatother
%
%\usepackage{babel}
%\begin{document}

\pagestyle{plain}
\chapter*{\label{chap:abstract}Abstract}

Quantum technologies are enjoying an unprecedented popularity, and
some applications are already in the market. This thesis studies two
phenomena that are behind a lot of quantum technologies: entanglement
and nonlocality. We focus on multipartite systems, and ask what configurations
of those systems are more useful than others. `Usefulness' takes on
different meanings depending on the context, but, roughly speaking,
we aim for more entanglement or more nonlocality.

Chapter \ref{chap:maxent} is motivated by an important issue with
traditional resource theories of multipartite entanglement: they give
rise to isolated states and inequivalent forms of entanglement. We
propose two new resource theories that do not give rise to these problems:
the resource theory of non-full-separability under full separability-preserving
operations, and the resource theory of genuine multipartite entanglement
(GME) under biseparability-preserving operations. Further, the latter
theory gives rise to a unique maximally GME state.

Chapters \ref{chap:gmnl} and \ref{chap:mixed} focus on quantum networks,
that is, configurations where pairs of parties share entangled states,
and parties are bipartitely entangled to one or more of the others.
First, we assume all shared states are pure. It is known that all
connected networks of bipartite pure entangled states are GME (which
is a necessary requirement for being nonlocal) so we ask what networks
give rise to genuine multipartite nonlocality (GMNL). Surprisingly,
they all do: any connected network of bipartite pure entangled states
is GMNL. Next, we allow for the presence of noise, and study networks
of mixed states taking isotropic states as a noise model. Not even
GME is guaranteed in these networks, so our first task is to find
out what networks, in terms of both noise and geometry, give rise
to GME. We find that, unlike in the case of pure states, topology
plays a crucial role: for any non-zero noise, tree networks and polygonal
networks become biseparable if the number of parties is large enough.
In sharp contrast, a completely connected network of isotropic states
is GME for any number of parties as long as the noise is below a threshold.
We further deduce that, while non-steerability of the shared states
can compromise GMNL or even render the networks fully local, taking
many copies of bilocal networks can restore GMNL. Thus, we obtain,
to our knowledge, the first example of superactivation of GMNL.

The thesis so far assumes that quantum theory is an apt description
of Nature. While there are good reasons to believe so, it is possible
that Nature allows for correlations that are stronger than those predicted
by quantum theory, and which we have not yet observed. In order to
find out whether Nature is quantum, one possibility is to devise physical
principles that act as constraints that can rule out post-quantum
theories which are consistent with experimentally observed results.
In a departure from the main ideas explored in previous chapters,
Chapter \ref{chap:agreement} is devoted to developing one such principle.
The principle is inspired in a seminal result in epistemics, which
is the formal study of knowledge and beliefs. We derive two notions
of disagreement in agents' observations of perfectly correlated events:
common certainty of disagreement, and singular disagreement. Both
notions are impossible in classical and quantum agents. Thus, we contend,
the principle of no disagreement must hold for any theory of Nature.

This thesis provides different ways of classifying multipartite systems
in terms of their entanglement and nonlocality. As such, it opens
up new questions, including the possibility of measuring multipartite
entanglement in a unique way, whether it is possible to generate multipartite
nonlocality from a single noiseless entangled state, and which network
topologies are needed to obtain multipartite nonlocality when the
shared states contain noise. Further, the task of constraining the
possible theories of Nature via external principles is a complex one.
We have proposed one step in this direction that, in addition, provides
a way to link quantum information and epistemics more closely. A natural
question that arises is to confirm whether our principle is intrinsic
to Nature, or whether correlations can be found experimentally which
do not satisfy the principle. Alternatively, a strong operational
grounding for the best current description of Nature would be achieved
by completing the list of physical principles that might characterise
quantum theory.
%\end{document}

%%% LyX 2.3.0rc1 created this file.  For more info, see http://www.lyx.org/.
%%% Do not edit unless you really know what you are doing.
%\documentclass[11pt,oneside,english]{book}
%\usepackage[T1]{fontenc}
%\usepackage[utf8]{inputenc}
%\usepackage{geometry}
%\geometry{verbose,tmargin=1in,bmargin=1in,lmargin=1in,rmargin=1in}
%\setcounter{secnumdepth}{3}
%\setcounter{tocdepth}{3}
%\usepackage{amsmath}
%\usepackage{amsthm}
%\usepackage{amssymb}
%\usepackage{setspace}
%\onehalfspacing
%
%\makeatletter
%%%%%%%%%%%%%%%%%%%%%%%%%%%%%%% User specified LaTeX commands.
%\usepackage{physics}
%\usepackage{cite}
%\usepackage{bbm}
%\usepackage{enumitem}
%\usepackage{amsfonts}
%
%
%\renewcommand{\theenumi}{(\roman{enumi})}
%\renewcommand{\labelenumi}{\theenumi}
%\providecommand{\ketbra}[1]{\ket{#1}\bra{#1}}
%\providecommand{\trace}{\textnormal{tr}}
%\providecommand{\FSP}{\mathcal{FSP}}
%
%\newcommand{\ui}{\mathrm{i}}
%\newcommand{\ue}{\mathrm{e}}
%
%
%
%\usepackage{babel}
%\providecommand{\definitionname}{Definition}
%\providecommand{\theoremname}{Theorem}
%
%\makeatother
%
%\usepackage{babel}
%\begin{document}

\chapter*{Resumen}

Las tecnologías cuánticas gozan actualmente de una popularidad sin
precedentes, y ya tienen aplicaciones en el mercado. Esta tesis estudia
dos fenómenos que están detrás de muchas de estas tecnologías: el
entrelazamiento y la no localidad. Nos centramos en sistemas multipartitos,
y tratamos de averiguar qué configuraciones de estos sistemas son
más útiles. La noción de utilidad varía según el contexto pero, en
términos generales, aspiramos a conseguir más entrelazamiento o más
no localidad.

El capítulo \ref{chap:maxent} viene motivado por un problema importante
en las teorías de recursos de entrelazamiento multipartito tradicionales:
dan lugar a estados aislados y a formas de entrelazamiento no equivalentes.
En este capítulo proponemos dos nuevas teorías de recursos que no
generan estos problemas: la teoría de recursos de no-separabilidad-completa
bajo operaciones que preservan separabilidad completa, y la teoría
de recursos de entrelazamiento multipartito genuino (GME, por sus
siglas en inglés) bajo operaciones que preservan biseparabilidad.
Además, esta última teoría da lugar a un estado \sloppy máximamente GME único.

Los capítulos \ref{chap:gmnl} y \ref{chap:mixed} se centran en redes
cuánticas, esto es, configuraciones donde se comparten estados entrelazados
entre pares de agentes, y cada agente está conectado de esta manera
a uno o varios más. Primero asumimos que todos los estados que se
comparten son puros. Se sabe que todas las redes conexas de estados
puros bipartitos entrelazados son GME (condición necesaria para ser
no locales), con lo que nos \sloppy preguntamos qué redes dan lugar a la no
localidad multipartita genuina (GMNL, por sus siglas en inglés). Sorprendentemente,
esto ocurre para todas las redes: cualquier red conexa de estados
puros bipartitos entrelazados es GMNL. A continuación, estudiamos
las redes de estados mezcla para analizar los efectos del ruido. Empleamos
los estados isotrópicos como modelo de ruido. Ni siquiera está garantizado
que estas redes sean GME, así que la primera tarea es investigar qué
redes (a nivel tanto de ruido como de geometría) dan lugar a GME.
Al contrario que en las redes de estados puros, vemos que la topología
juega un papel fundamental: para cualquier nivel de ruido (distinto
de cero), cualquier red en forma de árbol o de polígono se vuelve
biseparable si el número de nodos es lo suficientemente grande. En
el otro extremo, una red totalmente conexa de estados isotrópicos
es GME para cualquier número de nodos si el ruido está por debajo
de un umbral. Deducimos además que, si los estados compartidos no
son direccionables (\emph{steerable}), la red puede volverse bilocal
o incluso completamente local. Sin embargo, la GMNL se puede recuperar
tomando muchas copias de una red bilocal.

Hasta ahora, la tesis asume que la teoría cuántica es una descripción
válida de la naturaleza. Hay razones muy convincentes para creer que
esto es así, aunque sigue siendo posible que en la naturaleza se puedan
dar correlaciones que sean más fuertes de lo que predice la teoría
cuántica, y que aún no se hayan observado. Para investigar si la naturaleza
es cuántica, una posibilidad es desarrollar principios físicos que
actúen de restricciones para eliminar teorías poscuánticas que pudieran
describir los resultados experimentales. Saliéndonos de las ideas
principales de los capítulos anteriores, el capítulo \ref{chap:agreement}
desarrolla uno de estos principios. El principio está inspirado en
un resultado muy influyente en el estudio científico del conocimiento.
Desarrollamos dos nociones de desacuerdo que se aplican a las observaciones
de eventos perfectamente correlacionados por parte de dos agentes:
certeza común de desacuerdo, y desacuerdo singular. Ni los agentes
clásicos, ni tampoco los cuánticos, son susceptibles a estos tipos
de desacuerdo. Por eso argumentamos que el principio de no desacuerdo
debería darse en cualquier teoría de la naturaleza.

Esta tesis ofrece diferentes maneras de clasificar los sistemas multipartitos
de acuerdo a su entrelazamiento y no localidad. Sin duda da lugar
a nuevas preguntas, como la posibilidad de medir el entrelazamiento
multipartito de una manera única, si es posible generar no localidad
multipartita genuina utilizando un único estado entrelazado sin ruido,
y qué topologías se necesitan para obtener no localidad multipartita
a partir de redes de estados con ruido. Además, la tarea de restringir
las posibles teorías de la naturaleza a través de principios externos
es compleja. Hemos propuesto un paso en esta dirección que, además,
da lugar a conexiones más estrechas entre la información cuántica
y el estudio científico del conocimiento. Surge naturalmente la pregunta
de confirmar si este principio es intrínseco a la naturaleza, o si
es posible generar correlaciones experimentales que no lo satisfacen.
Por otro lado, si se completara la lista de principios físicos que
puedan caracterizar la teoría cuántica, esto dotaría de una base operacional
muy sólida a la mejor descripción de la naturaleza que tenemos actualmente.
%\end{document}

%%% LyX 2.3.0rc1 created this file.  For more info, see http://www.lyx.org/.
%%% Do not edit unless you really know what you are doing.
%\documentclass[11pt,oneside,english]{book}
%\usepackage[T1]{fontenc}
%\usepackage[utf8]{inputenc}
%\usepackage{geometry}
%\geometry{verbose,tmargin=1in,bmargin=1in,lmargin=1in,rmargin=1in}
%\setcounter{secnumdepth}{3}
%\setcounter{tocdepth}{3}
%\usepackage{amsmath}
%\usepackage{amsthm}
%\usepackage{setspace}
%\onehalfspacing
%
%\makeatletter
%%%%%%%%%%%%%%%%%%%%%%%%%%%%%%% User specified LaTeX commands.
%\usepackage{physics}
%\usepackage{cite}
%\usepackage{bbm}
%\usepackage{enumitem}
%\usepackage{amsmath}
%\usepackage{amsfonts}
%
%\renewcommand\theenumi{(\roman{enumi})}
%\renewcommand\labelenumi{\theenumi}
%\providecommand{\ketbra}[1]{\ket{#1}\bra{#1}}
%\providecommand{\trace}{\textnormal{tr}}
%\providecommand{\FSP}{\mathcal{FSP}}
%
%\newcommand{\ui}{\mathrm{i}}
%\newcommand{\ue}{\mathrm{e}}
%
%\makeatother
%
%\usepackage{babel}
%\begin{document}
\newpage\vspace*{2cm}
This thesis is based on the following works:

\vspace*{\bigskipamount}

Chapter \ref{chap:maxent}: Patricia Contreras-Tejada, Carlos Palazuelos,
and Julio I. de Vicente, Resource Theory of Entanglement with
a Unique Multipartite Maximally Entangled State, \emph{Phys. Rev. Lett.}
\textbf{122}, 120503 (2019).

Proposition \ref{prop:fsp-noninclusion} and the final observations
are new results.

\vspace*{\bigskipamount}

Chapter \ref{chap:gmnl}: Patricia Contreras-Tejada, Carlos Palazuelos,
and Julio I. de Vicente, Genuine Multipartite Nonlocality Is
Intrinsic to Quantum Networks. \emph{Phys. Rev. Lett.} \textbf{126},
040501 (2021).

The final observation is a new result.

\vspace*{\bigskipamount}

Chapter \ref{chap:mixed}: Patricia Contreras-Tejada, Carlos Palazuelos,
and Julio I. de Vicente, Asymptotic Survival of Genuine Multipartite Entanglement in Noisy Quantum Networks Depends on the Topology. \emph{Phys. Rev. Lett.} \textbf{128}, 220501 (2022).

\vspace*{\bigskipamount}

Chapter \ref{chap:agreement}: Patricia Contreras-Tejada, Giannicola
Scarpa, Aleksander M. Kubicki, Adam Brandenburger, and Pierfrancesco
La Mura, Observers of quantum systems cannot agree to disagree.
\emph{Nat. Commun.} \textbf{12}, 7021 (2021).

\vspace*{1cm}

Funding from the Spanish MINECO through grants SEV-2015-0554-16-3
and MTM2017-88385-P, and from the Comunidad de Madrid through grant
QUITEMAD+CMP2018/TCS-4342, is gratefully acknowledged.
%\end{document}

%%% LyX 2.3.0rc1 created this file.  For more info, see http://www.lyx.org/.
%%% Do not edit unless you really know what you are doing.
%\documentclass[11pt,oneside,english]{book}
%\usepackage[T1]{fontenc}
%\usepackage[utf8]{inputenc}
%\usepackage{geometry}
%\geometry{verbose,tmargin=1in,bmargin=1in,lmargin=1in,rmargin=1in}
%\setcounter{secnumdepth}{3}
%\setcounter{tocdepth}{3}
%\usepackage{amsmath}
%\usepackage{amsthm}
%\usepackage{amssymb}
%\usepackage{setspace}
%\onehalfspacing
%
%\makeatletter
%
%%%%%%%%%%%%%%%%%%%%%%%%%%%%%%% LyX specific LaTeX commands.
%%% Because html converters don't know tabularnewline
%\providecommand{\tabularnewline}{\\}
%
%%%%%%%%%%%%%%%%%%%%%%%%%%%%%%% User specified LaTeX commands.
%\usepackage{physics}
%\usepackage{cite}
%\usepackage{bbm}
%\usepackage{enumitem}
%\usepackage{amsmath}
%\usepackage{amsfonts}
%\usepackage{mathtools}
%
%\renewcommand\theenumi{(\roman{enumi})}
%\renewcommand\labelenumi{\theenumi}
%\providecommand{\ketbra}[1]{\ket{#1}\bra{#1}}
%\providecommand{\trace}{\textnormal{tr}}
%\providecommand{\FSP}{\mathcal{FSP}}
%
%\newcommand{\ui}{\mathrm{i}}
%\newcommand{\ue}{\mathrm{e}}
%
%\makeatother
%
%\usepackage{babel}
%\begin{document}

\chapter*{Notation}

Here are some symbols that appear frequently throughout the text. These are their meanings, unless stated otherwise:

\begin{tabular}{rll}
$n$ & \multicolumn{2}{l}{number of parties}\tabularnewline
$d$ & \multicolumn{2}{l}{local dimension, i.e., dimension of the Hilbert space pertaining to
each party}\tabularnewline
$[n]$ & \multicolumn{2}{l}{${1,...,n}$, the set of $n$ parties}\tabularnewline
$\succcurlyeq$ & \multicolumn{2}{l}{positive semidefinite}\tabularnewline
$\rho,\sigma,\tau$ & \multicolumn{2}{l}{usually denote mixed states ($\sigma$ is usually separable)}\tabularnewline
$\psi$ ($\phi$) & \multicolumn{2}{l}{$\ketbra{\psi}$ ($\ketbra{\phi}$), pure states}\tabularnewline
$\ket{\phi_{d}^{+}}$ ($\ket{\phi^{+}}$) & \multicolumn{2}{l}{maximally entangled state in dimension $d$ (2)}\tabularnewline
$\mathbbm1$ & \multicolumn{2}{l}{identity operator of local dimension $d$}\tabularnewline
$\tilde{\mathbbm1}$ & \multicolumn{2}{l}{$=\mathbbm1/d^{2}$, normalised identity state of local dimension
$d$}\tabularnewline
$\mathcal{H}$ & Hilbert space & \tabularnewline
$\mathcal{B}(\mathcal{H})$ & \multicolumn{2}{l}{the set of bounded operators on $\mathcal{H}$}\tabularnewline
$A,B,A_{i},B_{i}$ & \multicolumn{2}{l}{parties called Alice, Bob, Alice$_{i}$, Bob$_{i}$ respectively}\tabularnewline
$M|\overline{M}$ & \multicolumn{2}{l}{a bipartition $\{M,\overline{M}\}$ of the $n$ parties, where $M\cap\overline{M}=\emptyset$
and $M\cup\overline{M}=[n]$}\tabularnewline
$E_{a|x}$ ($F_{b|y}$) & \multicolumn{2}{l}{POVM element with output $a$ and input $x$ (output $b$ and input
$y$)}\tabularnewline
\multicolumn{1}{r}{$\begin{rcases}P&\\p&\\\mathsf{P}\end{rcases}$} & \multicolumn{2}{l}{probability distributions of $\begin{cases}
\textnormal{inputs and outputs}\\
\textnormal{hidden variables}\\
\text{states of the world}
\end{cases}$ }\tabularnewline
\end{tabular}
%\end{document}

%% ************************* Capitulos *************************
%% Se recomienda escribir cada capitulo en un archivo distinto
%% para evitar grandes tamano de archivos

\pagestyle{cuerpo-tesis} 	%Cambiar formato de pagina
 
%%% LyX 2.3.0rc1 created this file.  For more info, see http://www.lyx.org/.
%%% Do not edit unless you really know what you are doing.
%\documentclass[11pt,oneside,english]{book}
%\usepackage[T1]{fontenc}
%\usepackage[latin9]{inputenc}
%\usepackage{geometry}
%\geometry{verbose,tmargin=1in,bmargin=1in,lmargin=1in,rmargin=1in}
%\setcounter{secnumdepth}{3}
%\setcounter{tocdepth}{3}
%\usepackage{amsmath}
%\usepackage{amsthm}
%\usepackage{amssymb}
%\usepackage{setspace}
%\onehalfspacing
%
%\makeatletter
%%%%%%%%%%%%%%%%%%%%%%%%%%%%%%% Textclass specific LaTeX commands.
%\theoremstyle{plain}
%\newtheorem{thm}{\protect\theoremname}
%\theoremstyle{definition}
%\newtheorem{defn}[thm]{\protect\definitionname}
%
%%%%%%%%%%%%%%%%%%%%%%%%%%%%%%% User specified LaTeX commands.
%\usepackage{physics}
%\usepackage{cite}
%\usepackage{bbm}
%\usepackage{enumitem}
%\usepackage{amsfonts}
%
%
%\renewcommand{\theenumi}{(\roman{enumi})}
%\renewcommand{\labelenumi}{\theenumi}
%\providecommand{\ketbra}[1]{\ket{#1}\bra{#1}}
%\providecommand{\trace}{\textnormal{tr}}
%\providecommand{\FSP}{\mathcal{FSP}}
%
%\newcommand{\ui}{\mathrm{i}}
%\newcommand{\ue}{\mathrm{e}}
%
%
%
%\usepackage{babel}
%\providecommand{\definitionname}{Definition}
%\providecommand{\theoremname}{Theorem}
%
%\makeatother
%
%\usepackage{babel}
%\providecommand{\definitionname}{Definition}
%\providecommand{\theoremname}{Theorem}
%
%\begin{document}

\chapter{\label{chap:intro}Introduction}

Quantum technologies are enjoying an unprecedented popularity. Companies
and governments alike are investing large amounts of money in quantum
research, with high hopes that they will revolutionise our already
highly technified lives. Long-term uses range from making artificial
intelligence more powerful to establishing a quantum internet to communicate
securely, synchronise clocks, perform faster computations or combine
distant telescopes to form much more powerful ones. More modest, but
still noteworthy, applications are already in the market: they include
quantum random number generators, quantum key distribution for very
secure cryptography, or even the GPS.

These applications are only made possible by our ability to control
quantum systems. The special phenomena that are needed for quantum
technologies to work require the physical apparatus used to be extremely
isolated from its surroundings, something which gets harder the bigger
the system is. Otherwise, the apparatus reverts to behaving classically.
Hence the applications that require larger systems are still out of
our reach. However, studying systems of many quantum particles from
a theoretical standpoint is crucial to developing these applications:
we need to know the theory well, in order to know how to exploit it.
And, no less importantly, it is interesting in itself, as it deepens
our understanding of Nature.

This thesis will study two phenomena that are behind a lot of quantum
technologies: entanglement, which is a property of quantum states,
and nonlocality, which is a property of the correlations in the classical
information that we are able to extract from those states. Neither
phenomenon is present in classical systems, hence they sometimes appear
unintuitive since they challenge assumptions that are naturally (and
strongly) held by macroscopic beings such as humans. We will focus
on multipartite systems, and ask what configurations of those systems
are more useful than others. `Usefulness' will take on different meanings
depending on the context, but, roughly speaking, we aim for more entanglement
or more nonlocality. In Chapter \ref{chap:maxent}, we come up with
a way of ordering multipartite states according to the amount of entanglement
they contain. Chapters \ref{chap:gmnl} and \ref{chap:mixed} focus
on quantum networks, that is, configurations where pairs of parties
share entangled states, and parties are bipartitely entangled to one
or more of the others. First, we assume all shared states are pure
(that is, free of noise). It is known that all connected networks
of bipartite pure entangled states are entangled, so we ask what networks
give rise to nonlocality. Surprisingly, they all do: any connected
network of bipartite pure entangled states is nonlocal. Next, we allow
for the presence of noise, and study networks of a particular type
of noisy states. Not even entanglement is guaranteed in these networks,
so our first task is to find out what networks, in terms of both noise
and geometry, give rise to entanglement. We find that, unlike in the
case of pure states, topology plays a crucial role, as does the amount
of noise in the bipartite states. We further deduce that some of these
entangled networks display nonlocality too.

Evidently, quantum technologies rely on the world being quantum. That
is, they take quantum theory at face value, like we have done so far,
and exploit its properties. The fact that more and more quantum effects
are being confirmed in the laboratory to very high precision, and
even made into commercial apparatuses, is a very good indication that
quantum theory is an apt description of Nature. And indeed, it is
the best theory we have so far. However, experimental evidence is
inevitably limited. For example, generating nonlocal correlations
experimentally implies that Nature is not classical. However, it does
not guarantee that it is quantum, since there could be a different
theory underlying the experimentally observed correlations. Indeed,
it is possible that Nature allows for correlations that are stronger
than those predicted by quantum theory, and which we have not yet
observed. So, how can we find out whether Nature is quantum? One possibility
is to devise physical principles that act as constraints that can
rule out post-quantum theories which are consistent with experimentally
observed results. In a departure from the main ideas explored in previous
chapters, Chapter \ref{chap:agreement} is devoted to developing one
such principle. The principle is inspired in a seminal result in epistemics,
which is the formal study of knowledge and beliefs. We derive a notion
of agreement in agents' observations that holds for both classical
and quantum agents, and, we contend, must hold for any theory of Nature.

The rest of this chapter introduces the main technical notions that
will be used throughout this thesis, as well as fixing the notation.
It also motivates and summarises the main results presented in this
thesis.

\section{Quantum formalism}

We present the main notions of quantum theory that are used in this
thesis. Ref. \cite{nielsen_quantum_2000} is by now a classic textbook
of quantum information where the interested reader can find a much
more complete treatment of the ideas in this and the next section.

According to the first postulate of quantum mechanics, the state of
a system is described by a density operator $\rho\in\mathcal{B}(\mathcal{H})$,
the set of bounded operators on a Hilbert space $\mathcal{H}$. Vectors
on $\mathcal{H}$ are denoted by $\ket{\cdot}$, while $\bra{\cdot}$
denotes their duals. States $\rho$ must be positive semidefinite,
$\rho\succcurlyeq0$, and have unit trace, $\tr(\rho)=1$. If a density
operator $\rho$ is a rank-1 projector, i.e. $\rho^{2}=\rho$ or,
equivalently, $\tr\rho^{2}=1$, the system is said to be in a \emph{pure}
state. In this case, the unit vector $\ket{\psi}\in\mathcal{H}$ such
that $\rho=\ketbra{\psi}$ is sufficient to describe the state of
the quantum system, and vectors are often used in place of projective
operators. We use $\psi\equiv\ketbra{\psi}$ whenever a state is specified
as pure. Otherwise, the \emph{mixed} state $\rho$ can be written
as a probabilistic mixture of pure states, 
\begin{equation}
\rho=\sum_{i\in I}p_{i}\ketbra{\psi_{i}},
\end{equation}
where $I$ is an index set, often omitted from the notation, and $\{p_{i}\}_{i\in I}$
is a probability distribution, i.e., $0\leq p_{i}\leq1$ for all $i\in I$,
and $\sum_{i\in I}p_{i}=1$.

States are often written in terms of the \emph{computational basis},
which, for $\dim\mathcal{H}=d$ (so that $\mathcal{H}\cong\mathbb{C}^{d}$),
is represented symbolically as \sloppy $\{\ket{0},\ket{1},...,\ket{d-1}\}$.
Each $\ket{i},i=0,...,d-1$, is a column vector with a 1 in the $(i+1)$th
position and 0 on the rest.

One of the themes of this thesis is to analyse quantum states that
are shared by multiple parties. It is useful to assign names to the
parties, and, following convention, Alice and Bob will be our protagonists,
while Charlie will make an occasional appearance. We refer to parties,
agents, particles, and subsystems interchangeably, denoting them by
name or initial. We will start by defining relevant notions for bipartite
systems, before turning to their multipartite analogues.

The Hilbert space of a joint system is the tensor product of the Hilbert
spaces of its components. Therefore, if $\rho_{A}\in\mathcal{B}(\mathcal{H}_{A})$,
$\rho_{B}\in\mathcal{B}(\mathcal{H}_{B})$ are the states of systems
$A$, $B$ respectively, and the systems are independent, then the
state of the joint system $AB$ is $\rho_{AB}:=\rho_{A}\otimes\rho_{B}\in\mathcal{B}(\mathcal{H}_{A}\otimes\mathcal{H}_{B})$.
However, the state of a joint system is not always the tensor product
of the component states. Indeed, a pure state such as 
\begin{equation}
\ket{\psi}_{AB}=\frac{1}{\sqrt{2}}(\ket{0}_{A}\otimes\ket{1}_{B}+\ket{1}_{A}\otimes\ket{0}_{B})\label{eq:intro:entstate}
\end{equation}
cannot be written as 
\begin{equation}
\ket{a}_{A}\otimes\ket{b}_{B}
\end{equation}
for any $\ket{a}_{A}$, $\ket{b}_{B}$, as can be easily shown by
writing $\ket{a}_{A}$, $\ket{b}_{B}$ in the basis $\{\ket{0},\ket{1}\}$.
For vectors $\ket{\cdot}$, it is customary to omit the tensor product
symbol: $\ket{a}_{A}\otimes\ket{b}_{B}\equiv\ket{a}_{A}\ket{b}_{B}\equiv\ket{ab}_{AB}$.
Moreover, a balance between clarity and readability is sought for
when including or omitting subscripts denoting subsystems.

Given the state of the joint system $AB$, the state of a subsystem,
say $A$, can be found via the partial trace, denoted $\tr_{A}$,
$\tr_{B}$, etc., where the subscript refers to the subsystem(s) to
be traced out: 
\begin{equation}
\tr_{B}(\rho)=\sum_{k=0}^{d-1}\left(\mathbbm1_{A}\otimes\bra{k}_{B}\right)\rho\left(\mathbbm1_{A}\otimes\ket{k}_{B}\right),
\end{equation}
and similarly for $\tr_{A}(\rho)$. As evident from the mathematical
characterisation, the partial trace can only be applied to density
matrices, not vectors. The state of a subsystem obtained by applying
the partial trace on all other subsystems is called the \emph{reduced
state} of that subsystem. For example, if Alice and Bob share the
state $\ket{\psi}_{AB}$ in equation (\ref{eq:intro:entstate}), then
Alice's reduced state is 
\begin{equation}
\begin{aligned}\tr_{B}(\psi_{AB})= & \left(\mathbbm1_{A}\otimes\bra{0}_{B}\right)\left[\frac{1}{2}(\ketbra{01}_{AB}+\ketbra{10}_{AB})\right]\left(\mathbbm1_{A}\otimes\ket{0}_{B}\right)\\
 & +\left(\mathbbm1_{A}\otimes\bra{1}_{B}\right)\left[\frac{1}{2}(\ketbra{01}_{AB}+\ketbra{10}_{AB})\right]\left(\mathbbm1_{A}\otimes\ket{1}_{B}\right)\\
= & \frac{1}{2}\left(\ketbra{0}_{A}+\ketbra{1}_{A}\right).
\end{aligned}
\end{equation}
This state is called the \emph{maximally mixed state}, or the \emph{identity
state} (since it corresponds to a normalised identity on $\mathcal{H}_{A}$).
Similarly, 
\begin{equation}
\tr_{A}(\psi_{AB})=\frac{1}{2}\left(\ketbra{0}_{B}+\ketbra{1}_{B}\right).
\end{equation}
And yet, taking the tensor product of each party's reduced states
does not give back the original state: 
\begin{equation}
\tr_{B}(\psi_{AB})\otimes\tr_{A}(\psi_{AB})\ne\psi_{AB}.
\end{equation}
This would happen if and only if $\psi_{AB}$ could be written as
a tensor product $\ket{a}_{A}\otimes\ket{b}_{B}$. This is the key
observation behind the concept of entanglement, which we will review
in detail shortly.

The second postulate of quantum mechanics characterises the time evolution
of quantum states. We consider closed quantum systems first, which
are those which do not interact with their environment. The evolution
of closed quantum systems is described by unitary operators (sometimes
called `gates' in computational contexts), that is, operators $U$
such that $U^{\dagger}U=\mathbbm1$. If the state $\rho$ evolves
according to a unitary $U$, the evolved state $\rho^{\prime}$ is
\begin{equation}
\rho^{\prime}=U\rho U^{\dagger}=\sum_{i}p_{i}U\ketbra{\psi_{i}}U^{\dagger}.
\end{equation}
That is, if $\rho$ is a mixture of pure states $\ket{\psi_{i}}$,
the evolved state $\rho^{\prime}$ under $U$ is the mixture of the
evolved states $U\ket{\psi_{i}}$ with the same weights. Unitarity
ensures that $U\ket{\psi_{i}}$ and $\rho^{\prime}$ are still well-defined
quantum states.

More generally, the evolution of a quantum system takes into account
the presence of an environment. The system under study together with
its environment do form a closed system, whose evolution must be unitary.
Then, by Stinespring's dilation theorem \cite{stinespring_positive_1955},
the evolution of the system under study must be described by maps
$\Lambda$ that are required to be completely positive and trace-preserving.
Complete positivity ensures that, if $\rho_{S,E}$ is a positive operator
(e.g. describing the joint state of the system $S$ under study and
the environment $E$), then 
\begin{equation}
(\Lambda\otimes\mathbbm1)(\rho_{S,E})
\end{equation}
is still positive. Added to trace preservation, which makes $\tr[\Lambda(\cdot)]=\tr(\cdot)$,
the evolved state of the system $\rho^{\prime}=\tr_{E}[(\Lambda\otimes\mathbbm1)(\rho_{S,E})]$
is guaranteed to be a well-defined quantum state.

Finally, the third postulate concerns measurements. In order to extract
classical information out of quantum states, we can measure them.
Measurement in quantum mechanics is defined by Positive Operator-Valued
Measures, or POVMs, which are sets $\left\{ E_{a}\right\} _{a\in\mathcal{A}}$
of operators $E_{a}$ acting on $\mathcal{H}$ such that $E_{a}\succcurlyeq0$
for all $a\in\mathcal{A}$ and $\sum_{a\in\mathcal{A}}E_{a}=\mathbbm1_{A}$,
where $\mathcal{A}$ is the set of possible \emph{outcomes }or \emph{outputs}
of the measurement $\left\{ E_{a}\right\} _{a\in\mathcal{A}}$. The
probability that state $\rho\in\mathcal{B}(\mathcal{H})$ will give
outcome $a$ on being measured according to $\left\{ E_{a}\right\} _{a\in\mathcal{A}}$
is given by the Born rule: 
\begin{equation}
P(a)=\tr(E_{a}\rho).\label{eq:Bornrule}
\end{equation}
If the POVM elements $E_{a}$ are projectors, then the measurement
is termed \emph{projective}. A particular projective measurement can
be constructed by choosing the rank-1 projectors onto a basis of the
Hilbert space, which are thus positive and add to the identity. In
this case, we sometimes speak of measuring \emph{onto} a basis. A
common choice is the computational basis measurement, $E_{a}=\ketbra{a}$,
for $a=0,...,d-1$.

We often consider measurements performed separately on two parts of
a system, so that Alice performs $\left\{ E_{a}\right\} _{a\in\mathcal{A}}$
and Bob performs a different POVM $\left\{ F_{b}\right\} _{b\in\mathcal{B}}$.
Then, the set of tensor-product operators $\left\{ E_{a}\otimes F_{b}\right\} _{(a,b)\in\mathcal{A}\times\mathcal{B}}$
is a POVM too (or a projective measurement, if both $\left\{ E_{a}\right\} _{a\in\mathcal{A}}$
and $\left\{ F_{b}\right\} _{b\in\mathcal{B}}$ are), as can be straightforwardly
shown.

It is commonly assumed that, upon measurement, states evolve non-unitarily
and `collapse' onto another state that depends on the outcome obtained.\footnote{Taking collapse too literally poses a multitude of problems, both
physical and conceptual. To abate them, it is best to view a measurement
as an operation that entangles the state being measured with the measurement
apparatus, and in fact with its whole environment. By considering
the joint system of the state and the environment, the evolution becomes
unitary, and environment-induced decoherence makes sure that the parts
of the state corresponding to different outcomes do not interact with
each other. In particular, macroscopic observers such as humans never
see states corresponding to more than one outcome in the same measurement
process. Then, collapse can be understood as the effect of this decohering
process, which causes an observer to perceive evolution non-unitarily.
This is only a first step towards tackling the conceptual issues surrounding
measurement in quantum mechanics, which lie outside of the scope of
this thesis.} This evolution can be defined for POVMs and projective measurements
alike, but we will only need the projective case. If $\rho$ is measured
according to a projective measurement $\left\{ E_{a}\right\} _{a\in\mathcal{A}}$
and outcome $a$ is obtained, the post-measurement state $\rho_{a}^{\prime}$
is 
\begin{equation}
\rho_{a}^{\prime}=\frac{E_{a}\rho E_{a}}{\tr(E_{a}\rho)}.
\end{equation}

A very useful tool to study composite systems in quantum mechanics
is the \emph{Schmidt decomposition.} We have seen that, at least outside
of a measurement context, evolution of closed systems is unitary.
In fact, applying a unitary operator amounts only to changing the
basis in which the quantum state is represented. Indeed, equivalence
classes of states that are equal up to local unitaries are often considered.
The Schmidt decomposition of a state is a standard form for these
equivalence classes. 
\begin{thm}
For any pure, bipartite state $\ket{\psi}\in\mathcal{H}_{A}\otimes\mathcal{H}_{B}$,
there exist orthonormal bases $\left\{ \ket{i}_{A}\right\} _{i}\subset\mathcal{H}_{A}$,
$\left\{ \ket{i}_{B}\right\} _{i}\subset\mathcal{H}_{B}$ such that
\begin{equation}
\ket{\psi}=\sum_{i=0}^{d-1}\sqrt{\lambda_{i}}\ket{i}_{A}\ket{i}_{B},
\end{equation}
where the \emph{Schmidt coefficients }$\lambda_{i}$ are real numbers
satisfying $\lambda_{i}\geq0$ and $\sum_{i}\lambda_{i}=1$, and $d=\min\{d_{A},d_{B}\}$
where $d_{A},d_{B}$ is the dimension of $\mathcal{H}_{A},\mathcal{H}_{B}$
respectively. The basis $\left\{ \ket{i}_{A}\ket{i}_{B}\right\} _{i}\subset\mathcal{H}_{A}\otimes\mathcal{H}_{B}$
is called the \emph{Schmidt basis }of the state\emph{ $\ket{\psi}$.}
\end{thm}

This theorem first appeared in Ref. \cite{schmidt_zur_1907}, and
can be proven by arranging the coefficients of $\ket{\psi}$, written
in terms of any basis, in a $d_{A}\times d_{B}$ matrix, and finding
its singular value decomposition.

\subsection{Entanglement in bipartite systems\label{subsec:ent-bipart}}

Entanglement is a key notion in quantum mechanics, as there is no
classical analogue. Conceptually, entanglement is a property of composite
systems whereby a full description of each of the subsystems does
not provide the full information about the system as a whole. In the
early days of quantum mechanics, this was one of the most puzzling
phenomena that gave quantum mechanics an aura of being difficult,
weird or unintuitive. Faced with an entangled system, Einstein, Podolsky
and Rosen famously concluded that quantum theory could not be complete
\cite{einstein_can_1935}. The fact that there could be more information
in the system than that contained in its subsystems challenged the
previously held assumption that subsystems have inherent properties
that can be measured.

Mathematically, entanglement is a consequence of the tensor product
structure. As anticipated above, a pure state is said to be \emph{entangled}
if it cannot be written as the tensor product of the states of its
subsystems. Otherwise, it is \emph{separable.} For example, $\ket{\psi}_{AB}$
in equation (\ref{eq:intro:entstate}) above is entangled. More generally,
a mixed state is entangled if it cannot be decomposed into pure, separable
states: 
\begin{defn}
\label{def:intro-ent-bipart}A state $\rho\in\mathcal{B}(\mathcal{H}_{A}\otimes\mathcal{H}_{B})$
is \emph{separable} if there exist states $\ket{\eta_{i}}_{A}\in\mathcal{H}_{A}$,
$\ket{\chi_{i}}_{B}\in\mathcal{H}_{B}$ and probability distribution
$\{p_{i}\}_{i}$ such that 
\begin{equation}
\rho=\sum_{i}p_{i}\ketbra{\eta_{i}}_{A}\otimes\ketbra{\chi_{i}}_{B}.
\end{equation}
Otherwise, $\rho$ is \emph{entangled}. 
\end{defn}

To show that a state is separable, it is sufficient to find states
$\eta_{i},\chi_{i}$ to decompose it, according to Definition \ref{def:intro-ent-bipart}.
Still, this problem is NP-hard \cite{gurvits_classical_2003,gharibian_strong_2010}.
Showing that a state is entangled is clearly not as straightforward
in principle. However, a class of operators known as \emph{entanglement
witnesses} \cite{horodecki_separability_1996,terhal_bell_2000}, proves
very useful for this purpose: 
\begin{defn}
An operator $W$ acting on $\mathcal{H}_{A}\otimes\mathcal{H}_{B}$
is an \emph{entanglement witness} (or simply a \emph{witness}) if,
for all separable states $\sigma$, 
\begin{equation}
\tr(W\sigma)\geq0.
\end{equation}
If a state $\rho$ is such that 
\begin{equation}
\tr(W\rho)<0,
\end{equation}
the witness $W$ is said to \emph{detect} $\rho$. 
\end{defn}

Often, witnesses are required not to be positive operators, otherwise
they are useless as they do not detect any state. The existence of
witnesses follows from the Hahn-Banach theorem \cite{hahn_ueber_1927,banach_sur_1929,banach_sur_1929-1}:
since the set of separable states is convex and compact, there exists
a hyperplane separating it from any entangled state, which is a point
outside of the separable set.

Another useful criterion to detect entanglement is the PPT criterion,
or Peres-Horodecki criterion \cite{peres_separability_1996,horodecki_separability_1996}.
It is based on the partial transpose, an operation defined on composite
states that transposes the part corresponding to one party while leaving
the other untouched. Let $\ket{\psi}$ be a pure state, with Schmidt
decomposition $\ket{\psi}=\sum_{i}\sqrt{\lambda_{i}}\ket{ii}_{AB}$.
Then, its partial transpose with respect to system $A$ is 
\begin{equation}
\begin{aligned}\psi^{\Gamma_{A}} & =\sum_{i,j}\sqrt{\lambda_{i}\lambda_{j}}\left(\ket{i}\bra{j}_{A}\right)^{T}\left(\ket{i}\bra{j}_{B}\right)\\
 & =\sum_{i,j}\sqrt{\lambda_{i}\lambda_{j}}\left(\ket{j}\bra{i}_{A}\right)\left(\ket{i}\bra{j}_{B}\right)\\
 & \equiv\sum_{i,j}\sqrt{\lambda_{i}\lambda_{j}}\ket{ji}\bra{ij}_{AB}.
\end{aligned}
\end{equation}
Since the Schmidt coefficients $\lambda_{i}$ are real, we have $\psi^{\Gamma_{A}}=\psi^{\Gamma_{B}}\equiv\psi^{\Gamma}$.
This definition extends linearly to mixed states. The transpose is
not a completely positive operation, hence $\rho^{\Gamma}$ may not
be a positive operator although any quantum state $\rho$ is. In fact,
if the partial transpose of a quantum state (with respect to either
party) is not a positive operator, then the state is entangled. The
converse is only true when the dimensions of the local Hilbert spaces
are $2\times2$ or $2\times3$. A state with a positive partial transpose
is termed \emph{PPT}.

This criterion can be used to find witnesses that are positive for
all PPT states. While some entangled states will escape detection,
these witnesses come in useful in many situations: they can be found
with a semi-definite program, and it is simple to prove that an operator
is a PPT witness: an operator $W$ is a PPT witness if and only if
it is decomposable, i.e., if there exist positive operators $P,Q$
such that $W=P+Q^{\Gamma}$\cite{lewenstein_optimization_2000}.

Entangled states are interesting from a fundamental point of view
since there is no analogue in classical systems. But further, from
a practical point of view, they are a useful asset to own, since they
help perform some communication-related tasks more efficiently. For
example, they can be used to encode information in fewer bits than
are needed without the help of entanglement, to generate a key for
cryptographic protocols, or even as channels to transmit quantum information.

One prime example of the usefulness of entangled states is the \emph{teleportation
protocol} \cite{bennett_teleporting_1993}, which uses a maximally
entangled state together with classical communication in order to
send an unknown state from one place to another (thus transmitting
quantum information). Of course, if Alice knew the description of
the state she wanted to transmit, she might be able to send it to
Bob via a classical message, in order for Bob to recreate the state
in his laboratory. However, teleportation still works if Alice does
not know what state she has, and even if Alice's particle is entangled
to some other particle: teleportation will preserve that entanglement.

In more detail, suppose that Alice wants to teleport a 2-dimensional
state $\ket{\psi}_{A^{\prime}}=\alpha\ket{0}_{A^{\prime}}+\beta\ket{1}_{A^{\prime}}$,
where $|\alpha|^{2}+|\beta|^{2}=1$, to Bob. We assume $\ket{\psi}$
is pure for the sake of simplicity, but the extension to mixed states
is immediate by linearity. The protocol requires that Alice and Bob
share a state $\ket{\phi_{AB}^{+}}:=(\ket{00}_{AB}+\ket{11}_{AB})/\sqrt{2}$
(this state is the maximally entangled state in dimension 2, as we
shall see later). Therefore, the joint state of Alice's and Bob's
particles is
\begin{equation}
\ket{\psi}_{A^{\prime}}\ket{\phi^{+}}_{AB}=\frac{1}{\sqrt{2}}\left(\alpha\ket{0}_{A^{\prime}}(\ket{00}_{AB}+\ket{11}_{AB})+\beta\ket{1}_{A^{\prime}}(\ket{00}_{AB}+\ket{11}_{AB})\right).
\end{equation}
First, Alice applies a CNOT gate to her particles. This is a Controlled-NOT
gate, meaning that it flips the state of the second qubit (from $\ket{0}$
to $\ket{1}$ and vice versa) if and only if the first qubit is in
state $\ket{1}$. Therefore, the state becomes
\begin{equation}
\frac{1}{\sqrt{2}}\left(\alpha\ket{0}_{A^{\prime}}(\ket{00}_{AB}+\ket{11}_{AB})+\beta\ket{1}_{A^{\prime}}(\ket{10}_{AB}+\ket{01}_{AB})\right).
\end{equation}
Next, Alice applies a Hadamard gate to her first qubit. A Hadamard
gate $H$ is the unique unitary gate that maps $\ket{0}$ to $\ket{+}:=(\ket{0}+\ket{1})/\sqrt{2}$
and $\ket{1}$ to $\ket{-}:=(\ket{0}-\ket{1})/\sqrt{2}$. Therefore,
the joint state is now
\begin{equation}
\frac{1}{2}\left(\alpha(\ket{0}_{A^{\prime}}+\ket{1}_{A^{\prime}})(\ket{00}_{AB}+\ket{11}_{AB})+\beta(\ket{0}_{A^{\prime}}-\ket{1}_{A^{\prime}})(\ket{10}_{AB}+\ket{01}_{AB})\right).
\end{equation}
But, by regrouping the states of the particles that Alice holds, this
state can be rewritten as
\begin{equation}
\begin{aligned}\frac{1}{2} & \left[\ket{00}_{A^{\prime}A}(\alpha\ket{0}_{B}+\beta\ket{1}_{B})+\ket{01}_{A^{\prime}A}(\alpha\ket{1}_{B}+\beta\ket{0}_{B})\right.\\
 & \left.+\ket{10}_{A^{\prime}A}(\alpha\ket{0}_{B}-\beta\ket{1}_{B})+\ket{11}_{A^{\prime}A}(\alpha\ket{1}_{B}-\beta\ket{0}_{B})\right].
\end{aligned}
\end{equation}
Now, Alice can measure her qubits in the computational basis and thus
find out what state Bob's particle is in: if she obtains 00, Bob's
particle is in the original state $\ket{\psi}$ that Alice held. If
she obtains 01, 10, or 11, Bob can recover the state $\ket{\psi}$
by performing a local unitary on his particle. Therefore, when Alice
communicates her measurement outcome to Bob (which requires 2 bits
of classical communication), he can recover the state that she originally
held, even if the values of $\alpha$ and $\beta$ are not known to
either party. By the end of the protocol, the state $\phi^{+}$ is
no longer available to the parties: the entanglement they shared has
been consumed.

Indeed, entanglement is consumed after other protocols too. Further,
entangled states cannot be generated by acting separately on each
system---not even with the help of classical communication. This
is why they are a precious resource which will be the object of study
of a large part of this thesis. Protocols such as teleportation, as
well as superdense coding, quantum key distribution, and many others,
all of which rely on entanglement, act on entangled states by applying
operations locally on each subsystem, and communicating the results
of some of these operations via classical channels. This is the paradigm
of allowed operations that is often assumed when studying entangled
states: Local Operations and Classical Communication, or LOCC.\textbf{
}While its mathematical description is somewhat involved in general
(see, e.g., \cite{chitambar_everything_2014,hebenstreit_measurement_2021}),
since local operations may depend on prior measurement results and
communication rounds, conceptually it is a very useful tool. As hinted
above, LOCC operations cannot generate entanglement. Moreover, while
most well-known LOCC protocols consume the entanglement completely,
in general they might only degrade it, and indeed LOCC cannot increase
the amount of entanglement contained in a state (a precise way of
measuring the entanglement of quantum states will be introduced shortly).

This is the basic intuition behind the notion of a resource theory.
The aim of a resource theory is to order states according to their
usefulness for practical tasks. While the framework of resource theories
can be applied to many different scenarios \cite{chitambar_quantum_2019},
including coherence \cite{streltsov_colloquium_2017}, reference frame
alignment \cite{bartlett_reference_2007,gour_resource_2008}, noncontextuality
\cite{amaral_noncontextual_2018,duarte_resource_2018}, thermodynamics
\cite{brandao_resource_2013,gour_resource_2015}, nonlocality \cite{vicente_nonlocality_2014,gallego_nonlocality_2017},
steering \cite{gallego_resource_2015}, and many more, in this work
we shall be concerned only with the resource theory of entanglement.
In this context, LOCC is often taken as the set of \emph{free operations},
i.e., those which are accessible to the agents at no cost. Hence,
the set of \emph{free states }contains those states that can be prepared
using only the free operations, namely, the set of separable states.
Entangled states (which cannot be prepared by LOCC) become a resource,
and the free operations determine their relative usefulness: if a
state $\phi$ can be converted by a free operation into another state
$\psi$, then $\phi$ is at least as useful as $\psi$. Indeed, if
$\psi$ is needed for some task, but only $\phi$ is available, free
operations can be used to convert $\phi$ into $\psi$ before performing
the task, while the converse need not be true. Any entangled state
can be converted into a separable state using LOCC (a simple protocol
is to ignore the input state and generate the desired separable state),
therefore any entangled state is more useful than all separable states.

A resource theory considers all the possible conversions between resource
states so as to obtain the induced partial order on this set. The
possible LOCC interconversions among pure bipartite states where characterised
by Nielsen \cite{nielsen_conditions_1999}, who showed that the LOCC
ordering reduces to majorisation \cite{marshall_inequalities_2011}:
$\phi$ can be converted into $\psi$ if and only if the Schmidt coefficients
$\lambda_{i},\mu_{i}$ of $\phi,\psi$ respectively are such that,
for each $k=0,...,d-1$, 
\begin{equation}
\sum_{i=0}^{k-1}\lambda_{i}^{\downarrow}\leq\sum_{i=0}^{k-1}\mu_{i}^{\downarrow},
\end{equation}
where the superscript indicates that the coefficients are taken in
descending order. There are some entangled states that are incomparable
under LOCC, since neither set of coefficients majorises the other,
and hence none can be converted into the other using these free operations.
However, this ordering gives rise to a \emph{maximally entangled state}
for every dimension $d$, 
\begin{equation}
\ket{\phi_{d}^{+}}=\frac{1}{\sqrt{d}}\sum_{i=0}^{d-1}\ket{ii},\label{eq:intro-maxent}
\end{equation}
which can be converted to any other state of the same dimension using
LOCC. It is taken as the gold standard to measure entanglement and,
unsurprisingly, this state is also the most useful one in some common
protocols such as teleportation or superdense coding. The 2-dimensional
maximally entangled state is denoted $\ket{\phi^{+}}$.

The order that a resource theory induces on entangled states can be
used to quantify their entanglement. \emph{Entanglement measures}
are mappings $E:\mathcal{B}(\mathcal{H})\rightarrow\mathbb{R}^{+}$
from density operators into non-negative real numbers such that $E(\rho)=0$
if $\rho$ is separable. Moreover, $E$ must not increase under LOCC
operations $\Lambda:\mathcal{B}(\mathcal{H})\rightarrow\mathcal{B}(\mathcal{H})$
performed on $\rho$: 
\begin{equation}
E[\Lambda(\rho)]\leq E(\rho).
\end{equation}
This entails that, if $\rho$ is at least as entangled as $\tau$
according to the LOCC resource-theoretic ordering, then $E(\rho)\geq E(\tau)$.

One commonly used entanglement measure is the \emph{relative entropy
of entanglement} \label{entropyent} \cite{vedral_quantifying_1997,vedral_entanglement_1998},
which, for pure states, is the von Neumann entropy $H$ of the single-party
reduced state: letting $(\psi)_{A}:=\tr_{B}(\psi)$, 
\begin{equation}
E(\psi)=H[(\psi)_{A}]=-\tr[(\psi)_{A}\log(\psi)_{A}].\label{eq:intro-entropy}
\end{equation}
The generalisation to mixed states is called the \emph{entanglement
of formation, }and it is done via the convex roof construction, i.e.,
by minimising over all possible decompositions of a mixed state into
pure states:
\begin{equation}
E(\rho)=\min_{\{p_{i},\psi_{i}\}}\sum_{i}p_{i}E(\psi_{i}).
\end{equation}

In fact, it is often required that other entanglement measures reduce
to the entropy of entanglement when considering pure states \cite{plenio_introduction_2007}.
The $d$-dimensional maximally entangled state (equation (\ref{eq:intro-maxent}))
has $E(\phi_{d}^{+})=\log d$, which is the maximum value of $E$
(see equation (\ref{eq:intro-entropy})). This is also a common normalisation
requirement for entanglement measures.

In addition, the \emph{robustness} $R$ \label{robustness}\cite{vidal_robustness_1999}
intuitively captures the distance of a state $\rho$ from the set
of separable states. More specifically, the robustness of $\rho$
quantifies the weight needed to mix $\rho$ with the best choice of
separable state, in order to obtain a separable state:
\begin{equation}
R(\rho)=\min_{\sigma\in\mathcal{S}}R(\rho||\sigma),
\end{equation}
where 
\begin{equation}
R(\rho||\sigma)=\min\left\{ s:\frac{\rho+s\sigma}{1+s}\in\mathcal{S}\right\} ,
\end{equation}
and $\mathcal{S}$ is the set of separable states. If the state $\sigma$
is not required to be separable, the corresponding measure is termed
\emph{generalised robustness}.

Another useful measure is the \emph{geometric measure of entanglement}
$G$ \label{geommeas} \cite{wei_geometric_2003}. For pure states,
it is related to the maximum overlap of the state with a separable
state: 
\begin{equation}
G(\psi)=1-\max_{\sigma\in\mathcal{S}}\tr(\psi\sigma),
\end{equation}
where $\mathcal{S}$ is the set of separable states (note that $\sigma$
can be assumed pure, without loss of generality). To cover mixed states,
we use the convex roof construction, so that 
\begin{equation}
G(\rho)=\min_{\{p_{i},\psi_{i}\}}\sum_{i}p_{i}G(\psi_{i}),
\end{equation}
where $\rho=\sum_{i}p_{i}\psi_{i}$. This ensures that $G$ is a well-defined
measure.

Many more entanglement measures are known (see, e.g., Ref. \cite{plenio_introduction_2007}),
but we shall only be concerned with the above ones in this work.

A related concept is the \emph{fidelity} between two quantum states,
which is a measure of their closeness. The fidelity $F$ is defined
as 
\begin{equation}
F(\rho,\sigma)=\tr^{2}\sqrt{\sqrt{\rho}\sigma\sqrt{\rho}},\label{eq:intro-fidelity}
\end{equation}
and it is symmetric (i.e., $F(\rho,\sigma)=F(\sigma,\rho)$). If one
of the states is pure (say, if $\sigma=\ketbra{\psi}$), it takes
a simpler form: 
\begin{equation}
F(\rho,\psi)=\tr(\rho\psi)=\bra{\psi}\rho\ket{\psi}.
\end{equation}
(The fidelity is sometimes defined as the square root of the quantity
in equation (\ref{eq:intro-fidelity}), but the choice made here means
the $F$ in equation (\ref{eq:intro-fidelity}) is linear in each
argument when the other is fixed.)\label{fidelity}

While LOCC cannot be used to transform a less entangled state into
a more entangled one, it is possible to transform \emph{many }less
entangled states into \emph{few }more entangled ones. This process
is called \emph{distillation of entanglement,} and usually refers
to transforming many copies of any input state into few copies of
the maximally entangled state. In practice, perfect maximally entangled
states cannot be achieved, but rather, the goal is to obtain states
that are close (in terms of their fidelity) to the maximally entangled
state. As the number of copies of the input state grows unboundedly,
the fidelity can be made arbitrarily close to 1. The \emph{rate }of
a distillation protocol is the ratio between the number of copies
of the input state and the number of copies of the output state. The
best achievable rate for a given input state is termed the \emph{distillable
entanglement} of that state, which happens to be a measure of entanglement
(in fact, for pure states it is equal to the relative entropy of entanglement),
although it will not be used in this work.

\subsection{Entanglement in multipartite systems\label{subsec:ent-multipart}}

When considering more than two systems, the definition of entanglement
becomes ambiguous: indeed, we can have Alice, Bob and Charlie holding
three subsystems in tensor product, or Alice and Bob sharing an entangled
state which is separable from Charlie's, or all three sharing a truly
tripartite entangled state, and these cases do not fit Definition
\ref{def:intro-ent-bipart}. Instead, for a system of $n$ parties
we define: 
\begin{defn}
\label{def:intro-ent-multipart}A state $\rho\in\mathcal{B}(\mathcal{H}_{1}\otimes\dots\otimes\mathcal{H}_{n})$,
where $n\in\mathbb{N}$, is \emph{fully separable }if there exist
$\ket{\psi_{i,1}}\in\mathcal{H}_{1}$, ..., $\ket{\psi_{i,n}}\in\mathcal{H}_{n}$
and a probability distribution $\{p_{i}\}_{i}$ such that 
\begin{equation}
\rho=\sum_{i}p_{i}\psi_{i,1}\otimes\dots\otimes\psi_{i,n}.
\end{equation}

If $\rho$ is not fully separable, it is \emph{entangled.}

The state $\rho$ is \emph{biseparable} if there exist $\ket{\psi_{i,M}}\in\bigotimes_{j\in M}\mathcal{H}_{j}$,
$\ket{\psi_{i,\overline{M}}}\in\bigotimes_{j\in\overline{M}}\mathcal{H}_{j}$
for bipartitions $\{M,\overline{M}\}$ of $[n]$ and a probability
distribution $\{p_{i}\}_{i}$ such that 
\begin{equation}
\rho=\sum_{i,M}p_{i,M}\psi_{i,M}\otimes\psi_{i,\overline{M}}.
\end{equation}

If $\rho$ is not biseparable, it is \emph{genuine multipartite entangled
(GME)}. 
\end{defn}

Finer-grained versions of this definition can be found by considering
partitions of $[n]$ of different numbers $k\in[n]$ of elements,
giving rise to the concept of $k$-separability, but we shall not
be concerned with this here.

Two well-known GME states are the W state, 
\begin{equation}
\ket{W}=\frac{1}{\sqrt{2}}(\ket{001}_{ABC}+\ket{010}_{ABC}+\ket{100}_{ABC}),
\end{equation}
and the Greenberger-Horne-Zeilinger (GHZ) state, 
\begin{equation}
\ket{GHZ}=\frac{1}{\sqrt{2}}(\ket{000}_{ABC}+\ket{111}_{ABC}),
\end{equation}
which can be generalised to any dimension $d$ and number of parties
$n$ as 
\begin{equation}
\ket{GHZ(n,d)}=\frac{1}{\sqrt{d}}\sum_{i=0}^{d-1}\ket{i}^{\otimes n}.\label{eq:intro-ghz-gen}
\end{equation}

A consequence of the definition of biseparability that will be exploited
in this work is that it is not closed under tensor products: for example,
the tensor product of the states $\ket{\phi^{+}}_{AB}\otimes\ket{0}_{C}$
and $\ket{\phi^{+}}_{AC}\otimes\ket{0}_{B}$, which are both biseparable,
is 
\begin{equation}
\frac{1}{2}\left(\ket{00,00,00}_{ABC}+\ket{01,00,01}_{ABC}+\ket{10,10,00}_{ABC}+\ket{11,10,01}_{ABC}\right),\label{eq:intro-gme-superact}
\end{equation}
which cannot be written as $\ket{\psi_{M}}\otimes\ket{\psi_{\overline{M}}}$
for any bipartition $M|\overline{M}$ of $A,B,C$. Since the state
in equation (\ref{eq:intro-gme-superact}) is pure, this means it
is GME.

Since the sets of fully separable and biseparable states are convex,
witnesses can be used to detect entanglement in multipartite systems
as well as bipartite ones. An operator which has positive trace with
all fully separable states detects states that contain, at least,
bipartite entanglement; a witness which detects GME has positive trace
with all biseparable states. ($\mathcal{W}$ or other letters will
be used to denote multipartite witnesses if there is risk of confusion
with the state $W$.)

Similarly to the bipartite case, to study GME states it is sometimes
useful to consider a slightly larger set of states, namely, the set
of PPT mixtures \cite{jungnitsch_taming_2011}. This set contains
states that can be written as a mixture of PPT states, possibly in
different bipartitions. Just like bipartite separability implies PPT,
but not vice versa, the set of PPT mixtures strictly contains that
of biseparable states. Further, operators which have positive trace
with all PPT mixtures are, in particular, GME witnesses. These witnesses
$W$ can be written as $W=P_{M}+Q_{M}^{\Gamma_{M}}$, where $P_{M},Q_{M}$
are positive operators for each bipartition $M|\overline{M}$.

\section{Probability distributions\label{sec:prob}}

As mentioned in the previous section, classical information (i.e.,
information that us humans can access) can be obtained out of quantum
states by measuring them. It turns out that measurements on some quantum
states give results that could not have arisen out of any classical
system.

While POVMs are sets of operators indexed by a classical variable,
one can also consider a dependence on a classical input. Thus, if
a composite system is measured locally, i.e., Alice and Bob each measure
their particle with inputs $x,y$ and outputs $a,b$ respectively,
their results will be distributed according to a probability $P(a,b|x,y)$
of the outputs given the inputs. Such probability distributions can
be correlated much more strongly if the state shared by Alice and
Bob is entangled than if it is separable, a phenomenon known as \emph{nonlocality}.
This is one piece of classical evidence that can be searched for to
confirm whether Nature is post-classical. And indeed, several experiments
\cite{aspect_proposed_1976,brendel_experimental_1992,weihs_violation_1998,rowe_experimental_2001,hensen_loophole-free_2015,giustina_significant-loophole-free_2015,abellan_challenging_2018},
each more reliable than the last, have revealed correlations that
could not have been generated by any purely classical system.

In this thesis we shall only be concerned with the theoretical aspects
of nonlocality. To analyse them, we now define the most relevant notions:
we use $E_{a|x},F_{b|y}$ to denote the POVM elements with output
$a,b$ and input $x,y$ respectively, and let $\mathcal{A},\mathcal{B},\mathcal{X},\mathcal{Y}$
be the sets of Alice's and Bob's outputs and inputs respectively.
These sets will often be left implicit. If these POVMs act on a quantum
state $\rho$, the Born rule (equation (\ref{eq:Bornrule})) yields
a conditional probability $P(a,b|x,y)$. Conversely, probabilities
of this form are said to be quantum if they could have arisen from
a quantum state and measurements: 
\begin{defn}
A probability distribution $\left\{ P(a,b|x,y)\right\} _{(a,b,x,y)\in\mathcal{A}\times\mathcal{B}\times\mathcal{X}\times\mathcal{Y}}$
is \emph{quantum }if it can be written in the form

\begin{equation}
P(a,b|x,y)=\tr(E_{a|x}\otimes F_{b|y}\rho)\label{eq:intro-quantumP}
\end{equation}
for some quantum state $\rho$ and POVMs $\left\{ E_{a|x}\right\} _{a\in\mathcal{A}},\left\{ F_{b|y}\right\} _{b\in\mathcal{B}}$
for each $x\in\mathcal{X},y\in\mathcal{Y}$. 
\end{defn}

Note that, strictly speaking, $\left\{ P(a,b|x,y)\right\} _{(a,b)\in\mathcal{A}\times\mathcal{B}}$
is a probability distribution for each $(x,y)\in\mathcal{X}\times\mathcal{Y}$.
However, throughout this work we refer to $\left\{ P(a,b|x,y)\right\} _{(a,b,x,y)\in\mathcal{A}\times\mathcal{B}\times\mathcal{X}\times\mathcal{Y}}$
as a probability distribution.

Quantum distributions arising from some entangled states are especially
interesting, since they cannot be generated using only classical resources.
In fact, classical resources can only give rise to \emph{local} probability
distributions, which we now define. 
\begin{defn}
\label{def:intro-local-bipart}A probability distribution $\left\{ P(a,b|x,y)\right\} _{(a,b,x,y)\in\mathcal{A}\times\mathcal{B}\times\mathcal{X}\times\mathcal{Y}}$
is \emph{local} if it can be written in the form 
\begin{equation}
P(a,b|x,y)=\sum_{\lambda\in\Lambda}p(\lambda)P_{A}(a|x,\lambda)P_{B}(b|y,\lambda),\label{eq:intro-localP}
\end{equation}
for some distributions $\left\{ P_{A}(a|x,\lambda)\right\} _{(a,x,\lambda)\in\mathcal{A}\times\mathcal{X}\times\Lambda},\left\{ P_{B}(b|y,\lambda)\right\} _{(b,y,\lambda)\in\mathcal{B}\times\mathcal{Y}\times\Lambda}$,
and where $\lambda$ is a `hidden variable' taking values in some
set $\Lambda$ and distributed according to $\{p(\lambda)\}_{\lambda\in\Lambda}$.
Equation (\ref{eq:intro-localP}) is a \emph{Local Hidden-Variable
(LHV)} model for the probability distribution $P$. More parties can
be accounted for by adding more distributions, correlated only by
the hidden variable $\lambda$. A quantum state is local if, for any
POVMs, it can only give rise to local distributions. 
\end{defn}

(Unfortunately for our purposes, $\lambda$ is the most common choice
for both Schmidt coefficients and hidden variables---context will
determine what $\lambda$ stands for throughout this work).

It is simple to show that separable states can only give rise to local
distributions, i.e., all separable states are local. Suppose Alice
and Bob share a state $\rho=\sum_{i}p_{i}\ketbra{\eta_{i}}_{A}\otimes\ketbra{\chi_{i}}_{B}$,
and apply measurements given by $E_{a|x},F_{b|y}$. Then, 
\begin{equation}
\begin{aligned}P(a,b|x,y) & =\tr\left[\left(E_{a|x}\otimes F_{b|y}\right)\left(\sum_{i}p_{i}\ketbra{\eta_{i}}_{A}\otimes\ketbra{\chi_{i}}_{B}\right)\right]\\
 & =\sum_{i}p_{i}\tr\left(E_{a|x}\ketbra{\eta_{i}}_{A}\right)\tr\left(F_{b|y}\ketbra{\chi_{i}}_{B}\right),
\end{aligned}
\end{equation}
which is of the form of equation (\ref{eq:intro-localP}) by letting
\begin{equation}
\begin{aligned}P_{A}(a|x,\lambda) & =\tr\left(E_{a|x}\ketbra{\eta_{i}}_{A}\right)\\
P_{B}(b|y,\lambda) & =\tr\left(F_{b|y}\ketbra{\chi_{i}}_{B}\right)\\
p(\lambda) & =p_{i}.
\end{aligned}
\end{equation}

In fact, every local distribution can be written as a quantum distribution
arising from a separable state and measurements, that is, the set
of local distributions is included in the set of quantum distributions.
This inclusion is strict: there are entangled states and measurements
that give rise to \emph{nonlocal} distributions, i.e., those which
cannot be expressed like equation (\ref{eq:intro-localP}), and this
is the content of Bell's seminal theorem \cite{bell_einstein_1964}.

While all pure entangled states can display nonlocality (for the right
choice of measurements) \cite{gisin_bells_1991}, entanglement is
not equivalent to nonlocality. One prime example of this phenomenon
is given by isotropic states \cite{horodecki_reduction_1999}
\begin{equation}
\rho_{p}=p\phi_{d}^{+}+(1-p)\frac{\mathbbm1}{d^{2}},\label{eq:intro-isotropic}
\end{equation}
whose parameter $p$ is termed the \emph{visibility}.\emph{ }These
states are entangled if and only if $p>1/(d+1)$ \cite{horodecki_reduction_1999},
where $d$ is the local dimension, and local if $p>p_{L}$. While
the exact value of $p_{L}$ is not known, the bounds $\Theta(3/\ue d)\leq p_{L}\leq C\log^{2}d/d$,
where $C$ is a constant, were given in Refs. \cite{almeida_noise_2007,palazuelos_largest_2014},
and imply that, for $1/(d+1)<p<\Theta(3/\ue d)$, isotropic states
are entangled but cannot give rise to any nonlocal distributions.
Werner states may also be entangled and local \cite{werner_quantum_1989,barrett_nonsequential_2002}.
Still, taking many copies of an entangled, local mixed state sometimes
yields nonlocality, a phenomenon termed \emph{superactivation} of
nonlocality \cite{palazuelos_super-activation_2012}.

Nonlocal distributions have a wide variety of applications including
cryptography \cite{gisin_quantum_2002,pirandola_advances_2020}, randomness
extraction, amplication and certification \cite{acin_certified_2016},
communication complexity reduction \cite{buhrman_nonlocality_2010},
etc. Indeed, they are one of the reasons why some entangled states
are useful for communication-related tasks. Nonlocal distributions
are also useful as a means of certification of quantum entanglement.
For this, as well as for establishing sufficient conditions on quantum
states to give rise to nonlocality, it is necessary to have a reliable
way of knowing when a distribution is nonlocal, that depends only
on the distribution itself. The main tool for this purpose is \emph{Bell
inequalities}.

A Bell inequality is a functional $I$ that acts on distributions
$\left\{ P(a,b|x,y)\right\} _{(a,b,x,y)\in\mathcal{A}\times\mathcal{B}\times\mathcal{X}\times\mathcal{Y}}$
by assigning coefficients $c_{a,b,x,y}$ to each element $P(a,b|x,y)$,
and is such that, for all local $P$, 
\begin{equation}
\left\langle I,P\right\rangle \equiv\sum_{a,b,x,y}c_{a,b,x,y}P(a,b|x,y)\leq c_{0}
\end{equation}
for some $c_{0}\in\mathbb{R}$. If a distribution $P$ is such that
\begin{equation}
\left\langle I,P\right\rangle >c_{0},
\end{equation}
then $P$ is said to \emph{violate }the inequality $I$. Bell inequalities
play the role of nonlocality witnesses, in analogy to entanglement
witnesses. In fact, they are also hyperplanes separating a nonlocal
point from the (convex) set of local distributions.

In some cases, Bell inequalities can be given a physical meaning by
viewing them as \emph{nonlocal games}. In a nonlocal game played cooperatively
by Alice and Bob, a referee picks questions $x,y$, drawn from some
alphabets $\mathcal{X},\mathcal{Y}$, for Alice and Bob respectively,
with a probability $\pi(x,y)$. Without knowing each other's questions,
the players must each return an answer $a,b$, drawn from some alphabets
$\mathcal{A},\mathcal{B}$ respectively. The referee then decides
whether the players win or lose according to a publicly known verification
function $V(a,b,x,y)$ which depends on the questions and answers,
and is equal to 1 if they win and to 0 otherwise. In this setting,
the distribution $\left\{ P(a,b|x,y)\right\} _{(a,b,x,y)\in\mathcal{A}\times\mathcal{B}\times\mathcal{X}\times\mathcal{Y}}$
captures the probability of the players answering $a,b$ given questions
$x,y$, and thus encodes their strategy. Then, the overall winning
probability can be straightforwardly calculated as 
\begin{equation}
\sum_{x,y}\pi(x,y)\sum_{a,b}V(a,b,x,y)P(a,b|x,y).
\end{equation}

By associating the coefficients 
\begin{equation}
G_{ab|xy}=\pi(x,y)V(a,b,x,y)
\end{equation}
for each $a,b,x,y$, games can be viewed as functionals. Then, the
winning probability is the action $\left\langle G,P\right\rangle $
of the game functional on the probability distribution. In this sense,
games are a particular type of Bell inequalities, which have non-negative
coefficients (since $\pi(x,y),V(a,b,x,y)\geq0$ for all $a,b,x,y$).

One well-known nonlocal game is the CHSH game \cite{clauser_proposed_1969},
where inputs and outputs take values 0 or 1, the probability distribution
of the questions, $\pi(x,y)$, is uniform, and the verification function
$V(a,b,x,y)=1$ if and only if $a\oplus b=xy$. That is, Alice and
Bob win if they provide equal outputs whenever $x$ or $y$ are 0,
and different outputs if $x=y=1$. Local strategies give a maximum
winning probability of 3/4, while the maximally entangled state and
certain measurements can be used to achieve a winning probability
of $(2+\sqrt{2})/4\simeq0.85$. In fact, this is the maximum winning
probability achievable by a quantum strategy, as Tsirelson's seminal
result showed \cite{cirelson_quantum_1980}.

Aside from inequalities, nonlocality can also be detected by giving
a series of conditions that a distribution $P$ can only meet if it
is nonlocal. An example of this is Hardy's paradox \cite{hardy_quantum_1992,hardy_nonlocality_1993}:
\begin{equation}
P(0,0|0,0)>0=P(0,1|0,1)=P(1,0|1,0)=P(0,0|1,1).
\end{equation}
It is not difficult to prove that local distributions cannot satisfy
Hardy's paradox, and, in fact, nor can maximally entangled states.
However, all other pure entangled states can satisfy the paradox.

It can be shown that distributions of one input or output are always
local, so the study of nonlocality only takes off for larger distributions.
Hence, the simplest Bell inequalities are those of two inputs and
two outputs. To decide whether distributions of more inputs and outputs
(and, as will become important later, parties) are nonlocal, Bell
inequalities that account for this greater complexity can be devised.
Alternatively, one can apply local manipulations to the distribution
of more inputs, outputs or parties in order to achieve an effective
distribution of two inputs, outputs and parties, that is thus local
if the original distribution is. More generally, any distribution
$P$ where inputs and outputs take values on $\mathcal{A},\mathcal{B},\mathcal{X},\mathcal{Y}$
(and similar sets for any extra parties) can be mapped to an effective
distribution $\tilde{P}$ shared by Alice and Bob where inputs and
outputs take values on $\tilde{\mathcal{A}},\tilde{\mathcal{B}},\tilde{\mathcal{X}},\tilde{\mathcal{Y}}$.
Any inputs in $\mathcal{X}$ that are not in $\tilde{\mathcal{X}}$
can be simply ignored, and similarly for $\mathcal{Y}$. Outputs in
$\mathcal{A}$ can be grouped into $|\tilde{\mathcal{A}}|$ sets,
so that effective probabilities $\tilde{P}$ are the sum of some probabilities
$P(a,b|x,y)$ over certain values $a$, and similarly for $\mathcal{B}$.
Finally, one can restrict attention to a particular input and output
for any extra parties who are not Alice and Bob. It is not difficult
to show that all of these transformations are local; that is, if $P$
is local, then so is $\tilde{P}$. Conversely, if $\tilde{P}$ is
obtained from $P$ in the above way, and $\tilde{P}$ violates a Bell
inequality, then $P$ is nonlocal.

Equivalently, one can transform Bell inequalities, to make them account
for more inputs, outputs and parties. To do this, it is useful to
view them geometrically. For a given number of inputs and outputs
(i.e., when $\mathcal{A},\mathcal{B},\mathcal{X},\mathcal{Y}$ have
fixed size), the set of local distributions forms a polytope. Each
of these polytopes is completely characterised by a finite set of
Bell inequalities that corresponds to one of its facets (that is,
faces of maximal dimension). Thus, facet inequalities characterise
the border between the local and nonlocal regions. It turns out that
the facet inequalities of one local polytope provide necessary conditions
that larger polytopes (i.e. for more inputs, outputs and parties)
must meet, and hence it is possible to derive facet Bell inequalities
for larger polytopes starting from facet Bell inequalities of smaller
polytopes. This process is called \emph{lifting} Bell inequalities,
and is detailed in Ref. \cite{pironio_lifting_2005}. The lifting
process preserves the local bound as long as the original Bell inequality
has a local bound of zero. However, any Bell inequality can be written
so that this is the case, by expressing any constant as $\sum_{a,b}P(a,b|x,y)$
times the constant, for any fixed $x,y$.

Starting from a facet inequality with a local bound of zero that holds
for the local polytope defined by $\mathcal{A},\mathcal{B},\mathcal{X},\mathcal{Y}$,
more inputs can be accounted for by assigning coefficient 0 to any
inputs outside of $\mathcal{X},\mathcal{Y}$ (which corresponds to
ignoring the relevant inputs in the distribution that the inequality
acts on). Further, each new output outside of $\mathcal{A},\mathcal{B}$
gets assigned the same coefficient as one of the outputs in the smaller
polytope (so that, effectively, outputs are grouped into sets). Finally,
to account for more parties, the coefficients pertaining to the smaller
polytope are made equal to those of the larger polytope for some fixed
input and output of the extra parties, and all other inputs and outputs
get assigned coefficient 0 (corresponding to fixing the input and
output of all extra parties). Transforming a Bell inequality for a
smaller polytope in this way gives rise to a facet inequality for
the larger polytope. These transformations of Bell inequalities are
proven in Ref. \cite{pironio_lifting_2005}.

As well as generating nonlocal distributions, entangled states may
also exhibit the property of quantum steering \cite{wiseman_steering_2007,cavalcanti_quantum_2016,uola_quantum_2020}.
Suppose Alice prepares a bipartite quantum state and sends one of
the particles to Bob. They both measure their respective particles
and communicate classically. Suppose they repeat this many times.
Can Alice convince Bob that the state she prepares is entangled? The
answer is yes only if the state is \emph{steerable}, that is, if there
is no \emph{Local Hidden-State (LHS)} model describing it. 
\begin{defn}
A state $\rho$ has a Local Hidden-State model if, for any POVM $\left\{ E_{a|x}\right\} _{a\in\mathcal{A}}$,
the distribution arising from measuring the state can be written as
\begin{equation}
\tr_{A}(E_{a|x}\otimes\mathbbm1\rho)=\sum_{\lambda}P_{A}(a|x,\lambda)p(\lambda)\sigma_{\lambda},\label{eq:intro-LHS}
\end{equation}
where $\sigma_{\lambda}$ is a single-party state which depends on
the hidden variable $\lambda$. 
\end{defn}

\begin{defn}
If a state $\rho$ has an LHS model of the form of equation (\ref{eq:intro-LHS}),
then $\rho$ is non-steerable from Alice to Bob. Otherwise, $\rho$
is steerable\emph{ }from Alice to Bob\emph{.}
\end{defn}

If $\rho$ is non-steerable, Alice can simply prepare an ensemble
of $\{\sigma_{\lambda}\}_{\lambda}$ which she sends to Bob, and deliver
output $a$ for input $x$ with probability $P_{A}(a|x,\lambda)$.
Thus, she is only able to convince Bob that the state she prepared
is entangled if it is steerable. By considering $\tr(E_{a|x}\otimes F_{b|y}\rho)=\tr_{B}(\tr_{A}(E_{a|x}\otimes\mathbbm1\rho)F_{b|y})$
and denoting $P_{B}(b|y,\lambda)=\tr(F_{b|y}\sigma_{\lambda})$, it
is evident that all non-steerable states are local, but the converse
is not true.

This notion of steerability concerns the case where Alice can steer
Bob's state. One can consider the analogous concept interchanging
Alice and Bob. In fact, steerability is an asymmetric notion: there
are states that are steerable from Alice to Bob, but not from Bob
to Alice \cite{bowles_one-way_2014}.

We have seen that local probability distributions can be generated
using only classical resources, while quantum distributions can be
generated using only quantum resources. It is possible in principle
to consider larger, post-quantum sets of distributions, and one such
set turns out to be very useful: the set of \emph{nonsignalling }distributions. 
\begin{defn}
\label{def:intro-ns}A probability distribution $\left\{ P(a,b|x,y)\right\} _{(a,b,x,y)\in\mathcal{A}\times\mathcal{B}\times\mathcal{X}\times\mathcal{Y}}$
is nonsignalling if it is such that 
\begin{equation}
\begin{aligned}\sum_{a}P(a,b|x,y) & =\sum_{a}P(a,b|x^{\prime},y),\\
\sum_{b}P(a,b|x,y) & =\sum_{b}P(a,b|x,y^{\prime}),
\end{aligned}
\end{equation}
for all $x^{\prime}\neq x$, $y^{\prime}\neq y$. 
\end{defn}

These conditions ensure that the marginal distributions $P(a|x)$,
$P(b|y)$ are well-defined. Also, it follows from the definition that
nonsignalling distributions cannot be used by Alice and Bob to communicate
information faster than the speed of light: indeed, unless Alice is
allowed to send messages to Bob, he cannot know her input, and vice
versa. Further, all quantum distributions are nonsignalling, a condition
ensured by the normalisation of the POVMs. In fact, the set of quantum
distributions is strictly included in the set of nonsignalling distributions.

Nonsignalling distributions are often depicted as a \emph{nonsignalling
box}, which is a ficticious device that Alice and Bob can share while
spatially separated, admits inputs $x,y$ for each party respectively,
and gives outputs $a,b$ to each party respectively. By extension,
one can abstract away from the physical realisation of quantum and
local distributions, and imagine ficticious devices that generated
them. Thus, we often talk about quantum and local boxes, to mean special
cases of nonsignalling boxes whose underlying probability distributions
are quantum or local, respectively.

The nonsignalling conditions are an example of a minimal requirement
that a probability distribution should meet in order to be physical.
Indeed, while quantum theory is widely accepted as a very reliable
description of Nature, and is confirmed by experiment to a great level
of precision, it is not the \emph{only }possible theory compatible
with observations so far. Nonsignalling conditions were first introduced
as a physical principle that any physical theory should obey \cite{popescu_quantum_1994}.
As well as gaining a better understanding of Nature, properties of
nonsignalling distributions establish necessary conditions on quantum
distributions. Since the nonsignalling set forms a polytope, unlike
the quantum set (which is convex, but does not have a finite number
of extremal points), it is often simpler to analyse. In light of this,
it is helpful to have a means of knowing when a nonsignalling box
is quantum. Two main tools will be relevant for this work: Tsirelson's
theorem and quantum voids. They will be defined in Chapter \ref{chap:agreement}.

We have seen that a distribution is local if it admits a Local Hidden-Variable
model. These models are often given a physical interpretation in the
context of epistemics and game theory, where the hidden variables
represent the possible `states of the world'. Thus, the model is defined
by the probability space $(\Omega,\mathcal{E},\mathsf{P})$, where
$\Omega$ is the set of possible states of the world (which are commonly
denoted by $\omega$ instead of $\lambda$ in this context), $\mathcal{E}$
is the power set of $\Omega$, i.e., the set of events, and $\mathsf{P}$
is a probability measure on $\Omega$. Agents do not know which is
the true state of the world $\omega^{*}$, but they have limited information
about it: Alice and Bob partition the state space according to $\mathcal{P}_{A}$,
$\mathcal{P}_{B}$ respectively, and they know which partition element
contains the true state of the world. From this information, they
can each calculate the conditional probability of any event (which
is a set of states of the world) given the element $\mathcal{P}_{A,B}(\omega^{*})$
of their respective partitions that contains the true state of the
world.

As mentioned above, Bell's theorem prevents the extension of this
model to nonlocal settings, so, in particular, this model does not
apply to parties who share quantum devices. However, this model can
be made valid for all nonsignalling distributions via a simple relaxation:
allowing $\mathsf{P}$, the distribution of the states of the world,
to be a \emph{quasi-probability measure}, i.e., to take negative values
as well as non-negative ones, with $\sum_{\omega\in\Omega}\mathsf{P}(\omega)=1$.
These constructions are referred to in the literature as \emph{ontological
models}, which are a quasi-probability space $(\Omega,\mathcal{E},\mathsf{P})$
together with a set of partitions. In the literature (see, e.g., \cite{ferrie_quasi-probability_2011}),
they often include a set of preparations underlying the distribution
over the state space, and the partitions are usually phrased in terms
of measurements and outcomes. However, we consider preparations implicit
and use the language of partitions to bridge the gap between the fields
of epistemics and quantum information more smoothly. Thus, any nonsignalling
box can be associated to an ontological model, and vice versa. This
was derived in Refs. \cite{abramsky_sheaf-theoretic_2011,abramsky_operational_2014}
from sheaf-theoretic concepts, but in Chapter \ref{chap:agreement}
we provide a much more direct proof that is more suitable for the
purposes of this work.

In the following, distributions $p$ or $\mathsf{P}$ will always
be assumed to take non-negative values, unless otherwise mentioned.

\subsection{Multipartite nonlocality\label{subsec:prob-multipart}}

Just like in the case of entanglement, the definitions of local, quantum
and nonsignalling distributions can be extended in more than one way
to multipartite settings. As hinted above, a natural way of extending
the definition of local distributions to more than two parties is
to add single-party distributions to equation (\ref{eq:intro-localP}),
correlated only by the hidden variable. However, one can also imagine
distributions that group parties into two sets, and are local only
across this bipartition, or convex mixtures of these distributions.
Analogously to entanglement, we define fully local, bilocal and genuine
multipartite nonlocal distributions: 
\begin{defn}
\label{def:intro-GMNL}A probability distribution $\left\{ P(a_{1},...,a_{n}|x_{1},...,x_{n})\right\} _{(a_{i},x_{i})\in\mathcal{A}_{i}\times\mathcal{X}_{i},i\in[n]}$,
for some sets $\mathcal{A}_{i},\mathcal{X}_{i}$ for each $i\in[n]$,
where $n\in\mathbb{N}$, is \emph{fully local }if it can be written
in the form

\begin{equation}
P(a_{1},...,a_{n}|x_{1},...,x_{n})=\sum_{\lambda}p(\lambda)\prod_{i\in[n]}P_{i}(a_{i}|x_{i},\lambda)\label{eq:intro-fullylocal}
\end{equation}
for some probability distributions $\{P_{i}(a_{i}|x_{i},\lambda)\}_{(a_{i},x_{i},\lambda)\in\mathcal{A}_{i}\times\mathcal{X}_{i}\times\Lambda}$
for each $i\in[n]$, and $\{p(\lambda)\}_{\lambda\in\Lambda}$.

It is \emph{bilocal} if it can be written in the form 
\begin{equation}
P(a_{1},...,a_{n}|x_{1},...,x_{n})=\sum_{M\subsetneq[n]}\sum_{\lambda}p_{M}(\lambda)P_{M}(\{a_{i}\}_{i\in M}|\{x_{i}\}_{i\in M},\lambda)P_{\overline{M}}(\{a_{i}\}_{i\in\overline{M}}|\{x_{i}\}_{i\in\overline{M}},\lambda),
\end{equation}
for some distributions \sloppy$\{P_{M}(\{a_{i}\}_{i\in M}|\{x_{i}\}_{i\in M},\lambda)\}_{(a_{i},x_{i},\lambda)\in\mathcal{A}_{i}\times\mathcal{X}_{i}\times\Lambda,i\in M}$,
$\{P_{\overline{M}}(\{a_{i}\}_{i\in\overline{M}}|\{x_{i}\}_{i\in\overline{M}},\lambda)\}_{(a_{i},x_{i},\lambda)\in\mathcal{A}_{i}\times\mathcal{X}_{i}\times\Lambda,i\in\overline{M}}$
which are nonsignalling, for each bipartition $M|\overline{M}$ of
the parties, and where $p_{M}(\lambda)\geq0$, $\sum_{M,\lambda}p_{M}(\lambda)=1$.

Otherwise, it is \emph{genuine multipartite nonlocal (GMNL)}. 
\end{defn}

Bilocal distributions are sometimes termed ``hybrid'' or ``mixed''.
This is because the term ``bilocal'' has alternative definitions
in the literature, notably regarding hidden variable models where
variables correlate only pairs of parties \cite{tavakoli_correlations_2017,renou_limits_2019,renou_genuine_2019,gisin_constraints_2020,krivachy_neural_2020}.
Since there will be no ambiguity in this respect, in this work we
use the term ``bilocal'' in analogy to ``biseparable''.

The requirement that the distributions $P_{M},P_{\overline{M}}$ of
the bipartitions are nonsignalling is an added subtlety that the definition
of biseparability does not give rise to. Indeed, the original definition
of bilocality, given by Svetlichny in Ref. \cite{svetlichny_distinguishing_1987},
left these distributions unrestricted; however, this has been shown
to lead to operational problems \cite{buscemi_all_2012,gallego_operational_2012,bancal_definitions_2013,geller_quantifying_2014,vicente_nonlocality_2014,gallego_nonlocality_2017}.
Hence, like most recent works on the topic, we assume these distributions
are nonsignalling, which captures most physical situations better
\cite{schmid_type-independent_2020,wolfe_quantifying_2020}.

Contrary to locality, the nonsignalling conditions extend unambiguously
to the multipartite setting. It is sufficient to assume that the marginal
distribution of all but one party is independent of the input of this
party, and that this holds for all parties, in order to conclude that
the marginal of any subset of parties is independent of the inputs
outside this subset. 
\begin{defn}
A probability distribution $\left\{ P(a_{1},...,a_{n}|x_{1},...,x_{n})\right\} _{(a_{i},x_{i})\in\mathcal{A}_{i}\times\mathcal{X}_{i},i\in[n]}$,
for some sets $\mathcal{A}_{i},\mathcal{X}_{i}$ for each $i\in[n]$,
where $n\in\mathbb{N}$, is \emph{nonsignalling} if, for any party
$i\in[n]$, 
\begin{equation}
\sum_{a_{i}}P(a_{1},...,a_{n}|x_{1},...,x_{n})
\end{equation}
is independent of $x_{i}$. 
\end{defn}

It is also worth remarking that alternative definitions of GME and
GMNL are currently being proposed. Refs. \cite{schmid_type-independent_2020,wolfe_quantifying_2020}
consider nonlocality (among other quantum phenomena) from the point
of view of a resource theory whose free operations are Local Operations
and Shared Randomness (LOSR). They also lay the ground for Ref. \cite{navascues_genuine_2020}
to propose the notion of \emph{genuine network entanglement}, a stricter
notion than GME which rules out states which are tensor products of
biseparable states, and one could imagine a similar notion of network
nonlocality. However, we are interested mainly in quantum networks
where pairs of parties share entangled states, so these notions will
not help us determine which network configurations are more useful
than others.

GMNL distributions, defined as in Definition \ref{def:intro-GMNL},
have numerous applications, for example in multiparty cryptography
\cite{aolita_fully_2012}, the understanding of condensed matter physics
\cite{tura_detecting_2014,tura_energy_2017}, and the development
of quantum networks \cite{cavalcanti_quantum_2011,gisin_all_2017,supic_measurement-device-independent_2017,tavakoli_correlations_2017,renou_genuine_2019,krivachy_neural_2020},
particularly for quantum computation \cite{cirac_distributed_1999-1,howard_nonlocality_2012,howard_contextuality_2014}
and correlating particles which never interacted \cite{branciard_characterizing_2010,branciard_bilocal_2012}.

It is easy to see that biseparable states can only give rise to bilocal
distributions. Indeed, a pure biseparable state is separable along
a bipartition $M|\overline{M}$ of the parties, therefore, as shown
above for the bipartite case, it can only give rise to a distribution
that is local along that bipartition. Moreover, the distributions
$P_{M},P_{\overline{M}}$ are nonsignalling, since they arose from
a quantum state and measurements. Biseparable mixed states thus give
rise to convex combinations of bilocal distributions, which is still
bilocal. Like in the bipartite case, GME is not sufficient for nonlocality,
as there exist GME mixed states that are bilocal \cite{augusiak_entanglement_2015,augusiak_constructions_2018}
or even fully local \cite{bowles_genuinely_2016}. While pure GME
states are never fully local \cite{popescu_generic_1992,gachechiladze_completing_2017},
it is not known whether all pure GME states can give rise to GMNL.

In general, distributions arising from the multipartite state $\rho$,
and POVMs $\{E_{a_{i}|x_{i}}\}_{a_{i}\in\mathcal{A}_{i}}$ for each
$x_{i}$, for each party $i\in[n]$, are of the form
\begin{equation}
P(a_{1},...,a_{n}|x_{1},...,x_{n})=\tr(\bigotimes_{i=1}^{n}E_{a_{i}|x_{i}}\rho).
\end{equation}
To find out which of these distributions are GMNL, one can use \emph{GMNL
inequalities}, which, in analogy to Bell inequalities for bipartite
systems, are functionals that are bounded when acting on bilocal distributions:
\begin{equation}
\left\langle I,P\right\rangle \equiv\sum_{\substack{a_{i},x_{i}\\
i\in[n]
}
}c_{a_{1},...,a_{n},x_{1},...,x_{n}}P(a_{1},...,a_{n}|x_{1},...,x_{n})\leq c_{0}
\end{equation}
for all bilocal $P$.

\section{Our contribution}

Having reviewed the main technical tools that will be used throughout
this thesis, we now give an outline of the problems that we tackle
and the main results we develop. Chapters \ref{chap:maxent}-\ref{chap:mixed}
focus on multipartite systems, studying their entanglement properties
and how to obtain nonlocality from network states. Chapter \ref{chap:agreement}
introduces a principle that we contend should be satisfied by all
physical theories.

\subsection{A nontrivial resource theory of multipartite entanglement}

As advanced above, entanglement is a striking feature of quantum theory
with no classical analogue. Although initially studied to address
foundational issues \cite{schrodinger_discussion_1935,einstein_can_1935},
the development of quantum information theory \cite{nielsen_quantum_2000}
in the last few decades has elevated it to a resource that allows
the implementation of tasks which are impossible in classical systems.
The resource theory of entanglement \cite{plenio_introduction_2007,horodecki_quantum_2009}
aims at providing a rigorous framework to qualify and quantify entanglement
and, ultimately, to understand fully its capabilities and limitations
within the realm of quantum technologies. However, this theory is
much more firmly developed for bipartite than multipartite systems.
In fact, although a few applications have been proposed within the
latter setting such as secret sharing \cite{hillery_quantum_1999,gottesman_theory_2000},
the one-way quantum computer \cite{raussendorf_one-way_2001} and
metrology \cite{giovannetti_advances_2011,toth_quantum_2014}, a deeper
understanding of the complex structure of multipartite entangled states
might inspire further protocols in quantum information science and
better tools for the study of condensed-matter systems.

As already reviewed in Section \ref{subsec:ent-bipart}, the wide
applicability of the formulation of entanglement theory as a resource
theory has motivated an active line of work that studies different
quantum effects from this point of view \cite{chitambar_quantum_2019}.
Although it was introduced above in the context of entanglement under
LOCC operations, the framework of resource theories can be expressed
independently of the object of study: the main question a resource
theory addresses is to order the set of states and provide means to
quantify their nature as a resource. The free operations are crucial
to this task. This is a subset of transformations, which the given
scenario dictates can be implemented at no cost. Thus, all states
that can be prepared with these operations are free states. Conversely,
non-free states acquire the status of a resource: granted such states,
the limitations of the corresponding scenario might be overcome. Moreover,
the concept of free operations allows an order relation to be defined.
If a state $\rho$ can be transformed into $\sigma$ by some free
operation, then $\rho$ cannot be less resourceful than $\sigma$
since any task achievable by $\sigma$ is also achievable by $\rho$
as the corresponding transformation can be freely implemented. However,
the converse is not necessarily true. Furthermore, one can introduce
resource quantifiers as functionals that preserve this order.

Since entanglement is a property of systems with many constituents
which may be far away, the natural choice for free operations in this
resource theory is local operations and classical communication (LOCC).
Indeed, parties bound to LOCC can only prepare separable states, and
entangled states become a resource to overcome the constraints imposed
by LOCC manipulation. In pure bipartite states, the ordering induced
by LOCC reduces to majorisation \cite{nielsen_conditions_1999,marshall_inequalities_2011},
and there is a unique maximally entangled state for fixed local dimension.
This is because this state can be transformed by LOCC into any other
state of that dimension but no other state of that dimension can be
transformed into it.

Importantly, the situation changes drastically in the multipartite
case. Here, Ref. \cite{dur_three_2000} and subsequent work \cite{verstraete_four_2002,briand_moduli_2004}
have shown that there exist inequivalent forms of entanglement: the
state space is divided into classes, the so-called stochastic LOCC
(SLOCC) classes, of states which can be interconverted with non-zero
probability by LOCC but cannot be transformed outside the class by
LOCC, even probabilistically. This in particular shows that no maximally
entangled state can exist for multipartite states. Still, one could
in principle study the ordering induced by LOCC within each SLOCC
class. Recent work \cite{de_vicente_maximally_2013,spee_maximally_2016,hebenstreit_maximally_2016,spee_entangled_2017,de_vicente_entanglement_2017,gour_almost_2017}
in this direction has revealed, however, an extreme feature that culminates
with the result of Ref.\ \cite{sauerwein_transformations_2018}:
almost all pure states of more than three parties are \textit{isolated},
i.e.\ they cannot be obtained from nor transformed to another inequivalent
pure state of the same local dimensions by LOCC. This means that almost
all pure states are incomparable by LOCC, inducing a trivial ordering
and a meaningless arbitrariness in the construction of entanglement
measures. In this sense, one may say that the resource theory of multipartite
entanglement with LOCC is generically trivial.

We believe this calls for a critical reexamination of the resource
theory of entanglement and, in particular, for the notion of LOCC
as the ordering-defining relation. Indeed, although LOCC transformations
have a clear operational interpretation, this is not, in fact, the
most general class of transformations that maps the set of separable
states into itself. In other words, LOCC is strictly included in the
class of non-entangling operations. Thus, from the abstract point
of view of resource theories other consistent theories of entanglement
(i.e.\ with separable states being the free states) are possible
where the set of free operations is larger than LOCC. Hence, in principle,
these could give a more meaningful ordering and revealing structure
in the set of multipartite entangled states. To study such possibility
is precisely the goal of Chapter \ref{chap:maxent}. A similar approach
has been taken to address other unsatisfying features of the resource
theory of entanglement under LOCC such as irreversibility of state
transformations for an arbitrarily large number of copies \cite{vidal_irreversibility_2001}.
Remarkably, Ref. \cite{brandao_entanglement_2008,brandao_generalization_2010}
has shown that shifting the paradigm from LOCC to asymptotic non-entangling
operations provides a reversible theory of asymptotic entanglement
interconversion with a unique entanglement measure and this result
has been extended in \cite{brandao_reversible_2015} to arbitrary
resource theories under asymptotic resource-non-generating operations
\cite{chitambar_quantum_2019}. Also, in the absence of a clear set
of physical constraints determining the free operations, certain quantum
resource theories have been constructed by first defining the set
of free states and then considering classes of operations that preserve
this set. This is the case of the resource theory of coherence \cite{baumgratz_quantifying_2014},
which has been found useful in e.g.\ metrology applications \cite{biswas_interferometric_2017}
and quantum channel discrimination \cite{napoli_robustness_2016}
and which has subsequently given rise to a fruitful research line
considering an operational interpretation for the set of free operations
(see \cite{chitambar_critical_2016,streltsov_colloquium_2017} and
references therein).

Since we seek whether a non-trivial theory is at all possible for
single-copy manipulations, here we consider the resource theory of
entanglement under the largest possible class of free operations in
this regime: strictly non-entangling operations. However, multipartite
entanglement comes in two different forms, as seen in Definition \ref{def:intro-ent-multipart}.
Thus, one can formulate two theories: one in which entangled states
are considered a resource and where the free operations are full separability-preserving
(FSP), and the analogous with GME states and biseparability-preserving
(BSP) operations. Interestingly, our first result is that both formalisms
lead to non-trivial theories: no resource state is isolated in any
of these scenarios. Moreover, we show that there are no inequivalent
forms of entanglement. Then, we consider whether there exists a unique
multipartite maximally entangled state in these theories like in the
bipartite case. While we find a negative answer (at least in the simplest
non-trivial case of 3-qubit states) for FSP operations, our main result
is that the question is answered affirmatively in the resource theory
of GME under BSP operations. The maximally GME state turns out to
be the generalised Greenberger-Horne-Zeilinger (GHZ) state.

\subsection{Pure pair-entangled network states}

As shown in Section \ref{sec:prob}, correlations between quantum
particles may be much stronger than those between classical particles.
Their applications are manifold: cryptography \cite{gisin_quantum_2002,pirandola_advances_2020},
randomness extraction, amplification and certification \cite{acin_certified_2016},
communication complexity reduction \cite{buhrman_nonlocality_2010},
etc., and the study of these nonlocal correlations has led to the
growing field of device-independent quantum information processing
\cite{mayers_quantum_1998,acin_device-independent_2007,colbeck_quantum_2011}
(see also Ref. \cite{brunner_bell_2014}).

While bipartite nonlocality has been well researched in the past three
decades, much less is known about the multipartite case. Still, correlations
in quantum multicomponent systems have gained increasing attention
recently, with applications in multiparty cryptography \cite{aolita_fully_2012},
the understanding of condensed matter physics \cite{tura_detecting_2014,tura_energy_2017},
and the development of quantum networks \cite{cavalcanti_quantum_2011,gisin_all_2017,supic_measurement-device-independent_2017,tavakoli_correlations_2017,renou_genuine_2019,gisin_constraints_2020,krivachy_neural_2020},
particularly for quantum computation \cite{cirac_distributed_1999-1,howard_nonlocality_2012,howard_contextuality_2014}
and correlating particles which never interacted \cite{branciard_characterizing_2010,branciard_bilocal_2012}.

A necessary condition for nonlocality is quantum entanglement. Indeed,
this is one reason why entangled states are useful for communication-related
tasks. However, not all entangled states are nonlocal: some bipartite
entangled states only yield local distributions \cite{werner_quantum_1989,barrett_nonsequential_2002}.
Still, for pure bipartite states, entanglement \textit{is} sufficient
for nonlocality, which is the content of Gisin's theorem \cite{gisin_bells_1991,gisin_maximal_1992},
and multipartite entangled pure states are never fully local \cite{popescu_generic_1992,gachechiladze_completing_2017}.
Interestingly, distributing certain bipartite entangled states in
certain multipartite networks yields nonlocality even if the involved
states are individually local \cite{sende_entanglement_2005,cavalcanti_quantum_2011,cavalcanti_nonlocality_2012,supic_measurement-device-independent_2017,luo_nonlocality_2018,luo_nonlocal_2019}.

Multipartite nonlocality is in principle harder to generate than bipartite
nonlocality. By exploring the relationship between entanglement and
nonlocality in the multipartite regime, in Chapter \ref{chap:gmnl}
we show that \emph{pair-entangled network states }simplify the job
considerably: distributing arbitrarily low node-to-node entanglement
is sufficient to observe truly multipartite nonlocal effects involving
all parties in the network independently of its geometry. Added to
its practical consequences for applications, this fact points to a
deep property of quantum networks.

We show that the nonlocality arising from networks of bipartite pure
entangled states is a generic property and manifests in its strongest
form, GMNL. Specifically, we obtain that any connected network of
bipartite pure entangled states is GMNL. It was already known that
a star network of maximally entangled states is GMNL \cite{cavalcanti_quantum_2011},
but we provide a full, qualitative generalisation of this result by
making it independent of both the amount of entanglement shared and
the network topology. Thus, we show GMNL is an intrinsic property
of networks of pure bipartite entangled states.

Further, there are known mixed GME states that are bilocal \cite{augusiak_entanglement_2015,augusiak_constructions_2018}---some
are even fully local \cite{bowles_genuinely_2016}. Still, it is not
known whether Gisin's theorem extends to the genuine multipartite
regime. Recent results show that, for pure $n$-qubit symmetric states
\cite{chen_test_2014} and all pure 3-qubit states \cite{yu_tripartite_2013},
GME implies GMNL (at the single-copy level).\footnote{Hidden GMNL for three parties beyond qubits can be shown if some form
of preprocessing is allowed.} Our result above shows that all pure GME states that have a network
structure are GMNL; interestingly, we further apply this property
to establish a second result: all pure GME states are GMNL in the
sense that measurements can be found on finitely many copies of any
GME state to yield a GMNL behaviour. We thus tighten the relationship
between multipartite entanglement and nonlocality.

Our construction exploits the fact that the set of bilocal states
is not closed under tensor products. That is, GME can be superactivated
by taking tensor products of states that are unentangled across different
bipartitions. Thus, GME can be achieved by distributing bipartite
entangled states among different pairs of parties. To obtain our results,
we extend the superactivation property \cite{navascues_activation_2011,palazuelos_super-activation_2012,caban_activation_2015}
from the level of states to that of probability distributions, i.e.
GMNL can be superactivated by taking Cartesian products of probability
distributions that are local across different bipartitions. In fact,
when considering copies of quantum states, we only consider local
measurements performed on each copy separately, thus pointing at a
stronger notion of superactivation to achieve GMNL.

\subsection{Mixed pair-entangled network states}

As advanced in the previous section, quantum networks make it possible
to generate GMNL using only bipartite entanglement. Since bipartite
entanglement is in principle easier to distribute than truly multipartite
entanglement such as GHZ-state entanglement, this makes networks a
very useful tool for the wide variety of applications that require
GME and GMNL. Added to the operational motivation, their conceptual
simplicity and well-defined mathematical properties makes them a good
platform in which to explore the relationship between entanglement
and nonlocality in many-body systems. Indeed, it is experimentally
much simpler to distribute bipartite entanglement between different
nodes of a network than to establish genuine network entanglement
\cite{navascues_genuine_2020} between the nodes. Therefore, understanding
the behaviour of mixed pair-entangled network states is crucial to
gauge the full potential of near-term quantum technologies. Quantum
networks are widely studied, for example, as a means to achieve long-range
entanglement starting from smaller entanglement links, as well as
to entangle more than two parties \cite{pant_routing_2019,vardoyan_stochastic_2019}.
Applications such as cryptography \cite{gisin_quantum_2002}, quantum
error correction \cite{bennett_mixed-state_1996}, quantum metrology
\cite{giovannetti_advances_2011}, quantum sensor networks \cite{ren_clock_2012,eldredge_optimal_2018,khabiboulline_quantum-assisted_2019,qian_optimal_2020},
multi-party quantum communication \cite{hillery_quantum_1999,zhu_w-state_2015,murta_quantum_2020},
or computation \cite{raussendorf_one-way_2001,dhondt_computational_2006},
all require multipartite entanglement and thus drive the need to devise
ways of supplying end-to-end entanglement to nodes who request them.
The main theoretical and practical challenges in this respect are
ensuring high entanglement generation rates, high fidelity and long
coherence times, but proposals for optimal ways of generating entanglement
between two end nodes continue to be put forward \cite{caleffi_optimal_2017,chakraborty_distributed_2019,shi_modeling_2019,dai_optimal_2020,bauml_linear_2020,li_effective_2021}.
Further, realistic implementations of the quantum internet will rely
on existing infrastructure on which to build a quantum network \cite{rabbie_designing_2020},
so it is important to find the best ways of distributing entanglement
in a given network configuration, in the line of the recent work \cite{bugalho_distributing_2021}.

The results in Chapter \ref{chap:gmnl} show that networks of pure
states exhibit a very simple behaviour since, as long as they are
connected, they yield GMNL independently of their topology and the
amount of entanglement contained in the states on the edges. In particular,
this means that all connected networks of pure entangled states are
GME, and thus are sufficient to implement the above applications.
However, to derive those results, we use properties that are exclusive
of pure bipartite states: all pure entangled states are nonlocal,
and, moreover, they can satisfy Hardy's paradox \cite{hardy_quantum_1992,hardy_nonlocality_1993}
(if they are not maximally entangled) and exhibit full nonlocality
\cite{elitzur_quantum_1992,barrett_maximally_2006} (if they are).
However, mixed states are very different to pure states in terms of
the interplay between entanglement and nonlocality, even in the bipartite
case: there exist mixed states, such as some isotropic states \cite{horodecki_reduction_1999}
and some Werner states \cite{werner_quantum_1989,barrett_nonsequential_2002},
which are entangled and local. Multipartite mixed states may also
be GME and bilocal \cite{augusiak_entanglement_2015,augusiak_constructions_2018},
or even fully local \cite{bowles_genuinely_2016}.

In experimental settings, noise is unavoidable, and thus pure states
are out of experimental reach. Hence, studying networks of bipartite
mixed states is essential if they are to be used for applications.
While, by continuity, the results about pure states obtained in Chapter
\ref{chap:gmnl} must be robust to some noise, in general not even
GME is guaranteed for mixed-state networks (although connected networks
of entangled states are never fully separable). We focus mainly on
isotropic states\emph{ }as a noise model\emph{.} These are the only
states that are invariant under the action of $U\otimes U^{*}$ for
any unitary $U$ whose complex conjugate is $U^{*}$. Thus, in addition
of representing a standard noise model in which a maximally entangled
state is mixed with white noise, their symmetry properties make them
a convenient object to study theoretically. For qubits, any state
with a negative partial transpose (which is thus entangled) can be
transformed into an entangled isotropic state by \emph{twirling},
an LOCC operation consisting on averaging over all unitaries $U$
\cite{horodecki_mixed-state_1998}. (For larger dimensions, this only
happens if the fidelity with the maximally entangled state is large
enough.) Indeed, the first step of many protocols such as distillation
protocols is transforming the input state into an isotropic state.
Twirling is an LOCC operation, therefore biseparability is closed
under twirling. For this reason, if a given network of isotropic states
is GME, substituting some or all of the states for NPT states that
can be transformed into them preserves GME. In the particular case
of qubits, all networks of entangled states would be GME if and only
if all networks of isotropic states were.

We show that networks of mixed states exhibit very different properties
to those of pure states. First, node-to-node entanglement does not
necessarily imply that the network is GME. We give two examples of
tripartite networks with entangled isotropic states on the edges,
but which are nevertheless biseparable. Further, by studying networks
of three parties we can already show a dependence of the entanglement
properties on the topology of the network, unlike in the case of pure
states.

Then, we find that, in the case of larger networks, this dependence
manifests itself in the most extreme form. We show that all networks
in the form of a tree graph (i.e., a graph which contains no cycles)
or a polygon become biseparable for a sufficiently large number of
edges, as long as the visibility of the states on the edges is strictly
smaller than 1 (i.e., the noise parameter is strictly positive). Thus,
any given experimental limitation to the preparation of pure-state
entanglement prevents the observation of GME for these network configurations
if the number of parties is large enough: too few connections in a
network compromise entanglement.

Remarkably, GME depends crucially on the geometry of the network.
We show, in contrast to the above results, that a completely connected
network of isotropic states (i.e., a network where all vertices are
connected to all others) remains GME for any number of parties for
all visibilities above a threshold. As a consequence, GME in the completely
connected network holds for any number of parties as long as the visibility
is large enough. Since GME is a necessary condition for GMNL, we show
that distributing nonlocal states is not sufficient to generate GMNL
and, in particular, that the GMNL of networks of mixed states can
depend on the topology and the amount of entanglement present in the
network.

We also explore the nonlocality properties of some networks of isotropic
states. Beyond practical applications, the symmetry properties of
isotropic states and the fact that they can be entangled while local
makes this family of states particularly interesting for theoretical
study. We find that non-steerability is the main factor compromising
GMNL in these networks, and we find that a star network with a non-steerable
state on one edge and a maximally entangled state on the rest is GME
but bilocal. And yet, steerability does not guarantee GMNL. As a consequence
of the previous result, we provide an example of a steerable state
which, when distributed in a star network, is bilocal. Further, we
show that a star network of non-steerable states is fully local. Still,
by taking many copies of the bilocal network, we obtain, to our knowledge,
the first example of superactivation of GMNL from bilocality. In fact,
our construction can be used to obtain more examples of superactivation
of GMNL from many copies of bilocal networks.

\subsection{A physical principle from observers' agreement}

So far, we have focussed on analysing what resources can be extracted
from different configurations of quantum states. We have classified
the resource of multipartite entanglement, finding that there is a
way to obtain a maximally resourceful state, and understood what network
states give rise to GME and GMNL. All of these results take quantum
theory as a given, and their applicability depends on quantum theory
being an apt description of Nature. This is a reasonable assumption
to make, as numerous quantum effects have been confirmed by experiments
to very high precision. However, in principle, other post-classical
theories are possible, and the fact that post-quantum correlations
have not been observed so far does not mean that Nature does not allow
for them. In a departure from the multipartite considerations that
are the focus of the first part of the thesis, Chapter \ref{chap:agreement}
aims to rule out at least some post-quantum theories as possible descriptions
of Nature. In particular, we postulate a principle external to quantum
theory, but satisfied by it, and contend that it should hold in all
reasonable theories of Nature.

Quantum mechanics famously made its creators uncomfortable. Its differences
with classical physics are so structural that the theory seems highly
counterintuitive even today. Almost a century after its introduction,
it still sparks much conceptual and philosophical discussion. Indeed,
an active line of research in quantum foundations deals with the problem
of singling out quantum theory from other post-classical physical
theories. This field is a delicate balance between proposals for new
theories that are `tidier' than quantum mechanics \cite{spekkens_evidence_2007,larsson_contextual_2012}
and proposals for desirable physical principles that such theories
should obey \cite{popescu_quantum_1994,clifton_characterizing_2003,pawlowski_information_2009,sun_no_2019,yan_quantum_2013}.

In Chapter \ref{chap:agreement} we propose a new principle inspired
by a famous result in epistemics, which is the formal study of knowledge
and beliefs. In the domain of classical probability theory, Aumann
proved that Bayesian agents cannot agree to disagree \cite{aumann_agreeing_1976}.
A slightly more general restatement of Aumann's theorem, which we
will refer to as the \emph{classical agreement theorem}, states that,
if Alice and Bob, based on their partial information, assign probabilities
$q_{A},\,q_{B}$, respectively, to perfectly correlated events, and
these probabilities are common certainty between them, then $q_{A}=q_{B}$.
``Certainty'' means assigning probability 1, and ``common certainty''
means that Alice is certain about $q_{B}$; Bob is certain about $q_{A}$;
Alice is certain about Bob being certain about $q_{A}$; Bob is certain
about Alice being certain about $q_{A}$; and so on infinitely.

This result is considered a basic requirement in classical epistemics,
and we contend it should apply to all physical theories. The classical
agreement theorem has been used to show that two risk-neutral agents,
starting from a common prior, cannot agree to bet with each other
\cite{sebenius_dont_1983}, to prove ``no-trade'' theorems for efficient
markets \cite{milgrom_information_1982}, and to establish epistemic
conditions for Nash equilibrium \cite{aumann_epistemic_1995}. These
applications are all external to physics. Of course, the theorem holds
equally in the physical domain, provided that classical probability
theory applies.

But in the quantum domain, the classical model does not apply, and
so we cannot assume that the same facts about agreement and disagreement
between Bayesian agents hold when they observe quantum phenomena.
In particular, a fundamental result of quantum mechanics is that no
local hidden-variable theory can model the results of all quantum
experiments \cite{bell_einstein_1964}. This implies that the classical
Bayesian model does not apply, so the classical agreement theorem
need not hold. The question then arises: can observers of quantum
mechanical phenomena agree to disagree? We address this question by
exploring it in the broader nonsignalling setting.

First, we establish that, in general, nonsignalling agents \textit{can}
agree to disagree about perfectly correlated events, and we give explicit
examples of disagreeing nonsignalling distributions. In the particular
case of two inputs and two outputs, we characterise the distributions
that give rise to common certainty of disagreement. One might think
that the fact that nonsignalling agents can agree to disagree is a
direct consequence of the multitude of uncertainty relations in quantum
mechanics, all of which put a limit on the precision with which the
values of incompatible observables can be measured and have even been
linked to epistemic inconsistencies in quantum mechanics \cite{frauchiger_quantum_2018}.
Somewhat surprisingly, our next finding shows that this is not the
case. We find that disagreeing nonsignalling distributions of two
inputs and outputs cannot be quantum---i.e., the agreement theorem
holds for quantum agents in this setting. Then, we go beyond this
restriction and show that any disagreeing nonsignalling distribution
with more than two inputs or outputs induces a disagreeing distribution
with two inputs and outputs. Since the agreement theorem holds for
quantum agents sharing distributions of two inputs and outputs, it
does so for larger distributions too. Thus, even if quantum mechanics
features uncertainty relations, this does not apply to observers'
estimations of perfectly correlated events.

Next, we ask if nonsignalling and quantum agents can disagree in other
ways. We define a new notion of disagreement, which we call \emph{singular
disagreement}, by removing the requirement of common certainty and,
instead, imposing $q_{A}=1,$ $q_{B}=0,$ and we ask whether it holds
for classical, quantum and nonsignalling agents. We find the same
pattern: singular disagreement does not hold for classical or quantum
agents, but can occur in nonsignalling settings, where we characterise
the distributions that feature it. We then put our two characterisations
together and search for distributions that satisfy both common certainty
of disagreement and singular disagreement: we find that the PR box
\cite{popescu_quantum_1994} is of this kind---i.e., it displays
extremal disagreement in the above sense. This is neat, as the PR
box is known to exhibit the most extreme form of nonsignalling correlations
\cite{barrett_nonlocal_2005}.

Finally, we contend that agreement between observers could be a convenient
principle for testing the consistency of new postquantum theories.
Our results exhibit a clear parametrisation of the set of the probability
distributions that allow observer disagreement. This set is easy to
work with, thanks to its restriction to two observations and two outcomes
per observer. If a new theory can be used to generate such a distribution,
this might raise a red flag, as this theory violates a reasonable
and intuitive and, importantly, testable property that quantum mechanics
satisfies.

Aumann's theorem has appeared elsewhere in the physics literature.
However, it has been examined in a different context \cite{khrennikov_quantum_2015,khrennikov_possibility_2014},
where agents are assumed to use Born's rule as their probability update
rule. The authors conclude that Aumann's theorem does not hold for
this type of agents. Instead, our setting assumes that the agents
are macroscopic and merely share a quantum state or a nonsignalling
box. Our setting is appealing in quantum information for its applications
to communication complexity, cryptography, teleportation, and many
other scenarios. In turn, Ref. \cite{abramsky_non-locality_2019}
introduces a different notion of disagreement in a nonsignalling context.
The disagreement in that work concerns pieces of information about
some variables, and agreement refers to consistency in the information
provided about the variables. Hence, it is unrelated to the epistemic
notion of disagreement that Aumann's theorem defines, and that the
present work revisits from a nonsignalling perspective.

\chapter{\label{chap:maxent}A nontrivial resource theory of multipartite
entanglement}

\chaptermark{Resource theory of multipartite entanglement}

Entanglement theory is formulated as a quantum resource theory in
which the free operations are local operations and classical communication
(LOCC). This defines a partial order among bipartite pure states that
makes it possible to identify a maximally entangled state, which turns
out to be the most relevant state in applications. However, the situation
changes drastically in the multipartite regime. Not only do there
exist inequivalent forms of entanglement forbidding the existence
of a unique maximally entangled state, but recent results have shown
that LOCC induces a trivial ordering: almost all pure entangled multipartite
states are incomparable (i.e. LOCC transformations
among them are almost never possible). In order to cope with this
problem we consider alternative resource theories in which we relax
the class of LOCC to operations that do not create entanglement. We
consider two possible theories depending on whether resources correspond
to multipartite entangled or genuinely multipartite entangled (GME)
states and we show that they are both non-trivial: no inequivalent
forms of entanglement exist in them and they induce a meaningful partial
order (i.e. every pure state is transformable to more weakly entangled
pure states). Moreover, we prove that the resource theory of GME that
we formulate here has a unique maximally entangled state, the generalised
GHZ state, which can be transformed to any other state by the allowed
free operations.

\section{Definitions and preliminaries}

We will consider $n$-partite systems with local dimension $d$, i.e.\ states
in the Hilbert space $\mathcal{H}=\mathcal{H}_{1}\otimes\cdots\otimes\mathcal{H}_{n}=(\mathbb{C}^{d})^{\otimes n}$.
Given a subset $M$ of $[n]=\{1,\ldots,n\}$ and its complement $\bar{M}$,
we denote by $\mathcal{H}_{M}$ the tensor product of the Hilbert
spaces corresponding to the parties in $M$ and analogously with $\mathcal{H}_{\bar{M}}$.
Reviewing Definition \ref{def:intro-ent-multipart} for pure states,
we say $|\psi\rangle\in\mathcal{H}$ is FS (otherwise entangled) if
$|\psi\rangle=|\psi_{1}\rangle\otimes|\psi_{2}\rangle\otimes\cdots\otimes|\psi_{n}\rangle$
for some states $|\psi_{i}\rangle\in\mathcal{H}_{i}$ $\forall i$,
while it is BS (otherwise GME) if $|\psi\rangle=|\psi_{M}\rangle\otimes|\psi_{\overline{M}}\rangle$
for some states $|\psi_{M}\rangle\in\mathcal{H}_{M}$ and $|\psi_{\overline{M}}\rangle\in\mathcal{H}_{\overline{M}}$
and $M\subsetneq[n]$. These notions are extended to mixed states
by the convex hull and we define the sets of FS and BS states by 
\begin{equation}
\mathcal{FS}=\conv\{\psi:|\psi\rangle\textrm{ is FS}\},\;\mathcal{BS}=\conv\{\psi:|\psi\rangle\textrm{ is BS}\}.
\end{equation}

\begin{defn}
A completely positive and trace preserving (CPTP) map $\Lambda:\mathcal{B}(\mathcal{H})\rightarrow\mathcal{B}(\mathcal{H})$
is \emph{full separability-preserving (FSP)} if $\Lambda(\rho)\in\mathcal{FS}$
$\forall\rho\in\mathcal{FS}$. It is \emph{biseparability-preserving
(BSP)} if $\Lambda(\rho)\in\mathcal{BS}$ $\forall\rho\in\mathcal{BS}$.
\end{defn}

We will say that a functional $E$ taking operators on $\mathcal{H}$
to non-negative real numbers is an FSP measure (BSP measure) if $E(\rho)\geq E(\Lambda(\rho))$
for every state $\rho$ and FSP (BSP) map $\Lambda$. This is completely
analogous to entanglement measures, which are required to be non-increasing
under LOCC maps. Although LOCC is a strict subset of the FSP and BSP
maps, some well-known entanglement measures are still FSP or BSP measures
and this will play an important role in assessing which transformations
are possible within the two formalisms that we consider here. Indeed,
measures of the form 
\begin{equation}
E_{\mathcal{X}}(\rho)=\inf_{\sigma\in\mathcal{X}}E(\rho||\sigma),\label{measures}
\end{equation}
where $\mathcal{X}$ stands for either $\mathcal{FS}$ or $\mathcal{BS}$,
have the corresponding monotonicity property as long as the distinguishability
measure $E(\rho||\sigma)$ is \emph{contractive}, i.e.\ $E(\Lambda(\rho)||\Lambda(\sigma))\leq E(\rho||\sigma)$
for every CPTP map $\Lambda$. This includes the relative entropy
of entanglement \cite{vedral_quantifying_1997,vedral_entanglement_1998}
for $E(\rho||\sigma)=\tr(\rho\log\rho)-\tr(\rho\log\sigma)$ and the
robustness ($R_{\mathcal{X}}$) \cite{vidal_robustness_1999} for
\begin{equation}
E(\rho||\sigma)=R(\rho||\sigma)=\min\left\{ s:\frac{\rho+s\sigma}{1+s}\in\mathcal{X}\right\} .\label{defrobustness}
\end{equation}
If one uses the fidelity, $E(\rho||\sigma)=1-F(\rho||\sigma)=1-\tr^{2}\sqrt{\sqrt{\rho}\sigma\sqrt{\rho}}$,
in the case of pure states equation\ (\ref{measures}) boils down
to the geometric measure \cite{wei_geometric_2003}, which we will
denote by $G_{\mathcal{X}}$ and which is then seen to be a measure
under maps that preserve $\mathcal{X}$. We will only need to consider
$G_{\mathcal{X}}$ for pure states:
\begin{equation}
G_{\mathcal{X}}(\cdot)=1-\left(\max_{\ket{\phi}\in\mathcal{X}}\left|\bra{\phi}\ket{\cdot}\right|\right)^{2}
\end{equation}
(see also the definitions on pages \pageref{entropyent}-\pageref{fidelity}).
Notice, however, that, as has been recently shown in the bipartite
case in \cite{chitambar_entanglement_2020}, not all LOCC measures
remain monotonic under non-entangling maps since the latter formalism
allows state conversions that the former does not. In the following,
in order to understand the ordering of resources induced by these
theories, we study which transformations are possible among pure states
under FSP and BSP maps. However, first one should point out that whenever
there exist maps $\Lambda$ and $\Lambda'$ in the corresponding class
of free operations such that $\Lambda(\psi)=\phi$ and $\Lambda^{\prime}(\phi)=\psi$,
then the states $\psi$ and $\phi$ are equally resourceful and should
be regarded as equivalent in the corresponding theory. This is moreover
necessary so as to have a well-defined partial order. Hence, although
for simplicity we will talk about properties of states, one should
have in mind that one is actually speaking about equivalence classes.
Specifically, it is known that two pure states are interconvertible
by LOCC if and only if they are related by local unitary transformations
\cite{gingrich_properties_2002}. Interestingly, we will see that
the equivalence classes are wider in the resource theory of GME under
BSP. It should be stressed that, to our knowledge, this is the first
time that a resource theory of GME is formulated. Notice that the
restriction to LOCC can only have FS states as free states. Furthermore,
allowing a strict subset of parties to act jointly and classical communication
does not fit the bill either as $\mathcal{BS}$ is not closed under
these operations.

Throughout the proofs of the results we will use repeatedly that,
if $\rho_{1}$ and $\rho_{2}$ are density matrices, the map 
\begin{equation}
\Lambda(\rho)=\textnormal{tr}(A\rho)\rho_{1}+\textnormal{tr}[(\mathbbm1-A)\rho]\rho_{2}\label{eq:lambdacptp}
\end{equation}
is CPTP if $0\leq A\leq\mathbbm1$ (see e.g. \cite{chitambar_entanglement_2020}).

\section{Non-triviality of the theories}

Our first two results are valid in both the FSP and BSP regimes. Thus,
following the notation above, the two possible classes of maps will
be referred to as $\mathcal{X}$-preserving.
\begin{thm}[Collapse of the SLOCC classes]
\label{noSLOCC} In a resource theory of entanglement where the free
operations are $\mathcal{X}$-preserving maps, all resource states
are interconvertible with non-zero probability, i.e.\ given any pure
$\psi_{1},\psi_{2}\notin\mathcal{X}$, there exists a completely positive
and trace non-increasing $\mathcal{X}$-preserving map $\Lambda$
such that $\Lambda(\psi_{1})=p\psi_{2}$ with $p\in(0,1]$.
\end{thm}

\begin{proof}
The proof is based on explicitly constructing a completely positive
and trace non-increasing $\mathcal{X}$-preserving map $\Lambda$
such that $\Lambda(\psi_{1})=p\psi_{2}$, for some $p$ that can be
ensured to be strictly larger than 0.

Notice that, since $\psi_{1}\,,\psi_{2}\notin\mathcal{X}$ and both
the geometric measure and the robustness are faithful measures \cite{wei_geometric_2003,vidal_robustness_1999},
$R_{\mathcal{X}}(\psi_{2})\,,\:G_{\mathcal{X}}(\psi_{1})>0$. Also,
$G_{\mathcal{X}}(\psi_{1})<1$ because the fully (bi-)separable states
span the whole Hilbert space. Pick $p\in]0,1]$ such that 
\begin{equation}
p\leq\frac{1}{R_{\mathcal{X}}(\psi_{2})}\frac{G_{\mathcal{X}}(\psi_{1})}{1-G_{\mathcal{X}}(\psi_{1})}\label{eq:p_noSFSPclasses}
\end{equation}
and let 
\begin{equation}
\Lambda(\eta)=p\tr(\psi_{1}\eta)\psi_{2}+\tr\left[(\mathbbm1-\psi_{1})\eta\right]\rho_{\mathcal{\mathcal{X}}}\,.\label{eq:lambdaprob}
\end{equation}
Here $\rho_{\mathcal{\mathcal{X}}}\in\mathcal{X}$ is the state which
gives the corresponding robustness of $\psi_{2}\,,$ i.e., $R_{\mathcal{X}}(\psi_{2})=R(\psi_{2}||\rho_{\mathcal{X}})$---cf.
equation (\ref{defrobustness}). (Note that $\Lambda$ can be completed
to a CPTP $\mathcal{X}$-preserving map by adding a term of the form
$\Lambda'(\eta)=(1-p)\tr(\psi_{1}\eta)\rho_{\mathcal{\mathcal{X}}}\,.)$
Then $\Lambda(\psi_{1})=p\psi_{2}$ and it remains to be shown that
$\Lambda$ is $\mathcal{X}$-preserving. Let $\sigma\in\mathcal{\mathcal{X}}\,.$
Then 
\begin{equation}
\Lambda(\sigma)\propto\psi_{2}+\frac{1}{p}\left(\frac{1}{\tr(\psi_{1}\sigma)}-1\right)\rho_{\mathcal{\mathcal{X}}}\,,
\end{equation}
so $\Lambda(\sigma)/\tr(\Lambda(\sigma))\in\mathcal{\mathcal{X}}$
iff $\frac{1}{p}\left(\frac{1}{\tr(\psi_{1}\sigma)}-1\right)\geq R_{\mathcal{X}}(\psi_{2})\,.$
But this holds from equation (\ref{eq:p_noSFSPclasses}) and using
$\tr(\psi_{1}\sigma)\leq1-G_{\mathcal{X}}(\psi_{1})$ $\forall\sigma\in\mathcal{X}.$
\end{proof}
\begin{thm}[No isolation]
\label{noisolation} In a resource theory of entanglement where the
free operations are $\mathcal{X}$-preserving maps, no resource state
is isolated, i.e.\ given any pure $\psi_{1}\notin\mathcal{X}$ on
$\mathcal{H}$, there exists an inequivalent pure $\psi_{2}\notin\mathcal{X}$
on $\mathcal{H}$ and a CPTP $\mathcal{X}$-preserving map $\Lambda$
such that $\Lambda(\psi_{1})=\psi_{2}$.
\end{thm}

\begin{proof}
This result arises as a corollary of Theorem \ref{noSLOCC}. Given
any $\psi_{1}\notin\mathcal{X}$, continuity arguments show that there
always exists an inequivalent $\psi_{2}\notin\mathcal{X}$ with $R_{\mathcal{X}}(\psi_{2})$
small enough so that one can take $p=1$ in equation\ (\ref{eq:p_noSFSPclasses})
and construct a CPTP map.

Consider the map (\ref{eq:lambdaprob}) from the proof of Theorem
\ref{noSLOCC}. This map can be made deterministic if $R_{\mathcal{X}}(\psi_{2})$
is sufficiently smaller than $G_{\mathcal{X}}(\psi_{1})\,.$ Indeed,
if 
\begin{equation}
\frac{1}{R_{\mathcal{X}}(\psi_{2})}\frac{G_{\mathcal{X}}(\psi_{1})}{1-G_{\mathcal{X}}(\psi_{1})}>1\,,\label{eq:nontrivialhierarchy}
\end{equation}
then we can pick $p=1$ in the map (\ref{eq:lambdaprob}) so $\Lambda$
is CPTP (see equation (\ref{eq:lambdacptp})). Since robustness is
a continuous function of the input state \cite{vidal_robustness_1999},
it can be arbitrarily close to zero and so there exists $\psi_{2}$
such that the above condition is fulfilled for any $\psi_{1}\,.$
Further, $\psi_{1},\psi_{2}$ are inequivalent if they have different
robustness, but $R(\psi_{2})$ can always be picked to be different
from $R(\psi_{1})$ and still satisfying equation (\ref{eq:nontrivialhierarchy}).
\end{proof}
Theorem \ref{noSLOCC} proves that in our case there are no inequivalent
forms of entanglement. This is in sharp contrast to LOCC where, leaving
aside the case $\mathcal{H}=(\mathbb{C}^{2})^{\otimes3}$, the state
space splits into a cumbersome zoology of infinitely many different
SLOCC classes of unrelated entangled states. Theorem \ref{noisolation}
provides the non-triviality of our theories. While almost all states
turn out to be isolated under LOCC \cite{sauerwein_transformations_2018},
our classes of free operations induce a meaningful partial order structure
where, as in the case of bipartite entanglement, every pure state
can be transformed into a more weakly entangled pure state. It is
important to mention that the result of \cite{sauerwein_transformations_2018}
proves generic isolation when transformations are restricted among
GME states with the rank of all $n$ single-particle reduced density
matrices equal to $d$. However, Theorem \ref{noisolation} still
holds under this restriction.

Theorems \ref{noSLOCC} and \ref{noisolation} show that limitations
of the resource theory of multipartite entanglement under LOCC can
be overcome if one considers FSP or BSP operations instead. These
positive results raise the question of whether the induced structure
is powerful enough to have a unique multipartite maximally entangled
state. If this were so, our theories would point to a relevant class
of states that should be at the heart of the applications of multipartite
entanglement in a similar fashion to the maximally entangled state
in the bipartite case. In order to answer this question, we first
provide an unambiguous definition of a maximally resourceful state
which, on the analogy of the bipartite case, depends on the number
of parties $n$ and local dimension $d$: a state $\psi$ on $\mathcal{H}$
is the maximally resourceful state on $\mathcal{H}$ if it can be
transformed by means of the free operations into any other state on
$\mathcal{H}$.\footnote{Notice that this already implies that there exists no free operation
that transforms an inequivalent state on $\mathcal{H}$ into $\psi$.} 

\section{Existence of a maximally entangled state}

\subsection{FSP regime\label{sec:FSP-regime}}

We analyse the case of FSP operations first, where we find no maximally
resourceful state exists.
\begin{thm}
\label{fspth} In the resource theory of entanglement where the free
operations are FSP maps, there exists no maximally entangled state
on $\mathcal{H}=(\mathbb{C}^{2})^{\otimes3}$.
\end{thm}

To prove this result, we use that, if a maximally entangled state
in this case existed, it would need to be the W state $|W\rangle=(|001\rangle+|010\rangle+|100\rangle)/\sqrt{3}$.
This is because it has been shown in \cite{chen_computation_2010}
that the W state is the unique state in this Hilbert space that achieves
the maximal possible value of $G_{\mathcal{FS}}$, which we have shown
above to be an FSP measure. Thus, if there existed a maximally entangled
state, it would be necessary that the W state could be transformed
by FSP into any other state. However, we show that there exists no
FSP map transforming the W state into the GHZ state
\begin{equation}
\ket{GHZ(3,2)}=\frac{1}{\sqrt{2}}(\ket{000}+\ket{111}).
\end{equation}
To verify this last claim, it suffices to find an FSP measure $E$
such that $E(GHZ)>E(W)$. However, as discussed above, not many FSP
measures are known and, as with the geometric measure, it is also
known that the relative entropy of entanglement of the W is larger
than that of the GHZ state \cite{markham_entanglement_2007}. This
leaves us then with the robustness measure $R_{\mathcal{FS}}$, for
which we are able to show that $R_{\mathcal{FS}}(W)=R_{\mathcal{FS}}(GHZ)=2$.
This alone does not forbid that $W\to_{FSP}GHZ$ but, from the insight
developed in computing these quantities, an obstruction to such transformation
can be found even though they are equally robust. It is worth mentioning
that, to our knowledge, this is the first time that the robustness
is computed for multipartite states and we have reasons to conjecture
that the W and GHZ states attain its maximal value on $\mathcal{H}$,
and they are the only states that do so.

To prove this theorem, it is useful to introduce the following two
lemmas in order to compute the robustness of the W and GHZ states.
\begin{lem}
\label{lem:robGHZ} $R_{\mathcal{FS}}(GHZ)=2\,.$
\end{lem}

\begin{proof}
The robustness can be bounded from above from the definition (equations
(\ref{measures}),~(\ref{defrobustness})), as any fully separable
state which is a convex combination of the GHZ state with a fully
separable state will give an upper bound to the robustness. Ref. \cite{brandao_quantifying_2005}
provides a dual characterisation in terms of entanglement witnesses
which we use to bound the robustness from below: 
\begin{equation}
R_{\mathcal{FS}}(\rho)=\max\left\{ 0,-\min_{\mathcal{W}\in\mathcal{M}}\tr(\mathcal{W}\rho)\right\} \,.\label{eq:witnessrobustness}
\end{equation}
A witness for a state $\rho$ is an operator $\mathcal{W}$ such that
$\tr(\mathcal{W}\sigma)\geq0$ for all $\sigma\in\mathcal{FS}$ and
$\tr(\mathcal{W}\rho)<0\,.$ If the witness also satisfies $\tr(\mathcal{W}\sigma)\leq1$
for all $\sigma\in\mathcal{FS}$ (which defines the set $\mathcal{M}$
above), then $-\tr(\mathcal{W}\rho)$ is a lower bound to the robustness.

First, we show $R_{\mathcal{FS}}(GHZ)\leq2\,.$ We will use the following
notation as a means to characterise full separability of certain states
(this is a simplified version of the separability criterion in \cite[\S 2.1]{dur_separability_1999}):
a state of the form 
\begin{equation}
\begin{aligned}\rho(\lambda^{+},\lambda^{-},\lambda)=\lambda^{+}GHZ\,+\, & \lambda^{-}GHZ_{-}\,+\frac{\lambda}{6}\sum_{i=001}^{110}\ketbra{i}\,,\end{aligned}
\label{eq:DCTform}
\end{equation}
where $\ket{GHZ_{-}}=(\ket{000}-\ket{111})/\sqrt{2}$ and the summation
index $i$ ranges from 001 to 110 in binary, is fully separable iff
\begin{equation}
|\lambda^{+}-\lambda^{-}|\leq\lambda/3\,.\label{eq:critsep}
\end{equation}
We must also have $\lambda^{+}+\lambda^{-}+\lambda=1$ for normalisation,
and $\lambda^{\pm},\lambda\geq0$ for $\rho(\lambda^{+},\lambda^{-},\lambda)$
to be positive. Thus, the set of fully separable states of the form
(\ref{eq:DCTform}) is a polytope, and this property will be used
later.

Consider the following state: 
\begin{equation}
\frac{1}{3}\left(GHZ+2\rho\left(0,\frac{1}{4},\frac{3}{4}\right)\right)=\rho\left(\frac{1}{3},\frac{1}{6},\frac{1}{2}\right)\,,\label{eq:robGHZ}
\end{equation}
It is straightforward to check that both $\rho\left(0,\frac{1}{4},\frac{3}{4}\right)$
and $\rho\left(\frac{1}{3},\frac{1}{6},\frac{1}{2}\right)$ satisfy
(\ref{eq:critsep}) with equality, so $R_{\mathcal{FS}}(GHZ)\leq2\,.$

Next, we show $R_{\mathcal{FS}}(GHZ)\geq2\,.$ Let

\begin{equation}
\mathcal{W}=\frac{2}{3}\mathbbm1-\frac{8}{3}GHZ+\frac{4}{3}GHZ_{-}\label{eq:robwitGHZ}
\end{equation}
be a candidate witness for this purpose. To show $0\leq\textnormal{tr}(\mathcal{W}\sigma)\leq1$
for all fully separable states $\sigma$, it is enough to restrict
to states $\sigma$ of the form (\ref{eq:DCTform}), as can be shown
by considering the twirling map $T_{GHZ}$ onto the GHZ-symmetric
subspace. This map is defined in \cite{hayashi_bounds_2006}, but
we will only need the following properties: it is FSP and self-dual,
it maps all states onto states of the form (\ref{eq:DCTform}), i.e.
\begin{equation}
T_{GHZ}(\tau)=\rho(\lambda^{+},\lambda^{-},\lambda)
\end{equation}
for every state $\tau$ on $\mathcal{H}$ and for some $\lambda^{\pm},\lambda$
and, moreover, these states are fixed points: \sloppy $T_{GHZ}(\rho(\lambda^{+},\lambda^{-},\lambda))=\rho(\lambda^{+},\lambda^{-},\lambda)$
for all $\lambda^{\pm},\lambda\,.$ In particular, $T_{GHZ}(GHZ)=GHZ$
and the witness $\mathcal{W}$ in equation (\ref{eq:robwitGHZ}) is
such that $T_{GHZ}\,(\mathcal{W})=\mathcal{W}\,,$ and so 
\begin{equation}
\tr(\mathcal{W}\sigma)=\tr(T_{GHZ}\,(\mathcal{W})\,\sigma)=\tr(\mathcal{W}\,T_{GHZ}(\sigma))
\end{equation}
holds for any state $\sigma\,.$ Therefore, if $0\leq\textnormal{tr}(\mathcal{W}\sigma)\leq1$
holds for all $\sigma\in\mathcal{FS}$ such that $T_{GHZ}(\sigma)=\sigma\,,$
i.e. those of the form (\ref{eq:DCTform}) where (\ref{eq:critsep})
holds \cite{eltschka_entanglement_2012,eltschka_optimal_2013}, then
it is guaranteed to hold for any $\sigma\in\mathcal{FS}\,.$

As the space of fully separable GHZ-symmetric states is a polytope,
it is enough to show that $0\leq\trace(\mathcal{W}\sigma)\leq1$ at
the vertices of the polytope, which are (cf. \cite{eltschka_entanglement_2012,eltschka_optimal_2013}):
\begin{align}
\sigma_{1} & =\rho(0,0,1)\nonumber \\
\sigma_{2} & =\rho\left(0,\frac{1}{4},\frac{3}{4}\right)\nonumber \\
\sigma_{3} & =\rho\left(\frac{1}{2},\frac{1}{2},0\right)\\
\sigma_{4} & =\rho\left(\frac{1}{4},0,\frac{3}{4}\right)\,.\nonumber 
\end{align}
It is straightforward to check that $0\leq\trace(\mathcal{W}\sigma_{j})\leq1$
for all $j=1,...,4\,.$ Since $\trace(\mathcal{W\,}GHZ)=-2<0\,,$
$\mathcal{W}$ is a witness for the GHZ-state that meets the required
condition and so $R_{\mathcal{FS}}(GHZ)\geq2\,.$
\end{proof}
\begin{lem}
$R_{\mathcal{FS}}(W)=2\,.$
\end{lem}

\begin{proof}
The strategy is similar to the proof of Lemma \ref{lem:robGHZ}. First,
we prove $R_{\mathcal{FS}}(W)\leq2\,.$ We will show that 
\begin{equation}
\eta=\frac{1}{3}(W+2\tau)\,,\label{eq:robW}
\end{equation}
where 
\begin{equation}
\eta=\frac{9}{16}\ketbra{000}+\frac{3}{16}\ketbra{111}+\frac{1}{16}W+\frac{3}{16}\overline{W}\label{eq:rhoconvcomb}
\end{equation}
and 
\begin{equation}
\tau=\frac{3}{8}\ketbra{000}+\frac{1}{8}\ketbra{111}+\frac{3}{8}W+\frac{1}{8}\overline{W}\label{eq:sigmaconvcomb}
\end{equation}
are both fully separable. Here and in what follows, $\overline{W}$
denotes the qubit-flipped version of the $W$-state, 
\begin{equation}
\ket{\overline{W}}=\frac{1}{\sqrt{3}}\left(\ket{110}+\ket{101}+\ket{011}\right)\,.
\end{equation}
As shown in Theorem 6.2 of \cite{eckert_quantum_2002}, if a symmetric
3-qubit state remains positive after partial transposition (PPT),
then it is FS. Since both $\eta$ and $\tau$ are symmetric 3-qubit
states, it is enough to check that they are PPT, which is readily
done, to conclude that they are fully separable.

Another way to see this is by writing $\eta$ and $\tau$ as a convex
combination of fully separable states using a result from \cite{hayashi_entanglement_2008}.
Observe that 
\begin{equation}
\begin{aligned}\eta=\frac{5}{9}\ketbra{000}+ & \frac{4}{9}\left(\frac{1}{2^{6}}\ketbra{000}+\frac{27}{2^{6}}\ketbra{111}+\frac{9}{2^{6}}W+\frac{27}{2^{6}}\overline{W}\right)\end{aligned}
\label{eq:rhoconvcomb-1}
\end{equation}
and 
\begin{equation}
\begin{aligned}\tau=\frac{1}{9}\ketbra{111}+ & \frac{8}{9}\left(\frac{27}{2^{6}}\ketbra{000}+\frac{1}{2^{6}}\ketbra{111}+\frac{27}{2^{6}}W+\frac{9}{2^{6}}\overline{W}\right)\end{aligned}
\label{eq:sigmaconvcomb-1}
\end{equation}
where, in each case, the first term is clearly fully separable. As
we shall see, the second term is of the form

\begin{equation}
\begin{aligned}\tr & (\phi^{\otimes3}\ketbra{000})\ketbra{000}+\tr(\phi^{\otimes3}\ketbra{111})\ketbra{111}\\
+ & \tr(\phi^{\otimes3}\,W\,)\,W+\tr(\phi^{\otimes3}\,\overline{W}\,)\,\overline{W}
\end{aligned}
\label{eq:twirled}
\end{equation}
for some qubit state $\phi$. Ref. \cite{hayashi_entanglement_2008}
shows that all states of this form are fully separable. Writing 
\begin{equation}
\ket{\phi}=\cos\alpha\ket{0}+\ue^{\ui\beta}\sin\alpha\ket{1}\,.
\end{equation}
and inserting it into equation (\ref{eq:twirled}), the parameter
$\beta$ cancels in all terms and the state in equation (\ref{eq:twirled})
can be written in terms of $\alpha$ alone with $\alpha=\pi/3$ for
$\eta$ and $\alpha=\pi/6$ for $\tau\,.$

Next, we prove $R_{\mathcal{FS}}(W)\geq2\,.$ We will show that 
\begin{equation}
\begin{aligned}A= & \ketbra{000}-3\,W+\ketbra{001}+\ketbra{010}+\ketbra{100}+3\,\overline{W}\end{aligned}
\end{equation}
is a witness for the state $\ketbra{W}$ such that 
\begin{equation}
\tr(AW)=-2
\end{equation}
and 
\begin{equation}
0\leq\tr(A\sigma)\leq1\label{eq:trace1separable}
\end{equation}
for all $\sigma\in\mathcal{FS}\,.$

Let $\sigma\in\mathcal{FS}\,.$ Without loss of generality, to prove
(\ref{eq:trace1separable}) we can assume $\sigma=\ketbra{\psi}$
is pure. So we want to show 
\begin{equation}
0\leq\tr(A\ketbra{\psi})\leq1\,.\label{eq:trApsi01}
\end{equation}
Notice that $A$ is permutationally invariant, and that we can express
$A$ in the basis of Pauli matrices as 
\begin{equation}
A=\sum_{ijk\in{x,y,z}}\lambda_{ijk}\sigma_{i}\otimes\sigma_{j}\otimes\sigma_{k}+\frac{\mathbbm1_{8}}{2}
\end{equation}
for some $\lambda_{ijk}\in\mathbb{R}$ and where $\mathbbm1_{d}$
is the $d$-dimensional identity, so that 
\begin{equation}
A'=A-\frac{\mathbbm1_{8}}{2}
\end{equation}
has no identity component in the basis of Pauli matrices. That is,
$A'$ contains only full correlation terms, and it is still permutationally
invariant so it satisfies the conditions of Corollary 5 (ii) in \cite{hubener_geometric_2009}.
In particular, $A'$ can be viewed as a symmetric three-linear form
acting on $\mathbb{R}^{3}$. This means that 
\begin{equation}
\max_{\ket{\psi}\in\mathcal{FS}}\left|\tr(A'\ketbra{\psi})\right|
\end{equation}
can be attained by a symmetric state $\ket{\psi}=\ket{a}\ket{a}\ket{a}\equiv\ket{aaa}\,.$
The qubit $\ket{a}$ can be expressed in terms of two real parameters
as 
\begin{equation}
\ket{a}=\cos\alpha\ket{0}+\ue^{\ui\beta}\sin\alpha\ket{1}
\end{equation}
and so 
\begin{equation}
\left|\tr(A'\ketbra{aaa})\right|=\frac{1}{2}\left|\cos6\alpha\right|\leq\frac{1}{2}\,.
\end{equation}
But this completes the proof, since, by linearity, to show 
\begin{equation}
-\frac{1}{2}\leq\tr(A'\ketbra{\psi})\leq\frac{1}{2}
\end{equation}
(which is equivalent to (\ref{eq:trApsi01})) it suffices to show
\begin{equation}
\max_{\ket{\psi}\in\mathcal{FS}}\left|\tr(A'\ketbra{\psi})\right|\leq\frac{1}{2}\,.
\end{equation}
This can be seen by viewing $\tr(A'\ketbra{\psi})$ as a symmetric
three-linear form in $\mathbb{R}^{3}\,.$ If the maximum absolute
value is attained by some state $\ket{a^{*}}\,,$ then the state $\ket{\tilde{a}^{*}}$
which flips the sign of the vector which the three-linear form acts
on will give a minimum of the expression equal to minus the maximum.
Hence, 
\begin{equation}
\begin{aligned}\max_{\ket{\psi}\in\mathcal{FS}}\left|\tr(A'\ketbra{\psi})\right| & =\max_{\ket{\psi}\in\mathcal{FS}}\tr(A'\ketbra{\psi})\\
 & =-\min_{\ket{\psi}\in\mathcal{FS}}\tr(A'\ketbra{\psi})\,.
\end{aligned}
\end{equation}
Therefore (\ref{eq:trApsi01}) holds true and hence the witness $A$
gives the stated lower bound for the FS robustness of $W$.
\end{proof}
We note that the values obtained for the robustness $R_{\mathcal{FS}}$
of the W and GHZ states show that, unlike in the bipartite case, the
robustness can be strictly larger than the generalised robustness.
The generalised robustness, $R_{G}(\cdot)\,,$ is defined as 
\begin{equation}
R_{G}(\cdot)=\min_{\tau\in\mathcal{H}}R(\cdot||\tau)
\end{equation}
where, this time, $\tau$ may be separable or entangled. Hence $R_{G}(\cdot)\leq R(\cdot)$
but, in addition, it was shown in \cite{harrow_robustness_2003} that
$R_{G}(\cdot)=R(\cdot)$ for bipartite pure states. However, the generalised
robustness of the W state has been computed in \cite{hayashi_entanglement_2008}
to be $5/4,$ and that of the GHZ state was shown to be $1$ in \cite{hayashi_bounds_2006},
so they are both strictly less than the robustness of these states.
To the best of our knowledge, this is the first time that states such
that $R_{G}(\cdot)<R(\cdot)$ have been found.

We are now ready to prove Theorem \ref{fspth}.
\begin{proof}
As we outlined in the main text, the only candidate for a maximally
entangled state of three qubits is the W state, as it is the unique
state on $\mathcal{H}=\left(\mathbb{C}^{2}\right)^{\otimes3}$ that
achieves the maximum value of the FSP-measure $G_{\mathcal{FS}}$
(among both pure and mixed states, since the convex-roof extension
of $G_{\mathcal{FS}}$ to mixed states ensures that the maximum value
will always be achieved by a pure state). So, if there existed a maximally
entangled state, it would need to be possible that the W state be
transformed into any other state via an FSP map. We will assume that
there exists an FSP map $\Lambda$ such that $\Lambda(W)=GHZ\,,$
and will arrive at a contradiction by showing that there exists a
state $\eta\in\mathcal{FS}$ such that $\Lambda(\eta)\not\in\mathcal{FS}\,.$

Let $\Lambda$ be an FSP map such that $\Lambda(W)=GHZ$ and let 
\begin{equation}
\eta=\frac{1}{3}W+\frac{2}{3}\tau\in\mathcal{FS}\,,\label{eq:R(W)upper}
\end{equation}
where $\tau,\eta\in\mathcal{FS}\,,$ be the convex combination that
gives the upper bound to $R_{\mathcal{FS}}(W)$ in equations (\ref{eq:robW})-(\ref{eq:sigmaconvcomb}).
Let $T_{GHZ}$ be the twirling map onto the $GHZ$-symmetric subspace
(defined in \cite{hayashi_entanglement_2008}; see also the proof
of Lemma \ref{lem:robGHZ}). Then, 
\begin{equation}
\begin{aligned}\eta'= & T_{GHZ}\left(\Lambda\left(\frac{1}{3}W+\frac{2}{3}\tau\right)\right)\\
= & \frac{1}{3}GHZ+\frac{2}{3}T_{GHZ}(\Lambda(\tau))\,.
\end{aligned}
\end{equation}
Since both $T_{GHZ}$ and $\Lambda$ are full separability-preserving,
it is the case that $\eta',\Lambda(\tau)\,,$ $T_{GHZ}(\Lambda(\tau))\in\mathcal{FS}\,.$
Now, recall that $\tau$ has a non-zero $W$ component: 
\[
\tau=p\ketbra{W}+(1-p)\xi
\]
for some $p\in(0,1)$ and some state $\xi$, so that 
\begin{equation}
\eta'=\frac{1}{3}GHZ+\frac{2}{3}\left[p\,GHZ+(1-p)T_{GHZ}(\Lambda(\xi))\right]\,.\label{eq:tghzlambda}
\end{equation}
But, as we shall now show, the FS $GHZ$-symmetric state $\upsilon$
such that 
\begin{equation}
\frac{1}{3}GHZ+\frac{2}{3}\upsilon\in\mathcal{FS}\label{eq:GHZ_fs}
\end{equation}
is unique, i.e. if equation (\ref{eq:GHZ_fs}) holds then necessarily
$\upsilon=\rho(0,1/4,3/4)$ as in equation (\ref{eq:robGHZ}). However,
the state appearing in equation (\ref{eq:tghzlambda}) is not $\upsilon$
(since $\tr(\upsilon\,GHZ)=0$) hence, contrary to our assumption,
$\eta'$ cannot be FS.

Recall, from the proof of Lemma \ref{lem:robGHZ} (equation (\ref{eq:DCTform})),
that all $GHZ$-symmetric states are of the form

\begin{equation}
\begin{aligned}\begin{aligned}\rho(\lambda^{+},\lambda^{-},\lambda)= & \lambda^{+}GHZ\,+\lambda^{-}GHZ_{-}+\frac{\lambda}{6}\sum_{i=001}^{110}\ketbra{i}\end{aligned}
\end{aligned}
\label{eq:DCTform-1}
\end{equation}
so that equation (\ref{eq:GHZ_fs}) can be expressed in terms of the
$\lambda$ parameters as 
\begin{equation}
\frac{1}{3}GHZ+\frac{2}{3}\rho\left(\lambda^{+},\lambda^{-},\lambda\right)=\rho\left(\frac{1}{3}+\frac{2}{3}\lambda^{+},\frac{2}{3}\lambda^{-},\frac{2}{3}\lambda\right)\,.\label{eq:GHZ_fs_param}
\end{equation}
States of the form (\ref{eq:DCTform-1}) are fully separable iff 
\begin{equation}
|\lambda^{+}-\lambda^{-}|\leq\lambda/3\,.\label{eq:fs_condition}
\end{equation}
Since this condition must hold for both states $\rho(\cdot,\cdot,\cdot)$
in equation (\ref{eq:GHZ_fs_param}), we must also have 
\begin{equation}
\left|\frac{1}{3}+\frac{2}{3}\lambda^{+}-\frac{2}{3}\lambda^{-}\right|\leq\frac{2}{9}\lambda
\end{equation}
and, for normalisation, we need 
\begin{equation}
\lambda^{+}+\lambda^{-}+\lambda=1\,.
\end{equation}
It is straightforward to check that these three conditions hold only
if 
\begin{equation}
\lambda^{+}=0;\:\lambda^{-}=1/4;\:\lambda=3/4\,,
\end{equation}
which corresponds to the state $\upsilon$ as claimed above.

Therefore $\eta$ in equation (\ref{eq:R(W)upper}) is fully separable,
yet $\eta'=T_{GHZ}(\Lambda(\eta))$ is not fully separable. So $\Lambda$
is not FSP and hence the theorem is proven.
\end{proof}
Theorem \ref{fspth} forbids the existence of a multipartite maximally
entangled state under FSP in the simplest case of $\mathcal{H}=(\mathbb{C}^{2})^{\otimes3}$.
However, it is instructive to compare with the LOCC scenario since
these values of $n$ and $d$ make up the only case where no state
is isolated in the latter formalism (aside from the bipartite case).
Whenever no single maximally entangled state exists one needs to consider
a maximally entangled set (MES) \cite{de_vicente_maximally_2013},
defined as the minimal set of states on $\mathcal{H}$ such that any
state on $\mathcal{H}$ can be obtained by means of the free operations
from a state in this set. The MES under LOCC for $n=3$ and $d=2$
has been characterised in \cite{de_vicente_maximally_2013}, and it
is found to be relatively small in the sense that it has measure zero
on $\mathcal{H}$ (in contrast, for other values of $n$ and $d$
the fact that isolation is generic imposes that the MES has full measure
on $\mathcal{H}$). However, interestingly, the MES under FSP is smaller
even in this case, given that it is strictly included in the MES under
LOCC. This is because, as we will now show, the W and GHZ states can
be transformed by FSP operations into inequivalent states that are
in the MES under LOCC. It is worth mentioning that the target states
may be chosen to lie in different SLOCC classes with respect to the
initial states, and so this gives an explicit example of deterministic
FSP conversions among states in different SLOCC classes.

Let $\psi_{GHZ}^{+}$ denote states of the form 
\begin{equation}
\ket{\psi_{GHZ}^{+}}=\sqrt{K}\left(\ket{000}+\ket{\phi_{A}\phi_{B}\phi_{C}}\right)\label{eq:ghzplus}
\end{equation}
where 
\begin{equation}
\begin{aligned}\ket{\phi_{A}}= & \cos\alpha\ket{0}+\sin\alpha\ket{1}\,,\\
\ket{\phi_{B}}= & \cos\beta\ket{0}+\sin\beta\ket{1}\,,\\
\ket{\phi_{C}}= & \cos\gamma\ket{0}+\sin\gamma\ket{1}\,,
\end{aligned}
\end{equation}
$\alpha,\beta,\gamma\in(0,\pi/2]$ and $K=(2(1+\cos\alpha\cos\beta\cos\gamma))^{-1}$
is a normalisation factor. States of the form $\psi_{GHZ}^{+}$ are
in the MES under LOCC, since they cannot be reached by any LOCC map
regardless of the input state on $\mathcal{H}=(\mathbb{C}^{2})^{\otimes3}$
\cite{dur_three_2000,turgut_deterministic_2010,de_vicente_maximally_2013}.
So the following proposition does not hold in the LOCC regime.
\begin{prop}
There exists an FSP map $\Lambda$ such that $\Lambda(W)=\psi_{GHZ}^{+}$
for some state of the form $\psi_{GHZ}^{+}\,.$
\end{prop}

\begin{proof}
Let 
\begin{equation}
\Lambda(\eta)=\tr(W\eta)\psi_{GHZ}^{+}+\tr\left[(\mathbbm1-W)\eta\right]\tau_{\mathcal{FS}},
\end{equation}
where $\tau_{\mathcal{FS}}\in\mathcal{FS}$ is the state that gives
the robustness of the state $\psi_{GHZ}^{+}\,.$ Clearly, $\Lambda(W)=\psi_{GHZ}^{+}$
and it remains to be shown that $\Lambda$ is FSP. As argued in Theorems
\ref{noSLOCC} and \ref{noisolation}, this happens when 
\begin{equation}
R_{\mathcal{FS}}(\psi_{GHZ}^{+})\leq\frac{G_{\mathcal{FS}}(W)}{1-G_{\mathcal{FS}}(W)}=\frac{5}{4}\,.
\end{equation}
But, by continuity of the robustness, such a state $\psi_{GHZ}^{+}$
can always be found by picking the parameters $\alpha,\beta,\gamma$
sufficiently close to zero since in this case the states $\psi_{GHZ}^{+}$
approach the set of FS states.

Anyway, for the sake of completeness, we provide an explicit quantitative
upper bound in what follows. Consider the invertible local operations
\begin{align}
A & =\left(\begin{array}{cc}
1 & \cos\alpha\\
0 & \sin\alpha
\end{array}\right)\,,\nonumber \\
B & =\left(\begin{array}{cc}
1 & \cos\beta\\
0 & \sin\beta
\end{array}\right)\,,\\
C & =\left(\begin{array}{cc}
1 & \cos\gamma\\
0 & \sin\gamma
\end{array}\right)\,.\nonumber 
\end{align}
Applying these to the FS states in equation (\ref{eq:robGHZ}) used
to bound the robustness of the GHZ state, 
\begin{equation}
A\otimes B\otimes C\left(\frac{1}{3}GHZ+\frac{2}{3}\upsilon\right)A^{\dagger}\otimes B^{\dagger}\otimes C^{\dagger}\,,
\end{equation}
gives a state proportional to 
\begin{equation}
\frac{1}{3}\left(1+\cos\alpha\cos\beta\cos\gamma\right)\psi_{GHZ}^{+}+\frac{2}{3}\frac{4-\cos\alpha\cos\beta\cos\gamma}{4}\upsilon'\,,\label{eq:robpsighzplus}
\end{equation}
where $\upsilon'=A\otimes B\otimes C\upsilon A^{\dagger}\otimes B^{\dagger}\otimes C^{\dagger}$
is still fully separable since local operations cannot create entanglement.
For the same reason, the state in equation (\ref{eq:robpsighzplus})
is fully separable, and hence the robustness of the state $\psi_{GHZ}^{+}$
cannot exceed\footnote{This shows, in particular, that the robustness of all $\psi_{GHZ}^{+}$
states is always less than or equal to 2.} 
\begin{equation}
R_{\mathcal{FS}}(\psi_{GHZ}^{+})\leq\frac{4-\cos\alpha\cos\beta\cos\gamma}{2(1+\cos\alpha\cos\beta\cos\gamma)}\,.
\end{equation}
Clearly, there exist $\alpha,\beta,\gamma\in(0,\pi/2]$ such that
this bound is lower than or equal to 5/4, as required. For an example,
take $\alpha=\beta=\pi/2$ and $\gamma$ such that $\cos\gamma\geq6/7\,.$
\end{proof}
We will now show the converse result: there are FSP maps which take
the $GHZ$-state to states in the $W$-class which are in the MES
under LOCC. Such states are of the form 
\begin{equation}
\begin{aligned}\ket{\psi_{W}}= & \sqrt{x_{1}}\ket{001}+\sqrt{x_{2}}\ket{010}+\sqrt{x_{3}}\ket{100}\end{aligned}
\label{eq:wform}
\end{equation}
where $x_{1}+x_{2}+x_{3}=1$. They are in the MES under LOCC, as no
LOCC map can reach these states for any input state on $\mathcal{H}=(\mathbb{C}^{2})^{\otimes3}$
\cite{dur_three_2000,kintas_transformations_2010,de_vicente_maximally_2013},
but (as we will now prove) not under FSP. 
\begin{prop}
There exists an FSP map $\Lambda$ such that $\Lambda(GHZ)=\psi_{W}$
for some state of the form $\psi_{W}\,.$
\end{prop}

\begin{proof}
Since $G_{\mathcal{FS}}(GHZ)=1/2\,,$ it suffices to find a state
$\psi_{W}$ such that $R_{\mathcal{FS}}(\psi_{W})\leq1\,,$ which
can be done since the robustness is continuous and there are states
$\psi_{W}$ arbitrarily close to the set of FS states. Then, 
\begin{equation}
\Lambda(\eta)=\tr(GHZ\eta)\psi_{W}+\tr\left[(\mathbbm1-GHZ)\eta\right]\tau_{\mathcal{FS}},
\end{equation}
where $\tau_{\mathcal{FS}}\in\mathcal{FS}$ is the state such that
that $R_{\mathcal{FS}}(\psi_{W})=R(\psi_{W}||\tau_{\mathcal{FS}})\,,$
is the required map.
\end{proof}

\subsection{BSP regime}

Finally, we study the resource theory under BSP operations where,
remarkably, we find a unique maximally GME state for any value of
$n$ and $d$, given by the generalised GHZ state 
\begin{equation}
|GHZ(n,d)\rangle=\frac{1}{\sqrt{d}}\sum_{i=1}^{d}|i\rangle^{\otimes n}.\label{ghznd}
\end{equation}

\begin{thm}
\label{bspth} In the resource theory of entanglement where the free
operations are BSP maps, there exists a maximally GME state on every
$\mathcal{H}$. Namely, $\forall|\psi\rangle\in(\mathbb{C}^{d})^{\otimes n}$,
there exists a CPTP BSP map $\Lambda$ such that $\Lambda(GHZ(n,d))=\psi$.
\end{thm}

\begin{proof}
The main idea behind the proof is to use the construction of the proof
of Theorems \ref{noSLOCC} and \ref{noisolation} again, which shows
that there is a CPTP BSP map $\Lambda$ that converts $GHZ(n,d)$
into $\psi$ if the robustness of $\psi$ is bounded above by an expression
involving the geometric measure of $GHZ(n,d)$. However, unlike for
the FS case, $G_{\mathcal{BS}}$ is straightforward to compute. Finally,
a simple estimate shows that $R_{\mathcal{BS}}(\psi)\leq d-1$ $\forall|\psi\rangle\in(\mathbb{C}^{d})^{\otimes n}$,
which leads to the desired result.

For every given $\ket{\psi}\in\left(\mathbb{C}^{d}\right)^{\otimes n}\,,$
let 
\begin{equation}
\Lambda(\eta)=\tr(\eta\,GHZ(n,d))\:\psi\,+\tr\left[(\mathbbm1-GHZ(n,d))\eta\right]\rho_{\mathcal{BS}}
\end{equation}
where $\rho_{\mathcal{BS}}\in\mathcal{BS}$ is the state which gives
the (biseparable) robustness of $\psi$ (i.e. $R_{\mathcal{BS}}(\psi)=R(\psi||\rho_{\mathcal{BS}})$).
Then, $\Lambda(GHZ(n,d))=\psi$ and it remains to be shown that $\Lambda$
is BSP. As argued in the proofs of Theorems \ref{noSLOCC} and \ref{noisolation},
this happens iff

\begin{equation}
R_{\mathcal{BS}}(\psi)\leq\frac{G_{\mathcal{BS}}(GHZ(n,d))}{1-G_{\mathcal{BS}}(GHZ(n,d))}\,.\label{eq:EgeomRobustnessBSP}
\end{equation}
However, unlike for the FS case, $G_{\mathcal{BS}}$ is straightforward
to compute \cite{biswas_genuine-multipartite-entanglement_2014} in
terms of the Schmidt decomposition across every possible bipartite
splitting of the parties $M|\overline{M}$ (i.e., $|\psi\rangle=\sum_{i}\sqrt{\lambda_{i}^{M|\overline{M}}}|i\rangle_{M}|i\rangle_{\overline{M}}$
for each $M$) as
\begin{equation}
G_{\mathcal{BS}}(\psi)=1-\max_{M\subsetneq[n]}\lambda_{1}^{M|\overline{M}},\label{gbs}
\end{equation}
where $\lambda_{1}^{M|\overline{M}}$ is the largest Schmidt coefficient
of $\psi$ in the corresponding splitting. This immediately shows
that the generalised GHZ state has maximal value of the geometric
measure, $G_{\mathcal{BS}}(GHZ(n,d))=(d-1)/d$. Therefore, $\Lambda$
is BSP iff $R_{\mathcal{BS}}(\psi)\leq d-1\,.$ It is shown in \cite{vidal_robustness_1999}
that for every bipartite pure state $\psi_{A|B}$ with Schmidt decomposition
\begin{equation}
\psi_{A|B}=\sum_{i}\sqrt{\lambda_{i}^{A|B}}\ket{i}_{A}\ket{i}_{B}
\end{equation}
it holds that 
\begin{equation}
R_{\mathcal{BS}}(\psi)=\left(\sum_{i}\sqrt{\lambda_{i}^{A|B}}\right)^{2}-1\,.
\end{equation}
Thus 
\begin{equation}
\begin{aligned}R_{\mathcal{BS}}(\psi) & \leq\min_{M\subsetneq[n]}\left(\sum_{i}\sqrt{\lambda_{i}^{M|\overline{M}}}\right)^{2}-1\\
 & \leq d-1,
\end{aligned}
\end{equation}
where the latter inequality follows from considering the state with
all eigenvalues $\lambda_{i}=1/d\,.$ Hence, $\forall\ket{\psi}\in\left(\mathbb{C}^{d}\right)^{\otimes n}$
there exists a BSP map $\Lambda$ such that $\Lambda(GHZ)=\psi\,.$
\end{proof}
It follows from the proof that it suffices to have maximal $G_{\mathcal{BS}}$
to be convertible to any other state by BSP operations. Thus, any
state fulfilling that $G_{\mathcal{BS}}=(d-1)/d$ must automatically
maximise any other BSP measure. More importantly, this also shows
that any two states achieving this value of the geometric measure
are deterministically interconvertible by BSP operations and, therefore,
belong to the same GME-equivalence class despite potentially not being
related by local unitary transformations. An example of such class
when $d=2$ are GME graph states for which it is known that $G_{\mathcal{BS}}=1/2$
\cite{toth_detecting_2005}. Hence, all graph states including the
generalised GHZ state are in the equivalence class of the maximally
GME state in this theory. It is remarkable to find that this very
relevant family of states \cite{hein_entanglement_2006} in quantum
computation and error correction has this feature in a resource theory
of GME and we believe this is worth further research. Another previously
considered family of states that belongs to this equivalence class
is that of absolutely maximally entangled (AME) states \cite{helwig_absolute_2012},
which is defined as those states for which all reduced density matrices
are proportional to the identity in the maximum possible dimensions.
It follows from equation\ (\ref{gbs}) that $G_{\mathcal{BS}}=(d-1)/d$
holds for all AME states (for those values of $n$ and $d$ for which
they exist). Equation (\ref{gbs}) also tells us that a necessary
condition for a state to be in the equivalence class of the maximally
GME state is that all single-particle reduced density matrices must
be proportional to the $d$-dimensional identity. However, this condition
is not sufficient: the state in $(\mathbb{C}^{2})^{\otimes4}$ $|\phi\rangle=\sqrt{p}|\phi^{+}\rangle_{12}|\phi^{+}\rangle_{34}+\sqrt{1-p}|\phi^{-}\rangle_{12}|\phi^{-}\rangle_{34}$
($|\phi^{\pm}\rangle=(|00\rangle\pm|11\rangle)/\sqrt{2}$) is a GME
state (if $p\neq0,1$) with this property but $G_{\mathcal{BS}}(\phi)<1/2$
(if $p\neq1/2$).

\section{Comparison between the regimes}

The fact that the set of FS states is (strictly) contained in the
set of BS states might lead us to believe that the set of FSP operations
is contained in the set of BSP operations. Evidently, if such an inclusion
were to hold, it would be strict. Indeed, a BSP operation such as
the one that transforms the GHZ state into, say, the W state, is manifestly
not FSP. Otherwise, the results in Section \ref{sec:FSP-regime} would
not hold.

We show now that the reverse inclusion does not hold either: there
exist FSP operations that are not BSP.
\begin{prop}
\label{prop:fsp-noninclusion}The set of FSP operations is not included
in the set of BSP operations.
\end{prop}

\begin{proof}
We provide an example of an FSP map that is not BSP. Consider the
$GHZ$-symmetric states, which were introduced in the proof of Lemma
\ref{lem:robGHZ}. Denoting them by $\rho$, they can be alternatively
parametrised \cite{eltschka_entanglement_2012} in terms of two parameters
$x,y$, where
\begin{equation}
\begin{aligned}x(\rho)= & \frac{1}{2}\left[\bra{GHZ_{+}}\rho\ket{GHZ_{+}}-\bra{GHZ_{-}}\rho\ket{GHZ_{-}}\right],\\
y(\rho)= & \frac{1}{\sqrt{3}}\left[\bra{GHZ_{+}}\rho\ket{GHZ_{+}}+\bra{GHZ_{-}}\rho\ket{GHZ_{-}}-\frac{1}{4}\right].
\end{aligned}
\end{equation}
We denote a $GHZ$-symmetric state $\rho$ with parameters $x,y$
as $\rho(x,y)$. This set forms a triangle on the $x,y$ plane, and
the sets $\mathcal{FS},\mathcal{BS}$ are well characterised in terms
of lines on the plane, which correspond to witnesses. In particular,
the witness which detects FS states is
\begin{equation}
\mathcal{W}_{FS}(x,y)=-x-\frac{\sqrt{3}}{6}y+\frac{1}{8}\geq0,
\end{equation}
while the one which detects BS states is
\begin{equation}
\mathcal{W}_{BS}(x,y)=-x-\frac{\sqrt{3}}{2}y+\frac{3}{8}\geq0.\label{eq:Wbs}
\end{equation}

The candidate for our proof is the map $\Lambda$ defined by
\begin{equation}
\Lambda(\eta)=\tr(W\eta)\rho\left(\frac{5}{16},\frac{\sqrt{3}}{4}\right)+\tr((\mathbbm1-W)\eta)\rho\left(-\frac{1}{8},0\right).
\end{equation}
Its possible outputs are
\begin{equation}
\sigma(p)=p\rho\left(\frac{5}{16},\frac{\sqrt{3}}{4}\right)+(1-p)\rho\left(-\frac{1}{8},0\right),
\end{equation}
where $p$ is the trace of the input state $X$ with $W$. Now, $\tr(WX)$
achieves a maximum of $4/9$ for FS states $X$, as implied by the
geometric measure of $W$ \cite{chen_computation_2010}. The state
that gives such a maximum is mapped under $\Lambda$ to 
\begin{equation}
\sigma\left(\frac{4}{9}\right)=\rho\left(\frac{5}{72},\frac{\sqrt{3}}{9}\right),
\end{equation}
which is FS since
\begin{equation}
\mathcal{W}_{FS}\left(\frac{5}{72},\frac{\sqrt{3}}{9}\right)=0.
\end{equation}
Therefore, $\Lambda$ is FSP.

However, if $X$ is BS, then $\tr(WX)$ can be as high as $2/3$ (take,
for example, $\ket{\psi}=(\ket{001}+\ket{010})/\sqrt{2}$). For this
input state, we have
\begin{equation}
\Lambda(\psi)=\sigma\left(\frac{2}{3}\right)=\rho\left(\frac{1}{6},\frac{\sqrt{3}}{6}\right).
\end{equation}
But its overlap with the GME witness in equation (\ref{eq:Wbs}) is
negative:
\begin{equation}
\mathcal{W}_{BS}\left(\frac{1}{6},\frac{\sqrt{3}}{6}\right)=-\frac{1}{24}<0,
\end{equation}
meaning that $\Lambda$ is not BSP.
\end{proof}

\section{Looking beyond}

While the resource theory of GME leads to a unique maximally entangled
state, the set of free states (i.e., biseparable states) is not closed
under tensor products. Indeed, as mentioned in Chapter \ref{chap:intro},
there are many examples of biseparable states which, when tensored,
lead to a GME state. This means that, if taking copies is allowed,
biseparable states can be a resource. In particular, bipartite entangled
states distributed in a network are potential resources to obtain
GME. This observation will be key in the following chapters, where
we will study which pair-entangled network states lead to GME and
GMNL (we will often identify networks with graphs, which should not
be confused with the graph states mentioned above). In fact, some
such networks are maximally GME in the resource theory where the free
operations are BSP operations:\textbf{\vspace{2mm}
}

\noindent \textbf{Observation. }\emph{Consider a regular graph such
that the number of edge-disjoint paths between every pair of vertices
is equal to the degree. Any such graph where each pair of parties
shares a maximally entangled state in dimension $d$ is maximally
GME in the sense of Theorem \ref{bspth}.}
\begin{proof}
Let $\rho$ denote the state in the statement, whose local dimension
is $d^{k}$ where $k$ is the degree of the graph. It is sufficient
to show that $G_{\mathcal{BS}}(\rho)=(d^{k}-1)/d^{k}$, i.e., that
the largest Schmidt coefficient in each bipartition is $1/d^{k}$.
To show this, it is enough to check that one can obtain a $d^{k}$-dimensional
maximally entangled state between any two parties by LOCC. But such
a maximally entangled state is equivalent to $k$ copies of a $d$-dimensional
maximally entangled state. Then, the LOCC protocol is as follows:
let $A,B$ be any two parties. The idea is to use each of the other
parties that are intermediate in the paths connecting $A$ and $B$,
who already share a maximally entangled state with $A$, as a bridge.
Each of these other parties can teleport the particle they hold to
$B$, using the maximally entangled state they share with $B$ as
a channel. This achieves the goal.
\end{proof}
Less ambitiously, if we only require GME, having a connected network
of entangled pure states is sufficient. \textbf{\vspace{2mm}
}

\noindent \textbf{Observation. }\emph{Any connected network of entangled
pure states is GME.}
\begin{proof}
The observation holds because the partial trace over any strict subset
of parties yields a mixed state: let
\begin{equation}
\rho=\bigotimes_{k=1}^{K}\psi_{k}
\end{equation}
be the state of the network, where $k=1,...,K$ labels the edges,
and the parties are left implicit. Then, taking the partial trace
over one party, say, $A_{i}$, amounts to taking the partial trace
of each $\psi_{k}$ that is incident to $A_{i}$, while the remaining
$\psi_{k}$ are left untouched:
\begin{equation}
\tr_{A_{i}}(\rho)=\tr_{A_{i}}\bigotimes_{k=1}^{K}\psi_{k}=\bigotimes_{k\in I}\tr_{A_{i}^{k}}\psi_{k}\otimes\bigotimes_{k\notin I}\psi_{k},
\end{equation}
where $I$ is the set of edges $k$ incident to party $A_{i}$, and
$A_{i}^{k}$ is the particle held by party $A_{i}$ corresponding
to edge $k$. For all $i,$ $k\in I$, the state $\tr_{A_{i}^{k}}\psi_{k}$
is mixed, since $\psi_{k}$ is entangled. Therefore, the state $\tr_{A_{i}}(\rho)$
is mixed. This does not change if the partial trace is taken over
more parties. Therefore, $\rho$ is not biseparable in any bipartition,
which is a sufficient condition for GME in the case of pure states.
\end{proof}
Networks of states like the ones just considered will be studied in
Chapters \ref{chap:gmnl} and \ref{chap:mixed}. As we shall see,
all connected networks of pure states are not only genuine multipartite entangled, but also genuine multipartite nonlocal,
while allowing for noise on the shared states can compromise even
their entanglement properties. However, like in the result above,
entanglement is stronger if the network is completely connected. Indeed,
we find that these networks are always GME if the noise in the shared
states is below a threshold.

%\bibliographystyle{amsalpha}
%\bibliography{entnonloc}
%
%\end{document}

%%% LyX 2.3.0rc1 created this file.  For more info, see http://www.lyx.org/.
%%% Do not edit unless you really know what you are doing.
%\documentclass[11pt,oneside,english]{book}
%\usepackage[T1]{fontenc}
%\usepackage[latin9]{inputenc}
%\usepackage{geometry}
%\geometry{verbose,tmargin=1in,bmargin=1in,lmargin=1in,rmargin=1in}
%\setcounter{secnumdepth}{3}
%\setcounter{tocdepth}{3}
%\usepackage{amsmath}
%\usepackage{amsthm}
%\usepackage{amssymb}
%\usepackage{graphicx}
%\usepackage{setspace}
%\onehalfspacing
%
%\makeatletter
%%%%%%%%%%%%%%%%%%%%%%%%%%%%%%% Textclass specific LaTeX commands.
%\theoremstyle{plain}
%\newtheorem{thm}{\protect\theoremname}
%\theoremstyle{plain}
%\newtheorem{prop}[thm]{\protect\propositionname}
%
%%%%%%%%%%%%%%%%%%%%%%%%%%%%%%% User specified LaTeX commands.
%\usepackage{physics}
%\usepackage{cite}
%\usepackage{bbm}
%\usepackage{enumitem}
%\usepackage{amsmath}
%\usepackage{amsfonts}
%\usepackage{amsthm}
%
%\renewcommand\theenumi{(\roman{enumi})}
%\renewcommand\labelenumi{\theenumi}
%\providecommand{\ketbra}[1]{\ket{#1}\bra{#1}}
%\providecommand{\trace}{\textnormal{tr}}
%\providecommand{\FSP}{\mathcal{FSP}}
%
%\newcommand{\ui}{\mathrm{i}}
%\newcommand{\ue}{\mathrm{e}}
%
%\makeatother
%
%\usepackage{babel}
%\providecommand{\propositionname}{Proposition}
%\providecommand{\theoremname}{Theorem}
%
%\begin{document}
%\setcounter{chapter}{2}

\chapter{\label{chap:gmnl}Pure pair-entangled network states}

\chaptermark{Pure pair-entangled networks}

Quantum entanglement and nonlocality are inextricably linked. However,
while entanglement is necessary for nonlocality, it is not always
sufficient in the standard Bell scenario. We derive sufficient conditions
for entanglement to give rise to genuine multipartite nonlocality
in networks. We find that any network where the parties are connected
by bipartite pure entangled states is genuine multipartite nonlocal,
independently of the amount of entanglement in the shared states and
of the topology of the network. As an application of this result,
we also show that all pure genuine multipartite entangled states are
genuine multipartite nonlocal in the sense that measurements can be
found on finitely many copies of any genuine multipartite entangled
state to yield a genuine multipartite nonlocal behaviour. Our results
pave the way towards feasible manners of generating genuine multipartite
nonlocality using any connected network.

\section{Definitions and preliminaries}

We consider distributions arising from GME states, and ask whether
they are bilocal (see Definitions \ref{def:intro-ent-multipart} and
\ref{def:intro-GMNL}). The set of bilocal distributions is a polytope:
indeed, the set of local distributions across each bipartition $M|\overline{M}$
is a polytope, and convex combinations preserve that structure. We
call this $n$-partite polytope $\mathcal{B}_{n}.$

In particular, we consider distributions arising from networks where
each party measures individually on each particle they hold. Therefore,
we reserve the usual notation for inputs and outputs, $x,y;a,b$ respectively,
for those corresponding to \emph{particles}, while we denote inputs
and outputs for each \emph{party} as $\chi,\upsilon;\alpha,\beta$
respectively.

We use results from Ref. \cite{pironio_lifting_2005} to lift inequalities
to account for more parties, inputs and outputs. They consider the
fully local polytope $\mathcal{L}$, which only includes distributions
\begin{equation}
P(\alpha\beta|\chi\upsilon)=\sum_{\lambda}p(\lambda)P_{A}(\alpha|\chi,\lambda)P_{B}(\beta|\upsilon,\lambda)\label{eq:FLdistr}
\end{equation}
where each party may have different numbers of inputs and outputs
(more parties may be considered by adding more distributions correlated
only by $\lambda$). Polytope $\mathcal{B}_{n}$ includes convex combinations
of distributions that are local across different bipartitions $M|\overline{M}$
of the parties, but the lifting results in \cite{pironio_lifting_2005}
still hold. Indeed, to check an inequality holds for a polytope, it
is sufficient by convexity to check the extremal points. As all extremal
points in $\mathcal{B}_{n}$ are contained in some polytope $\mathcal{L}$
(splitting the parties in two as per the bipartition $M|\overline{M}$),
lifting results for $\mathcal{L}$ can be straightforwardly extended
to $\mathcal{B}_{n}$.

We also use the EPR2 decomposition \cite{elitzur_quantum_1992}: any
bipartite distribution $P$ can be expressed (nonuniquely) as 
\begin{equation}
P(\alpha\beta|\chi\upsilon)=p_{L}P_{L}(\alpha\beta|\chi\upsilon)+(1-p_{L})P_{NS}(\alpha\beta|\chi\upsilon)\label{eq:EPR2-bipartite}
\end{equation}
for some $0\leq p_{L}\leq1,$ where $P_{L}$ is local (i.e. satisfies
equation (\ref{eq:FLdistr})) and $P_{NS}$ is nonsignalling (since
so is $P$). $P$ is nonlocal if all such decompositions have $p_{L}<1,$
and fully nonlocal\footnote{Not to be confused with ``nonfully local'', which is the opposite
of ``fully local''. ``fully nonlocal'' is a particular case of
``nonfully local''.} if all such decompositions have $p_{L}=0.$ A quantum state $\rho$
is fully nonlocal if, for all $\varepsilon>0,$ there exist local
measurements giving rise to a distribution $P$ such that any decomposition
(\ref{eq:EPR2-bipartite}) has $p_{L}<\varepsilon.$

The EPR2 decomposition can be extended to the multipartite case \cite{almeida_multipartite_2010}
as 
\begin{equation}
\begin{aligned}P(\alpha_{1}...\alpha_{n}|\chi_{1}...\chi_{n})= & \sum_{M\subsetneq[n]}p_{L}^{M}P_{L}^{M}(\alpha_{1}...\alpha_{n}|\chi_{1}...\chi_{n})\\
 & +p_{NS}P_{NS}(\alpha_{1}...\alpha_{n}|\chi_{1}...\chi_{n})
\end{aligned}
\label{eq:EPR2-multipartite}
\end{equation}
where $p_{L}^{M}\geq0$ for every $M,$ $p_{NS}\geq0$ and 
\begin{equation}
\sum_{M}p_{L}^{M}+p_{NS}=1,
\end{equation}
$P_{L}^{M}$ is local across the bipartition $M|\overline{M}$ (i.e.
satisfies equation (\ref{eq:FLdistr})), and $P_{NS}$ is nonsignalling.
We are interested in decompositions which maximise the local EPR2
components, in order to deduce properties about the distributions.
For a distribution $P,$ we define 
\begin{equation}
EPR2(P)=\max\left\{ \sum_{M}p_{L}^{M}:P=\sum_{M}p_{L}^{M}P_{L}^{M}+p_{NS}P_{NS},\sum_{M}p_{L}^{M}+p_{NS}=1\right\} 
\end{equation}
and, for a state $\rho,$ we define (with a slight abuse of notation)
\begin{equation}
EPR2(\rho)=\inf\left\{ EPR2(P):P=\tr\left(\bigotimes_{i=1}^{n}E_{\alpha_{i}|\chi_{i}}^{i}\rho\right)\right\} ,
\end{equation}
where the infimum is taken over local measurements $E_{\alpha_{i}|\chi_{i}}^{i}$
on each particle such that 
\begin{equation}
E_{\alpha_{i}|\chi_{i}}^{i}\succcurlyeq0\:\forall\alpha_{i},\chi_{i},\:\sum_{\alpha_{i}}E_{\alpha_{i}|\chi_{i}}^{i}=\mathbbm1\:\forall\chi_{i},\:\forall i\in[n],
\end{equation}
with any number of inputs and outputs. Then, a distribution $P$ or
a state $\rho$ are GMNL if $EPR2(\cdot)<1,$ while they are fully
GMNL if $EPR2(\cdot)=0.$ Notice that the optimisation for probability
distributions yields a maximum since the number of inputs and outputs
is fixed. Instead, the optimisation for a state may involve measurements
with an arbitrarily large number of inputs or outputs, as is the case
for the maximally entangled state \cite{barrett_maximally_2006}.
In this work, the number of inputs and outputs is always finite, and
this will become relevant when bounding the EPR2 components of distributions
arising from maximally entangled states in Theorems \ref{thm:gmnl-from-bipartite-ent}
and \ref{thm:copies}.

\section{GMNL from bipartite entanglement}

Our first result shows that any connected network of pure bipartite
entanglement (see Figure \ref{fig:GMNL-from-bipartite}) is GMNL.

\begin{figure}
\begin{centering}
\includegraphics[width=0.6\textwidth]{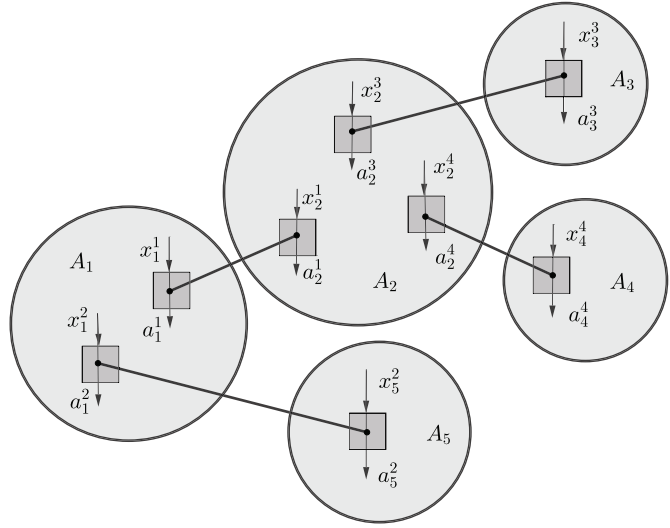}
\par\end{centering}
\caption{Connected network of bipartite entanglement. For each $i\in[n],$
party $A_{i}$ has input $x_{i}^{k}$ and output $a_{i}^{k}$ on the
particle at edge $k.$ Particles connected by an edge are entangled.
\label{fig:GMNL-from-bipartite}}
\end{figure}

\begin{thm}
\label{thm:gmnl-from-bipartite-ent} Any connected network of bipartite
pure entangled states is GMNL.
\end{thm}

We first establish the main ideas of the proof by outlining it for
a tripartite network, before turning to the general case. Since it
is sufficient to consider tree graphs, i.e., graphs without cycles,
we consider a Lambda network where $A_{1}$ is entangled to each of
$A_{2}$ and $A_{3}$.
\begin{proof}[Proof for the tripartite case]
Since it turns out to be sufficient to measure individually on each
party's different particles (see Figure \ref{fig:GMNL-from-bipartite}
for the $n$-partite structure), the shared distribution \sloppy
$P(a_{1}^{1}a_{1}^{2},a_{2}^{1},a_{3}^{2}|x_{1}^{1}x_{1}^{2},x_{2}^{1},x_{3}^{2})$
takes the form 
\begin{equation}
P_{1}(a_{1}^{1}a_{2}^{1}|x_{1}^{1}x_{2}^{1})P_{2}(a_{1}^{2}a_{3}^{2}|x_{1}^{2}x_{3}^{2}),\label{eq:P3partite}
\end{equation}
where parties $A_{i},A_{j}$ are connected by edge $k$ (we label
vertices and edges independently), and $P_{k}(a_{i}^{k}a_{j}^{k}|x_{i}^{k}x_{j}^{k})$
is the distribution arising from the state at edge $k.$

We consider three cases, depending on whether none, one or both of
the shared states are maximally entangled. If none are, we devise
inequalities to detect bipartite nonlocality at each edge of the network,
and combine them to form a multipartite inequality. Then, we find
measurements on the shared states to violate it. If both states are
maximally entangled, existing results show the network is fully GMNL
\cite{cavalcanti_quantum_2011,almeida_multipartite_2010}. Combining
these two cases for a heterogeneous network completes the proof.

To prove the first case, we take bipartite inequalities between $A_{1}$
and each other party, lift them to three parties and combine them
using Refs. \cite{pironio_lifting_2005,curchod_versatile_2019}, to
obtain the following GMNL inequality: 
\begin{equation}
\begin{aligned}I_{3}= & I^{1}+I^{2}+P(00,0,0|00,0,0)\\
 & -\sum_{a_{1}^{2}=0,1}P(0a_{1}^{2},0,0|00,0,0)-\sum_{a_{1}^{1}=0,1}P(a_{1}^{1}0,0,0|00,0,0)\leq0.
\end{aligned}
\label{eq:Iabc44}
\end{equation}
Here, 
\begin{equation}
\begin{aligned}I^{1}=\sum_{a_{1}^{2}=0,1} & \left[P(0a_{1}^{2},0,0|00,0,0)-P(0a_{1}^{2},1,0|00,1,0)\right.\\
- & \left.P(1a_{1}^{2},0,0|10,0,0)-P(0a_{1}^{2},0,0|10,1,0)\right]\leq0;
\end{aligned}
\label{eq:IAB44}
\end{equation}
\begin{equation}
\begin{aligned}I^{2}=\sum_{a_{1}^{1}=0,1} & \left[P(a_{1}^{1}0,0,0|00,0,0)-P(a_{1}^{1}0,0,1|00,0,1)\right.\\
- & \left.P(a_{1}^{1}1,0,0|01,0,0)-P(a_{1}^{1}0,0,0|01,0,1)\right]\leq0
\end{aligned}
\label{eq:IAC44}
\end{equation}
are liftings of 
\begin{equation}
I=P(00|00)-P(01|01)-P(10|10)-P(00|11)\leq0\label{eq:CHSHseed}
\end{equation}
to three parties with $A_{1}$ having 4 inputs and 4 outputs. Inequality
(\ref{eq:CHSHseed}) is equivalent to the CHSH inequality \cite{clauser_proposed_1969}
for nonsignalling distributions \cite{curchod_versatile_2019}. Thus,
inequalities (\ref{eq:IAB44}), (\ref{eq:IAC44}) are satisfied by
distributions that are local across $A_{1}|A_{2}$ and $A_{1}|A_{3}$
respectively. To see that equation (\ref{eq:Iabc44}) is a GMNL inequality
it is sufficient to check it holds for distributions that are local
across some bipartition. This is straightforwardly done by observing
the cancellations that occur when $I^{1}$ or $I^{2}$ are $\leq0.$

Since both states are less-than-maximally entangled, $A_{1}$ can
satisfy Hardy's paradox \cite{hardy_quantum_1992,hardy_nonlocality_1993}
with each other party, achieving 
\begin{equation}
P_{k}(00|00)>0=P_{k}(01|01)=P_{k}(10|10)=P_{k}(00|11)
\end{equation}
for both $k$ (the proof for qubits in Refs. \cite{hardy_quantum_1992,hardy_nonlocality_1993}
is extended to qudits by measuring on a two-dimensional subspace,
see Proposition \ref{prop:d-dimHardy}). Then, each negative term
in $I^{1}$ and $I^{2}$ is zero, as 
\begin{equation}
\begin{aligned}\sum_{a_{1}^{2}=0,1}P(0a_{1}^{2},1,0|00,1,0) & =P_{1}(01|01)\sum_{a_{1}^{2}=0,1}P_{2}(a_{1}^{2}0|00)\end{aligned}
\end{equation}
and similarly for the others. Hence, only 
\begin{equation}
P(00,0,0|00,0,0)=P_{1}(00|00)P_{2}(00|00)>0
\end{equation}
survives, violating the inequality.

If, instead, $A_{1}A_{2}$ share a maximally entangled state, and
$A_{2}A_{3}$ share a less-than-maximally entangled state, then $A_{1}A_{3}$
can measure so that $P_{2}$ satisfies Hardy's paradox; hence $\exists\,\varepsilon>0$
such that its local component in any EPR2 decomposition satisfies
\begin{equation}
p_{L,2}\leq1-\varepsilon.\label{eq:plocal-Hardy}
\end{equation}
Since the maximally entangled state is fully nonlocal \cite{barrett_maximally_2006},
for this $\varepsilon,$ $A_{1}A_{2}$ can measure such that any EPR2
decomposition of $P_{1}$ satisfies 
\begin{equation}
p_{L,1}<\varepsilon.\label{eq:plocal-maxent}
\end{equation}

Then, we assume for a contradiction that $P(a_{1}^{1}a_{1}^{2},a_{2}^{1},a_{3}^{2}|x_{1}^{1}x_{1}^{2},x_{2}^{1},x_{3}^{2})$
is bilocal and decompose it in its bipartite splittings, 
\begin{equation}
\begin{aligned}P(a_{1}^{1}a_{1}^{2}, & a_{2}^{1},a_{3}^{2}|x_{1}^{1}x_{1}^{2},x_{2}^{1},x_{3}^{2})\\
=\sum_{\lambda} & \left[p_{L}(\lambda)P_{A_{1}A_{2}}(a_{1}^{1}a_{1}^{2},a_{2}^{1}|x_{1}^{1}x_{1}^{2},x_{2}^{1},\lambda)P_{A_{3}}(a_{3}^{2}|x_{3}^{2},\lambda)\right.\\
 & +q_{L}(\lambda)P_{A_{1}A_{3}}(a_{1}^{1}a_{1}^{2},a_{3}^{2}|x_{1}^{1}x_{1}^{2},x_{3}^{2},\lambda)P_{A_{2}}(a_{2}^{1}|x_{2}^{1},\lambda)\\
 & \left.+r_{L}(\lambda)P_{A_{1}}(a_{1}^{1}a_{1}^{2}|x_{1}^{1}x_{1}^{2},\lambda)P_{A_{2}A_{3}}(a_{2}^{1},a_{3}^{2}|x_{2}^{1},x_{3}^{2},\lambda)\right]
\end{aligned}
\label{eq:PmixBL}
\end{equation}
where $\sum_{\lambda}\left[p_{L}(\lambda)+q_{L}(\lambda)+r_{L}(\lambda)\right]=1.$

Summing equation (\ref{eq:PmixBL}) over $a_{1}^{2},a_{3}^{2}$ and
using equation (\ref{eq:P3partite}), we get an EPR2 decomposition
of $P_{1}$ with local components $q_{L},r_{L}.$ By equation (\ref{eq:plocal-maxent}),
this entails $\sum_{\lambda}\left[q_{L}(\lambda)+r_{L}(\lambda)\right]<\varepsilon,$
so 
\begin{equation}
\sum_{\lambda}p_{L}(\lambda)>1-\varepsilon.
\end{equation}

Summing, instead, equation (\ref{eq:PmixBL}) over $a_{1}^{1},a_{2}^{1},$
we obtain an EPR2 decomposition of $P_{2}$ whose only nonnegligible
component, $\sum_{\lambda}p_{L}(\lambda),$ is local in $A_{1}|A_{3},$
contradicting equation (\ref{eq:plocal-Hardy}). Therefore, $P$ must
be GMNL.
\end{proof}
\begin{proof}[Proof of Theorem \ref{thm:gmnl-from-bipartite-ent}]
Turning to the fully general case, we consider the network as a connected
graph where vertices are parties and edges are states. The graph is
such that, at each vertex, there is one particle for every incident
edge.\footnote{Throughout the proof we assume $k\geq2.$ If $k=1,$ there are only
two parties sharing bipartite pure entangled states, so the network
is nonlocal by Refs. \cite{gisin_bells_1991,gisin_maximal_1992}.} We label the edges as $k=1,...,K$ (where $K$ is the number of edges
of the graph) and the parties as $A_{1},...,A_{n}.$ Since it will
be enough to consider individual measurements on each particle, we
denote the input and output of party $A_{i}$ at edge $k$ as $x_{i}^{k},a_{i}^{k}$
respectively. We group the inputs and outputs of each party as $\chi_{i}=\{x_{i}^{k}\}_{k\in E_{i}},$
$\alpha_{i}=\{a_{i}^{k}\}_{k\in E_{i}}$ where $E_{i}$ is the set
of edges incident to vertex $i.$ Then, the shared distribution is
of the form 
\begin{equation}
P(\alpha_{1},...,\alpha_{n}|\chi_{1},...,\chi_{n})=\prod_{k=1}^{K}P_{k}(a_{i}^{k}a_{j}^{k}|x_{i}^{k}x_{j}^{k}),\label{eq:Pnpartite}
\end{equation}
where parties $A_{i},A_{j}$ are connected by edge $k$ (notice that
we label vertices and edges independently), and $P_{k}(a_{i}^{k}a_{j}^{k}|x_{i}^{k}x_{j}^{k})$
is the distribution arising from the state at edge $k.$ It will be
sufficient to consider tree graphs, i.e., graphs such that every pair
of vertices (parties) is connected by exactly one path of edges. If
the given graph is not a tree, any extra edges can be ignored.

Depending on the nature of the shared states, we consider three cases: 
\begin{enumerate}
\item every shared state is less-than-maximally entangled;\label{enu:Hardy} 
\item every shared state is maximally entangled;\label{enu:maxent} 
\item some shared states are maximally entangled, some are not.\label{enu:mix} 
\end{enumerate}
\textbf{Case \ref{enu:Hardy}: }if all states are less-than-maximally
entangled, we prove the result by deriving an inequality that detects
GMNL and finding measurements on the shared states to violate it.
To derive the inequality, we will find bipartite inequalities that
can be violated by the state at each edge $k$, lift them to more
inputs, outputs and parties using the techniques in Ref. \cite{pironio_lifting_2005}
and combine them to obtain a GMNL inequality using tools in Ref. \cite{curchod_versatile_2019}.
We will consider 2-input 2-output measurements on each particle. Thus,
the global distribution will have $2^{|E_{i}|}$ inputs and outputs
for each party $A_{i}.$

We start from the inequality 
\begin{equation}
I=P(00|00)-P(01|01)-P(10|10)-P(00|11)\leq0,\label{eq:Iorig}
\end{equation}
which is a facet inequality equivalent to the CHSH inequality \cite{clauser_proposed_1969}
for nonsignalling distributions \cite{curchod_versatile_2019}. This
inequality detects any bipartite nonlocality present in any bipartition
that splits the parties connected by edge $k$ \cite{curchod_versatile_2019}.
To lift it to $n$ parties, each with $2^{|E_{i}|}$ inputs and outputs
(see Ref. \cite{pironio_lifting_2005}), we must set the inputs and
outputs of the parties that are not connected by edge $k$ to a fixed
value (0, wlog). For the parties $i$ that are connected by edge $k,$
any extra inputs other than $x_{i}^{k}=0_{i}^{k},1_{i}^{k}$ can be
ignored. Outputs must be grouped, by summing over some of their digits,
in order to get an effective 2-output distribution. It will be convenient
to add over the output components $a_{i}^{\bar{k}}$ that do not correspond
to edge $k,$ varying only the digit $a_{i}^{k}=0_{i}^{k},1_{i}^{k}.$
Thus, we obtain the following $n$-partite inequality at each edge
$k:$ 
\begin{equation}
\begin{aligned}I^{k}=\sum_{\overrightarrow{a}_{i}^{\bar{k}},\overrightarrow{a}_{j}^{\bar{k}}} & \left[P\left(\left.0_{i}^{k}\overrightarrow{a}_{i}^{\bar{k}},0_{j}^{k}\overrightarrow{a}_{j}^{\bar{k}},\overrightarrow{0}_{\bar{i},\bar{j}}\right|0_{i}^{k}0_{i}^{\bar{k}},0_{j}^{k}0_{j}^{\bar{k}},\overrightarrow{0}_{\bar{i},\bar{j}}\right)\right.\\
 & -P\left(\left.0_{i}^{k}\overrightarrow{a}_{i}^{\bar{k}},1_{j}^{k}\overrightarrow{a}_{j}^{\bar{k}},\overrightarrow{0}_{\bar{i},\bar{j}}\right|0_{i}^{k}0_{i}^{\bar{k}},1_{j}^{k}0_{j}^{\bar{k}},\overrightarrow{0}_{\bar{i},\bar{j}}\right)\\
 & -P\left(\left.1_{i}^{k}\overrightarrow{a}_{i}^{\bar{k}},0_{j}^{k}\overrightarrow{a}_{j}^{\bar{k}},\overrightarrow{0}_{\bar{i},\bar{j}}\right|1_{i}^{k}0_{i}^{\bar{k}},0_{j}^{k}0_{j}^{\bar{k}},\overrightarrow{0}_{\bar{i},\bar{j}}\right)\\
 & \left.-P\left(\left.0_{i}^{k}\overrightarrow{a}_{i}^{\bar{k}},0_{j}^{k}\overrightarrow{a}_{j}^{\bar{k}},\overrightarrow{0}_{\bar{i},\bar{j}}\right|1_{i}^{k}0_{i}^{\bar{k}},1_{j}^{k}0_{j}^{\bar{k}},\overrightarrow{0}_{\bar{i},\bar{j}}\right)\right]\leq0,
\end{aligned}
\label{eq:Ik}
\end{equation}
where the sum is over each binary digit $a_{i}^{\bar{k}},a_{j}^{\bar{k}}$
of the outputs of parties $i,j$ (which are connected by edge $k$),
except digits $a_{i}^{k},a_{j}^{k}$ which are fixed to 0 or 1 in
each term. The term $\overrightarrow{0}_{\bar{i},\bar{j}}$ denotes
input or output 0 for all components of all parties that are not $i,j.$
Thus, each inequality $I^{k}$ detects the bipartite nonlocality present
in the distribution $P$ across any bipartition that splits the parties
connected by edge $k.$ In the particular case of the distribution
(\ref{eq:Pnpartite}), it tells whether the component $P_{k}$ is
nonlocal.

Now, we can combine the inequalities $I^{k}$ to form a GMNL inequality:
\begin{equation}
I_{n}=\sum_{k=1}^{K}I^{k}+P(\overrightarrow{0},\overrightarrow{0}|\overrightarrow{0},\overrightarrow{0})-\sum_{k=1}^{K}\sum_{\overrightarrow{a}_{i}^{\bar{k}},\overrightarrow{a}_{j}^{\bar{k}}}P\left(\left.0_{i}^{k}\overrightarrow{a}_{i}^{\bar{k}},0_{j}^{k}\overrightarrow{a}_{j}^{\bar{k}},\overrightarrow{0}_{\bar{i},\bar{j}}\right|0_{i}^{k}0_{i}^{\bar{k}},0_{j}^{k}0_{j}^{\bar{k}},\overrightarrow{0}_{\bar{i},\bar{j}}\right)\leq0\,.\label{eq:InGMNL}
\end{equation}
To show that this is indeed a GMNL inequality, we must show that it
holds for any distribution $P$ that is local across some bipartition.
A bipartition of the network defines a cut of the graph. Because the
graph is assumed connected, for every cut there exists an edge $k_{0}$
which crosses the cut. Therefore, if $P$ is local across a bipartition
which is crossed by edge $k_{0},$ then by Ref. \cite{curchod_versatile_2019}
we have 
\begin{equation}
I^{k_{0}}\leq0.
\end{equation}
Hence, 
\begin{equation}
\begin{aligned}I_{n}\leq & \sum_{\substack{k=1\\
k\neq k_{0}
}
}^{K}I^{k}+P(\overrightarrow{0},\overrightarrow{0}|\overrightarrow{0},\overrightarrow{0})-\sum_{k=1}^{K}\sum_{\overrightarrow{a}_{i}^{\bar{k}},\overrightarrow{a}_{j}^{\bar{k}}}P\left(\left.0_{i}^{k}\overrightarrow{a}_{i}^{\bar{k}},0_{j}^{k}\overrightarrow{a}_{j}^{\bar{k}},\overrightarrow{0}_{\bar{i},\bar{j}}\right|0_{i}^{k}0_{i}^{\bar{k}},0_{j}^{k}0_{j}^{\bar{k}},\overrightarrow{0}_{\bar{i},\bar{j}}\right).\end{aligned}
\end{equation}

For each $k\neq k_{0},$ the only nonnegative term gets subtracted
in the final summation. The term $P(\overrightarrow{0},\overrightarrow{0}|\overrightarrow{0},\overrightarrow{0})$
then cancels out with the first term in the final summation for $k=k_{0},$
leaving only negative terms in the expression as required.

To complete the proof, we find local measurements for each party to
violate inequality (\ref{eq:InGMNL}). Since all shared states are
nonseparable and less-than-maximally entangled, the parties can choose
local measurements on each particle such that all resulting distributions
satisfy Hardy's paradox \cite{hardy_quantum_1992,hardy_nonlocality_1993}:
\begin{equation}
P_{k}(00|00)>0=P_{k}(01|01)=P_{k}(10|10)=P_{k}(00|11)\label{eq:Hardy}
\end{equation}
for each $k=1,...,K.$ This was proven for qubits in Refs. \cite{hardy_quantum_1992,hardy_nonlocality_1993},
and we show the extension to any local dimension in Proposition \ref{prop:d-dimHardy}
below. Because the distribution is of the form (\ref{eq:Pnpartite}),
each term in each inequality (\ref{eq:Ik}) simplifies significantly.
For example, the second term gives 
\begin{equation}
\begin{aligned}\sum_{\overrightarrow{a}_{i}^{\bar{k}},\overrightarrow{a}_{j}^{\bar{k}}} & P\left(\left.0_{i}^{k}\overrightarrow{a}_{i}^{\bar{k}},1_{j}^{k}\overrightarrow{a}_{j}^{\bar{k}},\overrightarrow{0}_{\bar{i},\bar{j}}\right|0_{i}^{k}0_{i}^{\bar{k}},1_{j}^{k}0_{j}^{\bar{k}},\overrightarrow{0}_{\bar{i},\bar{j}}\right)\\
 & =P_{k}(0_{i}^{k}1_{j}^{k}|0_{i}^{k}1_{j}^{k})\prod_{\ell}\sum_{a_{i}^{\ell}}P_{\ell}(a_{i}^{\ell}0_{j'}^{\ell}|0_{i}^{\ell}0_{j'}^{\ell})\prod_{\ell'}\sum_{a_{j}^{\ell'}}P_{\ell}(0_{i'}^{\ell'}a_{j}^{\ell'}|0_{i'}^{\ell'}0_{j}^{\ell'})\prod_{m}P_{m}(0_{i'}^{m}0_{j'}^{m}|0_{i'}^{m}0_{j'}^{m})\\
 & =P_{k}(0_{i}^{k}1_{j}^{k}|0_{i}^{k}1_{j}^{k})\,p_{k}\,,
\end{aligned}
\label{eq:summation1}
\end{equation}
where edges $\ell$ connect party $i$ to party $j'\neq j,$, edges
$\ell'$ connect party $j$ to party $i'\neq i,$ and edges $m$ connect
parties $i'$ and $j'$ where $i',j'\neq i,j.$ (Depending on the
structure of the graph, there may be no edges $\ell$, $\ell'$ or
$m$ for a given pair of parties $i,j,$ but that does not affect
the proof.)

The product of the terms $P_{\ell},\,P_{\ell'}$ and $P_{m}$ will
give a number $p_{k}.$ This is similar for the third and fourth terms,
which factorise to 
\begin{equation}
\begin{aligned}P_{k}(1_{i}^{k}0_{j}^{k}|1_{i}^{k}0_{j}^{k})\, & p_{k},\\
P_{k}(0_{i}^{k}0_{j}^{k}|1_{i}^{k}1_{j}^{k})\, & p_{k}
\end{aligned}
\label{eq:summation2}
\end{equation}
respectively. The first term of each $I^{k}$ cancels out with the
last summation in $I_{n},$ and the only term that remains is 
\begin{equation}
P(\overrightarrow{0},\overrightarrow{0}|\overrightarrow{0},\overrightarrow{0})=\prod_{k=1}^{K}P_{k}(0_{i}^{k}0_{j}^{k}|0_{i}^{k}0_{j}^{k}).
\end{equation}
Since $P_{k}$ satisfies Hardy's paradox for every $k$, then the
components of each $P_{k}$ appearing in equations (\ref{eq:summation1}),
(\ref{eq:summation2}) are all zero, while the only surviving term,
$P(\overrightarrow{0},\overrightarrow{0}|\overrightarrow{0},\overrightarrow{0}),$
is strictly greater than zero. Thus, the inequality $I_{n}$ is violated,
showing that $P$ is GMNL.

\textbf{Case \ref{enu:maxent}:} for every bipartition, there is an
edge that crosses the corresponding cut, and each of these edges already
contains a maximally entangled state. Therefore, the present network
meets the requirements of Theorem 2 in \cite{almeida_multipartite_2010},
so the network is GMNL---in fact it is fully GMNL.

\textbf{Case \ref{enu:mix}: }assume wlog that each edge $k=1,...,K_{0}$
contains a less-than-maximally entangled state, while each edge $k=K_{0}+1,...,K$
contains a maximally entangled state. Let 
\begin{equation}
P=P_{H}P_{+}\label{eq:PHP+}
\end{equation}
where 
\begin{equation}
\begin{aligned}P_{H}(\{a_{i}^{k}\}_{k\leq K_{0},i\in[n]}|\{x_{i}^{k}\}_{k\leq K_{0},i\in[n]}) & =\prod_{k=1}^{K_{0}}P_{k}(a_{i}^{k}a_{j}^{k}|x_{i}^{k}x_{j}^{k}),\\
P_{+}(\{a_{i}^{k}\}_{k>K_{0},i\in[n]}|\{x_{i}^{k}\}_{k>K_{0},i\in[n]}) & =\prod_{k=K_{0}+1}^{K}P_{k}(a_{i}^{k}a_{j}^{k}|x_{i}^{k}x_{j}^{k})
\end{aligned}
\end{equation}
where, on the right-hand side, parties $i,j$ are connected by edge
$k.$ For $k=1,...,K_{0},$ terms $P_{k}$ satisfy Hardy's paradox
(equation (\ref{eq:Hardy})), as they arise from the measurements
performed in Case \ref{enu:Hardy}. For $k=K_{0}+1,...,K,$ the terms
$P_{k}$ arise from measurements on the maximally entangled state
to be specified later. We now classify bipartitions depending on whether
or not they are crossed by an edge $k\leq K_{0}$ or $k>K_{0}:$ let
$S_{\leq K_{0}}$ be the set of bipartitions $M|\overline{M}$ (indexed
by $M$) which are crossed by an edge $k\leq K_{0},$ and $T_{\leq K_{0}}$
be its complement, i.e. the set of bipartitions which are \emph{not}
crossed by an edge $k\leq K_{0}.$ Similarly, $S_{>K_{0}}$ (respectively,
$T_{>K_{0}}$) is the set of bipartitions which are (not) crossed
by an edge $k>K_{0}.$

Let $I_{H}^{k}$ be an inequality detecting nonlocality on edge $k,$
for the distribution $P_{H}.$ That is, $I_{H}^{k}$ is as in equation
(\ref{eq:Ik}) but where the sum over $\overrightarrow{a}_{i}^{\bar{k}},\overrightarrow{a}_{j}^{\bar{k}}$
concerns only the components of parties $A_{i},A_{j}$ that belong
only to edges $k'\leq K_{0},k'\neq k.$ Then, consider the following
functional acting on distributions of the form of $P_{H}:$ 
\begin{equation}
I_{H}=\sum_{k=1}^{K_{0}}I_{H}^{k}+P(\overrightarrow{0},\overrightarrow{0}|\overrightarrow{0},\overrightarrow{0})-\sum_{k=1}^{K_{0}}\sum_{\overrightarrow{a}_{i}^{\bar{k}},\overrightarrow{a}_{j}^{\bar{k}}}P\left(\left.0_{i}^{k}\overrightarrow{a}_{i}^{\bar{k}},0_{j}^{k}\overrightarrow{a}_{j}^{\bar{k}},\overrightarrow{0}_{\bar{i},\bar{j}}\right|0_{i}^{k}0_{i}^{\bar{k}},0_{j}^{k}0_{j}^{\bar{k}},\overrightarrow{0}_{\bar{i},\bar{j}}\right).
\end{equation}
Again, the summation in the last term concerns only components that
belong to edges $k'\leq K_{0},k'\neq k.$ We claim that the functional
$I_{H}$ is non-positive for any distribution $P$ that is local across
a bipartition of type $S_{\leq K_{0}},$ i.e. one that is crossed
by an edge $k_{0}\leq K_{0}.$ The reasoning is similar to that in
Case \ref{enu:Hardy}: if $P$ is local across a bipartition crossed
by an edge $k_{0}\leq K_{0},$ then $I_{H}^{k_{0}}\leq0$ will be
satisfied, and so 
\begin{equation}
I_{H}\leq\sum_{\substack{k=1\\
k\neq k_{0}
}
}^{K_{0}}I_{H}^{k}+P(\overrightarrow{0},\overrightarrow{0}|\overrightarrow{0},\overrightarrow{0})-\sum_{k=1}^{K_{0}}\sum_{\overrightarrow{a}_{i}^{\bar{k}},\overrightarrow{a}_{j}^{\bar{k}}}P\left(\left.0_{i}^{k}\overrightarrow{a}_{i}^{\bar{k}},0_{j}^{k}\overrightarrow{a}_{j}^{\bar{k}},\overrightarrow{0}_{\bar{i},\bar{j}}\right|0_{i}^{k}0_{i}^{\bar{k}},0_{j}^{k}0_{j}^{\bar{k}},\overrightarrow{0}_{\bar{i},\bar{j}}\right).
\end{equation}
Now, for each $k\neq k_{0},$ the only nonnegative term gets subtracted
in the final summation. The term $P(\overrightarrow{0},\overrightarrow{0}|\overrightarrow{0},\overrightarrow{0})$
then cancels out with the first term in the final summation for $k=k_{0},$
leaving only negative terms in the expression as required.

We now show that, for $P=P_{H},$ we have $I_{H}>0.$ Indeed, the
terms in $I_{H}^{k}$ simplify in a similar manner to Case \ref{enu:Hardy}.
Then, since each $P_{k},$ $k\leq K_{0}$ satisfies Hardy's paradox,
the second, third and fourth terms in each $I_{H}^{k}$ are zero,
the first cancels out with the last summation, and the only surviving
term is 
\begin{equation}
P(\overrightarrow{0},\overrightarrow{0}|\overrightarrow{0},\overrightarrow{0})=\prod_{k=1}^{K_{0}}P_{k}(0_{i}^{k}0_{j}^{k}|0_{i}^{k}0_{j}^{k})>0.
\end{equation}

This means that there exists an $\varepsilon>0$ such that, for any
EPR2 decomposition of $P_{H},$ 
\begin{equation}
P_{H}=\sum_{M}p_{L,H}^{M}P_{L,H}^{M}+p_{NS,H}P_{NS,H},
\end{equation}
we have that the terms where $P_{L,H}^{M}$ is local across a bipartition
such that $M\in S_{\leq K_{0}}$ satisfy 
\begin{equation}
\sum_{M\in S_{\le K_{0}}}p_{L,H}^{M}\leq1-\varepsilon.\label{eq:plocal-Hardy-1}
\end{equation}

Also, it can be deduced from Ref. \cite{almeida_multipartite_2010}
that, given the $\varepsilon$ above, the parties can choose suitable
measurements such that $P_{+}$ is fully nonlocal across all bipartitions
$S_{>K_{0}}.$ That is, any multipartite EPR2 decomposition of $P_{+},$
\begin{equation}
P_{+}=\sum_{M}p_{L,+}^{M}P_{L,+}^{M}+p_{NS,+}P_{NS,+},
\end{equation}
is such that the terms where $P_{L,+}^{M}$ is local across a bipartition
such that $M\in S_{>K_{0}}$ satisfy 
\begin{equation}
\sum_{M\in S_{>K_{0}}}p_{L,+}^{M}<\varepsilon.\label{eq:plocal-maxent-1}
\end{equation}

To prove that the global distribution $P$ is GMNL, as is our goal,
we assume the converse, and we derive a contradiction from the nonlocality
properties of $P_{H}$ and $P_{+}.$ Assuming $P$ is bilocal, we
can express the distribution as 
\begin{equation}
P=\sum_{\lambda,M}p_{L}^{M}(\lambda)P_{M}(\{\alpha_{i}\}_{i\in M}|\{\chi_{i}\}_{i\in M},\lambda)P_{\overline{M}}(\{\alpha_{i}\}_{i\in\overline{M}}|\{\chi_{i}\}_{i\in\overline{M}},\lambda),\label{eq:PnotGMNL}
\end{equation}
where $p_{L}^{M}(\lambda)$ are nonegative numbers for every $M,\lambda$
such that 
\begin{equation}
\sum_{\lambda,M}p_{L}^{M}(\lambda)=1,
\end{equation}
for each $\alpha_{i},\chi_{i},i=1,...,n.$

Then, summing over the output components $a_{i}^{k}$ for all $k\leq K_{0}$
and all $i,$ we get $P_{+}$ on the left-hand side, from equation
(\ref{eq:PHP+}). On the right-hand side, we get two types of terms
(depending on the type of bipartition) that turn out to form an EPR2
decomposition of $P_{+}$.\footnote{Note that, while each of the terms on the right-hand side may depend
on the whole of each party's input $\chi_{i},$ the left-hand side
does not, because the distribution is of the form (\ref{eq:PHP+}).
That is, the resulting EPR2 decomposition of $P_{+}$ holds for any
fixed value of the inputs $\{x_{i}^{k}\}_{k\leq K}$ on the left-
and right-hand sides.} Indeed, the local terms are given by bipartitions such that $M\in S_{>K_{0}},$
while the nonlocal terms are given by bipartitions such that $M\in T_{>K_{0}}$
(since all terms are nonsignalling). By equation (\ref{eq:plocal-maxent-1}),
the choice of measurements on the particles involved in $P_{+}$ ensures
that 
\begin{equation}
\sum_{\lambda,M\in S_{>K_{0}}}p_{L}^{M}(\lambda)<\varepsilon,
\end{equation}
while 
\begin{equation}
\sum_{\lambda,M\in T_{>K_{0}}}p_{L}^{M}(\lambda)>1-\varepsilon.\label{eq:sumI1-1}
\end{equation}

If, instead, we sum over the output components $a_{i}^{k}$ for all
$k>K_{0}$ and all $i,$ we get $P_{H}$ on the left-hand side, from
equation (\ref{eq:PHP+}). On the right-hand side, by similar reasoning
we find an EPR2 decomposition of $P_{H}.$ This time, $S_{\leq K_{0}}$
will give the local terms and $T_{\leq K_{0}}$ will give the nonlocal
terms. By equation (\ref{eq:plocal-Hardy-1}), we have 
\begin{equation}
\sum_{\lambda,M\in S_{\leq K_{0}}}p_{L}^{M}(\lambda)\leq1-\varepsilon.\label{eq:sumJ2}
\end{equation}
Now, since the graph is connected, if a bipartition is not crossed
by an edge $k>K_{0},$ then it must be crossed by an edge $k\leq K_{0}.$
That is, $T_{>K_{0}}\subseteq S_{\leq K_{0}}.$ This means that equation
(\ref{eq:sumJ2}) also holds if the sum is over $T_{>K_{0}},$ but
this contradicts equation (\ref{eq:sumI1-1}). Therefore, the distribution
$P$ must be GMNL.
\end{proof}
In Theorem \ref{thm:gmnl-from-bipartite-ent} we assumed that all
less-than-maximally entangled states satisfy Hardy's paradox. This
is shown for qubits in \cite{hardy_nonlocality_1993}, and we now
extend the proof to any dimension.
\begin{prop}
\label{prop:d-dimHardy}Let $\ket{\psi}\in\mathcal{H}_{A}\otimes\mathcal{H}_{B}\cong\left(\mathbb{C}^{d}\right)^{\otimes2}$
be a nonseparable and less-than-maximally entangled pure state. Then,
$\ket{\psi}$ satisfies Hardy's paradox.
\end{prop}

\begin{proof}
Let $\ket{\psi}$ be as in the statement of the Proposition. We present
2-input, 2-output measurements for $\ket{\psi}$ to generate a distribution
which satisfies Hardy's paradox \cite{hardy_quantum_1992,hardy_nonlocality_1993}
using tools from Ref. \cite{curchod_versatile_2019}.

Consider the Schmidt decomposition 
\begin{equation}
\ket{\psi}=\sum_{i=0}^{d-1}\lambda_{i}^{1/2}\ket{ii}
\end{equation}
and assume the coefficients are ordered such that $0\neq\lambda_{0}\neq\lambda_{1}\neq0,$
which is always possible if the state is nonseparable and less-than-maximally
entangled. Wlog assume the Schmidt basis of the state is the canonical
basis. Let $\alpha\in]0,\pi/2[$ and $\delta\in\mathbb{R}$ and consider
the dual vectors 
\begin{equation}
\begin{aligned}\bra{e_{0|0}} & =\cos\alpha\bra{0}+\ue^{\ui\delta}\sin\alpha\bra{1}\\
\bra{e_{1|1}} & =\lambda_{0}\cos\alpha\bra{0}+\lambda_{1}\ue^{\ui\delta}\sin\alpha\bra{1}\\
\bra{f_{0|0}} & =\lambda_{1}^{3/2}\ue^{\ui\delta}\sin\alpha\bra{0}-\lambda_{0}^{3/2}\cos\alpha\bra{1}\\
\bra{f_{1|1}} & =\lambda_{1}^{1/2}\ue^{\ui\delta}\sin\alpha\bra{0}-\lambda_{0}^{1/2}\cos\alpha\bra{1}
\end{aligned}
\label{eq:Hardymts}
\end{equation}
(one can write the projectors in the Schmidt basis of the state instead
of assuming the state decomposes into the canonical basis). Define
the measurements $E_{a|x}$ for Alice, with input $x$ and output
$a,$ and $F_{b|y}$ for Bob, with input $y$ and output $b,$ given
by 
%\begin{equation}
\begin{align*}E_{0|0} & =\ketbra{e_{0|0}}\\
E_{1|0} & \propto\ketbra{e_{0|0}}^{\perp}\oplus\mathbbm1_{2,...,d-1}\\
E_{0|1} & \propto\ketbra{e_{1|1}}^{\perp}\\
E_{1|1} & \propto\ketbra{e_{1|1}}\oplus\mathbbm1_{2,...,d-1}\\
F_{0|0} & \propto\ketbra{f_{0|0}}\tag{\stepcounter{equation}\theequation}\\
F_{1|0} & \propto\ketbra{f_{0|0}}^{\perp}\oplus\mathbbm1_{2,...,d-1}\\
F_{0|1} & \propto\ketbra{f_{1|1}}^{\perp}\oplus\mathbbm1_{2,...,d-1}\\
F_{1|1} & \propto\ketbra{f_{1|1}}
\end{align*}
%\end{equation}
where $\ketbra{e_{0|0}}^{\perp}$ denotes the density matrix corresponding
to the vector orthogonal to $\ket{e_{0|0}}$ when restricted to the
subspace spanned by $\{\ket{0},\ket{1}\},$ and $\mathbbm1_{2,...,d-1}$
is the identity operator on the subspace spanned by $\left\{ \ket{i}\right\} _{i=2}^{d-1},$
for either Alice or Bob. Note that, since we are only interested in
whether some probabilities are equal or different from zero, normalisation
will not play a role.

We now show that the distribution given by 
\begin{equation}
P(ab|xy)=\tr(E_{a|x}\otimes F_{b|y}\ketbra{\psi})
\end{equation}
satisfies Hardy's paradox. Indeed, because of the probabilities considered
and the form of the measurements, only the terms in $i=0,1$ contribute
to the probabilities that appear in Hardy's paradox, therefore

\begin{equation}
\begin{aligned}P(01|01)\propto\left|\sum_{i=0}^{1}\lambda_{i}^{1/2}\left(\bra{e_{0|0}}\otimes\bra{f_{1|1}}\right)\ket{ii}\right|^{2} & =0\\
P(10|10)\propto\left|\sum_{i=0}^{1}\lambda_{i}^{1/2}\left(\bra{e_{1|1}}\otimes\bra{f_{0|0}}\right)\ket{ii}\right|^{2} & =0\\
P(00|11)\propto\left|\sum_{i=0}^{1}\lambda_{i}^{1/2}\left(\bra{e_{0|1}}\otimes\bra{f_{0|1}}\right)\ket{ii}\right|^{2} & =0.
\end{aligned}
\end{equation}
For $P(00|00),$ we find 
\begin{equation}
\begin{aligned}P(00|00) & \propto\left|\sum_{i=0}^{1}\lambda_{i}^{1/2}\left(\bra{e_{0|0}}\otimes\bra{f_{0|0}}\right)\ket{ii}\right|^{2}\\
 & =\left|\ue^{\ui\delta}\sin\alpha\cos\alpha\,\lambda_{0}^{1/2}\,\lambda_{1}^{1/2}\,(\lambda_{1}-\lambda_{0})\right|^{2},
\end{aligned}
\end{equation}
which is strictly greater than zero when $\alpha\in]0,\pi/2[$ and
$0\neq\lambda_{0}\neq\lambda_{1}\neq0,$ like we assumed. This proves
the claim.
\end{proof}

\section{GMNL from GME}

By Theorem \ref{thm:gmnl-from-bipartite-ent}, a star network whose
central node shares pure-state entanglement with all others is GMNL.
We now ask whether all GME states are GMNL (i.e. the genuine multipartite
extension of Gisin's theorem). We show $(n-1)$ copies of any pure
GME $n$-partite state suffice to generate $n$-partite GMNL. We do
this by generating a distribution from these copies that mimics the
star network configuration. 
\begin{figure}
\begin{centering}
\includegraphics[width=0.6\textwidth]{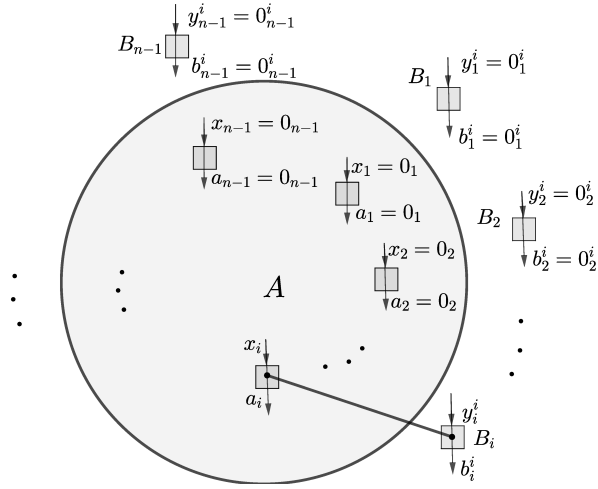}
\par\end{centering}
\caption{Element $i\in[n-1]$ of the star network of bipartite entanglement
created from a GME state $\ket{\Psi}.$ Parties $\{B_{j}\}_{j\in[n-1],j\protect\neq i}$
have already measured $\ket{\Psi}$ and are left unentangled. Alice
and party $B_{i}$ share a pure bipartite entangled state. Alice has
input $x_{i}$ and output $a_{i}$ while each party $B_{j},$ $j\in[n-1],$
has input $y_{j}^{i}$ and output $b_{j}^{i}.$ \label{fig:GMNL-from-GME}}
\end{figure}

We fix some notation that we will use in Theorem \ref{thm:copies}
below. The result considers a GME state $\ket{\Psi}\in\mathcal{H}_{A}\otimes\mathcal{H}_{B_{1}}...\otimes\mathcal{H}_{B_{n-1}}\cong(\mathbb{C}^{d})^{\otimes n},$
$n-1$ copies of which are shared between $n$ parties $A,B_{1},...,B_{n-1}.$
Each party measures locally on each particle, like in Theorem \ref{thm:gmnl-from-bipartite-ent}.
We denote Alice's input and output, respectively, as $\chi\equiv x_{1}...x_{n-1},$
$\alpha\equiv a_{1}...a_{n-1}$ in terms of the digits $x_{i},a_{i}$
corresponding to each particle $i\in[n-1].$ We let the measurement
made by party $B_{j}$ on copy $i$ have input $y_{j}^{i}$ and output
$b_{j}^{i}\,,$ where $i,j=1,...,n-1,$ and for each $j$ we denote
$\upsilon_{j}=y_{j}^{1}...y_{j}^{n-1}$ and $\beta_{j}=b_{j}^{1}...b_{j}^{n-1}$
digit-wise. Then, after measurement, the parties share a distribution
\begin{equation}
\left\{ P(\alpha\beta_{1}...\beta_{n-1}|\chi\upsilon_{1}...\upsilon_{n-1})\right\} _{\substack{\alpha,\beta_{1}...\beta_{n-1}\\
\chi,\upsilon_{1}...\upsilon_{n-1}
}
}.
\end{equation}
Because we are considering local measurements made on each particle,
this distribution is of the form 
\begin{equation}
P(\alpha\beta_{1}...\beta_{n-1}|\chi\upsilon_{1}...\upsilon_{n-1})=\prod_{i=1}^{n-1}P_{i}(a_{i}b_{1}^{i}...b_{n-1}^{i}|x_{i}y_{1}^{i}...y_{n-1}^{i}),\label{eq:Pnpartite-fromGME}
\end{equation}
where each $P_{i}$ is the distribution arising from copy $i$ of
the state $\ket{\Psi}.$ Each copy $i$ of the state $\ket{\Psi}$
will give an edge of a star network connecting Alice and party $B_{i}.$
Because of the structure of this particular network, we can simplify
the notation with respect to Theorem \ref{thm:gmnl-from-bipartite-ent}
and identify the index of each party $B_{i}$ with its corresponding
edge $i.$
\begin{thm}
\label{thm:copies}Any GME state $\ket{\Psi}\in\mathcal{H}_{1}\otimes...\otimes\mathcal{H}_{n}\cong(\mathbb{C}^{d})^{\otimes n}$
is such that $\ket{\Psi}^{\otimes(n-1)}$ is GMNL.
\end{thm}

We first outline the proof for the tripartite case, and then extend
it to the general case.
\begin{proof}[Proof for the tripartite case]
Since $n=3$, we consider two copies of the state. For each copy,
we derive measurements for Bob1 and Bob2 that leave Alice bipartitely
entangled with Bob2 and Bob1 respectively. This yields a network as
in equation (\ref{eq:P3partite}) but postselected on the inputs and
outputs of these measurements. We generalise Theorem \ref{thm:gmnl-from-bipartite-ent}
to show this network is also GMNL.

For $i,j=1,2,$ on copy $i,$ $B_{j}$'s measurements have input $y_{j}^{i}$
and output $b_{j}^{i}$ and Alice's measurement has input $x_{i}$
and output $a_{i}.$ We denote $B_{j}$'s inputs and outputs in terms
of their digits as $\upsilon_{j}=y_{j}^{1}y_{j}^{2}$ and $\beta_{j}=b_{j}^{1}b_{j}^{2}.$
Then, after measurement, the parties share a distribution 
\begin{equation}
\begin{aligned}P(\alpha & \beta_{1}\beta_{2}|\chi\upsilon_{1}\upsilon_{2})\\
 & =P_{1}(a_{1},b_{1}^{1}b_{2}^{1}|x_{1},y_{1}^{1}y_{2}^{1})P_{2}(a_{2},b_{1}^{2}b_{2}^{2}|x_{2},y_{1}^{2}y_{2}^{2})\,.
\end{aligned}
\label{eq:P3partite-fromGME}
\end{equation}

For each $i,j=1,2,$ $i\neq j,$ we assume $B_{j}$ uses input $0_{j}^{i}$
and output $0_{j}^{i}$ to project the $i$th copy of $\ket{\Psi}$
onto $\ket{\phi_{i}}_{AB_{i}},$ as shown in Figure \ref{fig:GMNL-from-GME}
for $n$ parties. Then, Refs. \cite{popescu_generic_1992,gachechiladze_completing_2017}
and a continuity argument serve to show we only have two possibilities
for each $i$: either there exists an input and output per party such
that $\ket{\phi_{i}}_{AB_{i}}$ is less-than-maximally entangled,
or there exists an input per party such that, for all outputs, $\ket{\phi_{i}}_{AB_{i}}$
is maximally entangled. In each case we generalise the proof in Theorem
\ref{thm:gmnl-from-bipartite-ent} to show $\ket{\Psi}^{\otimes2}$
is GMNL.

If both $\ket{\phi_{i}}_{AB_{i}},$ $i=1,2$ are less-than-maximally
entangled, we use the following expression, which is a GMNL inequality
by the same reasoning as in Theorem \ref{thm:gmnl-from-bipartite-ent}:
\begin{equation}
\begin{aligned}I_{3} & =\sum_{i=1}^{2}I^{i}+P(00,00,00|00,00,00)\\
 & -\sum_{i=1}^{2}\sum_{\substack{a_{j},b_{i}^{j},\\
b_{j}^{j}=0,1,\\
j\neq i
}
}P(0_{i}a_{j}\,,0_{i}^{i}b_{i}^{j}\,,0_{j}^{i}b_{j}^{j}\,|0_{i}0_{j}\,,0_{i}^{i}0_{i}^{j}\,,0_{j}^{i}0_{j}^{j})\leq0,
\end{aligned}
\end{equation}
where 
\begin{equation}
\begin{aligned}I^{i}=\sum_{\substack{a_{j},b_{i}^{j},b_{j}^{j}=0,1,\\
j\neq i
}
} & \left[P(0_{i}a_{j}\,,0_{i}^{i}b_{i}^{j}\,,0_{j}^{i}b_{j}^{j}\,|0_{i}0_{j}\,,0_{i}^{i}0_{i}^{j}\,,0_{j}^{i}0_{j}^{j})-P(0_{i}a_{j}\,,1_{i}^{i}b_{i}^{j}\,,0_{j}^{i}b_{j}^{j}\,|0_{i}0_{j}\,,1_{i}^{i}0_{i}^{j}\,,0_{j}^{i}0_{j}^{j})\right.\\
 & \left.-P(1_{i}a_{j}\,,0_{i}^{i}b_{i}^{j}\,,0_{j}^{i}b_{j}^{j}\,|1_{i}0_{j}\,,0_{i}^{i}0_{i}^{j}\,,0_{j}^{i}0_{j}^{j})-P(0_{i}a_{j}\,,0_{i}^{i}b_{i}^{j}\,,0_{j}^{i}b_{j}^{j}\,|1_{i}0_{j}\,,1_{i}^{i}0_{i}^{j}\,,0_{j}^{i}0_{j}^{j})\right].
\end{aligned}
\end{equation}

Evaluating the inequality on the distribution (\ref{eq:P3partite-fromGME}),
we find again that all negative terms in each $I^{i}$ can be sent
to zero. For each $i$ we get, for example, 
\begin{equation}
\begin{aligned}\sum_{\substack{a_{j},b_{i}^{j},b_{j}^{j}\\
=0,1
}
} & P(0_{i}a_{j}\,,1_{i}^{i}b_{i}^{j}\,,0_{j}^{i}b_{j}^{j}\,|0_{i}0_{j}\,,1_{i}^{i}0_{i}^{j}\,,0_{j}^{i}0_{j}^{j})\\
 & =P_{i}(0_{i}1_{i}^{i}0_{j}^{i}|0_{i}1_{i}^{i}0_{j}^{i})
\end{aligned}
\end{equation}
as the sum over $P_{j}$ is 1. But, conditioned on $B_{j}$'s input
and output being $0_{j}^{i},$ parties $AB_{i}$ can measure so $P_{i}$
satisfies Hardy's paradox, hence this term is zero, and similarly
for the other two negative terms. This means all terms in $I_{3}$
are zero except $P(00,00,00|00,00,00)>0,$ violating the inequality.
Therefore, $\ket{\Psi}^{\otimes2}$ is GMNL.

If, for both $i=1,2,$ there exists a local measurement for party
$B_{j},$ $j\neq i$ such that, for \emph{all }outputs, $\ket{\Psi}$
is projected onto a maximally entangled state $\ket{\phi_{i}}_{AB_{i}},$
then $\ket{\Psi}$ satisfies Theorem 2 in Ref. \cite{almeida_multipartite_2010},
so $\ket{\Psi}$ itself is GMNL. Therefore so is $\ket{\Psi}^{\otimes2}.$

Finally, if $\ket{\phi_{1}}_{AB_{1}}$ is maximally entangled for
all of $B_{2}$'s outputs, and $\ket{\phi_{2}}_{AB_{2}}$ is less-than-maximally
entangled, using Refs. \cite{popescu_generic_1992,almeida_multipartite_2010}
we deduce that the bipartite EPR2 components of $P_{1,2}$ across
$A|B_{1,2}$ respectively are bounded like in Theorem \ref{thm:gmnl-from-bipartite-ent}.
That is, $\exists\,\varepsilon>0$ such that the local component of
any EPR2 decomposition across $A|B_{2}$ satisfies 
\begin{equation}
p_{L,2}^{A|B_{2}}\leq1-\varepsilon\label{eq:pHlocal-copies-3parties}
\end{equation}
and, given this $\varepsilon,$ parties $AB_{1}$ can measure locally
such that all bipartite EPR2 decompositions across $A|B_{1}$ have
a local component 
\begin{equation}
p_{L,1}^{A|B_{1}}<\varepsilon.\label{eq:Plocal-maxent-copies-3parties}
\end{equation}

Then, we assume $P(\alpha\beta_{1}\beta_{2}|\chi\upsilon_{1}\upsilon_{2})$
is bilocal and decompose it in local terms across different bipartitions,
like in equation (\ref{eq:PmixBL}) in Theorem \ref{thm:gmnl-from-bipartite-ent}.
Summing over $a_{2},b_{j}^{2},$ $j=1,2$ gives an EPR2 decomposition
of $P_{1}$ whose local components can be bounded using equation (\ref{eq:Plocal-maxent-copies-3parties}).
Summing over $a_{1},b_{j}^{1},$ $j=1,2$ instead gives an EPR2 decomposition
of $P_{2}.$ But the bound on the local component of $P_{1}$ entails
a bound on that of $P_{2}$ which contradicts equation (\ref{eq:pHlocal-copies-3parties}),
proving $P$ is GMNL.
\end{proof}
\begin{proof}[Proof of Theorem \ref{thm:copies}]
We turn now to the case of $n$ parties, for any $n\in\mathbb{N}$.
For each copy $i=1,...,n-1$ of the state $\ket{\Psi},$ we will find
measurements for parties $\{B_{j}\}_{j\neq i}$ that leave Alice and
party $B_{i}$ with a bipartite entangled state. This will yield a
network in a similar configuration to Theorem \ref{thm:gmnl-from-bipartite-ent}
for a star network, but conditionalised on the inputs and outputs
of these measurements. We will generalise the result of Theorem \ref{thm:gmnl-from-bipartite-ent}
as it applies to a star network to show that this network is also
GMNL.

Let $i\in[n-1]$ and consider the $i$th copy of $\ket{\Psi}.$ Suppose
each party $B_{j},$ $j\neq i,$ performs a local, projective measurement
onto a basis $\left\{ \ket{b_{j}}\right\} _{b_{j}=0}^{d-1}.$ We pick
the computational basis on each party's Hilbert space to be such that
the measurement performed by the parties $B_{j},$ $j\neq i,$ leave
Alice and $B_{i}$ in state $|\phi_{\overrightarrow{b}}\rangle_{AB_{i}},$
where $\overrightarrow{b}=b_{1}...b_{i-1}b_{i+1}...b_{n-1}$ denotes
the output obtained by the parties $B_{j},$ $j\neq i$ (we briefly
omit the script $i$ referring to the copy of the state, for readability).
This means that we can write the state $\ket{\Psi}$ as 
\begin{equation}
\ket{\Psi}=\sum_{\overrightarrow{b}}\lambda_{\overrightarrow{b}}\ket{\phi_{\overrightarrow{b}}}_{AB_{i}}\ket{\overrightarrow{b}}_{B_{1}...B_{i-1}B_{i+1}...B_{n-1}.}
\end{equation}
Ref. \cite{popescu_generic_1992}, whose proof was completed in Ref.
\cite{gachechiladze_completing_2017}, showed that there always exist
measurements (i.e. bases) $\left\{ \ket{b_{j}}\right\} _{b_{j}=0}^{d-1}$
such that $\ket{\phi_{\overrightarrow{b}}}_{AB_{i}}$ is entangled
for a certain output $\overrightarrow{b}$. We now show that this
opens up only two possibilities for each $i$: either there exists
an output such that $\ket{\phi_{\overrightarrow{b}}}_{AB_{i}}$ is
less-than-maximally entangled, or for all outputs $\overrightarrow{b}$,
$\ket{\phi_{\overrightarrow{b}}}_{AB_{i}}$ is maximally entangled.
Indeed, the only option left to discard is one where, for some $\overrightarrow{b}=\overrightarrow{b^{*}},$
$\ket{\phi_{\overrightarrow{b^{*}}}}_{AB_{i}}$ is maximally entangled,
and for some other $\overrightarrow{b}=\overrightarrow{b^{**}},$
$\ket{\phi_{\overrightarrow{b^{**}}}}_{AB_{i}}$ is separable. But
it is easy to see, by using a continuity argument, that in this case
the bases $\left\{ \ket{b_{j}}\right\} _{b_{j}=0}^{d-1}$ can be modified
so that there exists one output for which $AB_{i}$ are projected
onto a less-than-maximally entangled state: it suffices to consider
one (normalised) element of the measurement basis to be $c_{0}\ket{b_{j}^{*}}+c_{1}\ket{b_{j}^{**}}$
for some values $c_{0},c_{1}\in\mathbb{C},$ for each $j.$

Therefore, we consider the following cases: 
\begin{enumerate}
\item for all $i\in[n-1],$ there exists an input and output for each $B_{j},j\neq i$
such that $\ket{\phi_{i}}_{AB_{i}}$ is less-than-maximally entangled;\label{enu:Hardy-copies} 
\item for all $i\in[n-1],$ there exists an input for each $B_{j},j\neq i$
such that $\ket{\phi_{i}}_{AB_{i}}$ is maximally entangled for all
outputs;\label{enu:maxent-copies} 
\item there exist $i,k\in[n-1]$ such that $\ket{\phi_{i}}_{AB_{i}}$ is
as in Case \ref{enu:maxent-copies} and $\ket{\phi_{k}}_{AB_{k}}$
is as in Case \ref{enu:Hardy-copies}.\label{enu:mix-copies} 
\end{enumerate}
\textbf{Case \ref{enu:Hardy-copies}:} let $i\in[n-1].$ Suppose parties
$\{B_{j}\}_{j\neq i}$ perform the measurements explained above that
leave Alice and $B_{i}$ less-than-maximally entangled. Then, Alice
and $B_{i}$ can perform local measurements on the resulting state
to satisfy Hardy's paradox. We will modify the inequality in Theorem
\ref{thm:gmnl-from-bipartite-ent} and show that these measurements
on $\ket{\Psi}^{\otimes(n-1)}$ give a distribution which violates
the inequality.

To modify the inequality in Theorem \ref{thm:gmnl-from-bipartite-ent},
we import the same strategy to lift inequality (\ref{eq:Iorig}) to
$n$ parties, each with $2^{n-1}$ inputs and outputs. We want $I^{AB_{i}}$
to detect bipartite nonlocality between Alice's $i$th particle and
$B_{i}$'s $i$th particle, that is, nonlocality in $a_{i}b_{i}^{i}|x_{i}y_{i}^{i}.$
Therefore, for each $i$ we now need to fix all other inputs $x_{j},y_{i}^{j},y_{j}^{j}$
and add over all other outputs $a_{j},b_{i}^{j},b_{j}^{j},$ $j\neq i,$
so that 
\begin{equation}
\begin{aligned}I^{AB_{i}}=\sum_{a_{\bar{i}},b_{i}^{\bar{i}},b_{\bar{i}}^{\bar{i}}=0,1} & \left[P(0_{i}a_{\bar{i}}\,,0_{i}^{i}b_{i}^{\bar{i}}\,,0_{\bar{i}}^{i}b_{\bar{i}}^{\bar{i}}\,|0_{i}0_{\bar{i}}\,,0_{i}^{i}0_{i}^{\bar{i}}\,,0_{\bar{i}}^{i}0_{\bar{i}}^{\bar{i}})-P(0_{i}a_{\bar{i}}\,,1_{i}^{i}b_{i}^{\bar{i}}\,,0_{\bar{i}}^{i}b_{\bar{i}}^{\bar{i}}\,|0_{i}0_{\bar{i}}\,,1_{i}^{i}0_{i}^{\bar{i}}\,,0_{\bar{i}}^{i}0_{\bar{i}}^{\bar{i}})\right.\\
 & \left.-P(1_{i}a_{\bar{i}}\,,0_{i}^{i}b_{i}^{\bar{i}}\,,0_{\bar{i}}^{i}b_{\bar{i}}^{\bar{i}}\,|1_{i}0_{\bar{i}}\,,0_{i}^{i}0_{i}^{\bar{i}}\,,0_{\bar{i}}^{i}0_{\bar{i}}^{\bar{i}})-P(0_{i}a_{\bar{i}}\,,0_{i}^{i}b_{i}^{\bar{i}}\,,0_{\bar{i}}^{i}b_{\bar{i}}^{\bar{i}}\,|1_{i}0_{\bar{i}}\,,1_{i}^{i}0_{i}^{\bar{i}}\,,0_{\bar{i}}^{i}0_{\bar{i}}^{\bar{i}})\right],
\end{aligned}
\end{equation}
where the outputs in the first term are denoted as follows: $0_{i}a_{\bar{i}}$
denotes output $\alpha=a_{1}...0_{i}...a_{n-1},\:$ $0_{i}^{i}b_{i}^{\bar{i}}$
denotes output $\beta_{i}=b_{i}^{1}...0_{i}^{i}...b_{i}^{n-1},\:$
and $0_{\bar{i}}^{i}b_{\bar{i}}^{\bar{i}}$ denotes output $\beta_{j}=b_{j}^{1}...0_{j}^{i}...b_{j}^{n-1}$
for \emph{all} $j\neq i.$ Inputs are denoted similarly, and the notation
is similar for the other three terms. Then, the inequality 
\begin{equation}
I_{n}=\sum_{i=1}^{n-1}I^{AB_{i}}+P(\overrightarrow{0},\overrightarrow{0}|\overrightarrow{0},\overrightarrow{0})-\sum_{i=1}^{n-1}\sum_{a_{\bar{i}},b_{i}^{\bar{i}},b_{\bar{i}}^{\bar{i}}=0,1}P(0_{i}a_{\bar{i}}\,,0_{i}^{i}b_{i}^{\bar{i}}\,,0_{\bar{i}}^{i}b_{\bar{i}}^{\bar{i}}\,|0_{i}0_{\bar{i}}\,,0_{i}^{i}0_{i}^{\bar{i}}\,,0_{\bar{i}}^{i}0_{\bar{i}}^{\bar{i}})\leq0
\end{equation}
is a GMNL inequality, by the same reasoning as in Theorem \ref{thm:gmnl-from-bipartite-ent}.

Evaluating the inequality on the distribution (\ref{eq:Pnpartite-fromGME}),
we find again that each term simplifies. For each $i$ we get, for
example,

%\begin{equation}
\begin{align*}\sum_{a_{\bar{i}},b_{i}^{\bar{i}},b_{\bar{i}}^{\bar{i}}=0,1}P & (0_{i}a_{\bar{i}}\,,1_{i}^{i}b_{i}^{\bar{i}}\,,0_{\bar{i}}^{i}b_{\bar{i}}^{\bar{i}}\,|0_{i}0_{\bar{i}}\,,1_{i}^{i}0_{i}^{\bar{i}}\,,0_{\bar{i}}^{i}0_{\bar{i}}^{\bar{i}})\tag{\stepcounter{equation}\theequation}\\\\
 & =P_{i}(0_{i}1_{i}^{i}0_{\bar{i}}^{i}|0_{i}1_{i}^{i}0_{\bar{i}}^{i})\prod_{\substack{j=1\\
j\neq i
}
}^{n-1}\sum_{\substack{a_{j},b_{k}^{j}=0,1\\
k\neq j
}
}P_{j}(a_{j}b_{1}^{j}...b_{j-1}^{j}b_{j+1}^{j}...b_{n-1}^{j}\,|0_{j}0_{1}^{j}...0_{j-1}^{j}0_{j+1}^{j}...0_{n-1}^{j})\\
 & =P_{i}(0_{i}1_{i}^{i}0_{\bar{i}}^{i}|0_{i}1_{i}^{i}0_{\bar{i}}^{i})
\end{align*}
%\end{equation}
and, similarly, 
\begin{equation}
\begin{aligned}\sum_{a_{\bar{i}},b_{i}^{\bar{i}},b_{\bar{i}}^{\bar{i}}=0,1}P(1_{i}a_{\bar{i}}\,,0_{i}^{i}b_{i}^{\bar{i}}\,,0_{\bar{i}}^{i}b_{\bar{i}}^{\bar{i}}\,|1_{i}0_{\bar{i}}\,,0_{i}^{i}0_{i}^{\bar{i}}\,,0_{\bar{i}}^{i}0_{\bar{i}}^{\bar{i}}) & =P_{i}(1_{i}0_{i}^{i}0_{\bar{i}}^{i}|1_{i}0_{i}^{i}0_{\bar{i}}^{i})\,;\\
\sum_{a_{\bar{i}},b_{i}^{\bar{i}},b_{\bar{i}}^{\bar{i}}=0,1}P(0_{i}a_{\bar{i}}\,,0_{i}^{i}b_{i}^{\bar{i}}\,,0_{\bar{i}}^{i}b_{\bar{i}}^{\bar{i}}\,|1_{i}0_{\bar{i}}\,,1_{i}^{i}0_{i}^{\bar{i}}\,,0_{\bar{i}}^{i}0_{\bar{i}}^{\bar{i}}) & =P_{i}(0_{i}0_{i}^{i}0_{\bar{i}}^{i}|1_{i}1_{i}^{i}0_{\bar{i}}^{i}).
\end{aligned}
\end{equation}
Also, 
\begin{equation}
P(\overrightarrow{0},\overrightarrow{0}|\overrightarrow{0},\overrightarrow{0})=\prod_{i=1}^{n-1}P_{i}(0_{i}0_{i}^{i}0_{\bar{i}}^{i}|0_{i}0_{i}^{i}0_{\bar{i}}^{i})\,.
\end{equation}
Now each $P_{i}$ in equation (\ref{eq:Pnpartite-fromGME}) arises
from measurements by $\{B_{j}\}_{j\neq i}$ to create a less-than-maximally
entangled state between Alice and $B_{i},$ who can then choose measurements
to satisfy Hardy's paradox. Hence all terms are zero except $P(\overrightarrow{0},\overrightarrow{0}|\overrightarrow{0},\overrightarrow{0})>0,$
and so the inequality is violated. Therefore, $\ket{\Psi}^{\otimes(n-1)}$
is GMNL.

\textbf{Case \ref{enu:maxent-copies}:} we assumed that, for all $i\in[n-1],$
there exist local measurements on $\ket{\Psi}$ for parties $\left\{ B_{j}\right\} _{j\neq i}$
that, for \emph{all} outcomes, create a maximally entangled state
$\ket{\phi_{i}}_{AB_{i}}$ shared between Alice and $B_{i}.$ Since
all bipartitions can be expressed as $A|B_{i}$ for some $i,$ we
find that $\ket{\Psi}$ meets the requirements of Theorem 2 in \cite{almeida_multipartite_2010},
and so $\ket{\Psi}$ is GMNL. That is, one copy of the shared state
$\ket{\Psi}$ is already GMNL, and therefore so is $\ket{\Psi}^{\otimes(n-1)}.$

\textbf{Case \ref{enu:mix-copies}:} assume wlog that the state $\ket{\phi_{i}}_{AB_{i}}$
is less-than-maximally entangled for $i=1,...,K_{0}$ and maximally
entangled for $i=K_{0}+1,...,n-1.$ We will show that $\ket{\Psi}^{\otimes(K_{0}+1)}$
is GMNL, which implies that $\ket{\Psi}^{\otimes(n-1)}$ is so too.

It will be useful to classify bipartitions $M|\overline{M}$ like
in Theorem \ref{thm:gmnl-from-bipartite-ent}. We will always assume
that Alice belongs to $M$ in order not to duplicate the bipartitions.
Let $S_{\leq K_{0}}$ be the set of bipartitions $M|\overline{M}$
(indexed by $M$) which are crossed by an edge $j\leq K_{0},$ i.e.,
where $\overline{M}$ contains at least one index $j\in\{1,...,K_{0}\},$
and $T_{\leq K_{0}}$ be its complement, i.e. the set of bipartitions
where $\overline{M}$ contains only indices $j\in\{K_{0}+1,...,n-1\}.$
Similarly, $S_{>K_{0}}$ (respectively, $T_{>K_{0}}$) is the set
of bipartitions which are (not) crossed by an edge $j>K_{0}.$ That
is, in $S_{>K_{0}},$ there is some $j\in\{K_{0}+1,...,n-1\}$ which
belongs to $\overline{M},$ while in $T_{>K_{0}},$ $\overline{M}$
contains only indices $j\in\{1,...,K_{0}\}.$

For each $i=1,...,K_{0},$ parties $AB_{i}$ can perform measurements
on their shared state $\ket{\phi_{i}}_{AB_{i}}$ which, together with
the measurements of parties $\{B_{j}\}_{j\neq i}$ that projected
$\ket{\Psi}$ onto $\ket{\phi_{i}}_{AB_{i}},$ give rise to a distribution
\begin{equation}
P_{i}(a_{i}b_{1}^{i}...b_{n-1}^{i}|x_{i}y_{1}^{i}...y_{n-1}^{i})\label{eq:PiHardy-copies}
\end{equation}
which satisfies Hardy's paradox when post-selected on the inputs and
outputs of parties $\{B_{j}\}_{j\neq i}.$ Then, the distribution
arising from the first $K_{0}$ copies of $\ket{\Psi}$ is 
\begin{equation}
P_{H}(\{a_{i}\}_{i\leq K_{0}}\{b_{j}^{i}\}_{i\leq K_{0},j\in[n-1]}|\{x_{i}\}_{i\leq K_{0}}\{y_{j}^{i}\}_{i\leq K_{0},j\in[n-1]})=\prod_{i=1}^{K_{0}}P_{i}(a_{i}b_{1}^{i}...b_{n-1}^{i}|x_{i}y_{1}^{i}...y_{n-1}^{i}),\label{eq:PHardy-copies}
\end{equation}
with $P_{i}$ as in equation (\ref{eq:PiHardy-copies}). This distribution
is similar to that in Case \ref{enu:Hardy-copies} when post-selected
on the inputs and outputs of parties $\{B_{j}\}_{j>K_{0}}.$ More
precisely, by the nonsignalling condition, we have 
\begin{equation}
\begin{aligned}P_{H}(\{a_{i}\}_{i\leq K_{0}}\{b_{j}^{i} & \}_{i\leq K_{0},j\leq K_{0}}\{b_{j}^{i}=0_{j}^{i}\}_{i\leq K_{0},j>K_{0}}|\{x_{i}\}_{i\leq K_{0}}\{y_{j}^{i}\}_{i\leq K_{0},j\leq K_{0}}\{y_{j}^{i}=0_{j}^{i}\}_{i\leq K_{0},j>K_{0}})\\
=P_{AB_{1}...B_{K_{0}}}( & \{a_{i}\}_{i\leq K_{0}}\{b_{j}^{i}\}_{i\leq K_{0},j\leq K_{0}}\\
 & |\{x_{i}\}_{i\leq K_{0}}\{y_{j}^{i}\}_{i\leq K_{0},j\leq K_{0}},\{b_{j}^{i}=0_{j}^{i}\}_{i\leq K_{0},j>K_{0}},\{y_{j}^{i}=0_{j}^{i}\}_{i\leq K_{0},j>K_{0}})\\
\times P_{B_{K_{0}+1}} & _{...B_{n-1}}(\{b_{j}^{i}=0_{j}^{i}\}_{i\leq K_{0},j>K_{0}}|\{y_{j}^{i}=0_{j}^{i}\}_{i\leq K_{0},j>K_{0}}),
\end{aligned}
\label{eq:PHcondtl-copies}
\end{equation}
where by Case \ref{enu:Hardy-copies} we know that $P_{AB_{1}...B_{K_{0}}}$
is GMNL in its parties. Then, $P_{H}$ must be $(K_{0}+1)$-way nonlocal
(i.e., GMNL when restricted to parties $A,B_{1},...,B_{K_{0}}$).
Indeed, if this were not the case, by equation (\ref{eq:PHcondtl-copies})
we could obtain a bilocal decomposition for $P_{AB_{1}...B_{K_{0}}},$
which would contradict the fact that this distribution is GMNL.

Therefore, there exists an $\varepsilon>0$ such that any EPR2 decomposition
of $P_{H}$ as 
\begin{equation}
P_{H}=\sum_{M}p_{L,H}^{M}P_{L,H}^{M}+p_{NS,H}P_{NS,H}
\end{equation}
we have that the terms where $P_{L,H}^{M}$ is local across a bipartition
such that $M\in S_{\leq K_{0}}$ satisfy 
\begin{equation}
\sum_{M\in S_{\leq K_{0}}}p_{L,H}^{M}\leq1-\varepsilon.\label{eq:pHlocal-copies}
\end{equation}

On the other hand, $\ket{\Psi}$ satisfies Theorem 1 in Ref. \cite{almeida_multipartite_2010}
for all bipartitions $A|B_{i}$ for $i=K_{0}+1,...,n-1,$ hence it
is fully nonlocal across all such bipartitions. This means that, for
any $\delta_{i}>0,$ there exist local measurements on $\ket{\Psi}$
(which depend on $i$) that lead to a distribution 
\begin{equation}
P_{+}(ab_{1}...b_{n-1}|xy_{1}...y_{n-1})\label{eq:P+-copies}
\end{equation}
such that any bipartite EPR2 decomposition across a bipartition $A|B_{i},$
for $i=K_{0}+1,...,n-1,$ 
\begin{equation}
P_{+}=p_{L,+}^{A|B_{i}}P_{L,+}^{A|B_{i}}+(1-p_{L,+}^{A|B_{i}})P_{NS,+}^{A|B_{i}}
\end{equation}
satisfies 
\begin{equation}
p_{L,+}^{A|B_{i}}<\delta_{i}.\label{eq:P+localdeltai-copies}
\end{equation}
Thus, considering the possibility of implementing all the above measurements
for each \emph{i} leads to a distribution of the form (\ref{eq:P+-copies})
in which equation (\ref{eq:P+localdeltai-copies}) holds for every
$i=K_{0}+1,...,n-1.$

Therefore, given the $\varepsilon$ above, the parties can choose
suitable $\delta_{i}$ to bound the bipartitely local components and
hence ensure that any multipartite EPR2 decomposition of $P_{+},$
\begin{equation}
P_{+}=\sum_{M}p_{L,+}^{M}P_{L,+}^{M}+p_{NS,+}P_{NS,+}
\end{equation}
is such that the terms where $P_{L,+}^{M}$ is local across a bipartition
such that $M\in S_{>K_{0}}$ satisfy 
\begin{equation}
\sum_{M\in S_{>K_{0}}}p_{L,+}^{M}<\varepsilon.\label{eq:Plocal-maxent-copies}
\end{equation}

Since we only need to consider $(K_{0}+1)$ copies of the state, we
denote the inputs and outputs of Alice and each party $B_{j},$ $j\in[n-1]$
by $\chi=x_{1}...x_{K_{0}+1},$ $\upsilon_{j}=y_{j}^{1}...y_{j}^{K_{0}+1};$
$\alpha=a_{1}...a_{K_{0}+1},$ $\beta_{j}=b_{j}^{1}...b_{j}^{K_{0}+1}$
respectively. Then, the global distribution obtained from $\ket{\Psi}^{\otimes(K_{0}+1)}$
is 
\begin{equation}
\begin{aligned}P(\alpha\beta_{1}... & \beta_{n-1}|\chi\upsilon_{1}...\upsilon_{n-1})=\\
 & P_{H}(\{a_{i}\}_{i\leq K_{0}}\{b_{j}^{i}\}_{i\leq K_{0},j\in[n-1]}|\{x_{i}\}_{i\leq K_{0}}\{y_{j}^{i}\}_{i\leq K_{0},j\in[n-1]})\\
 & \times P_{+}(a_{K_{0}+1}b_{1}^{K_{0}+1}...b_{n-1}^{K_{0}+1}|x_{K_{0}+1}y_{1}^{K_{0}+1}...y_{n-1}^{K_{0}+1}),
\end{aligned}
\label{eq:P+PH-copies}
\end{equation}
where $P_{H}$ comes from equation (\ref{eq:PHardy-copies}) and the
EPR2 components of $P_{H},P_{+}$ are as per equations (\ref{eq:pHlocal-copies}),
(\ref{eq:Plocal-maxent-copies}).

We now follow a similar strategy to that in Theorem \ref{thm:gmnl-from-bipartite-ent}.
To prove that the global distribution $P$ is GMNL, as is our goal,
we assume the converse, and we derive a contradiction from the nonlocality
properties of $P_{H}$ and $P_{+}.$ Assuming $P$ is bilocal, we
can express the distribution as 
\begin{equation}
\begin{aligned}P(\alpha\beta_{1}...\beta_{n-1}| & \chi\upsilon_{1}...\upsilon_{n-1})\\
 & =\sum_{\lambda,M}p_{L}^{M}(\lambda)P_{M}(\alpha\{\beta_{j}\}_{j\in M}|\chi\{\upsilon_{j}\}_{j\in M},\lambda)P_{\overline{M}}(\{\beta_{j}\}_{j\in\overline{M}}|\{\upsilon_{j}\}_{j\in\overline{M}},\lambda),
\end{aligned}
\label{eq:PnotGMNL-copies}
\end{equation}
where 
\begin{equation}
\sum_{\lambda,M}p_{L}^{M}(\lambda)=1,
\end{equation}
for \emph{each} $\alpha,\beta_{j},\chi,\upsilon_{j},j=1,...,n-1,$
where we recall that each $\beta_{j}=b_{j}^{1}...b_{j}^{K_{0}+1}$
and similarly for $\upsilon_{j}.$

Now, if we sum equation (\ref{eq:PnotGMNL-copies}) over $a_{i},b_{j}^{i}$
for $i=1,...,K_{0}$ and $j=1,...,n-1$ (that is, we sum over the
$i$th digit, $i\leq K_{0},$ of Alice and all parties $B_{j}$),
we obtain $P_{+}$ on the left-hand side, from equation (\ref{eq:P+PH-copies}).
On the right-hand side, we obtain, for each $M,$\footnote{Note that, once more, the distribution obtained by summing over only
some of the digits of a party's output still depends on the whole
input as it may be signalling in the different digits of the party's
input. However, as in Theorem \ref{thm:gmnl-from-bipartite-ent},
these extra inputs can be fixed to an arbitrary value as the left-hand
side is independent of them.} 
\begin{equation}
\sum_{\lambda}p_{L}^{M}(\lambda)P_{M}(a_{K+1}\{b_{j}^{K_{0}+1}\}_{j\in M}|\chi\{\upsilon_{j}\}_{j\in M},\lambda)P_{\overline{M}}(\{b_{j}^{K_{0}+1}\}_{j\in\overline{M}}|\{\upsilon_{j}\}_{j\in\overline{M}},\lambda)\,,\label{eq:PMMbar-P+}
\end{equation}
whose sum turns out to form an EPR2 decomposition of $P_{+}.$ Indeed,
local terms are given by bipartitions such that $M\in S_{>K_{0}},$
as in these terms there is some digit $b_{j}^{K_{0}+1}$ with $j>K_{0}$
appearing in $P_{\overline{M}},$ thus they are local across $A|B_{j}$
for some $j>K_{0}.$ The nonlocal terms are given by bipartitions
such that $M\in T_{>K_{0}}$ (since all terms are nonsignalling).
Therefore, the choice of measurements which generated $P_{+}$ ensures
(by equation (\ref{eq:Plocal-maxent-copies})) that 
\begin{equation}
\sum_{\lambda,M\in S_{>K_{0}}}p_{L}^{M}(\lambda)<\varepsilon\label{eq:sumS>K}
\end{equation}
and hence 
\begin{equation}
\sum_{\lambda,M\in T_{>K_{0}}}p_{L}^{M}(\lambda)>1-\varepsilon.\label{eq:sumT>K}
\end{equation}

Going back now to equation (\ref{eq:PnotGMNL-copies}), we sum over
$a_{K_{0}+1},b_{j}^{K_{0}+1}$ for $j=1,...,n-1$ (that is, we sum
over the $(K_{0}+1)$th digit of Alice and all parties $B_{j}$).
Then, we obtain $P_{H}$ on the left-hand side, from equation (\ref{eq:P+PH-copies}).
On the right-hand side, we obtain for each $M$, 
\begin{equation}
\sum_{\lambda}p_{L}^{M}(\lambda)P_{M}(\{a_{i}\}_{i\leq K_{0}}\{b_{j}^{i}\}_{i\leq K_{0},j\in M}|\chi\{\upsilon_{j}\}_{j\in M},\lambda)P_{\overline{M}}(\{b_{j}^{i}\}_{i\leq K_{0},j\in\overline{M}}|\{\upsilon_{j}\}_{j\in\overline{M}},\lambda)\,,
\end{equation}
whose sum over $M$ gives an EPR2 decomposition of $P_{H}.$ This
time, $S_{\leq K_{0}}$ will give the local terms, as $P_{\overline{M}}$
will contain at least some digit $b_{j}^{j}$ for $j\leq K_{0},$
while $T_{\leq K_{0}}$ will give the nonlocal terms. By equation
(\ref{eq:pHlocal-copies}), our choice of $\varepsilon$ implies that
\begin{equation}
\sum_{\lambda,M\in S_{\leq K_{0}}}p_{L}^{M}(\lambda)\leq1-\varepsilon.\label{eq:sumSleqK}
\end{equation}

Now, any bipartition in $T_{>K_{0}}$ is such that all $j\in\{K_{0}+1,...,n-1\}$
are in $M.$ Hence, there must be some $j\leq K_{0}$ in $\overline{M},$
otherwise $\overline{M}$ would be empty. Therefore, $P_{\overline{M}}$
always contains at least one digit $b_{j}^{j}$ for some $j\leq K_{0},$
and so terms where $M\in T_{>K_{0}}$ are local across the bipartition
$A|B_{j}$ for some $j\leq K_{0}.$ That is, $T_{>K_{0}}\subseteq S_{\leq K_{0}}.$

This means that equation (\ref{eq:sumSleqK}) also holds if the sum
is over $T_{>K_{0}},$ but this is in contradiction with equation
(\ref{eq:sumT>K}).
\end{proof}

\section{Looking beyond}

While continuity ensures that Theorem \ref{thm:gmnl-from-bipartite-ent}
is robust to some noise, other extensions of this result to mixed
states might be considered. In fact, a very simple construction can
be used to show that Theorem \ref{thm:gmnl-from-bipartite-ent} extends,
at least, to networks of some mixed states:

\noindent \textbf{\vspace{2mm}
}

\noindent \textbf{Observation. }\label{obs:GMNL-mixed-nwk}\emph{There
exist bipartite mixed states which, distributed in any connected network,
yield GMNL independently of the noise parameter.}
\begin{proof}
Consider a connected network of pure bipartite entangled states with
$K$ edges, where, at each edge $k,$ the less-than-maximally entangled
state $\rho_{k}$ gets measured with POVMs $E_{a_{i}^{k}|x_{i}^{k}},\,F_{a_{j}^{k}|x_{j}^{k}}$
(which depend on $k,$ but this is omitted from the notation for readability)
and gives rise to a distribution
\begin{equation}
P_{k}(a_{i}^{k},a_{j}^{k}|x_{i}^{k},x_{j}^{k})=\tr(E_{a_{i}^{k}|x_{i}^{k}}\otimes F_{a_{j}^{k}|x_{j}^{k}}\rho_{k})
\end{equation}
which violates the Bell inequality $I^{k},$ which is $\leq0$ for
all local distributions, with value $\omega_{P}>0$. Consider another
network of the same topology, but with less-than-maximally entangled
states $\sigma_{k}$ and measurements $G_{a_{i}^{k}|x_{i}^{k}},\,H_{a_{j}^{k}|x_{j}^{k}}$
(which also depend on $k$), giving rise to a distribution
\begin{equation}
Q_{k}(a_{i}^{k},a_{j}^{k}|x_{i}^{k},x_{j}^{k})=\tr(G_{a_{i}^{k}|x_{i}^{k}}\otimes H_{a_{j}^{k}|x_{j}^{k}}\sigma_{k})
\end{equation}
at each edge $k.$ Suppose $Q_{k}$ violates $I^{k}$ with value $\omega_{Q}>0$.
Then, for any $p\in(0,1),$ placing the mixed state $p\rho_{k}\oplus(1-p)\sigma_{k}$
and measurements $E_{a_{i}^{k}|x_{i}^{k}}\oplus G_{a_{i}^{k}|x_{i}^{k}},\,F_{a_{j}^{k}|x_{j}^{k}}\oplus H_{a_{j}^{k}|x_{j}^{k}}$
at each edge, one gets the distribution
\begin{equation}
\begin{aligned}R_{k}(a_{i}^{k},a_{j}^{k}|x_{i}^{k},x_{j}^{k}) & =\tr\left((E_{a_{i}^{k}|x_{i}^{k}}\oplus G_{a_{i}^{k}|x_{i}^{k}})\otimes(F_{a_{j}^{k}|x_{j}^{k}}\oplus H_{a_{j}^{k}|x_{j}^{k}})(p\rho_{k}\oplus(1-p)\sigma_{k})\right)\\
 & =pP_{k}(a_{i}^{k},a_{j}^{k}|x_{i}^{k},x_{j}^{k})+(1-p)Q_{k}(a_{i}^{k},a_{j}^{k}|x_{i}^{k},x_{j}^{k}),
\end{aligned}
\end{equation}
which is also a quantum distribution. By convexity, the value of each
$R_{k}$ on inequality $I^{k}$ is
\begin{equation}
I^{k}(R_{k})=pI^{k}(P_{k})+(1-p)I^{k}(Q_{k})=p\omega_{P}+(1-p)\omega_{Q}>0,
\end{equation}
which constitutes a violation. Therefore, by Theorem \ref{thm:gmnl-from-bipartite-ent},
any connected network with the mixed states $p\rho_{k}\oplus(1-p)\sigma_{k}$
at each edge is GMNL.
\end{proof}
However, this construction cannot be used to show GMNL in networks
of states that are a mixture of an entangled state with separable
noise. In fact, in such networks, not even GME is guaranteed in principle.
In Chapter \ref{chap:mixed} we will study networks of mixed states,
and find that even their entanglement properties depend on their topology,
as well as on the level of noise contained in the states.

%\bibliographystyle{amsalpha}
%\bibliography{entnonloc}
%
%\end{document}

%\documentclass[11pt,oneside,english,a4paper,twoside]{book}
%\usepackage[T1]{fontenc}
%\usepackage[latin9]{inputenc}
%\usepackage{geometry}
%\geometry{verbose,tmargin=1in,bmargin=1in,lmargin=1in,rmargin=1.5in}
%\setcounter{secnumdepth}{3}
%\setcounter{tocdepth}{3}
%\usepackage{verbatim}
%\usepackage{amsmath}
%\usepackage{amsthm}
%\usepackage{amssymb}
%\usepackage{setspace}
%\onehalfspacing
%
%\makeatletter
%
%%%%%%%%%%%%%%%%%%%%%%%%%%%%%%% LyX specific LaTeX commands.
%%% Because html converters don't know tabularnewline
%\providecommand{\tabularnewline}{\\}
%
%%%%%%%%%%%%%%%%%%%%%%%%%%%%%%% Textclass specific LaTeX commands.
%\theoremstyle{plain}
%\newtheorem{thm}{\protect\theoremname}
%\theoremstyle{plain}
%\newtheorem{lem}[thm]{\protect\lemmaname}
%
%%%%%%%%%%%%%%%%%%%%%%%%%%%%%%% User specified LaTeX commands.
%\usepackage{physics}
%\usepackage{cite}
%\usepackage{bbm}
%\usepackage{enumitem}
%\usepackage{amsmath}
%\usepackage{amsfonts}
%
%\renewcommand\theenumi{(\roman{enumi})}
%\renewcommand\labelenumi{\theenumi}
%\providecommand{\ketbra}[1]{\ket{#1}\bra{#1}}
%\providecommand{\trace}{\textnormal{tr}}
%\providecommand{\FSP}{\mathcal{FSP}}
%
%\newcommand{\ui}{\mathrm{i}}
%\newcommand{\ue}{\mathrm{e}}
%\renewcommand*\arraystretch{1.5}
%
%\makeatother
%
%\usepackage{babel}
%\providecommand{\lemmaname}{Lemma}
%\providecommand{\theoremname}{Theorem}
%
%\begin{document}
%%\setcounter{chapter}{3}

\chapter{\label{chap:mixed}Mixed pair-entangled network states}

\chaptermark{Mixed pair-entangled networks}

Pair-entangled networks of mixed states are widely studied as an experimentally
feasible way of achieving genuine multipartite quantum effects. Moreover,
they provide a good platform in which to explore the relationship
between entanglement and nonlocality in many-body systems. We focus
on networks where pairs of parties share isotropic states. First,
we provide bounds on the noise parameter needed to guarantee biseparability
or GME in tripartite networks. Next, we obtain no-go results which
show that tree networks and polygonal networks cannot be GME if the
number of parties is large enough. In sharp contrast, completely connected
networks are always GME if the visibility of the shared states is
above a threshold. Still, GME in any connected network of entangled
isotropic states can be recovered by taking many copies. In addition,
we find that sharing non-steerable states can compromise the GMNL
of a network, or even render it fully local. This leads us to provide
constructions of networks that are GME but not GMNL. However, these
limitations to the obtention of nonlocality can be overcome: surprisingly,
taking many copies of some bilocal networks make it possible to restore
the GMNL. Thus, genuine multipartite effects can be obtained from
some networks if enough copies of the pair-entangled states are available.
This result constitutes, to our knowledge, the first example of superactivation
of GMNL from bilocality.

\section{Entanglement in mixed-state networks}

We will consider networks where pairs of parties share isotropic states,
and analyse whether or not the network is biseparable. For notational
convenience, unless otherwise specified we will group the Hilbert
spaces in terms of the shared states (the edges of the network), although
the bipartitions considered for biseparability will always refer to
parties (the vertices, or nodes, of the network). So, for example,
a Lambda network where Alice shares isotropic states with each of
Bob1 and Bob2 will be denoted as
\begin{equation}
\begin{aligned}\rho_{p_{1},p_{2}}^{\otimes2} & =\rho_{p_{1},A_{1}B}\otimes\rho_{p_{2},A_{2}C},\end{aligned}
\label{eq:lambda-iso}
\end{equation}
with or without the subscripts referring to parties, where the isotropic
state $\rho_{p}$ is given in equation (\ref{eq:intro-isotropic}).

We will also use the flip, or swap, operator,
\begin{equation}
\Pi=\sum_{i,j=0}^{d-1}\ket{ij}\bra{ji},
\end{equation}
in any dimension $d$, which swaps the states of the particles in
a bipartite system.

We begin by exploring the tripartite setting, where the only two connected
networks are a Lambda network, where Alice shares bipartite states
with Bob and Charlie, and a triangle network, where each pair of Alice,
Bob and Charlie share bipartite states. We find that, in contrast
to the case of pure states, where any amount of entanglement yields
GME, networks of isotropic states are only GME if the visibility of
the states on the edges is large enough. Remarkably, this implies
that entanglement can be \emph{deactivated} in mixed-state networks.
In each case, we provide bounds on the visibilities needed to achieve
GME or biseparability.

For tripartite networks, we denote the parties as $A,B,C$, with subindices
wherever a party holds more than one particle. These results were
motivated by the numerical techniques in Ref. \cite{jungnitsch_taming_2011}.
\begin{thm}
\label{thm:lambda-iso-gme}A Lambda network where Alice shares a 2-dimensional
isotropic state $\rho_{p_{1}}$ with Bob, and another $\rho_{p_{2}}$
with Charlie, is GME when $p_{i}>1/3$, $p_{j}>1/(3p_{i})$, $i,j=1,2$.
\end{thm}

\begin{proof}
We will show that the operator 
\begin{equation}
W=\mathbbm1\otimes\mathbbm1+2\mathbbm1\otimes\phi^{+}+2\phi^{+}\otimes\mathbbm1-8\phi^{+}\otimes\phi^{+}\label{eq:witpi-1}
\end{equation}
is a GME witness and detects $\rho_{p_{1}}\otimes\rho_{p_{2}}$ for
the stated bounds.

To show $W$ is a witness, it suffices to show that $\tr(W\rho)\geq0$
for every $\rho$ that is a PPT mixture. In turn, for this it is enough
to see that there exist $P_{M},Q_{M}\succcurlyeq0$ such that $W=P_{M}+Q_{M}^{\Gamma_{M}}$
for $M=A,B,C$ \cite{jungnitsch_taming_2011}. It is straightforward
to verify that this is indeed the case (we only need to use that the
partial transpose of the flip operator $\Pi$ is twice the maximally
entangled state: $\Pi^{\Gamma}=2\phi^{+},\phi^{+}:=|\phi^{+}\rangle\langle\phi^{+}|$)
if:

\begin{align}
P_{A} & =2\phi^{+}\otimes(\mathbbm1-\phi^{+})+2(\mathbbm1-\phi^{+})\otimes\phi^{+},\\
Q_{A} & =\frac{1}{2}[(\mathbbm1-\Pi)\otimes(\mathbbm1+\Pi)+(\mathbbm1+\Pi)\otimes(\mathbbm1-\Pi)],\\
P_{B} & =0,\quad Q_{B}=(\mathbbm1+\Pi)\otimes(\mathbbm1-\phi^{+})+3(\mathbbm1-\Pi)\otimes\phi^{+},\\
P_{C} & =0,\quad Q_{C}=(\mathbbm1-\phi^{+})\otimes(\mathbbm1+\Pi)+3\phi^{+}\otimes(\mathbbm1-\Pi),
\end{align}
where we have used that $\phi^{+},\mathbbm1-\phi^{+},\mathbbm1\pm\Pi,\succcurlyeq0$
and that the sum and tensor product of positive semidefinite matrices
is positive semidefinite.

Next, we have
\begin{equation}
\tr(W\rho_{p_{1}}\otimes\rho_{p_{2}})=\frac{3}{2}\left(1-3p_{1}p_{2}\right),
\end{equation}
which is strictly smaller than zero whenever $p_{1}p_{2}>1/3$. Since
we must also have $p_{1},p_{2}\leq1$, the stated bounds follow.
\end{proof}
\begin{thm}
\label{thm:lambda-iso-BS}A Lambda network of $d$-dimensional isotropic
states $\rho_{p}$ is biseparable for $p\leq\left[(1+\sqrt{2})d-1\right]/(d^{2}+2d-1)$.
\end{thm}

\begin{proof}
Consider a tripartite network where Alice and each of Bob and Charlie
share an entangled $\rho_{p},$ i.e., let $p>1/(d+1).$ The state
of the network is
\begin{equation}
\begin{aligned}\rho_{p,A_{1}B} & \otimes\rho_{p,A_{2}C}\\
=p^{2} & \phi_{A_{1}B}^{+}\otimes\phi_{A_{2}C}^{+}+p(1-p)\phi_{A_{1}B}^{+}\otimes\tilde{\mathbbm1}_{A_{2}C}+p(1-p)\tilde{\mathbbm1}_{A_{1}B}\otimes\phi_{A_{2}C}^{+}+(1-p)^{2}\tilde{\mathbbm1}_{A_{1}B}\otimes\tilde{\mathbbm1}_{A_{2}C},
\end{aligned}
\end{equation}
where $\tilde{\mathbbm1}=\mathbbm1/4$ is the normalised identity,
and which can be rewritten as
\begin{equation}
\begin{aligned}\rho_{p,A_{1}B}\otimes & \rho_{p,A_{2}C}=(1-q)p^{2}\phi_{A_{1}B}^{+}\otimes\phi_{A_{2}C}^{+}+(1-p)^{2}\tilde{\mathbbm1}_{A_{1}B}\otimes\tilde{\mathbbm1}_{A_{2}C}\\
 & +\phi_{A_{1}B}^{+}\otimes\left(\frac{qp^{2}}{2}\phi_{A_{2}C}^{+}+p(1-p)\tilde{\mathbbm1}_{A_{2}C}\right)+\left(\frac{qp^{2}}{2}\phi_{A_{1}B}^{+}+p(1-p)\tilde{\mathbbm1}_{A_{1}B}\right)\otimes\phi_{A_{2}C}^{+}
\end{aligned}
\label{eq:lambdaisoBS}
\end{equation}
for any $q\in[0,1].$ Now, the first line can be seen as an (unnormalised)
isotropic state in parties $A|BC$ (with dimension $d^{2}$), while
the second line contains isotropic states in $A_{2}|C$ and $A_{1}|B$
respectively (with dimension $d$). Showing that each of these states
is separable will entail the result.

Denote these isotropic states by $\sigma_{0},\sigma_{1},\sigma_{2}$
in the order that they appear in equation (\ref{eq:lambdaisoBS}).
Normalising $\sigma_{0},$ we find
\begin{equation}
\sigma_{0}=\frac{(1-q)p^{2}\phi_{A_{1}B}^{+}\otimes\phi_{A_{2}C}^{+}+(1-p)^{2}\tilde{\mathbbm1}_{A_{1}B}\otimes\tilde{\mathbbm1}_{A_{2}C}}{(1-q)p^{2}+(1-p)^{2}},
\end{equation}
meaning it is separable in $A|BC$ if
\begin{equation}
\frac{(1-q)p^{2}}{(1-q)p^{2}+(1-p)^{2}}\leq\frac{1}{d^{2}+1},
\end{equation}
i.e., if
\begin{equation}
q\geq1-\frac{(1-p)^{2}}{p^{2}d^{2}}.
\end{equation}
Normalising $\sigma_{1},$ we obtain
\begin{equation}
\sigma_{1}=\frac{qp^{2}\phi_{A_{2}C}^{+}+2p(1-p)\tilde{\mathbbm1}_{A_{2}C}}{qp^{2}+2p(1-p)},
\end{equation}
which is separable if
\begin{equation}
\frac{qp^{2}}{qp^{2}+2p(1-p)}\leq\frac{1}{d+1}.
\end{equation}
Simplifying, this entails that
\begin{equation}
q\leq\frac{2-2p}{pd}.
\end{equation}
Reasoning symmetically, the separability of $\sigma_{2}$ gives the
same bound.

Both bounds on $q$ together entail that
\begin{equation}
1-\frac{(1-p)^{2}}{p^{2}d^{2}}\leq\frac{2-2p}{pd}
\end{equation}
and, solving for $p,$ we find that $\rho_{p,A_{1}B}\otimes\rho_{p,A_{2}C}$
is biseparable for
\begin{equation}
p\leq\frac{(1+\sqrt{2})d-1}{d^{2}+2d-1}.
\end{equation}
In particular, for any $d$ there exists $p$ such that
\begin{equation}
\frac{1}{d+1}<p\leq\frac{(1+\sqrt{2})d-1}{d^{2}+2d-1},
\end{equation}
showing that entanglement can be deactivated in a Lambda network.
\end{proof}
In the particular case where the states on the edges of a Lambda network
are the same, $\rho_{p_{1}}=\rho_{p_{2}}=\rho_{p}$, and for $d=2$,
Theorem \ref{thm:lambda-iso-gme} implies that the network is GME
for $p>1/\sqrt{3}\simeq0.577$, while Theorem \ref{thm:lambda-iso-BS}
implies that the network is biseparable for $p\leq(1+2\sqrt{2})/7\simeq0.547$.
Adding an extra edge to the network so that it forms a triangle makes
it possible to achieve GME with less entanglement on the edges, as
we now show.
\begin{thm}
\label{thm:triangle-GME}A triangle network of 2-dimensional states
$\rho_{p}$ is GME for $p>(2\sqrt{5}-3)/3\simeq0.491$.
\end{thm}

\begin{proof}
We will show the operator 
\begin{equation}
\begin{aligned}W= & \mathbbm1\otimes\mathbbm1\otimes\phi^{+}+\mathbbm1\otimes\phi^{+}\otimes\mathbbm1+\phi^{+}\otimes\mathbbm1\otimes\mathbbm1\\
 & -\mathbbm1\otimes\phi^{+}\otimes\phi^{+}-\phi^{+}\otimes\phi^{+}\otimes\mathbbm1-\phi^{+}\otimes\mathbbm1\otimes\phi^{+}-3\phi^{+}\otimes\phi^{+}\otimes\phi^{+},
\end{aligned}
\end{equation}
(where the Hilbert spaces are ordered as $A_{1}B_{1}A_{2}C_{1}B_{2}C_{2}$)
is a witness and detects the triangle with $\rho_{p}$ at each edge,
for all $p>(2\sqrt{5}-3)/3.$ To show $W$ is a witness, it is sufficient
to show that it can be decomposed as
\begin{equation}
W=P_{M}+Q_{M}^{\Gamma_{M}}
\end{equation}
for each bipartition $M=A,B,C$, where $P_{M},Q_{M}\succcurlyeq0$
for all $M.$ Indeed, we have
\begin{equation}
\begin{aligned}P_{A} & =\mathbbm1\otimes\phi^{+}\otimes(\mathbbm1-\phi^{+})+\phi^{+}\otimes(\mathbbm1-\phi^{+})\otimes(\mathbbm1-\phi^{+})\\
Q_{A} & =\frac{1}{2}[(\mathbbm1-\Pi)\otimes(\mathbbm1+\Pi)+(\mathbbm1+\Pi)\otimes(\mathbbm1-\Pi)]\otimes\phi^{+}\\
P_{B} & =\mathbbm1\otimes(\mathbbm1-\phi^{+})\otimes\phi^{+}+\phi^{+}\otimes(\mathbbm1-\phi^{+})\otimes(\mathbbm1-\phi^{+})\\
Q_{B} & =\frac{1}{2}[(\mathbbm1-\Pi)\otimes(\mathbbm1+\Pi)+(\mathbbm1+\Pi)\otimes(\mathbbm1-\Pi)]_{A_{1}B_{1}B_{2}C_{2}}\otimes\phi_{A_{2}C_{1}}^{+}\\
P_{C} & =(\mathbbm1-\phi^{+})\otimes\phi^{+}\otimes\mathbbm1+(\mathbbm1-\phi^{+})\otimes(\mathbbm1-\phi^{+})\otimes\phi^{+}\\
Q_{C} & =\phi^{+}\otimes\frac{1}{2}[(\mathbbm1-\Pi)\otimes(\mathbbm1+\Pi)+(\mathbbm1+\Pi)\otimes(\mathbbm1-\Pi)],
\end{aligned}
\end{equation}
where the Hilbert spaces of all operators are ordered as $A_{1}B_{1}A_{2}C_{1}B_{2}C_{2}$,
except $Q_{B}$, and where we use that $\mathbbm1,\phi^{+},\mathbbm1-\phi^{+},\mathbbm1\pm\Pi\succcurlyeq0.$

Then,
\begin{equation}
\tr(W\rho_{p}^{\otimes3})=\frac{3}{64}(11+15p-63p^{2}-27p^{3}),
\end{equation}
which is strictly smaller than zero when $p>(2\sqrt{5}-3)/3$.
\end{proof}
Since the visibilities required for biseparability or GME are different
in the case of a Lambda network and a triangle network, we find that,
unlike in the case of pure states, topology does influence the entanglement
of mixed-state networks. Indeed, in the case of $d=2$, a Lambda network
where $p_{1}=p_{2}=:p$ is biseparable for $p\leq(1+2\sqrt{2})/7\simeq0.547$
(by Theorem \ref{thm:lambda-iso-BS}), while a triangle network is
GME for $p>(2\sqrt{5}-3)/3\simeq0.491$ (by Theorem \ref{thm:triangle-GME}).
This means that adding extra connections in a network has an impact
on its entanglement. If, instead, we hold the number of shared states
fixed, but distribute them in networks of different numbers of parties
(as long as the networks are kept connected), the presence of cycles
makes it possible to achieve GME with less visibility on the edges.
Indeed, adding an extra party to a Lambda network such that the one
in Theorem \ref{thm:lambda-iso-BS}, so that $AB$, $BC$ and $CD$
each share an isotropic state, does not change the biseparability
bound. This is because a biseparable state can be decomposed into
states that are separable along a given bipartition. So the biseparable
decomposition of the Lambda network works also for the network with
an added party, substituting $C$ for $CD$ in each term.

In a similar pattern, a triangle network is biseparable for a smaller
range of visibilities than the triangle network, as we now show.
\begin{thm}
\label{thm:triangle-BS}A triangle network with a $d$-dimensional
isotropic state $\rho_{p}$ at each edge is biseparable for $p\leq3/(3+2d)$.
\end{thm}

\begin{proof}
We show that the state of the triangle network can be decomposed into
four matrices, three of which are separable along one bipartition
each, and the fourth of which is fully separable. Let $\rho_{p}^{\otimes3}$
be the state of the network. Then,
\begin{equation}
\rho_{p}^{\otimes3}=\frac{p(p^{2}-3p+3)}{3}\left(\sigma_{1}+\sigma_{2}+\sigma_{3}\right)+(1-p)^{3}\tilde{\mathbbm1}\otimes\tilde{\mathbbm1}\otimes\tilde{\mathbbm1},
\end{equation}
where
\begin{equation}
\begin{aligned} & \sigma_{1}=\\
 & \frac{p^{2}\phi^{+}\otimes\phi^{+}\otimes\phi^{+}+\left(3p(1-p)/2\right)(\phi^{+}\otimes\tilde{\mathbbm1}\otimes\phi^{+}+\tilde{\mathbbm1}\otimes\phi^{+}\otimes\phi^{+})+3(1-p)^{2}\tilde{\mathbbm1}\otimes\tilde{\mathbbm1}\otimes\phi^{+}}{p^{2}-3p+3}\\
 & \sigma_{2}=\\
 & \frac{p^{2}\phi^{+}\otimes\phi^{+}\otimes\phi^{+}+\left(3p(1-p)/2\right)(\phi^{+}\otimes\phi^{+}\otimes\tilde{\mathbbm1}+\tilde{\mathbbm1}\otimes\phi^{+}\otimes\phi^{+})+3(1-p)^{2}\tilde{\mathbbm1}\otimes\phi^{+}\otimes\tilde{\mathbbm1}}{p^{2}-3p+3}\\
 & \sigma_{3}=\\
 & \frac{p^{2}\phi^{+}\otimes\phi^{+}\otimes\phi^{+}+\left(3p(1-p)/2\right)(\phi^{+}\otimes\tilde{\mathbbm1}\otimes\phi^{+}+\phi^{+}\otimes\phi^{+}\otimes\tilde{\mathbbm1})+3(1-p)^{2}\phi^{+}\otimes\tilde{\mathbbm1}\otimes\tilde{\mathbbm1}}{p^{2}-3p+3}.
\end{aligned}
\end{equation}
Clearly, the matrix $\tilde{\mathbbm1}\otimes\tilde{\mathbbm1}\otimes\tilde{\mathbbm1}$
is fully separable. We will show that $\sigma_{1}$ is separable in
$A|BC$, and separability of $\sigma_{2}$ and $\sigma_{3}$ in $B|AC$
and $C|AB$ respectively will follow by symmetry. We have that
\begin{equation}
\sigma_{1}=\tau_{1}\otimes\phi_{B_{2}C_{2}}^{+},
\end{equation}
where showing that
\begin{equation}
\begin{aligned}\tau_{1} & =\\
 & \frac{p^{2}\phi_{A_{1}B_{1}}^{+}\otimes\phi_{A_{2}C_{1}}^{+}+\left(3p(1-p)/2\right)(\phi_{A_{1}B_{1}}^{+}\otimes\tilde{\mathbbm1}_{A_{2}C_{1}}+\tilde{\mathbbm1}_{A_{1}B_{1}}\otimes\phi_{A_{2}C_{1}}^{+})}{p^{2}-3p+3}\\
 & +\frac{3(1-p)^{2}\tilde{\mathbbm1}_{A_{1}B_{1}}\otimes\tilde{\mathbbm1}_{A_{2}C_{1}}}{p^{2}-3p+3}
\end{aligned}
\end{equation}
is separable in $A_{1}A_{2}|B_{1}C_{1}$ is sufficient to show that
$\sigma_{1}$ is separable in $A|BC$. Indeed, we can write
\begin{equation}
\begin{aligned}\tau_{1}= & \frac{(3-p)^{2}}{4(3-3p+p^{2})}\left(\frac{2p}{3-p}\phi_{A_{1}B_{1}}^{+}+\frac{3(1-p)}{3-p}\tilde{\mathbbm1}_{A_{1}B_{1}}\right)\otimes\left(\frac{2p}{3-p}\phi_{A_{2}C_{1}}^{+}+\frac{3(1-p)}{3-p}\tilde{\mathbbm1}_{A_{2}C_{1}}\right)\\
 & +\frac{3(1-p)^{2}}{4(3-3p+p^{2})}\tilde{\mathbbm1}_{A_{1}B_{1}}\otimes\tilde{\mathbbm1}_{A_{2}C_{1}}.
\end{aligned}
\label{eq:tauisotropic}
\end{equation}
The isotropic state
\begin{equation}
\frac{2p}{3-p}\phi_{A_{1}B_{1}}^{+}+\frac{3(1-p)}{3-p}\tilde{\mathbbm1}_{A_{1}B_{1}}
\end{equation}
is separable whenever $2p/(3-p)\leq1/(d+1),$ i.e., $p\leq3/(3+2d),$
therefore $\tau_{1}$ is fully separable in $A_{1}|A_{2}|B_{1}|C_{1}$
which guarantees the required separability of $\sigma_{1}.$ Reasoning
symmetrically, separability of $\sigma_{2}$ in $B|AC$ and of $\sigma_{3}$
in $C|AB$ follows for the same values of $p$, hence $\rho_{p}^{\otimes3}$
is biseparable for the stated bounds.
\end{proof}
The results of Theorems \ref{thm:lambda-iso-gme}-\ref{thm:triangle-BS}
are summarised in Table \ref{tab:lambdatriangle}.

\begin{table}[ht]
\centering %
\begin{tabular}{c|cc}
 & biseparable for $p\leq$ & GME for $p>$\tabularnewline
\hline 
$\Lambda$ & $(1+2\sqrt{2})/7\simeq0.547$ & $1/\sqrt{3}\simeq0.577$\tabularnewline
$\triangle$ & $3/7\simeq0.429$ & $(2\sqrt{5}-3)/3\simeq0.491$\tabularnewline
\end{tabular}\caption{Bounds for biseparability and GME in a Lambda network ($\Lambda$)
and a triangle network ($\triangle$) where the shared states are
isotropic states with visibility $p$ in dimension 2.}
\label{tab:lambdatriangle} 
\end{table}

We now explore the generalisation of these results to the case of
larger networks, where we find the dependency on the topology is most
extreme. Tree networks are those which contain no cycles, and are
thus the networks with fewest edges, $n-1$, for a fixed number of
parties $n$. Polygonal networks, i.e., networks in the form of a
closed chain, have only one more edge. We find that, if the number
of parties is large enough, distributing isotropic states on either
of these networks renders them biseparable no matter the visibility
(as long as it is $<1$). In sharp contrast, if the network is completely
connected (i.e., every pair of parties shares an isotropic state),
it remains GME for any number of parties for all visibilities above
a threshold.

In what follows, notation for parties will depend on the geometry
of the network. In general, parties will be denoted by $A_{i}$, but
we will also use $B_{i}$ whenever parties can be naturally divided
into two types (for example, in the case of a star network, with the
central node vs all others).
\begin{thm}
\label{thm:tree-BS}For any $p\in[0,1),$ there exists $K\in\mathbb{N}$
such that distributing an isotropic state with visibility $p$ in
any tree network of $K$ edges yields a biseparable state.
\end{thm}

\begin{proof}
The main idea is to write the state of the network as a tensor product
of the state of the edges, and observe that, since each bipartition
is crossed by exactly one edge, all terms with $\tilde{\mathbbm1}$
on at least one edge are biseparable. Then, the only GME term is the
one where all edges contain $\phi_{d}^{+}$. Distributing this term
among the terms with exactly one $\tilde{\mathbbm1}$ gives terms
with $\phi_{d}^{+}$ on all but one edge, and a mixture of $\tilde{\mathbbm1}$
and $\phi_{d}^{+}$ on the remaining edge where the weight of $\tilde{\mathbbm1}$
is inversely proportional to the number of edges. Thus, for a sufficiently
large number of edges, this term can be made biseparable too, proving
the result.

Consider a network in the form of a tree graph where each edge has
a copy of the state $\rho_{p}.$ Let $K$ be the number of edges in
the network, where the edges are indexed by $i$, and denote the state
of the network as $\rho_{p}^{\otimes K}.$ Expanding the tensor product,
we find
\begin{equation}
\rho_{p}^{\otimes K}=p^{K}\bigotimes_{i=1}^{K}\phi_{d,i}^{+}+p^{K-1}(1-p)\sum_{i=1}^{K}\tilde{\mathbbm1}_{i}\otimes\bigotimes_{j\neq i}\phi_{d,j}^{+}+\dots\,,\label{eq:treeBS}
\end{equation}
where the omitted terms are all separable along at least one bipartition,
since each bipartition is crossed by exactly one edge and at least
one edge in each term contains $\tilde{\mathbbm1}.$ Showing that
the above expression is biseparable for some $K$ is sufficient to
prove the claim. But we can rewrite the above as
\begin{equation}
\begin{aligned}\rho_{p}^{\otimes K} & =\frac{p^{K}}{K}\sum_{i=1}^{K}\bigotimes_{j=1}^{K}\phi_{d,j}^{+}+p^{K-1}(1-p)\sum_{i=1}^{K}\tilde{\mathbbm1}_{i}\otimes\bigotimes_{j\neq i}\phi_{d,j}^{+}+\dots\\
 & =\sum_{i=1}^{K}\left(\frac{p^{K}}{K}\phi_{d,i}^{+}+p^{K-1}(1-p)\tilde{\mathbbm1}_{i}\right)\otimes\bigotimes_{j\neq i}\phi_{d,j}^{+}+\dots\,.
\end{aligned}
\end{equation}
Here, each bracket has $\phi_{d}^{+}$ on $K-1$ edges, and the state
\begin{equation}
\frac{(p^{K}/K)\phi_{d}^{+}+p^{(K-1)}(1-p)\tilde{\mathbbm1}}{p^{K}/K+p^{(K-1)}(1-p)}=\frac{(p/K)\phi_{d}^{+}+(1-p)\tilde{\mathbbm1}}{p/K+1-p}\label{eq:treeBS-seppart}
\end{equation}
on the rest. But this is an isotropic state with visibility $(p/K)/(p/K+1-p)$,
which is thus guaranteed to become separable when the visibility is
smaller than or equal to $1/(d+1)$. For fixed $p,$ this can be achieved
by choosing $K\geq dp/(1-p).$ This bound, however, is not optimal,
as lower $K$ could be achieved by distributing the term $\bigotimes_{i=1}^{K}\phi_{d,i}^{+}$
among some or all of the omitted terms in equation (\ref{eq:treeBS})
as well, as in the proof of Theorem \ref{thm:lambda-iso-BS}.
\end{proof}
In fact, Theorem \ref{thm:tree-BS} holds for a more general class
of states, namely, convex mixtures of an entangled state and a separable
state that is not on the boundary of the set of separable states,
since the proof also holds for such states.

\begin{thm}
\label{thm:polygon-BS}For any $p\in[0,1)$, there exists $K\in\mathbb{N}$
such that distributing an isotropic state with visibility $p$ in
a polygonal network of $K$ edges yields a biseparable state.
\end{thm}

\begin{proof}
The proof is very similar to that of Theorem \ref{thm:tree-BS}. This
time, terms with $\tilde{\mathbbm1}$ on only one edge are not biseparable,
but those with $\tilde{\mathbbm1}$ on two or more edges are. Therefore,
we distribute the term where all edges contain $\phi_{d}^{+}$ among
the terms where two edges contain $\tilde{\mathbbm1}$, and the terms
with one $\tilde{\mathbbm1}$ among the terms with three or more.
Like in the case of a tree network, the weight of $\phi_{d}^{+}$
in each case decreases with the number of edges in the network, proving
the claim.

Consider a polygonal network where each edge has a copy of the state
$\rho_{p}=p\phi_{d}^{+}+(1-p)\tilde{\mathbbm1}$. Let $A_{i}$ denote
the parties, for each $i\in[K]$. Let $i$ also index the edge to
the right of party $A_{i}$, as well as the state on that edge. Denote
the state of the network as $\rho_{p}^{\otimes K}$. Expanding the
tensor product, we find
\begin{equation}
\begin{aligned}\rho_{p}^{\otimes K}= & p^{K}\bigotimes_{i=1}^{K}\phi_{d,i}^{+}+p^{K-1}(1-p)\sum_{i=1}^{K}\tilde{\mathbbm1}_{i}\otimes\bigotimes_{j\neq i}\phi_{d,j}^{+}+p^{K-2}(1-p)^{2}\sum_{\substack{i,j=1\\
i\neq j
}
}^{K}\tilde{\mathbbm1}_{i}\otimes\tilde{\mathbbm1}_{j}\otimes\bigotimes_{k\neq i,j}\phi_{d,k}^{+}+\dots\,,\end{aligned}
\label{eq:polygonBS}
\end{equation}
where all terms with two or more edges containing $\tilde{\mathbbm1}$
are separable along at least one bipartition, since each bipartition
of the polygon is crossed by exactly two edges. We will show that
the terms where fewer than two edges contain $\tilde{\mathbbm1}$
can be paired with separable terms in order to write $\rho_{p}^{\otimes K}$
as a convex mixture of separable states. Out of the terms containing
$\tilde{\mathbbm1}$ on two edges, there are $K$ such terms containing
$\tilde{\mathbbm1}_{i}\otimes\tilde{\mathbbm1}_{i+1}$ (identifying
$K+1\equiv1$) for some $i\in K$, i.e., where $\tilde{\mathbbm1}$
lies on two adjacent edges. This means that the term is separable
in $A_{i+1}|\{A_{j}\}_{j\neq i+1}$. Pairing these $K$ terms with
$p^{K}\bigotimes_{i=1}^{K}\phi_{d,i}^{+}$, we can write a fragment
of $\rho_{p}^{\otimes K}$ as
\begin{equation}
\begin{aligned} & p^{K}\bigotimes_{i=1}^{K}\phi_{d,i}^{+}+p^{K-2}(1-p)^{2}\left(\sum_{i=1}^{K}\tilde{\mathbbm1}_{i}\otimes\tilde{\mathbbm1}_{i+1}\otimes\bigotimes_{j\neq i,i+1}\phi_{d,j}^{+}\right)\\
 & =p^{K-2}\sum_{i=1}^{K}\left[\left(\frac{p^{2}}{K}\phi_{d,i}^{+}\otimes\phi_{d,i+1}^{+}+(1-p)^{2}\tilde{\mathbbm1}_{i}\otimes\tilde{\mathbbm1}_{i+1}\right)\otimes\bigotimes_{j\neq i,i+1}\phi_{d,j}^{+}\right].
\end{aligned}
\label{eq:fragment1}
\end{equation}
If, for each $i,$ the state
\begin{equation}
\frac{p^{2}}{K}\phi_{d,i}^{+}\otimes\phi_{d,i+1}^{+}+(1-p)^{2}\tilde{\mathbbm1}_{i}\otimes\tilde{\mathbbm1}_{i+1},\label{eq:iso2}
\end{equation}
once normalised, is separable in $A_{i+1}|\{A_{j}\}_{j\neq i+1},$
then the fragment of $\rho_{p}^{\otimes K}$ in equation (\ref{eq:fragment1})
will be biseparable.

Now, the state in equation\textbf{ }(\ref{eq:iso2}) is a convex mixture
of $\phi_{d,i}^{+}\otimes\phi_{d,i+1}^{+}$ and $\tilde{\mathbbm1}_{i}\otimes\tilde{\mathbbm1}_{i+1}$.
Therefore, the state (\ref{eq:iso2}) is guaranteed to become separable
when\textbf{ }$(1-p)^{2}/(p^{2}/K+(1-p)^{2})$ is close enough to
1. For fixed $p,$ this can be achieved by choosing a large enough
$K.$

Using a similar strategy, the terms containing $\tilde{\mathbbm1}$
on one edge can be combined with some of those containing $\tilde{\mathbbm1}$
on three appropriately chosen edges. Thus, assuming that $K>4$, another
fragment of $\rho_{p}^{\otimes K}$ can be written as
\begin{equation}
\begin{aligned} & p^{K-1}(1-p)\tilde{\mathbbm1}_{i}\otimes\bigotimes_{j\ne i}\phi_{d,j}^{+}+p^{K-3}(1-p)^{3}\sum_{\substack{j\neq i,i\pm1,\\
i-2
}
}\tilde{\mathbbm1}_{i}\otimes\tilde{\mathbbm1}_{j}\otimes\tilde{\mathbbm1}_{j+1}\otimes\bigotimes_{k\neq i,j,j+1}\phi_{d,k}^{+}\\
 & =p^{K-3}(1-p)\sum_{\substack{j\neq i,i\pm1,\\
i-2
}
}\left[\tilde{\mathbbm1}_{i}\otimes\bigotimes_{k\neq i,j,j+1}\phi_{d,k}^{+}\otimes\left(\frac{p^{2}}{K-4}\phi_{d,j}^{+}\otimes\phi_{d,j+1}^{+}+(1-p)^{2}\tilde{\mathbbm1}_{j}\otimes\tilde{\mathbbm1}_{j+1}\right)\right].
\end{aligned}
\label{eq:fragment2}
\end{equation}
Hence, it is sufficient to show that the state
\begin{equation}
\frac{p^{2}}{K-4}\phi_{d,j}^{+}\otimes\phi_{d,j+1}^{+}+(1-p)^{2}\tilde{\mathbbm1}_{j}\otimes\tilde{\mathbbm1}_{j+1},
\end{equation}
once normalised, is separable in $A_{j+1}|\{A_{k}\}_{k\neq j+1}$
to deduce that the fragment in equation (\ref{eq:fragment2}) is separable.
Again, for fixed $p$, this is guaranteed for large enough $K$. Since
every term that does not appear in fragments (\ref{eq:fragment1})
and (\ref{eq:fragment2}) is already biseparable, the claim follows.
\end{proof}
At the other extreme of connected networks, we explore completely
connected networks, which are those where every pair of parties shares
a bipartite state. It is to be expected that their genuine multipartite
entanglement is more robust to noise than in the case of tree or polygonal
networks. In fact, the contrast with tree and polygonal networks is
as sharp as it can be: we find that a completely connected network
of isotropic states remains GME for any number of parties for all
visibilities above a threshold. We first show the result for all visibilities
above a threshold if the number of parties is sufficiently large.
As a consequence, we show that, by considering high enough visibilities,
GME can be achieved for completely connected networks of any number
of edges.
\begin{thm}
\label{thm:completegraph} There exists $p_{1}<1$ such that, for
all $n\in\mathbb{N}$, a completely connected network of $n$ parties
where each pair of parties shares an isotropic state $\rho_{p}$ of
any dimension $d$ is GME for all $p_{1}<p\leq1$.
\end{thm}

The proof strategy is to show that, in the limit of large $n$, a
completely connected network of $n$ parties where each pair of parties
shares an isotropic state $\rho_{p}$ can be used to distill maximally
entangled states between any pair of parties with a fidelity unachievable
by any biseparable state of $n$ parties.

We first consider an LOCC protocol acting on a biseparable state of
$n$ parties, which distills maximally entangled states between any
given pair of parties. In Lemma \ref{lem:fidelity-bs}, we find an
upper bound for the fidelity of any such protocol, added over all
distinct pairs of parties. Next, to prove the Theorem, we consider
a specific LOCC protocol that uses the state of the network to distill
maximally entangled states between any pair of parties, and lower
bound its fidelity added over all pairs of parties. We find the bound
corresponding to the state of the network is strictly larger than
that of any biseparable state, proving the claim in the limit of large
$n$. Once GME is ensured to persist for a large number of parties,
it follows that, for a fixed, large enough visibility, GME can be
guaranteed for completely connected networks of any size.
\begin{lem}
\label{lem:fidelity-bs}Consider an LOCC transformation that maps
an $n$-partite biseparable state of any dimension $d$ to a 2-qubit
maximally entangled state shared by any two parties. The achievable
fidelity of this transformation, added over all distinct pairs of
parties, is bounded above by $(n-1)^{2}/2$.
\end{lem}

\begin{proof}
Let $\chi$ be an $n$-partite biseparable state. Consider an LOCC
protocol $\Lambda_{ij}:\mathcal{B}(\bigotimes_{i=1}^{n}\mathcal{H}_{i})\rightarrow\mathcal{B}(\mathcal{H}_{i})\otimes\mathcal{B}(\mathcal{H}_{j})$
that maps $n$-partite states to bipartite states shared between parties
$A_{i},A_{j}$, where $i<j\in[n]$. Since $\chi$ is biseparable,
we have
\begin{equation}
\chi=\sum_{M}p_{M}\chi_{M},
\end{equation}
where each $\chi_{M}$ is separable across the bipartition $M|\overline{M}$
and $\sum_{M}p_{M}=1$. Let $M_{ij}$ be the set of bipartitions that
split parties $A_{i}$ and $A_{j}$. Then, we can write the evolution
of $\chi$ under the protocol $\Lambda_{ij}$ as
\begin{equation}
\Lambda_{ij}(\chi)=\sum_{M\not\in M_{ij}}p_{M}\tau_{M}(i,j)+\sum_{M\in M_{ij}}p_{M}\sigma_{M}(i,j),
\end{equation}
where $\tau_{M}(i,j)$ and $\sigma_{M}(i,j)$ are bipartite states
of parties $A_{i},A_{j}$. The state $\tau_{M}(i,j)$ is in principle
unrestricted, so it can have up to unit fidelity with the 2-qubit
maximally entangled state $\phi_{ij}^{+}$. However, since LOCC operations
cannot create entanglement, $\sigma_{M}(i,j)$ must be separable,
therefore its fidelity with $\phi_{ij}^{+}$ cannot be larger than
$1/2$. Therefore, for any LOCC protocol $\Lambda_{ij}$, the fidelity
$F$ of $\Lambda_{ij}(\chi)$ with $\phi_{ij}^{+}$ is bounded above:
\begin{equation}
F(\Lambda_{ij}(\chi),\phi_{ij}^{+})\leq\sum_{M\not\in M_{ij}}p_{M}+\frac{1}{2}\sum_{M\in M_{ij}}p_{M}=1-\frac{1}{2}\sum_{M\in M_{ij}}p_{M}.
\end{equation}
Summing over all parties $i<j$, we obtain
\begin{equation}
\sum_{i<j}F(\Lambda_{ij}(\chi),\phi_{ij}^{+})\leq\frac{n(n-1)}{2}-\frac{n-1}{2}=\frac{(n-1)^{2}}{2}.\label{eq:fidelityBS}
\end{equation}
The first term is the number of distinct pairs $i<j$, while the second
comes from the following observation: the sums over $i<j$ and $M\in M_{ij}$
run through all bipartitions $M$, more than once. In fact, the number
of times each bipartition is counted is equal to the number of times
each bipartition is crossed by the edge connecting $i$ and $j$.
In turn, this number is equal to the number of edges crossing each
bipartition. A bipartition splitting $k$ parties from the remaining
$(n-k)$ is crossed by $k(n-k)$ edges, which is smallest when $k=1$.
That is, each bipartition $M\in M_{ij}$ appears at least $(n-1)$
times, and so the sum $\sum_{M}p_{M}$, running over \emph{all }$M$,
appears at least $(n-1)$ times. That is,
\begin{equation}
\sum_{i<j}\sum_{M\in M_{ij}}p_{M}\geq n-1.
\end{equation}
\end{proof}
\begin{proof}[Proof of Theorem \ref{thm:completegraph}]
Let $\rho$ be the state of the complete graph of $n$ parties, with
isotropic states $\rho_{p}$ on each edge. In the complete graph,
for all $i,j$, each party $A_{k},k\neq i,j$ shares a copy of $\rho_{p}$
with party $A_{i}$, which we denote $\rho_{p}(i,k)$, and another
with party $A_{j}$, denoted $\rho_{p}(j,k)$. The protocol starts
with each $A_{k}$, $k\neq i,j$, teleporting their half of $\rho_{p}(i,k)$
to $A_{j}$ by using the channel $\rho_{p}(j,k)$. This is a noisy
version of the standard teleportation protocol \cite{bennett_teleporting_1993}.
The teleported state will be a mixture of four terms, namely the four
combinations of teleporting half of a maximally entangled or maximally
mixed state along a maximally entangled or a maximally mixed channel.
The first term, with weight $p^{2}$, will give a maximally entangled
state, the other three turn out to give a maximally mixed state (as
can be simply checked by performing the calculations in the standard
protocol, replacing the teleported state and/or the channel by an
identity in each case). Therefore, the teleportation protocol yields
\begin{equation}
\rho_{p^{2}}(i,j)=p^{2}\phi_{d}^{+}+(1-p^{2})\tilde{\mathbbm1}.
\end{equation}
In fact, parties $A_{i},A_{j}$ end up sharing $(n-2)$ copies of
this state, one coming from each party $A_{k},k\neq i,j$. Parties
$A_{i},A_{j}$ can now apply a distillation protocol $D_{ij}$ to
obtain something close to a maximally entangled state, whose fidelity
approaches 1 in the limit of large $n$. More specifically,
\begin{equation}
F(\Lambda_{ij}(\rho),\phi_{ij}^{+})=F(D_{ij}(\rho_{p^{2}}^{\otimes(n-2)}),\phi_{ij}^{+})\geq1-\varepsilon_{n},\label{eq:fidelityGME}
\end{equation}
where $\varepsilon_{n}\rightarrow0$ as $n\rightarrow\infty$. To
compare to equation (\ref{eq:fidelityBS}) in Lemma \ref{lem:fidelity-bs},
there only remains to show that $\varepsilon_{n}\rightarrow0$ sufficiently
fast as $n$ grows.

Ref. \cite{bennett_mixed-state_1996} showed that a one-way distillation
protocol acting on isotropic states and having rate $R$ is equivalent
to a quantum error-correcting code on a depolarising channel with
the same rate $R$. In turn, Ref. \cite{hamada_exponential_2002}
proved a lower bound on the fidelity $F$ of a $d$-dimensional quantum
error-correcting code of rate $R$ acting on a certain class of memoryless
channels, which, in particular, include depolarising channels. After
$n$ uses of the channel, the lower bound\footnote{Notice that the definition of the fidelity $F$ used in Ref. \cite{hamada_exponential_2002}
is the square root of the one used in this work (see equation (\ref{eq:intro-fidelity}),
hence the factor of 2 appears here in front of $(n+1)$.} is
\begin{equation}
F\geq1-2(n+1)^{2(d^{2}-1)}d^{-nE}.
\end{equation}
Here, $E$ is a function of the rate $R$ and the noise parameter
of the depolarising channel, which, in turn, corresponds to a function
of the noise parameter $p$ of the isotropic state. It holds that
$E>0$ when the rate is strictly below the maximum achievable rate
of the channel. By the correspondence with distillation, this entails
that $\varepsilon_{n}$ in equation (\ref{eq:fidelityGME}) goes to
zero exponentially fast for one-way distillable isotropic states if
the rate is suboptimal. Since, for our protocol $\Lambda_{ij}$, we
are only interested in obtaining one copy of $\phi^{+}$, we can achieve
this exponential decay. Therefore,
\begin{equation}
\sum_{i<j}F(\Lambda_{ij}(\rho),\phi_{ij}^{+})\geq\frac{n(n-1)}{2}(1-\varepsilon_{n}),
\end{equation}
which, by Lemma \ref{lem:fidelity-bs}, is strictly larger than the
fidelity achievable by any biseparable state, if $n$ is large enough.

It is now possible to bound the visibility $p$ of the isotropic states
for which the statement holds if $n$ is large enough. The one-way
distillation protocol \cite{devetak_distillation_2005} requires that
the states $\eta$ to be distilled are such that
\begin{equation}
H(\tr_{A}(\eta))-H(\eta)>0,\label{eq:entropy-eta}
\end{equation}
where $H(\cdot)$ is the von Neumann entropy. It can be seen that
isotropic states $\eta=\rho_{p^{2}}$ of any dimension satisfy this
inequality if $p$ is large enough.

In the particular case of $d=2$, a bound for $p$ can be calculated
explicitly. Equation (\ref{eq:entropy-eta}) reduces to
\begin{equation}
\frac{3(1-p^{2})}{4}\log_{2}(1-p^{2})+\frac{1+3p^{2}}{4}\log_{2}(1+3p^{2})>1,
\end{equation}
which holds when $p_{0}<p\leq1$ for $p_{0}\simeq0.865$.

Therefore, there exists $n_{0}\in\mathbb{N}$ such that the statement
is true for all $n\geq n_{0}$, for $p>p_{0}$. Now, the claim must
hold as well for all $n<n_{0}$ because, for every fixed such $n$,
the complete graph with $p=1$ is GME (since, in this case, the edges
are maximally entangled) and, since the set of GME states is open,
the state must remain GME for all $p>p^{*}(n)$, where $p^{*}(n)<1$.
Thus, the statement of the Theorem holds by picking
\begin{equation}
p_{1}=\max\{p_{0},\max_{n<n_{0}}p^{*}(n)\}.
\end{equation}
\end{proof}
In fact, $p_{1}$ can be estimated as a function of $n_{0}$. Consider
a network of 2-dimensional isotropic states. Consider also an LOCC
protocol described by maps $\Lambda_{ij}$ (like in Theorem \ref{thm:completegraph})
acting on the network state $\rho$, which traces out all particles
except those corresponding to the isotropic state $\rho_{p}$ shared
by parties $A_{i},A_{j}$. Then, the fidelity with the maximally entangled
state is
\begin{equation}
F(\Lambda_{ij}(\rho),\phi_{ij}^{+})=\frac{1+3p}{4}
\end{equation}
for all $i<j$. Comparing to the biseparability bound in equation
(\ref{eq:fidelityBS}), the complete graph of $n$ vertices is GME
if $p>1-4/3n$. Hence, taking $p_{1}=1-4/3n_{0}$ we have that the
network is GME for all $n<n_{0}$ for all $p>p_{1}$.

The limitations to the obtention of GME from certain networks shown
in Theorems \ref{thm:tree-BS} and \ref{thm:polygon-BS} can be overcome
by taking many copies of such networks. Indeed, exploiting the lack
of closure under tensor products of the set of biseparable states,
taking copies makes it possible to obtain GME from any connected network
of entangled isotropic states, as we now show.
\begin{thm}
The state corresponding to many copies of any connected network where
the nodes share arbitrary entangled isotropic states is GME if the
number of copies is large enough.
\end{thm}

\begin{proof}
Many copies of the network in the statement are equivalent to a network
where the nodes share many copies of arbitrary entangled isotropic
states. Since all entangled isotropic states are distillable \cite{horodecki_reduction_1999},
this means that there exists a bipartite LOCC protocol that brings
sufficiently many copies of an isotropic state as close as desired
to at least one copy of a maximally entangled state. Performing these
LOCC protocols for all edges of the network is in itself an LOCC protocol
for the parties in the network. With this they can obtain a state
arbitrarily close to a connected network of maximally entanged states,
which is GME (this can be shown in many ways; for example, the results
of Chapter \ref{chap:gmnl} show that such a network is GMNL, which
is a sufficient condition for GME) and, since the set of biseparable
states is closed, the output state of this protocol can be GME as
well. Since the set of biseparable states is closed under LOCC, this
entails that the original network must also have been GME.
\end{proof}

\section{Locality in mixed pair-entangled networks}

The results above imply that, in contrast to networks of pure states,
distributing nonlocal mixed states in any network does not necessarily
lead to GME, let alone GMNL. Further, while GME is a necessary condition
for GMNL, it is not sufficient, and there are known examples of states
that are GME and bilocal \cite{augusiak_entanglement_2015,augusiak_constructions_2018}
or even fully local \cite{bowles_genuinely_2016}. We find that networks
provide a particularly good setting to find more examples of this
kind, and non-steerability of the shared states can compromise the
obtention of GMNL. We focus on star networks, and find that one non-steerable
edge is enough to render the network bilocal while still being GME.
Moreover, if all shared states in a star network are non-steerable,
then the network is fully local. While steerability in these settings
is necessary for GMNL, we show it is not sufficient: we provide an
example of a steerable state which, when distributed in a star network,
is bilocal. Still, GMNL can be recovered by taking many copies of
these networks: we find, to our knowledge, the first example of superactivation
of GMNL from bilocality.

Our bilocality result for a star network with one non-steerable edge
actually applies to a slightly larger class of networks, as shown
below:
\begin{thm}
\label{thm:bilocal-network}Consider a network such that party $A_{1}$
is connected to the rest of the network only by one state which is
non-steerable from $A_{1}$ to $A_{2}$. Then, the network is local
across that bipartition. In particular, it is bilocal.
\end{thm}

\begin{proof}
We show that any probability distribution arising from local POVMs
on the network state can be rewritten as a bipartite distribution
arising from POVMs acting on the non-steerable state. Since this state
is local, the bipartite distribution is local in the bipartition that
the state crosses, which, by assumption, splits one party from the
rest. The LHS model of the non-steerable state is then used to show
that the term corresponding to the rest of the parties is nonsignalling,
proving the Theorem.

We consider the network as a connected graph where vertices are parties
and edges are states. We label the nodes as $A_{i}$ for $i=1,...,n,$
and the edges as $k=1,...,K,$ where $K$ is the number of edges in
the graph. Thus, the Hilbert space corresponding to each node $i$
is $\mathcal{H}_{A_{i}}=\bigotimes_{k}\mathcal{H}_{A_{i}^{k}},$ where
the tensor product runs over all the edges incident to node $i.$
On each edge $k$ lies a quantum state $\rho_{k}$ (the parties that
edge $k$ connects are left implicit). We will label the parties such
that the non-steerable state lies on edge $1$ (which connects parties
$A_{1}$ and $A_{2}$) and we will call this state $\sigma_{1}.$
We will show that the network is local across the bipartition $A_{1}|A_{2}...A_{n}.$

If each party $A_{i}$ measures according to the POVM $\left\{ E_{a_{i}|x_{i}}^{i}\right\} _{a_{i}}$
with outputs $a_{i}$ and inputs $x_{i},$ they generate a probability
distribution of the form
\begin{equation}
P(a_{1}...a_{n}|x_{1}...x_{n})=\tr\left(\bigotimes_{i=1}^{n}E_{a_{i}|x_{i}}^{i}\bigotimes_{k\neq1}\rho_{k}\otimes\sigma_{1}\right).
\end{equation}
This distribution can be rewritten as
\begin{equation}
\begin{aligned}P(a_{1} & ...a_{n}|x_{1}...x_{n})=\tr_{A_{1}A_{2}^{1}}\left[E_{a_{1}|x_{1}}^{1}\otimes\tr_{A_{2}^{k\neq1}...A_{n}}\left[\bigotimes_{i=2}^{n}E_{a_{i}|x_{i}}^{i}\Biggl(\mathbbm1_{A_{2}^{1}}\otimes\bigotimes_{k\neq1}\rho_{k}\Biggr)\Biggr)\right]\sigma_{1}\right]\end{aligned}
\label{eq:bilocal-sigma}
\end{equation}
where we will show that
\begin{equation}
\begin{aligned}F_{a_{2}...a_{n}|x_{2}...x_{n}} & :=\tr_{A_{2}^{k\neq1}...A_{n}}\left[\bigotimes_{i=2}^{n}E_{a_{i}|x_{i}}^{i}\Biggl(\mathbbm1_{A_{2}^{1}}\otimes\bigotimes_{k\neq1}\rho_{k}\Biggr)\Biggr)\right]\end{aligned}
\label{eq:bilocal-mtsH}
\end{equation}
is a POVM element acting on $\mathcal{H}_{A_{2}^{1}}$. By denoting
$A_{2}^{1}=:A$, $A_{2}^{k\neq1}...A_{n}=:B$, $a_{2}...a_{n}=:a$,
$x_{2}...x_{n}=:x$, $\bigotimes_{i=2}^{n}E_{a_{i}|x_{i}}^{i}=:E_{a|x}^{AB}$,
and $\bigotimes_{k\neq1}\rho_{k}=:\tau_{B}$, $a_{2}...a_{n}=:a$,
and $x_{2}...x_{n}=:x$, we can rewrite this as
\begin{equation}
F_{a|x}=\tr_{B}\left(E_{a|x}^{AB}(\mathbbm1_{A}\otimes\tau_{B})\right).
\end{equation}

To show positivity, we notice that $\tau_{B}$ is a quantum state,
and thus can be written as a convex combination of pure states $\ketbra{\psi}$.
Therefore, to show positivity of $F_{a|x}$ we can assume
\begin{equation}
F_{a|x}=\tr_{B}\left(E_{a|x}^{AB}(\mathbbm1_{A}\otimes\ketbra{\psi})\right)\equiv\bra{\psi}E_{a|x}^{AB}\ket{\psi}.
\end{equation}
If $F_{a|x}$ were not positive, there would exist $\ket{x}\in\mathcal{H}_{A}$
such that
\begin{equation}
\bra{x}F_{a|x}\ket{x}<0,
\end{equation}
which would imply that
\begin{equation}
\bra{\psi}\bra{x}E_{a|x}^{AB}\ket{x}\ket{\psi}<0,
\end{equation}
and hence that $E_{a|x}^{AB}$ would not be positive. But this is
false, as $E_{a|x}^{AB}$ is a POVM element. Therefore, $F_{a|x}\succcurlyeq0$.

Normalisation of $F_{a|x}$ is guaranteed by the normalisation of
$E_{a|x}^{AB}$, as
\begin{equation}
\begin{aligned}\sum_{a}F_{a|x} & =\tr_{B}\left(\sum_{a}E_{a|x}^{AB}(\mathbbm1_{A}\otimes\tau_{B})\right)\\
 & =\tr_{B}\left(\mathbbm1_{A}\otimes\tau_{B}\right)=\mathbbm1_{A}.
\end{aligned}
\end{equation}

Therefore, since equation (\ref{eq:bilocal-sigma}) expresses the
distribution achieved by the parties in the network as two POVM elements
acting on the non-steerable state $\sigma_{1}$ (which is, in particular,
local), the network is local across the bipartition $A_{1}|A_{2}...A_{n}.$
Indeed, the distribution can be written as
\begin{equation}
\begin{aligned}P(a_{1}...a_{n}|x_{1}...x_{n}) & =\sum_{\lambda}p_{\lambda}P_{1}(a_{1}|x_{1}\lambda)P_{2}(a_{2}...a_{n}|x_{2}...x_{n}\lambda).\end{aligned}
\end{equation}

This is sufficient to prove the claim according to Svetlichny's definition
of bilocality \cite{svetlichny_distinguishing_1987}, where the probability
distributions of the bipartition elements are not required to be nonsignalling.
In fact, the proof so far applies to any local state $\sigma_{1}$
lying on any edge of the network. However, the operational definition
that is used throughout this work requires, in addition, that the
local components $P_{1}$ and $P_{2}$ be nonsignalling. To this end,
we use that, since $\sigma_{1}$ is a non-steerable state, it has
an LHS model \cite{wiseman_steering_2007}. This means that any distribution
arising from measuring $\sigma_{1}$ with the POVMs $E_{a_{1}|x_{1}}^{1}$,
$F_{a_{2}...a_{n}|x_{2}...x_{n}}$(of the form of equation (\ref{eq:bilocal-mtsH}))
can be written as
\begin{equation}
P(a_{1}...a_{n}|x_{1}...x_{n})=\sum_{\lambda}p_{\lambda}P_{1}(a_{1}|x_{1}\lambda)\tr(F_{a_{2}...a_{n}|x_{2}...x_{n}}\eta_{\lambda}).
\end{equation}
Since $P_{2}$ is now of the form
\begin{equation}
P_{2}(a_{2}...a_{n}|x_{2}...x_{n}\lambda)=\tr_{A_{2}^{k\neq1}...A_{n}}\left[\bigotimes_{i=2}^{n}E_{a_{i}|x_{i}}^{i}\Biggl(\mathbbm1_{A_{2}^{1}}\otimes\bigotimes_{k\neq1}\rho_{k}\Biggr)\Biggr)\eta_{\lambda}\right],
\end{equation}
the normalisation of each $E_{a_{i}|x_{i}}^{i}$ ensures that this
expression is nonsignalling. Also, $P_{1}$ is a single-party distribution,
so the concept of signalling is mute on this side, and hence both
distributions are nonsignalling, as required.
\end{proof}
Theorem \ref{thm:bilocal-network} can be used to construct a network
which is GME and not GMNL for any number of parties, as we now show.
The non-steerable state must be entangled, otherwise the network would
be biseparable along the bipartition $A_{1}|A_{2}...A_{n}$. If, in
addition, it is an isotropic state, and the network takes the form
of a star where all other edges contain maximally entangled states,
we obtain the desired construction:
\begin{thm}
\label{thm:star-gme}A star network where Alice and Bob1 share an
entangled isotropic state, and Alice and all other Bobs share maximally
entangled states, is GME.
\end{thm}

\begin{proof}
We give a local operator that transforms the given state into a state
which is GME, as detected by a witness. Since local operators cannot
generate entanglement, this proves the Theorem.

Consider a star network with Alice in the central node and parties
$B_{1},...,B_{n-1}$ in the rest. Let Alice and $B_{1}$ share an
isotropic state
\begin{equation}
\rho_{p}=p\phi_{d}^{+}+(1-p)\tilde{\mathbbm1},
\end{equation}
where $\tilde{\mathbbm1}$ is the normalised identity and $\phi_{d}^{+}$
is the maximally entangled state in dimension $d,$ with $p>1/(d+1).$
Let Alice and each of $B_{2},...,B_{n}$ share a maximally entangled
state (for ease of notation, for the remainder of this proof we omit
the subscript $d$ from $\phi_{d}^{+}$). Then, the state of the network
is
\begin{equation}
\begin{aligned}\tau & =\rho_{p,AB_{1}}\otimes\phi_{AB_{2}}^{+}\otimes\dots\otimes\phi_{AB_{n-1}}^{+}\\
 & =p\phi_{AB_{1}}^{+}\otimes\phi_{AB_{2}}^{+}\otimes\dots\otimes\phi_{AB_{n-1}}^{+}+(1-p)\tilde{\mathbbm1}_{AB_{1}}\otimes\phi_{AB_{2}}^{+}\otimes\dots\otimes\phi_{AB_{n-1}}^{+}.
\end{aligned}
\end{equation}

We will show that, if Alice applies
\begin{equation}
A=\sum_{i=0}^{d-1}\ket{i}\bra{i}^{\otimes n-1},
\end{equation}
and the Bobs apply the identity, the resulting state is GME. Since
local operators preserve biseparability, this will mean that $\tau$
is GME too. Let
\begin{equation}
\tilde{\tau}_{f}=\left(A_{A}\otimes\mathbbm1_{B_{1}...B_{n-1}}^{\otimes n-1}\right)\tau\left(A_{A}^{\dagger}\otimes\mathbbm1_{B_{1}...B_{n-1}}^{\otimes n-1}\right)
\end{equation}
be the unnormalised state after the parties apply their operations.
We can write the components of $\tau$ as 
\begin{equation}
\begin{aligned}\phi_{AB_{1}}^{+}\otimes\phi_{AB_{2}}^{+}\otimes\dots\otimes\phi_{AB_{n-1}}^{+} & =\frac{1}{d^{n-1}}\sum_{i,j=0}^{d^{n-1}-1}\ket{i}_{A}\ket{i}_{B_{1}...B_{n-1}}\bra{j}_{A}\bra{j}_{B_{1}...B_{n-1}}\\
\tilde{\mathbbm1}_{AB_{1}}\otimes\phi_{AB_{2}}^{+}\otimes\dots\otimes\phi_{AB_{n-1}}^{+} & =\frac{1}{d^{n}}\sum_{i,j=0}^{d-1}\sum_{k,\ell=0}^{d^{n-2}-1}\ket{ik}_{A}\ket{jk}_{B_{1}...B_{n-1}}\bra{i\ell}_{A}\bra{j\ell}_{B_{1}...B_{n-1}}.
\end{aligned}
\end{equation}
Applying $A$ to each $\ket{i}_{A}$, where $i=0,...,d^{n-1}-1$,
picks out the terms where all digits of $i$ are equal, and thus gives
simply $\ket{i}_{A}$ with $i=0,...,d-1$. Then, the digits of $\ket{i}_{B_{1}...B_{n-1}}$
must also be equal. Similarly, the action of $A$ on $\ket{ik}_{A}$,
where $i=0,...,d-1$ and $k=0,...,d^{n-2}-1$, makes $i=k=0,...,d-1$,
with the corresponding effect on the $k$ index of $\ket{jk}_{B_{1}...B_{n-1}}$.
Therefore, we obtain 
\begin{equation}
\begin{aligned}\tilde{\tau}_{f}= & \frac{p}{d^{n-1}}\sum_{i,j=0}^{d-1}\ket{i}_{A}\ket{i}_{B_{1}...B_{n-1}}^{\otimes n-1}\bra{j}_{A}\bra{j}_{B_{1}...B_{n-1}}^{\otimes n-1}\\
 & +\frac{1-p}{d^{n}}\sum_{i,j=0}^{d-1}\ket{i}_{A}\ket{j}_{B_{1}}\ket{i}_{B_{2}...B_{n}}^{\otimes n-2}\bra{i}_{A}\bra{j}_{B_{1}}\bra{i}_{B_{2}...B_{n}}^{\otimes n-2}.
\end{aligned}
\end{equation}
Hence, the normalised state after the transformation is
\begin{equation}
\tau_{f}=\frac{\tilde{\tau}_{f}}{\tr\tilde{\tau}_{f}}=p\ketbra{GHZ}+\frac{1-p}{d^{2}}\sum_{i,j=0}^{d-1}\ket{i}_{A}\ket{j}_{B_{1}}\ket{i}_{B_{2}...B_{n}}^{\otimes n-2}\bra{i}_{A}\bra{j}_{B_{1}}\bra{i}_{B_{2}...B_{n}}^{\otimes n-2}.
\end{equation}

To show that this state is GME, it is sufficient to find a witness
that detects it. The operator

\begin{equation}
W=\frac{1}{d}\mathbbm1-\ketbra{GHZ}
\end{equation}
fits the bill: since the maximum overlap of the GHZ state with a biseparable
state is $1/d$ \cite{biswas_genuine-multipartite-entanglement_2014},
we have that
\begin{equation}
\tr(W\sigma)\geq0
\end{equation}
for all biseparable states $\sigma.$ Moreover,
\begin{equation}
\tr(W\tau_{f})=\frac{d-1-p(d^{2}-1)}{d^{2}},
\end{equation}
which is strictly smaller than zero for all values of $p$ such that
$\rho_{p}$ is entangled, i.e., for all
\begin{equation}
p>\frac{1}{d+1}.
\end{equation}
\end{proof}
Theorems \ref{thm:bilocal-network} and \ref{thm:star-gme} imply
that, if Alice and Bob1 share an entangled, non-steerable state $\rho_{p}$,
while Alice and all other Bobs share maximally entangled states, then
the star network is GME and bilocal. Taking into account the bounds
for separability and steerability given in Refs. \cite{horodecki_reduction_1999,almeida_noise_2007}
respectively, this occurs if $1/(d+1)<p\leq(3d-1)(d-1)^{d-1}/[(d+1)d^{d}]$.

We have seen that a non-steerable state on one edge of a star network
compromises its GMNL. In fact, this behaviour is more extreme if \emph{all}
edges are made non-steerable, in which case the network becomes fully
local, as we show in Theorem \ref{thm:star-fl}. However, steerability
does not guarantee GMNL. We show this in Theorem \ref{thm:steerable-bl}
by finding a steerable state which, when distributed in a star network,
makes the network bilocal.
\begin{thm}
\label{thm:star-fl}Any star network of states which are non-steerable
from each external node to the centre node is fully local.
\end{thm}

\begin{proof}
We iterate the ideas used to prove Theorem \ref{thm:bilocal-network}.
We start by rewriting the probability distribution arising from the
$n$-partite network state as an $(n-1)$-partite distribution arising
from POVMs acting on all but one of the non-steerable states. Then,
the LHS model of the remaining state is used to show that the distribution
is local in the bipartition crossed by that state. Since all such
bipartitions split one party from the rest (as the network takes the
form of a star), the local model contains a single-party distribution
corresponding to one of the external nodes. Iterating this process
for each of the states of the network completes the proof.

Let Alice share a non-steerable state $\sigma_{i}$ acting on $\mathcal{H}_{AB_{i}}$
with each of Bob$_{i}$, $i=1,...,n-1$. Let Alice apply a POVM $\{E_{a|x}\}_{a}$
for each input $x$, and each Bob$_{i}$ apply a POVM $\{F_{b_{i}|y_{i}}^{i}\}_{b_{i}}$
for each input $y_{i}$. Then, the distribution obtained is
\begin{equation}
P(a,b_{1},...,b_{n-1}|x,y_{1},...,y_{n-1})=\tr\left[\left(E_{a|x}\otimes\bigotimes_{i=1}^{n-1}F_{b_{i}|y_{i}}^{i}\right)\bigotimes_{i=1}^{n-1}\sigma_{i}\right].
\end{equation}
Using similar ideas to the proof of Theorem \ref{thm:bilocal-network},
we can rewrite this distribution as
\begin{equation}
\begin{aligned}P(a,b_{1},..., & b_{n-1}|x,y_{1},...,y_{n-1})\\
 & =\tr_{A\backslash A_{1},B\backslash B_{1}}\left[\left(\tr_{A_{1}B_{1}}\left[\left(E_{a|x}\otimes F_{b_{1}|y_{1}}^{1}\right)\left(\sigma_{1}\otimes\mathbbm1_{A\backslash A_{1}}\right)\right]\otimes\bigotimes_{i\neq1}F_{b_{i}|y_{i}}^{i}\right)\bigotimes_{i\neq1}\sigma_{i}\right].
\end{aligned}
\label{eq:Pstar-fl-sigma1}
\end{equation}
Now, denoting $A\backslash A_{1}=:A^{\prime}$, we have
\begin{equation}
\tr_{A_{1}B_{1}}\left[\left(E_{a|x}\otimes F_{b_{1}|y_{1}}^{1}\right)\left(\sigma_{1}\otimes\mathbbm1_{A\backslash A_{1}}\right)\right]=\tr_{A_{1}}\left[E_{a|x}\left(\tr_{B_{1}}\left[\left(\mathbbm1_{A_{1}}\otimes F_{b_{1}|y_{1}}^{1}\right)\sigma_{1}\right]\otimes\mathbbm1_{A^{\prime}}\right)\right].\label{eq:trA1B1}
\end{equation}
Since $\sigma_{1}$ is non-steerable from $B_{1}$ to $A_{1}$, it
has an LHS model, therefore
\begin{equation}
\tr_{B_{1}}\left[\left(\mathbbm1_{A_{1}}\otimes F_{b_{1}|y_{1}}^{1}\right)\sigma_{1}\right]=\sum_{\lambda}p_{\lambda}\eta_{\lambda}P_{B_{1}}(b_{1}|y_{1},\lambda),
\end{equation}
where $\eta_{\lambda}$ is a state that depends on the hidden variable
$\lambda$, which is distributed according to $\{p_{\lambda}\}_{\lambda}$.
Therefore, equation (\ref{eq:trA1B1}) can be rewritten as
\begin{equation}
\sum_{\lambda}p_{\lambda}P_{B_{1}}(b_{1}|y_{1},\lambda)\tr_{A_{1}}\left[E_{a|x}\left(\eta_{\lambda_{}}\otimes\mathbbm1_{A^{\prime}}\right)\right].
\end{equation}

We want to show that, for all $x,\lambda$, 
\begin{equation}
\tr_{A_{1}}\left[E_{a|x}\left(\eta_{\lambda}\otimes\mathbbm1_{A^{\prime}}\right)\right]=:\tilde{E}_{a|x}^{\lambda}
\end{equation}
is a POVM element in $A^{\prime}$, following a similar strategy to
the proof of Theorem \ref{thm:bilocal-network}. Since $\eta_{\lambda_{}}$
is a state, it can be written as a convex combination of pure states.
Therefore, to show positivity of $\tilde{E}_{a|x}^{\lambda}$ we can
assume
\begin{equation}
\tilde{E}_{a|x}^{\lambda}=\tr_{A_{1}}\left[E_{a|x}\left(\ketbra{\psi_{\lambda}}\otimes\mathbbm1_{A^{\prime}}\right)\right]\equiv\bra{\psi_{\lambda}}E_{a|x}\ket{\psi_{\lambda}}.
\end{equation}
If $\tilde{E}_{a|x}^{\lambda}$ were not positive, there would exist
$\ket{x}\in\mathcal{H}_{A^{\prime}}$ such that
\begin{equation}
\bra{x}\tilde{E}_{a|x}^{\lambda}\ket{x}<0,
\end{equation}
which would imply that
\begin{equation}
\bra{\psi_{\lambda}}\bra{x}E_{a|x}\ket{x}\ket{\psi_{\lambda}}<0,
\end{equation}
and hence that $\tilde{E}_{a|x}^{\lambda}$ would not be positive.
But this is false, as $E_{a|x}$ is a POVM element. Therefore, $\tilde{E}_{a|x}^{\lambda}\succcurlyeq0$.

Normalisation of $\tilde{E}_{a|x}^{\lambda}$ follows from that of
$E_{a|x}$, as
\begin{equation}
\sum_{a}\tilde{E}_{a|x}^{\lambda}=\tr_{A_{1}}\left[\sum_{a}E_{a|x}\left(\eta_{\lambda}\otimes\mathbbm1_{A^{\prime}}\right)\right]=\tr_{A_{1}}\left[\left(\eta_{\lambda}\otimes\mathbbm1_{A^{\prime}}\right)\right]=\mathbbm1_{A^{\prime}}.
\end{equation}
Therefore, from equation (\ref{eq:Pstar-fl-sigma1}) we deduce
\begin{equation}
\begin{aligned}P(a,b_{1},...,b_{n-1}| & x,y_{1},...,y_{n-1})\\
 & =\sum_{\lambda}p_{\lambda}\tr_{A\backslash A_{1},B\backslash B_{1}}\left[\left(\tilde{E}_{a|x}^{\lambda}\otimes\bigotimes_{i\neq1}F_{b_{i}|y_{i}}^{i}\right)\bigotimes_{i\neq1}\sigma_{i}\right]P_{B_{1}}(b_{1}|y_{1},\lambda).
\end{aligned}
\end{equation}
Denoting $\lambda=:\lambda_{1},$ and $p_{\lambda}=:p_{\lambda_{1}}^{1}$,
this argument can now be iterated for $i=2,...,n-1$, obtaining new
POVMs of the form of $\tilde{E}_{a|x}^{\lambda_{1}}$ which depend
on more hidden variables $\lambda_{1},...,\lambda_{n-1}$, so that
\begin{equation}
P(a,b_{1},...,b_{n-1}|x,y_{1},...,y_{n-1})=\sum_{\lambda_{1},...,\lambda_{n-1}}\prod_{i=1}^{n-1}p_{\lambda_{i}}^{i}P_{B_{i}}(b_{i}|y_{i}\lambda_{i})P_{A}(a|x,\lambda_{n-1}).\label{eq:Pstar-FL}
\end{equation}
By unifying the hidden variables $\lambda_{1},...,\lambda_{n-1}$
into a single variable $\lambda$, such that $p_{\lambda}:=\prod_{i=1}^{n-1}p_{\lambda_{i}}^{i}$,
equation (\ref{eq:Pstar-FL}) shows that $P$ is fully local.
\end{proof}
We note that the previous proof still holds if one of the $\sigma_{i}$
is merely local, but not necessarily non-steerable: labelling the
parties such that $\sigma_{n-1}$ is steerable but local, while $\sigma_{i}$
for $i\in[n-2]$ are non-steerable, we have that the first $n-2$
iterations of the argument lead to a decomposition including $P_{B_{i}}$,
for $i\in[n-2]$, and a POVM element of the form of $\tilde{E}_{a|x}^{\lambda_{1}}$
but depending on all hidden variables $\lambda_{1},...,\lambda_{n-2}$
and acting on $\mathcal{H}_{A_{n-1}}$. Then, the action of this POVM
and $F_{b_{n-1}|y_{n-1}}^{n-1}$ on $\sigma_{n-1}$ gives a distribution
$\sum_{\lambda_{n-1}}p_{\lambda_{n-1}}^{n-1}P_{B_{n-1}}(b_{n-1}|y_{n-1},\lambda_{n-1})P_{A}(a|x,\lambda_{n-1})$,
which is local.

While non-steerability of a state in certain networks makes the network
bilocal, steerability does not always guarantee GMNL: a star network
of steerable states may be bilocal, as we now show. In fact, the existence
of a steerable state which makes a star network bilocal is a consequence
of Theorem \ref{thm:bilocal-network}, as a star network of steerable,
isotropic states is a convex mixture of star networks which satisfy
that Theorem.
\begin{thm}
\label{thm:steerable-bl}There exists a steerable state which, when
distributed in a star network of any number of parties, is bilocal.
\end{thm}

\begin{proof}
We write the network state as a convex mixture of terms that contain
a non-steerable state on at least one edge of the star network. By
Theorem \ref{thm:bilocal-network}, each such term is bilocal, and,
since bilocality is closed under convex mixtures, this completes the
proof.

Consider a star network where Alice shares the isotropic state
\begin{equation}
\rho_{p}=p\phi^{+}+(1-p)\tilde{\mathbbm1},
\end{equation}
with $p=p_{0}$ to be defined later, with each of Bob$_{k}$, $k=1,...,K,$
which we will denote as $\rho_{p_{0}}^{\otimes K}.$ Let $p_{S}$
be the steerability threshold of the isotropic state, i.e., $\rho_{p}$
is non-steerable for $p\leq p_{S}$ and steerable otherwise. We will
show that, for any $K\in\mathbb{N},$ there exists $p_{0}>p_{S}$
such that $\rho_{p_{0}}^{\otimes K}$ is bilocal.

Expanding the tensor product, we find
\begin{equation}
\begin{aligned}\rho_{p_{0}}^{\otimes K}= & \,p_{0}^{K}\phi^{+\otimes K}\\
 & +p_{0}^{K-1}(1-p_{0})\left(\phi^{+\otimes(K-1)}\otimes\tilde{\mathbbm1}+\phi^{+\otimes(K-2)}\otimes\tilde{\mathbbm1}\otimes\phi^{+}+\dots+\tilde{\mathbbm1}\otimes\phi^{+\otimes(K-1)}\right)+\dots,
\end{aligned}
\label{eq:starBS-1}
\end{equation}
where the omitted terms all contain at least one term $\tilde{\mathbbm1}$
acting on $\mathcal{H}_{AB_{k}}$ for some $k,$ hence they are local
(in fact, separable) across at least one bipartition $A\{B_{\bar{k}}\}_{\bar{k}\neq k}|B_{k}.$
We can rewrite the above as

\begin{equation}
\begin{aligned}\rho_{p_{0}}^{\otimes K}= & \left(\frac{p_{0}^{K}}{K}\phi^{+\otimes K}+p_{0}^{K-1}(1-p_{0})\phi^{+\otimes(K-1)}\otimes\tilde{\mathbbm1}\right)\\
 & +\left(\frac{p_{0}^{K}}{K}\phi^{+\otimes K}+p_{0}^{K-1}(1-p_{0})\phi^{+\otimes(K-2)}\otimes\tilde{\mathbbm1}\otimes\phi^{+}\right)\\
 & +\dots+\left(\frac{p_{0}^{K}}{K}\phi^{+\otimes K}+p_{0}^{K-1}(1-p_{0})\tilde{\mathbbm1}\otimes\phi^{+\otimes(K-1)}\right)+\dots\,.
\end{aligned}
\end{equation}
Here, each bracket has $\phi^{+}$ on $K-1$ branches, and the normalised
state
\begin{equation}
\sigma=\frac{(p_{0}/K)\phi^{+}+(1-p_{0})\tilde{\mathbbm1}}{p_{0}/K+1-p_{0}}\label{eq:starBS-seppart-1}
\end{equation}
on the remaining one. The state $\sigma$ is itself isotropic, and
hence non-steerable if the coefficient of $\phi^{+}$ is equal to
$p_{S},$ that is, if
\begin{equation}
p_{0}=\frac{Kp_{S}}{1+(K-1)p_{S}}.
\end{equation}
Since $p_{S}<1,$ we have
\begin{equation}
p_{0}=\frac{p_{S}+(K-1)p_{S}}{1+(K-1)p_{S}}>\frac{p_{S}+(K-1)p_{S}^{2}}{1+(K-1)p_{S}}=p_{S},
\end{equation}
showing that there exists $p_{0}>p_{S}$ such that the isotropic states
$\sigma$ are non-steerable and the states $\rho_{p_{0}}$ in the
star network are steerable.

This means that we can write $\rho_{p_{0}}^{\otimes K}$ as a convex
mixture of star networks where the state in at least one branch is
non-steerable: either an isotropic state with parameter $p_{S}$,
or the identity. By Theorem \ref{thm:bilocal-network}, each term
in the mixture is local across $A\{B_{\bar{k}}\}_{\bar{k}\neq k}|B_{k}.$
Further, each such bipartition has one element containing a single
party $B_{k},$ and both the isotropic state and the identity have
an LHS model. Therefore, the network is bilocal.
\end{proof}
In fact, the same argument can be applied to any tree network of nonlocal,
isotropic states to prove bilocality for Svetlichny's definition \cite{svetlichny_distinguishing_1987},
where the probability distributions of the bipartition elements are
not required to be nonsignalling. A tree network can be decomposed
into a convex combination of terms, each of which has an edge $k$
containing either an isotropic state or an identity. The isotropic
state is local for a certain value of the parameter $p_{0}$, therefore
each term is local across the bipartition crossed by the edge $k$.
However, not all bipartitions will have an element containing a single
party, therefore the argument for nonsignalling can't be imported.

\section{Superactivation of GMNL in networks}

We have seen that placing non-steerable states on either one or all
edges of a star network renders the network bilocal or even fully
local, by Theorems \ref{thm:bilocal-network} and \ref{thm:star-fl}
respectively. We now show that the bilocal network displays superactivation
of GMNL. Indeed, taking many copies of this network makes it possible
to win a generalisation of the Khot-Vishnoi game with a much higher
probability than with bilocal resources. We first introduce the game
and its extension to the star network via a Lemma, and then prove
a second Lemma to bound the probability of winning with a bilocal
strategy, before showing the superactivation result.

The Khot-Vishnoi game \cite{khot_unique_2005,buhrman_near-optimal_2012}
is parametrised by a number $v$, which is assumed to be a power of
2, and a noise parameter $\eta\in[0,1/2]$. Consider the group $\{0,1\}^{v}$
of all $v$-bit strings, with operation $\oplus$ denoting bitwise
modulo 2 addition, and the subgroup $H$ of all Hadamard codewords.
The subgroup $H$ partitions the group $\{0,1\}^{v}$ into $2^{v}/v$
cosets of $v$ elements each. These cosets will act as questions,
and answers will be elements of the question cosets. The referee chooses
a uniformly random coset $[x]$, as well as a string $z\in\{0,1\}^{v}$
where each bit $z(i)$ is chosen independently and is 1 with probability
$\eta$ and 0 otherwise. Alice's question is the coset $[x]$, which
can be thought of as $u\oplus H$ for a uniformly random $u\in\{0,1\}^{v}$,
while Bob's is the coset $[x\oplus z]$, which can be thought of as
$u\oplus z\oplus H$. The aim of the players is to guess the string
$z$, and thus they must output $a\in[x]$ and $b\in[x\oplus z]$
such that $a\oplus b=z$. Ref. \cite{buhrman_near-optimal_2012} showed
that any local strategy for Alice and Bob, implemented by a distribution
denoted by $P_{\text{local}}$, achieves a winning probability of
\begin{equation}
\left\langle G_{KV},P_{\text{local}}\right\rangle \leq\frac{v}{v^{1/(1-\eta)}}.
\end{equation}

A much higher winning probability can be obtained with a distribution
$P_{\text{max}}$ arising from certain projective measurements on
the maximally entangled state:
\begin{equation}
\left\langle G_{KV},P_{\text{max}}\right\rangle \geq(1-2\eta)^{2}.\label{eq:KV-quantum-bipart}
\end{equation}

Picking the value $\eta=1/2-1/\log v$ gives rise to the bounds
\begin{equation}
\left\langle G_{KV},P_{\text{local}}\right\rangle \leq\frac{C}{v}\quad\left\langle G_{KV},P_{\text{max}}\right\rangle \geq D/\log^{2}v,
\end{equation}
for universal constants $C,D$.
\begin{lem}
\label{lem:KVmulti}The Khot-Vishnoi game can be extended to the star
network by letting Alice and each of $Bob$$_{i}$, for $i=1,...,K$,
play the bipartite Khot-Vishnoi game. This defines a game whose coefficients
are normalised, i.e., satisfy
\begin{equation}
\sum_{\substack{x_{1},...,x_{K}\\
y_{1},...,y_{K}
}
}\max_{\substack{a_{1},...,a_{K}\\
b_{1},...,b_{K}
}
}\widetilde{G}_{a_{1}...a_{K}b_{1}...b_{K}|x_{1}...x_{K}y_{1}...y_{K}}\leq1.
\end{equation}
\end{lem}

\begin{proof}
Consider a star network of $K$ edges, each of which connect Alice
to Bob$_{i}$ for $i=1,...,K$. Alice will play the bipartite Khot-Vishnoi
game $G_{KV}$ with each Bob$_{i}$. Denoting Alice's inputs and outputs
as $x_{1},...,x_{K}$ and $a_{1},...,a_{K}$ respectively, and Bob$_{i}$'s
input and output as $y_{i},$ $b_{i}$ respectively, we denote the
coefficients of each game as $G_{a_{i}b_{i}|x_{i}y_{i}}$. Thus, the
$(K+1)$-partite game being played on the star network, which we denote
by $\widetilde{G}$, has coefficients
\begin{equation}
\widetilde{G}_{a_{1}...a_{K}b_{1}...b_{K}|x_{1}...x_{K}y_{1}...y_{K}}=\prod_{i=1}^{K}G_{a_{i}b_{i}|x_{i}y_{i}}.
\end{equation}
Since $G_{KV}$ is a game, so is $\widetilde{G}$, i.e. all of its
coefficients are positive. Moreover, one can check that the coefficients
$G_{a_{i}b_{i}|x_{i}y_{i}}$ satisfy the normalisation condition
\begin{equation}
\sum_{x_{i},y_{i}}\max_{a_{i},b_{i}}G_{a_{i}b_{i}|x_{i}y_{i}}\leq1\label{eq:KV-normn}
\end{equation}
for all $i\in[K]$, and hence
\begin{equation}
\begin{aligned}\sum_{\substack{x_{1},...,x_{K}\\
y_{1},...,y_{K}
}
}\max_{\substack{a_{1},...,a_{K}\\
b_{1},...,b_{K}
}
}\widetilde{G}_{a_{1}...a_{K}b_{1}...b_{K}|x_{1}...x_{K}y_{1}...y_{K}} & =\sum_{\substack{x_{1},...,x_{K}\\
y_{1},...,y_{K}
}
}\max_{\substack{a_{1},...,a_{K}\\
b_{1},...,b_{K}
}
}\prod_{i=1}^{K}G_{a_{i}b_{i}|x_{i}y_{i}}\\
 & =\prod_{i=1}^{K}\sum_{x_{i},y_{i}}\max_{a_{i},b_{i}}G_{a_{i}b_{i}|x_{i}y_{i}}\\
 & \leq1.
\end{aligned}
\label{eq:KVmulti-normn}
\end{equation}
In fact, this normalisation condition also holds if we take only a
subset of games $G_{KV}$, i.e. the product of $G_{a_{i}b_{i}|x_{i}y_{i}}$
for $i$ in some subset of $[K]$.
\end{proof}
\begin{lem}
\label{lem:KV-bl}The extension of the Khot-Vishnoi game to the star
network is such that the winning probability using any bilocal strategy
is bounded above by $C/v$, where $v$ is the parameter of the game
and $C$ is a universal constant.
\end{lem}

\begin{proof}
Consider the game in Lemma \ref{lem:KVmulti}. To bound the winning
probability of bilocal strategies, we consider the two possible types
of bilocal distributions $P_{BL}$: those local in a bipartition that
separates Alice from all the Bobs, and those where some of the Bobs
are in Alice's partition element. The first case can be upper bounded
by considering the parallel repetition of the Khot-Vishnoi game \cite{raz_parallel_1998},
and we prove an upper bound on the winning probability equal to that
of local strategies. The second case uses the normalisation constraint
in Lemma \ref{lem:KVmulti} to reduce to the parallel repetition case.

First, we take a distribution of the form
\begin{equation}
P_{1}(a_{1},...,a_{K}|x_{1},...,x_{K})P_{2}(b_{1},...,b_{K}|y_{1},...,y_{K})\label{eq:KVparallelprob}
\end{equation}
for each $a_{i},b_{i},x_{i},y_{i}$, $i\in[K]$. In a bilocal strategy,
$P_{1},P_{2}$ would need to be nonsignalling, however, an upper bound
on the winning probability can be obtained more easily by removing
this constraint. Using such a distribution to play $\widetilde{G}$,
the parties are effectively playing the $K$-fold parallel repetition
of $G_{KV}$, which we denote as $G_{KV}^{\otimes K}$. To bound their
winning probability we will use the same techniques as in Refs. \cite{amr_unbounded_2020,buhrman_near-optimal_2012}.
Recall that, for the bipartite game, the questions are cosets of $H$
in the group $\{0,1\}^{v}$, which can be thought of as $u\oplus H$
for Alice and $u\oplus z\oplus H$ for Bob (where $u\in\{0,1\}^{v}$
is sampled uniformly and $z\in\{0,1\}^{v}$ is sampled bitwise independently
with noise $\eta$), and the answers are elements of the question
cosets. Without loss of generality, we can assume Alice and Bob's
strategy is deterministic, and identify it with Boolean functions
$A,B:\{0,1\}^{v}\rightarrow\{0,1\}$ which take the value 1 for exactly
one element of each coset. That is, for each question, $A,B$ respectively
pick out Alice's and Bob's answer. Since the players win if and only
if their answers $a,b$ satisfy $a\oplus b=z$, we have that for all
$u,z$,
\begin{equation}
\sum_{h\in H}A(u\oplus h)B(u\oplus z\oplus h)
\end{equation}
is 1 if the players win on inputs $u\oplus H$, $u\oplus z\oplus H$,
and 0 otherwise. Therefore, the winning probability is
\begin{equation}
\begin{aligned}\mathop{\mathbb{E}}_{u,z}\left[\sum_{h\in H}A(u\oplus h)B(u\oplus z\oplus h)\right] & =\sum_{h\in H}\mathop{\mathbb{E}}_{u,z}\left[A(u\oplus h)B(u\oplus z\oplus h)\right]\\
 & =v\mathop{\mathbb{E}}_{u,z}\left[A(u)B(u\oplus z)\right],
\end{aligned}
\end{equation}
since for all $h,$ the distribution of $u\oplus h$ is uniform.

For the parallel repetition $G_{KV}^{\otimes K}$, Alice and Bob must
pick an answer for each copy of the game, so we can identify their
strategy with some new Boolean functions $A,B:\{0,1\}^{vK}=\{0,1\}^{v}\times\dots\times\{0,1\}^{v}\rightarrow\{0,1\}$
which, restricted to each set of $K$ questions (i.e. $K$ cosets),
take the value 1 for exactly one element. Alice's question of $G_{KV}^{\otimes K}$
is given by $u=(u_{1},...,u_{K})$ where each $u_{i}\in\{0,1\}^{v},\,i\in[K]$
is sampled uniformly. But this is equivalent to sampling $u$ uniformly
in $\{0,1\}^{vK}$. Similarly, $z=(z_{1},...,z_{K})$ is sampled bitwise
independently, since each $z_{i}$ is. Therefore, the winning probability
is given by
\begin{equation}
\begin{aligned}\mathop{\mathbb{E}}_{\substack{u_{i},z_{i}\\
i\in[K]
}
} & \left[\sum_{\substack{h_{i}\in H_{i}\\
i\in[K]
}
}A((u_{1},...,u_{K})\oplus(h_{1},...,h_{K}))B((u_{1},...,u_{K})\oplus(z_{1},...,z_{K})\oplus(h_{1},...,h_{K}))\right]\\
 & =\sum_{\substack{h_{i}\in H_{i}\\
i\in[K]
}
}\mathop{\mathbb{E}}_{u,z}\left[A(u\oplus(h_{1},...,h_{K}))B(u\oplus z\oplus(h_{1},...,h_{K}))\right]\\
 & =v^{K}\mathop{\mathbb{E}}_{u,z}\left[A(u)B(u\oplus z)\right]
\end{aligned}
\end{equation}
since there are $v^{K}$ choices of strings of the form $(h_{1},...,h_{K})\in H_{1}\times...\times H_{K}$.
To bound $\mathbb{E}_{u,z}\left[A(u)B(u\oplus z)\right]$, we follow
the computation of Ref. \cite[Theorem 4.1]{buhrman_near-optimal_2012},
which uses the Cauchy-Schwarz and hypercontractive inequalities, and
obtain
\begin{equation}
\mathbb{E}_{u,z}\left[A(u)B(u\oplus z)\right]\leq\frac{1}{v^{K/(1-\eta)}}.
\end{equation}

If, instead, the parties share a distribution of the form
\begin{equation}
P_{1}(\alpha,\{b_{i}\}_{i\leq k}|\chi,\{y_{i}\}_{i\leq k})P_{2}(\{b_{i}\}_{i>k}|\{y_{i}\}_{i>k}),\label{eq:KVblprob}
\end{equation}
for each $\alpha\equiv(a_{1},...,a_{K}),$ $\chi\equiv(x_{1},...,x_{K}),$
$b_{i},y_{i}$, $i\in[K]$, for some $k\in[K]$, then their winning
probability is given by
%\begin{equation}
\begin{align*}\Bigl\langle\widetilde{G},P_{1}P_{2}\Bigr\rangle=\sum_{\substack{a_{i},b_{i},x_{i},y_{i}\\
i\in[K]
}
}\prod_{i=1}^{K} & G_{a_{i}b_{i}|x_{i}y_{i}}P_{1}(\alpha,\{b_{i}\}_{i\leq k}|\chi,\{y_{i}\}_{i\leq k})P_{2}(\{b_{i}\}_{i>k}|\{y_{i}\}_{i>k})\\
=\sum_{\substack{a_{i},b_{i},x_{i},y_{i}\\
k<i\leq K
}
}\prod_{i=k+1}^{K} & G_{a_{i}b_{i}|x_{i}y_{i}}\left(\sum_{\substack{a_{i},b_{i},x_{i},y_{i}\\
1\leq i\leq k
}
}\prod_{i=1}^{k}G_{a_{i}b_{i}|x_{i}y_{i}}P_{1}(\alpha,\{b_{i}\}_{i\leq k}|\chi,\{y_{i}\}_{i\leq k})\right)\\
 & \times P_{2}(\{b_{i}\}_{i>k}|\{y_{i}\}_{i>k})\tag{\stepcounter{equation}\theequation}\\
=\sum_{\substack{a_{i},b_{i},x_{i},y_{i}\\
k<i\leq K
}
}\prod_{i=k+1}^{K} & G_{a_{i}b_{i}|x_{i}y_{i}}f(\{a_{i}\}_{i>k},\{x_{i}\}_{i>k})P_{2}(\{b_{i}\}_{i>k}|\{y_{i}\}_{i>k}),
\end{align*}
\label{eq:KV-local}
%\end{equation}
where we define
\begin{equation}
f(\{a_{i}\}_{i>k},\{x_{i}\}_{i>k})=\sum_{\substack{a_{i},b_{i},x_{i},y_{i}\\
1\leq i\leq k
}
}\prod_{i=1}^{k}G_{a_{i}b_{i}|x_{i}y_{i}}P_{1}(\alpha,\{b_{i}\}_{i\leq k}|\chi,\{y_{i}\}_{i\leq k}).
\end{equation}
We will now show that we can define a probability distribution $\tilde{P}(\{a_{i}\}_{i>k}|\{x_{i}\}_{i>k})$
all of whose components are greater than or equal to those of $f(\{a_{i}\}_{i>k},\{x_{i}\}_{i>k})$,
and use this, together with the previous parallel repetition result,
to bound $\left\langle \widetilde{G},P_{1}P_{2}\right\rangle $. First,
note that the function $f$ is pointwise positive and such that
\begin{equation}
\sum_{a_{k+1},...,a_{K}}f(\{a_{i}\}_{i>k},\{x_{i}\}_{i>k})\leq1
\end{equation}
for all $x_{k+1},...,x_{K}$. Indeed, fixing $x_{k+1},...,x_{K}$,
we have
\begin{equation}
\begin{aligned}\sum_{a_{k+1},...,a_{K}} & \sum_{\substack{a_{i},b_{i},x_{i},y_{i}\\
1\leq i\leq k
}
}\prod_{i=1}^{k}G_{a_{i}b_{i}|x_{i}y_{i}}P_{1}(\alpha,\{b_{i}\}_{i\leq k}|\chi,\{y_{i}\}_{i\leq k})\\
 & =\sum_{\substack{a_{i},b_{i},x_{i},y_{i}\\
1\leq i\leq k
}
}\prod_{i=1}^{k}G_{a_{i}b_{i}|x_{i}y_{i}}\left(\sum_{a_{k+1},...,a_{K}}P_{1}(\alpha,\{b_{i}\}_{i\leq k}|\chi,\{y_{i}\}_{i\leq k})\right)\\
 & \leq\sum_{\substack{x_{i},y_{i}\\
1\leq i\leq k
}
}\prod_{i=1}^{k}\max_{a_{i},b_{i}}G_{a_{i}b_{i}|x_{i}y_{i}}\left(\sum_{\substack{a_{i},i\in[K]\\
b_{i},i\leq k
}
}P_{1}(\alpha,\{b_{i}\}_{i\leq k}|\chi,\{y_{i}\}_{i\leq k})\right)\\
 & =\sum_{\substack{x_{i},y_{i}\\
1\leq i\leq k
}
}\prod_{i=1}^{k}\max_{a_{i},b_{i}}G_{a_{i}b_{i}|x_{i}y_{i}}\leq1,
\end{aligned}
\end{equation}
where the last inequality follows from equation (\ref{eq:KVmulti-normn}).
Thus, for each $x_{k+1},...,x_{K}$ we can define $\tilde{P}(\{a_{i}\}_{i>k}|\{x_{i}\}_{i>k})$
to have the same elements as $f$ except when all $a_{i}=0$:
\begin{equation}
\tilde{P}(\{a_{i}\}_{i>k}|\{x_{i}\}_{i>k})=\begin{cases}
f(\{a_{i}\}_{i>k},\{x_{i}\}_{i>k}) & \text{if }a_{i}\neq0\text{ for some }i>k,\\
1-\sum_{\left\{ a_{i}^{\prime}\right\} _{i>k}\neq\overrightarrow{0}}f(\{a_{i}^{\prime}\}_{i>k},\{x_{i}\}_{i>k}) & \text{if }a_{i}=0\text{ for all }i>k.
\end{cases}
\end{equation}
Then, $\tilde{P}$ is a probability distribution, all of whose components
are larger than or equal to those of $f$, and hence from equation
(\ref{eq:KV-local}) we deduce that
\begin{equation}
\left\langle \widetilde{G},P_{1}P_{2}\right\rangle \leq\sum_{\substack{a_{i},b_{i},x_{i},y_{i}\\
k<i\leq K
}
}\prod_{i=k+1}^{K}G_{a_{i}b_{i}|x_{i}y_{i}}\tilde{P}(\{a_{i}\}_{i>k}|\{x_{i}\}_{i>k})P_{2}(\{b_{i}\}_{i>k}|\{y_{i}\}_{i>k}).
\end{equation}
But the right-hand side is the winning probability of the $(K-k)$-fold
parallel repetition of $G_{KV}$ using the bilocal distribution $\tilde{P}P_{2}$
which, repeating the calculation above, can be found to be bounded
as
\begin{equation}
\left\langle \widetilde{G},P_{1}P_{2}\right\rangle \leq\frac{v^{K-k}}{v^{(K-k)/(1-\eta)}}.
\end{equation}
This exhausts the local strategies available to the players. Comparing
the bound just obtained to the one from the distribution in equation
(\ref{eq:KVparallelprob}), we find $v^{K}\geq v^{K-k}\geq v$, since
$v>1$ and $K>k$. Hence,
\begin{equation}
\frac{v^{K}}{v^{K/(1-\eta)}}\leq\frac{v^{K-k}}{v^{(K-k)/(1-\eta)}}\leq\frac{v}{v^{/(1-\eta)}}.
\end{equation}
Taking $\eta=1/2-1/\log v$, we have that, for any bilocal distribution
$P_{\text{BL}}$,
\begin{equation}
\left\langle \widetilde{G},P_{\text{BL}}\right\rangle \leq\frac{C}{v}\label{eq:KV-BLbound}
\end{equation}
for some constant $C$.
\end{proof}
We can now use the extension of the Khot-Vishnoi game to the star
network to show that GMNL can be activated in networks. We take advantage
of the large separation between the quantum and classical bounds obtained
in \cite{buhrman_near-optimal_2012} and used to prove superactivation
of nonlocality in \cite{palazuelos_super-activation_2012}. We extend
the superactivation property to the star network setting using similar
techniques as Ref. \cite{amr_unbounded_2020} uses in the triangle
network, with the construction in Ref. \cite{cavalcanti_all_2013}.
This enables us to show that taking many copies of the network which
was shown to be bilocal in Theorem \ref{thm:bilocal-network} gives
rise to GMNL.
\begin{thm}
\label{thm:star-bl-activation}A star network where Alice and Bob1
share any entangled isotropic state, and Alice and all other Bobs
share maximally entangled states, gives rise to GMNL by taking many
copies.
\end{thm}

\begin{proof}
Taking $L$ copies of the star network in the statement of the Theorem
is equivalent to taking a star network where Alice shares $L$ copies
of an isotropic state, $\rho_{p}^{\otimes L}$, with Bob$_{1}$, and
$L$ copies of a $d$-dimensional maximally entangled state, $\phi_{d}^{+\otimes L}$,
with each Bob$_{i}$, $i=2,...,K$. The latter state is in turn equivalent
to a $d^{L}$-dimensional maximally entangled state $\phi_{d^{L}}^{+}$.
We will use the superactivation result first proved in Ref. \cite{palazuelos_super-activation_2012}
and extended in Ref. \cite{cavalcanti_all_2013} to show that these
states allow the parties to win the game $\widetilde{G}$ with a higher
probability than if they use any bilocal strategy.

Given the structure of the game, the probability of winning $\widetilde{G}$
using the quantum state of the network is lower bounded by the product
of the probabilities of winning each $G_{KV}$ with the state at each
edge $i$, since the players can play every game independently. On
the maximally entangled edges, the probability of winning $G_{KV}$
is bounded by equation (\ref{eq:KV-quantum-bipart}). We will obtain
a similar bound for the isotropic edge. Let $\rho_{p}$ be a $d$-dimensional
isotropic state with a visibility $p$ such that the state is entangled
and local. Its entanglement fraction \cite{horodecki_reduction_1999}
is $F=\bra{\phi_{d}^{+}}\rho_{p}\ket{\phi_{d}^{+}}=p+(1-p)/d^{2}$,
which we can use to write the isotropic state in the form
\begin{equation}
\rho_{F}=F\phi_{d}^{+}+(1-F)\frac{\mathbbm1-\phi_{d}^{+}}{d^{2}-1}.
\end{equation}
We have that $\rho_{F}$ is entangled if and only if $F>1/d$. Expanding
the tensor product in the state $\rho_{F}^{\otimes L}$, we can write
it as
\begin{equation}
\rho_{F}^{\otimes L}=F^{L}\phi_{d}^{+\otimes L}+\dots=F^{L}\phi_{d^{L}}^{+}+\dots,\label{eq:isocopies}
\end{equation}
where the omitted terms are tensor products of $\phi_{d}^{+}$ and
$(\mathbbm1-\phi_{d}^{+})/(d^{2}-1)$ with coefficients that are products
of $F$ and $(1-F)$. Acting on $\rho_{F}^{\otimes L}$ with the same
projective measurements as above gives a probability distribution
$P_{\text{iso}}$ which is linear in the terms of equation (\ref{eq:isocopies}),
and the action of $G_{KV}$ on this distribution is linear too. Since
the coefficients $G_{a_{i}b_{i}|x_{i}y_{i}}$ are nonnegative, we
have
\begin{equation}
\left\langle G_{KV},P_{\text{iso}}\right\rangle \geq F^{L}\left\langle G_{KV},P_{1}\right\rangle ,\label{eq:KV-isobound}
\end{equation}
where $P_{1}$ is the probability distribution obtained from the projective
measurements acting on $\phi_{d^{L}}^{+}$. By equation (\ref{eq:KV-quantum-bipart}),
we find
\begin{equation}
F^{L}\left\langle G_{KV},P_{1}\right\rangle \geq F^{L}(1-2\eta)^{2}.
\end{equation}
Taking $\eta=1/2-1/\log v$ like in Lemma \ref{lem:KV-bl}, we find
\begin{equation}
F^{L}\left\langle G_{KV},P_{1}\right\rangle \geq F^{L}\frac{D}{\ln^{2}v},
\end{equation}
where $D$ is a universal constant. If $d$ is a power of 2, we can
choose $v=d^{L}$ to obtain
\begin{equation}
F^{L}\frac{D}{\ln^{2}v}=F^{L}\frac{D}{L^{2}\ln^{2}d}.
\end{equation}
Otherwise, the game can be modified like in Remark 1.1 of Ref. \cite{palazuelos_largest_2014},
obtaining a similar bound on the quantum winning probability but with
a different constant $D$. The bound on the classical winning probability
(equation (\ref{eq:KV-BLbound})) is unchanged.

Putting the bounds from both types of edges together, and denoting
by $P_{Q}$ the probability distribution obtained from the state of
the whole network and the projective measurements that are performed
on each edge, we obtain 
\begin{equation}
\left\langle \widetilde{G},P_{Q}\right\rangle \ge F^{L}\frac{D}{L^{2}\ln^{2}d}\left(\frac{D}{L^{2}\ln^{2}d}\right)^{K-1}=\frac{F^{L}D^{K}}{L^{2K}\ln^{2K}d}.\label{eq:KV-quantum}
\end{equation}

Finally, we use Lemma \ref{lem:KV-bl} to compare the local and quantum
bounds. Using equations (\ref{eq:KV-BLbound}) and (\ref{eq:KV-quantum}),
we find
\begin{equation}
\frac{\left\langle \widetilde{G},P_{Q}\right\rangle }{\sup_{P_{BL}\in\mathcal{BL}}\left\langle \widetilde{G},P_{\text{BL}}\right\rangle }\geq\frac{D^{K}}{CL^{2K}\ln^{2K}d}F^{L}d^{L},
\end{equation}
where $\mathcal{BL}$ is the set of bilocal distributions. Since $F>1/d$
(as $\rho_{F}$ is entangled), this expression tends to $\infty$
as $L$ grows unbounded. In particular, the ratio is $>1$, proving
that GMNL is obtained. Since one copy of the star network is bilocal,
there is superactivation of GMNL.
\end{proof}

\chapter{\label{chap:agreement}A physical principle from observers' agreement}

\chaptermark{Physical principle from observers' agreement}

Is the world quantum? An active research line in quantum foundations
is devoted to exploring what constraints can rule out the post-quantum
theories that are consistent with experimentally observed results.
We explore this question in the context of epistemics, and ask whether
agreement between observers can serve as a physical principle that
must hold for any theory of the world. Aumann's seminal Agreement
Theorem states that two (classical) agents cannot agree to disagree.
We propose an extension of this theorem to no-signaling settings.
In particular, we establish an Agreement Theorem for quantum agents,
while we construct examples of (post-quantum) no-signaling boxes where
agents can agree to disagree. The PR box is an extremal instance of
this phenomenon. These results make it plausible that agreement between
observers might be a physical principle, while they also establish
links between the fields of epistemics and quantum information that
seem worthy of further exploration.

\section{Classical agreement theorem}

\label{sec:classical-agreement} Aumann's theorem \cite{aumann_agreeing_1976}
is formulated on a probability space, and partial information of the
observers is represented by different partitions of the space. Each
observer knows which of their partition elements obtains, and estimates
the probability of an event of interest by Bayesian inference. We
refer to a probability space, together with some given partitions,
as a (classical) \emph{ontological model}. Ontological models appearing
in the literature (see, e.g., \cite{ferrie_quasi-probability_2011})
also contain a set of preparations underlying the distribution over
the state space, and the partitions are usually phrased in terms of
measurements and outcomes. However, we consider preparations implicit
and use the language of partitions to bridge the gap between classical
probability spaces and nonsignalling boxes more smoothly.

For the sake of simplicity and following Aumann, we will restrict
our analysis to two observers, Alice and Bob. Aumann's original theorem
considers common \emph{knowledge} about \emph{one single} event of
interest to both observers. We provide a slight generalisation with
common \emph{certainty} about \emph{two perfectly correlated} events
of interest, one for each observer. This allows us to jump into the
framework of nonsignalling boxes that we will use later. We refer
to this generalisation as the classical agreement theorem. This nomenclature
can be further motivated by the fact that, for purely classical situations,
both statements---the original Aumann's theorem and our formulation
with perfectly correlated events---can be proven to be \emph{equivalent}
(as long as states of the world with null probability are ignored,
as in Ref. \cite{aaronson_complexity_2005}).

Consider a probability space $(\Omega,\mathcal{E},\mathsf{P})$ where
$\Omega$ is a finite set of possible states of the world; $\mathcal{E}$
is its power set (i.e., the set of events); and $\mathsf{P}$ is a
probability measure over $\Omega$. We will consider two events $E_{A},E_{B}\in\mathcal{E}$
of interest to Bob and Alice, respectively (the choice of subscripts
will become clear later). We assume that they are \emph{perfectly
correlated}: $\mathsf{P}(E_{A}\backslash E_{B})=\mathsf{P}(E_{B}\backslash E_{A})=0.$

Fix partitions $\mathcal{P}_{A},\mathcal{P}_{B}$ of $\Omega$ for
Alice and Bob, respectively. For convenience, assume that all members
of the join (coarsest common refinement) of $\mathcal{P}_{A}$ and
$\mathcal{P}_{B}$ are non-null. For a state $\omega\in\Omega,$ $\mathcal{P}_{A,B}(\omega)$
is the partition element (of Alice's or Bob's, respectively) that
contains $\omega.$ For each $n\in\mathbb{N},$ fix numbers $q_{A},q_{B}\in[0,1]$
and consider the following sets: 
\begin{equation}
\begin{aligned}A_{0} & =\left\{ \omega\in\Omega:\mathsf{P}(E_{B}|\mathcal{P}_{A}(\omega))=q_{A}\right\} ,\\
B_{0} & =\left\{ \omega\in\Omega:\mathsf{P}(E_{A}|\mathcal{P}_{B}(\omega))=q_{B}\right\} ,\\
A_{n+1} & =\left\{ \omega\in A_{n}:\mathsf{P}(B_{n}|\mathcal{P}_{A}(\omega))=1\right\} ,\\
B_{n+1} & =\left\{ \omega\in B_{n}:\mathsf{P}(A_{n}|\mathcal{P}_{B}(\omega))=1\right\} .
\end{aligned}
\label{eq:A0B0}
\end{equation}
Here, the set $A_{0}$ is the set of states $\omega$ such that Alice
assigns probability $q_{A}$ to event $E_{B}$; the set $B_{1}$ is
the set of states $\omega$ such that Bob assigns probability $q_{B}$
to event $E_{A}$ \textit{and} probability 1 to the states in $A_{0}$---i.e.,
states where Bob assigns probability $q_{B}$ to $E_{A}$ and is certain
that Alice assigned probability $q_{A}$ to $E_{B}$; and so on, and
similarly for $B_{0}$, $A_{1}$, etc.

It is common certainty at a state $\omega^{*}\in\Omega$ that Alice
assigns probability $q_{A}$ to $E_{B}$ and that Bob assigns probability
$q_{B}$ to $E_{A}$ if 
\begin{equation}
\omega^{*}\in A_{n}\cap B_{n}\quad\forall n\in\mathbb{N}.\label{eq:cc-local}
\end{equation}
If equation \eqref{eq:cc-local} does not hold for all $n\in\mathbb{N}$,
but, instead, only for $n\leq N$ for a certain $N\in\mathbb{N},$
then we talk about $N$th-order mutual certainty.

We now state and prove the classical agreement theorem that will be
the basis of our work: 
\begin{thm}
\label{thm:Aumann} Fix a probability space $(\Omega,\mathcal{E},\mathsf{P})$,
where $E_{A}$ and $E_{B}$ are perfectly correlated events. If it
is common certainty at a state $\omega^{*}\in\Omega$ that Alice assigns
probability $q_{A}$ to $E_{B}$ and Bob assigns probability $q_{B}$
to $E_{A},$ then $q_{A}=q_{B}.$ 
\end{thm}

\begin{proof}
The main idea behind the proof is to notice that, since $\Omega$
is finite, there is a finite $N\in\mathbb{N}$ such that, for all
$n\geq N,$ $A_{n+1}=A_{n}$ and $B_{n+1}=B_{n}.$ Using the definition
of $A_{N+1},$ noticing that $A_{N}$ is a union of Alice's partition
elements, and using convex combination arguments together with the
perfect correlations between $E_{A}$ and $E_{B}$ leads to 
\begin{equation}
\mathsf{P}(E_{A}\cap E_{B}|A_{N}\cap B_{N})=q_{A}.
\end{equation}
Running the parallel argument for Bob entails that the same expression
is equal to $q_{B},$ proving the claim.

In more detail, since $\Omega$ is finite, there is a finite $N\in\mathbb{N}$
such that, for all $n\geq N,$ $A_{n+1}=A_{n}$ and $B_{n+1}=B_{n}.$
From the definition of $A_{N+1},$ we have that 
\begin{equation}
\mathsf{P}(B_{N}|\mathcal{P}_{A}(\omega))=1\ \forall\omega\in A_{N}.\label{eq:A_N}
\end{equation}
Now, $A_{N}$ is a union of partition elements of $\mathcal{P}_{A},$
i.e., $A_{N}=\bigcup_{i\in I}\pi_{i}$ where each $\pi_{i}\in\mathcal{P}_{A}$
and $I$ is a finite index set. From Equation (\ref{eq:A_N}), we
have 
\begin{equation}
\mathsf{P}(B_{N}|\pi_{i})=1\ \forall i\in I.
\end{equation}
Since $\mathsf{P}(B_{N}|A_{N})$ is a convex combination of $\mathsf{P}(B_{N}|\pi_{i})$
for $i\in I,$ we must have 
\begin{equation}
\mathsf{P}(B_{N}|A_{N})=1.\label{eq:probBNAN}
\end{equation}

Now, since $A_{N}\subseteq A_{0},$ then $\mathsf{P}(E_{B}|\pi_{i})=q_{A}$
for all $i\in I$ too. Using a convex combination argument once more,
this entails that 
\begin{equation}
\mathsf{P}(E_{B}|A_{N})=q_{A}.\label{eq:probEBAN}
\end{equation}
Equations (\ref{eq:probBNAN}) and (\ref{eq:probEBAN}) together imply
that 
\begin{equation}
\mathsf{P}(E_{B}|A_{N}\cap B_{N})=q_{A}.
\end{equation}

But events $E_{A}$ and $E_{B}$ are perfectly correlated, so that
\begin{equation}
\mathsf{P}(E_{A}\cap E_{B}|A_{N}\cap B_{N})=q_{A},
\end{equation}
as well.

Running the parallel argument with $A$ and $B$ interchanged, we
obtain 
\begin{equation}
\mathsf{P}(E_{A}\cap E_{B}|A_{N}\cap B_{N})=q_{B},
\end{equation}
which implies that $q_{A}=q_{B}.$ 
\end{proof}

\section{Mapping agreement to nonsignalling boxes}

We now map the classical agreement theorem into the nonsignalling
framework, in order to explore its applicability beyond the classical
realm.

We consider nonsignalling distributions, or boxes \cite{popescu_quantum_1994},
as per Definition \ref{def:intro-ns}.

We now show that we can associate a nonsignalling box with any ontological
model, and vice versa. Remarkably, this can be accomplished even in
the case in which the nonsignalling box is nonlocal, obtaining an
ontological model with a quasi-probability measure instead of standard
positive probabilities \cite{abramsky_sheaf-theoretic_2011}. (The
appearance of quasi-probabilities here should not surprise the reader.
In fact, one cannot hope to obtain ontological models with only non-negative
probabilities for post-classical nonsignalling boxes, as this would
provide local hidden-variable models that contradict, for instance,
Bell's theorem. In any case, the use of this mathematical tool has
been well rooted in the study of quantum mechanics since its origins---see
\cite{ferrie_quasi-probability_2011} for a nice review of this subject.)
This makes it possible to translate results from one framework to
the other, something that might be of interest in order to establish
further connections between epistemics and quantum theory. However,
once we establish these mappings, we focus on nonsignalling boxes
and leave this digression aside. 
\begin{prop}
Given any ontological model, there is a corresponding nonsignalling
box which reproduces the observable statistics of the model. If the
ontological model is classical, then the box is local. 
\end{prop}

\begin{proof}
Let $\mathcal{A},\mathcal{B},\mathcal{X},\mathcal{Y}$ be index sets.
Let $(\Omega,\mathcal{F},\mathsf{P})$ be a probability space, and,
for each $x\in\mathcal{X}$, let $\mathsf{A}_{a|x}$ be a partition
of the states $\omega\in\Omega$ where $a\in\mathcal{A}$ denotes
the partition elements. Similarly, for each $y\in\mathcal{Y}$, let
$\mathsf{B}_{b|y}$ be another partition of the states $\omega\in\Omega$,
where $b\in\mathcal{B}$ denotes the partition elements. According
to that, we can understand labels $x\in\mathcal{X}$, $y\in\mathcal{Y}$
as \emph{inputs}---this information fixes what partition Alice and
Bob look at---and $a\in\mathcal{A}$, $b\in\mathcal{B}$ as \emph{outputs}---this
is the information that the agents gain by observing their corresponding
partitions.

With all the above, \sloppy $\big\lbrace(\Omega,\mathcal{F},\mathsf{P})$,
$\lbrace\mathsf{A}_{a|x},\mathsf{B}_{b|y}\rbrace_{a,b,x,y}\big\rbrace$
is an ontological model that we now want to associate to a nonsignalling
box that reproduces its statistics. In this ontological model, given
inputs $x\in\mathcal{X}$, $y\in\mathcal{Y}$, the probability of
obtaining outputs $a\in\mathcal{A}$, $b\in\mathcal{B}$ is given
by $\mathsf{P}(\mathsf{A}_{a|x}\cap\mathsf{B}_{b|y})$. This simple
observation leads us to construct the nonsignalling box $\left\{ P(a,b|x,y)\right\} _{a,b,x,y}$,
where each probability is given by 
\begin{equation}
P(a,b|x,y):=\mathsf{P}\big(\mathsf{A}_{a|x}\cap\mathsf{B}_{b|y}\big),\qquad\forall(a,b,x,y)\in\mathcal{A}\times\mathcal{B}\times\mathcal{X}\times\mathcal{Y}.\label{OntToBox}
\end{equation}

Clearly, the probability that Alice and Bob make observations according
to the partitions given by $x,y$ and conclude that they are in the
partition element that corresponds to $a,b$ respectively is the same
as the probability that they input $x,y$ in their nonsignalling box
and obtain outputs $a,b$. That is, the nonsignalling box reproduces
the statistics of the ontological model.

Further, while an ontological model with a quasi-probability distribution
might be such that some states $\omega$ are such that $\mathsf{P}(\omega)<0$,
all partition elements must be observable, that is, $\mathsf{P}(\mathsf{A}_{a|x}),\mathsf{P}(\mathsf{B}_{b|y})\geq0$
for all $a,b,x,y$. In particular, the probabilities of all intersection
of partition elements is non-negative, therefore so is $P$. Normalisation
of $P$ follows from the normalisation of $\mathsf{P}$: since for
all $x,y$ we have $\bigcup_{a}\mathsf{A}_{a|x}=\bigcup_{b}\mathsf{B}_{b|y}=\Omega$,
and $\mathsf{P}(\Omega)=1$, then 
\begin{equation}
\sum_{a,b}P(a,b|x,y)=\mathsf{P}(\bigcup_{a}\mathsf{A}_{a|x}\cap\bigcup_{b}\mathsf{B}_{b|y})=\mathsf{P}(\Omega)=1,
\end{equation}
for all $x,y$.

The fact that $P$ is nonsignalling follows from the same ideas: for
all $x\neq x^{\prime}$, we have 
\begin{equation}
\sum_{a}P(a,b|x,y)=\mathsf{P}(\bigcup_{a}\mathsf{A}_{a|x}\cap\mathsf{B}_{b|y})=\mathsf{P}(\mathsf{B}_{b|y})=\mathsf{P}(\bigcup_{a}\mathsf{A}_{a|x^{\prime}}\cap\mathsf{B}_{b|y})=\sum_{a}P(a,b|x^{\prime},y),
\end{equation}
and $\sum_{b}P(a,b|x,y)=\sum_{b}P(a,b|x,y^{\prime})$ for all $y\neq y^{\prime}$
follows similarly.

Finally, if the ontological model is classical, it is an LHV model
for the box, where $\mathsf{P}(\omega)\equiv p(\lambda)$, $\mathsf{P}(\mathsf{A}_{a|x})\equiv P(a|x,\lambda)$,
and $\mathsf{P}(\mathsf{B}_{b|y})\equiv P(b|y,\lambda)$. By Definition
\ref{def:intro-local-bipart}, such a box is local. 
\end{proof}
\begin{prop}
\noindent Given any nonsignalling box, there is a (non-unique) corresponding
ontological model whose probabilities assigned to the states of the
world are not necessarily non-negative. 
\end{prop}

This result was already derived in \cite{abramsky_sheaf-theoretic_2011}
from sheaf-theoretic concepts, however we provide a much more direct
proof that is more suitable for the purposes of this work. 
\begin{proof}
Let $\left\{ P(ab|xy)\right\} _{a,b,x,y}$ be a nonsignalling box.
We construct its associated ontological model $\big\lbrace(\Omega,\mathcal{F},\mathsf{P})$,
$\lbrace\mathsf{A}_{a|x},\mathsf{B}_{b|y}\rbrace_{a,b,x,y}\big\rbrace\,.$
We provide the proof for $a,b,x,y\in\left\{ 0,1\right\} \,$ for ease
of notation, but the generalisation to more inputs and outputs is
immediate.

To construct the ontological model, we postulate the existence of
a set of states 
\begin{equation}
\omega_{a_{0}a_{1}b_{0}b_{1}}
\end{equation}
with quasi-probabilities 
\begin{equation}
\mathsf{P}_{a_{0}a_{1}b_{0}b_{1}}\equiv\mathsf{P}(\omega_{a_{0}a_{1}b_{0}b_{1}}).
\end{equation}
Each state corresponds to an \textit{instruction set} \cite{abramsky_operational_2014},
i.e., the state where Alice outputs $a_{0}$ on input $x=0$ and $a_{1}$
on input $x=1\,,$ and Bob outputs $b_{0}$ on input $y=0$ and $b_{1}$
on input $y=1.$ Then, each $\mathsf{P}_{a_{0}a_{1}b_{0}b_{1}}$ is
the quasi-probability of the corresponding instruction set. Of course,
if the given box is post-classical, not all of these quasi-probabilities
will be non-negative. In fact, in principle it need not even be guaranteed
that one can find a quasi-probability distribution over these states.
But we will use the probability distribution of the inputs and outputs
of the given nonsignalling box to derive a linear system of equations
over the quasi-probabilities, and show that it does have a solution.

There are 16 states in total, as there are two possible outputs for
each of the 4 inputs ($|\mathcal{A}|^{|\mathcal{X}|}\cdot|\mathcal{B}|^{|\mathcal{Y}|}$
in general). Then, each partition corresponds to a set of states as
follows: 
\begin{equation}
\begin{aligned}\mathsf{A}_{a|x} & =\left\{ \omega_{a_{0}a_{1}b_{0}b_{1}}:a_{x}=a\right\} \\
\mathsf{B}_{b|y} & =\left\{ \omega_{a_{0}a_{1}b_{0}b_{1}}:b_{y}=b\right\} \,.
\end{aligned}
\label{eq:defA,B}
\end{equation}

We associate the probabilities $P(ab|xy)$ of the nonsignalling box
to the probabilities $\mathsf{P}(\mathsf{A}_{a|x}\cap\mathsf{B}_{b|y})$
of each intersection of partitions, for each input pair $(x,y)$ and
output pair $(a,b)\,.$ This gives rise to a set of equations for
the probabilities $\mathsf{P}_{a_{0}a_{1}b_{0}b_{1}}.$ Indeed, the
probability of each intersection is given by 
\begin{equation}
\mathsf{P}(\mathsf{A}_{a|x}\cap\mathsf{B}_{b|y})=\sum_{a_{\bar{x}}\,,b_{\bar{y}}}\mathsf{P}_{a_{x}a_{\bar{x}}b_{y}b_{\bar{y}}}\label{eq:prob-intersection}
\end{equation}
where we denote the output corresponding to the input that is not
$x$ as $a_{\bar{x}}\,,$ and similarly for $b_{\bar{y}}\,,$ and
so we have, for each $a,b,x,y\,,$ 
\begin{equation}
\sum_{a_{\bar{x}}\,,b_{\bar{y}}}\mathsf{P}_{a_{x}a_{\bar{x}}b_{y}b_{\bar{y}}}=P(ab|xy)\,.\label{eq:linsyst}
\end{equation}
Since there are 16 values of $P(ab|xy)$ in the 2-input 2-output nonsignalling
box, we arrive at 16 equations ($|\mathcal{A}|\times|\mathcal{B}|\times|\mathcal{X}|\times|\mathcal{Y}|$
in general). Of course, there are some linear dependencies between
the equations, but we will show that the system still has a solution.

The system of equations can be expressed as 
\begin{equation}
M\overline{\mathsf{P}}=C
\end{equation}
where $M$ is the matrix of coefficients, $\overline{\mathsf{P}}$
is the vector of probabilities $\mathsf{P}_{a_{0}a_{1}b_{0}b_{1}}$
and $C$ is the vector of independent terms $P(ab|xy)\,.$ The system
has a solution (which is not necessarily unique) if and only if 
\begin{equation}
\textnormal{rank}(M)=\textnormal{rank}(M|C)\,.
\end{equation}
Since the rank of a matrix is the number of linearly independent rows,
it is trivially true that 
\begin{equation}
\textnormal{rank}(M)\leq\textnormal{rank}(M|C)\,,
\end{equation}
as including the independent terms can only remove some relations
of linear dependence, not add more. Equivalently, the number of relations
of linear dependence of $M|C$ is always smaller than or equal to
the number of relations of linear dependence of $M$. Therefore, to
show that their ranks are equal, it is sufficient to show that every
relation of linear dependence that we find in $M$ still holds in
$M|C$. That is, for every relation of linear dependence between the
probabilities $\mathsf{P}_{a_{0}a_{1}b_{0}b_{1}}$ that is contained
in $M\,,$ it is sufficient to show that the relation still holds
when the sums of probabilities are matched to the elements $P(ab|xy)$
of the nonsignalling box in order to show that the system of equations
has a solution.

Observe that $M$ contains only zeros and ones, as the equations (\ref{eq:linsyst})
are just sums of probabilities. Moreover, each column of $M$ corresponds
to the probability of a state $\omega_{a_{0}a_{1}b_{0}b_{1}}\,,$
while each row corresponds to an equation with independent term $P(ab|xy)\,.$
Because the equations (\ref{eq:linsyst}) correspond to intersections
of partitions of the set of $\omega_{a_{0}a_{1}b_{0}b_{1}}\,,$ we
can observe that each row of $M$ has a 1 in the column corresponding
to the states $\omega_{a_{0}a_{1}b_{0}b_{1}}$ contained in the corresponding
partition, and a 0 elsewhere. Put another way, in order to construct
$M$ one must first partition the set of $\omega_{a_{0}a_{1}b_{0}b_{1}}$
in four different ways, corresponding to 
\begin{equation}
\begin{aligned}\left\{ \mathsf{A}_{a|0}\right\} _{a}\,, & \left\{ \mathsf{A}_{a|1}\right\} _{a}\,,\\
\left\{ \mathsf{B}_{b|0}\right\} _{b}\,, & \left\{ \mathsf{B}_{b|1}\right\} _{b}
\end{aligned}
\end{equation}
for Alice and Bob respectively. This gives partitions of the columns
of $M\,.$ Then, the 16 possible ways of intersecting partitions of
Alice's with partitions of Bob's give the 16 equations with independent
term $P(ab|xy)\,.$ But notice now that the partition structure imposes
a certain relation of linear dependence between the rows of $M.$
Indeed, for each $b,y,$ we have 
\begin{equation}
\left\{ \bigcup_{a}\left(\mathsf{A}_{a|0}\cap\mathsf{B}_{b|y}\right)\right\} =\left\{ \bigcup_{a}\left(\mathsf{A}_{a|1}\cap\mathsf{B}_{b|y}\right)\right\} \,,
\end{equation}
as 
\begin{equation}
\bigcup_{a}\left(\mathsf{A}_{a|0}\cap\mathsf{B}_{b|y}\right)=\left(\bigcup_{a}\mathsf{A}_{a|0}\right)\cap\mathsf{B}_{b|y}=\Omega\cap\mathsf{B}_{b|y}=\left(\bigcup_{a}\mathsf{A}_{a|1}\right)\cap\mathsf{B}_{b|y}=\bigcup_{a}\left(\mathsf{A}_{a|1}\cap\mathsf{B}_{b|y}\right)\,,
\end{equation}
and, similarly, for each $a,x$ we have 
\begin{equation}
\left\{ \bigcup_{b}\left(\mathsf{A}_{a|x}\cap\mathsf{B}_{b|0}\right)\right\} =\left\{ \bigcup_{b}\left(\mathsf{A}_{a|x}\cap\mathsf{B}_{b|1}\right)\right\} \,.
\end{equation}
Using the correspondence of these partitions with the partitions of
the columns of $M$ gives 8 relations of linear dependence between
its rows. Now, noticing that a union of columns of $M$ corresponds
to a sum of probabilities $\mathsf{P}_{a_{0}a_{1}b_{0}b_{1}}\,,$
we find that these relations correspond exactly to the nonsignalling
conditions, as 
\begin{equation}
\mathsf{P}\left(\bigcup_{a}\left(\mathsf{A}_{a|x}\cap\mathsf{B}_{b|y}\right)\right)=\sum_{a}\mathsf{P}\left(\mathsf{A}_{a|x}\cap\mathsf{B}_{b|y}\right)
\end{equation}
so 
\begin{equation}
\sum_{a}\mathsf{P}\left(\mathsf{A}_{a|0}\cap\mathsf{B}_{b|y}\right)=\sum_{a}\mathsf{P}\left(\mathsf{A}_{a|1}\cap\mathsf{B}_{b|y}\right)
\end{equation}
and similarly for Bob. Of course, by definition of nonsignalling box,
these relations hold for the independent terms $P(ab|xy)$ as well,
since 
\begin{equation}
\begin{aligned}\sum_{a}\mathsf{P}\left(\mathsf{A}_{a|x}\cap\mathsf{B}_{b|y}\right) & =\sum_{a}P(ab|xy)\,,\\
\sum_{b}\mathsf{P}\left(\mathsf{A}_{a|x}\cap\mathsf{B}_{b|y}\right) & =\sum_{b}P(ab|xy)
\end{aligned}
\end{equation}
by construction of the linear system (see equations (\ref{eq:prob-intersection})
and (\ref{eq:linsyst})). Therefore, every relation of linear dependence
between the rows of $M$ holds also between the rows of $M|C\,,$
as required.

Notice also that the implication goes both ways: the linear system
has a solution \emph{only if} the set of probabilities $P(ab|xy)$
is nonsignalling. The coefficient matrix $M$ includes the nonsignalling
conditions by construction of the states $\omega_{a_{0}a_{1}b_{0}b_{1}}$
with probabilities $\mathsf{P}_{a_{0}a_{1}b_{0}b_{1}}\,.$ Therefore,
if these conditions do not hold for the independent terms $P(ab|xy)\,,$
then the rank of $M|C$ must be larger than that of $M\,,$ as $M|C$
contains more linearly independent rows than $M.$ 
\end{proof}
With the mapping between ontological models and nonsignalling boxes
in mind, we next define \emph{common certainty of disagreement} for
nonsignalling boxes. The idea is to reinterpret the definitions in
Section \ref{sec:classical-agreement} in this latter setting.

We first provide meaning for the events of interest (previously identified
as $E_{A}$, $E_{B}$) in the present setting. Now, these events correspond
to some set of outputs, given that the nonsignalling box was queried
with some particular inputs. For the sake of concreteness, we fix
these inputs to be $x=1,\,y=1$ and the \emph{outputs of interest}
to be $a=1$, $b=1$. This motivates us to consider the events $F_{A}=\lbrace(1,b,1,y)\rbrace_{b\in\mathcal{B},y\in\mathcal{Y}}$
(on Alice's side) and $F_{B}=\lbrace(a,1,x,1)\rbrace_{a\in\mathcal{A},x\in\mathcal{X}}$
(on Bob's side). Then, we say that $F_{A}$ and $F_{B}$ are perfectly
correlated when 
\begin{equation}
P(a,b|x=1,y=1)=0\textnormal{ for all }a\neq b.\label{eq:perfcorr}
\end{equation}

Given this, we assume that the agents will actually conduct their
measurements according to some other partitions. Again, for concreteness,
let us assume that those partitions are the ones associated with inputs
$x=0$, $y=0$. These inputs take on the role of partitions $\mathcal{P}_{A}$,
$\mathcal{P}_{B}$ in the ontological model picture. The outputs obtained
from these measurements are the nonsignalling box analogue to the
events $\mathcal{P}_{A}(\omega)$, $\mathcal{P}_{B}(\omega)$. In
order to make the following expressions more concrete, we assume,
when $x=0$, $y=0$ are inputted, that the outputs obtained are $a=0$
and $b=0$, respectively.

Therefore, given the perfectly correlated events $F_{A},\,F_{B}$
and numbers $q_{A},\,q_{B}\in[0,1]$, we define the sets 
\begin{align}
\alpha_{0} & =\left\{ a:P(b=1|a,x=0,y=1)=q_{A}\right\} \,,\\
\beta_{0} & =\left\{ b:P(a=1|b,x=1,y=0)=q_{B}\right\} \,,
\end{align}
and, for all $n\geq0,$ 
\begin{align}
\alpha_{n+1} & =\left\{ a\in\alpha_{n}:P(B_{n}|a,x=0,y=0)=1\right\} \,,\\
\beta_{n+1} & =\left\{ b\in\beta_{n}:P(A_{n}|b,x=0,y=0)=1\right\} \,,
\end{align}
where 
\begin{align}
A_{n} & =\alpha_{n}\times\mathcal{B}\times\mathcal{X}\times\mathcal{Y}\,,\\
B_{n} & =\mathcal{A}\times\beta_{n}\times\mathcal{X}\times\mathcal{Y}\,.
\end{align}
By analogy with the sets in equation \eqref{eq:A0B0}, the set $\alpha_{0}$
is the set of Alice's outcomes such that she assigns probability $q_{A}$
to $F_{B}$, having input $x=0$. The set $\beta_{1}$ is the set
of Bob's outcomes such that he is certain that Alice assigned probability
$q_{A}$ to $F_{B}$, and so on, and similarly for $\beta_{0}$, $\alpha_{1}$,
etc.

Suppose that Alice and Bob both input 0 and get output 0. Then, there
is \emph{common certainty of disagreement} about the event that Alice
assigns probability $q_{A}$ to $F_{B}$ and Bob assigns probability
$q_{B}$ to $F_{A}$ if $q_{A}\neq q_{B}$ and 
\begin{equation}
(a=0,b=0,x=0,y=0)\in A_{n}\cap B_{n}\quad\forall n\in\mathbb{N}.\label{def:ccd_in_NS}
\end{equation}
Notice the relationship between this definition and the previous one:
the $\omega^{*}$ in equation \eqref{eq:cc-local}, at which the disagreement
occurred, fixed the partition elements that Alice and Bob observed.
Here, disagreement occurs at the inputs and outputs $(a=0,b=0,x=0,y=0)$
that the agents obtain.

We are now in a position to state and prove the classical agreement
theorem in the nonsignalling language, i.e. for local boxes. We restrict
to boxes of two inputs and two outputs since, by Theorem \ref{thm:generalization},
any larger box exhibiting disagreement can be reduced to a 2-input
2-output box that also exhibits disagreement, while preserving its
locality properties. With the mapping defined above, the following
is now a corollary of Theorem \ref{thm:Aumann}, although we provide
a standalone proof of the result in the interest of readers more familiarised
with the language of nonsignalling boxes. Moreover, Theorem \ref{thm:Aumann}
and Corollary \ref{cor:Aumann} can be shown to be equivalent. 
\begin{cor}
\label{cor:Aumann} Suppose Alice and Bob share a local nonsignalling
box with underlying probability distribution $P$. Let $q_{A},q_{B}\in[0,1]$,
and let 
\begin{equation}
\begin{aligned}P(b=1|a=0,x=0,y=1) & =q_{A},\\
P(a=1|b=0,x=1,y=0) & =q_{B}.
\end{aligned}
\end{equation}
If $q_{A}$ and $q_{B}$ are common certainty between the agents,
then $q_{A}=q_{B}.$ 
\end{cor}

\begin{proof}
By definition of $q_{A},q_{B},$ and using the fact that the shared
distribution is local and hence satisfies Definition \ref{def:intro-local-bipart},
we have 
\begin{equation}
\begin{aligned}q_{A}\sum_{\lambda}p(\lambda)P_{A}(0|0\lambda) & =\sum_{\lambda}p(\lambda)P_{A}(0|0\lambda)P_{B}(1|1\lambda)\\
q_{B}\sum_{\lambda}p(\lambda)P_{B}(0|0\lambda) & =\sum_{\lambda}p(\lambda)P_{A}(1|1\lambda)P_{B}(0|0\lambda)\,.
\end{aligned}
\label{eq:qAqBlocal}
\end{equation}
In the proof of Theorem \ref{thm:2in2outcc} we show that, if $1\in\alpha_{n}$
or $1\in\beta_{n}$ for all $n\in\mathbb{N}$ then there is no common
certainty of disagreement for any nonsignalling distribution, and
these encompass local distributions. Hence there only remains to prove
the claim for $1\not\in\alpha_{n}$ and $1\not\in\beta_{n},$ for
some $n\in\mathbb{N}.$ This implies that 
\begin{equation}
\begin{aligned}P(b=0|a=0,x=0,y=0) & =1\\
P(a=0|b=0,x=0,y=0) & =1
\end{aligned}
\end{equation}
and hence 
\begin{equation}
\begin{aligned}\sum_{\lambda}p(\lambda)P_{A}(0|0\lambda) & =\sum_{\lambda}p(\lambda)P_{A}(0|0\lambda)P_{B}(0|0\lambda)\\
\sum_{\lambda}p(\lambda)P_{B}(0|0\lambda) & =\sum_{\lambda}p(\lambda)P_{A}(0|0\lambda)P_{B}(0|0\lambda),
\end{aligned}
\end{equation}
which implies, on the one hand, that 
\begin{equation}
\sum_{\lambda}p(\lambda)P_{A}(0|0\lambda)=\sum_{\lambda}p(\lambda)P_{B}(0|0\lambda)\label{eq:00equal}
\end{equation}
and, on the other, that 
\begin{equation}
\sum_{\lambda}p(\lambda)P_{A}(0|0\lambda)P_{B}(1|0\lambda)=\sum_{\lambda}p(\lambda)P_{A}(1|0\lambda)P_{B}(0|0\lambda)=0,
\end{equation}
that is, 
\begin{equation}
P_{A}(0|0\lambda)P_{B}(1|0\lambda)=P_{A}(1|0\lambda)P_{B}(0|0\lambda)=0\label{eq:cc-local-app}
\end{equation}
for all $\lambda.$ Therefore, there remains to prove only that 
\begin{equation}
\sum_{\lambda}p(\lambda)P_{A}(0|0\lambda)P_{B}(1|1\lambda)=\sum_{\lambda}p(\lambda)P_{A}(1|1\lambda)P_{B}(0|0\lambda).\label{eq:cc-local-apptoprove}
\end{equation}
Because the outputs for inputs $x=1,y=1$ are perfectly correlated,
we have 
\begin{equation}
P_{A}(0|1\lambda)P_{B}(1|1\lambda)=P_{A}(1|1\lambda)P_{B}(0|1\lambda)=0\label{eq:perfcorrlocal}
\end{equation}
for all $\lambda$ and, since $P_{A}(0|1\lambda)+P_{A}(1|1\lambda)=1$
and similarly for $P_{B},$ this implies 
\begin{equation}
P_{A}(1|1\lambda)=P_{B}(1|1\lambda).\label{eq:11equal}
\end{equation}
Then we can prove (\ref{eq:cc-local-apptoprove}) by simple manipulations
of the probability distributions of each party: 
%\begin{equation}
%\begin{aligned}\sum_{\lambda}p(\lambda)P_{A}(0|0\lambda)P_{B}(1|1\lambda) & =\sum_{\lambda}p(\lambda)P_{A}(0|0\lambda)P_{A}(1|1\lambda)\left[P_{B}(0|0\lambda)+P_{B}(1|0\lambda)\right]\\
% & =\sum_{\lambda}p(\lambda)P_{A}(0|0\lambda)P_{A}(1|1\lambda)P_{B}(0|0\lambda)\\
% & =\sum_{\lambda}p(\lambda)\left[P_{A}(0|0\lambda)+P_{A}(1|0\lambda)\right]P_{A}(1|1\lambda)P_{B}(0|0\lambda)\\
% & =\sum_{\lambda}p(\lambda)\left[P_{A}(0|1\lambda)+P_{A}(1|1\lambda)\right]P_{A}(1|1\lambda)P_{B}(0|0\lambda)\\
% & =\sum_{\lambda}p(\lambda)P_{A}(1|1\lambda)P_{A}(1|1\lambda)P_{B}(0|0\lambda)\\
% & =\sum_{\lambda}p(\lambda)P_{A}(1|1\lambda)P_{B}(0|0\lambda)
%\end{aligned}
%\end{equation}
where we have used the fact that $\sum_{b\in\mathcal{B}}P_{B}(b|y\lambda)=1$
for all $y,\lambda$ in the first equality, (\ref{eq:cc-local-app})
in the second and third, $\sum_{a\in\mathcal{A}}P_{A}(a|x\lambda)=1$
for all $x,\lambda$ in the fourth, (\ref{eq:perfcorrlocal}) again
in the fifth, and $P_{A}(1|1\lambda)^{2}=P_{A}(1|1\lambda)$ for all
$\lambda$ (since $P_{A}(a|x\lambda)$ can be assumed to be either
1 or 0 for every $a,x,\lambda$) in the last. 
\end{proof}

\section{Nonsignalling agents can agree to disagree}

Given the mapping exhibited above, as well as the restatement of the
agreement theorem for local boxes, it is now natural to ask whether
the agreement theorem holds when dropping the locality constraint.
When we generalise the setting and allow the agents to share a generic
nonsignalling box, we find that the agreement theorem does not hold.
That is, nonsignalling observers \emph{can} agree to disagree, and
we characterise the distributions that give rise to common certainty
of disagreement. Later, we find that no such distribution can be quantum---i.e.,
quantum observers \emph{cannot} agree to disagree.

We first present the following theorem in which the nonsignalling
box has two inputs and two outputs, but we will show in Theorem \ref{thm:generalization}
that the result is fully general. In place of ``common certainty
of disagreement about the event that Alice assigns probability $q_{A}$
to $F_{B}=\lbrace(a,1,x,1)\rbrace_{a\in\mathcal{A},x\in\mathcal{X}}$
and Bob assigns probability $q_{B}$ to $F_{A}=\lbrace(1,b,1,y)\rbrace_{b\in\mathcal{B},y\in\mathcal{Y}}$,
at event $(0,0,0,0)$,'' we simply say ``common certainty of disagreement.'' 
\begin{thm}
\label{thm:2in2outcc} A 2-input 2-output nonsignalling box gives
rise to common certainty of disagreement if and only if it takes the
form of Table \ref{tab:nsccd}.

\begin{table}[ht]
\centering %
\begin{tabular}{|c|c|c|c|c|}
\hline 
\label{isabox}$xy\backslash ab$  & 00  & 01  & 10  & 11\tabularnewline
\hline 
\hline 
00  & $r$  & 0  & 0  & $1-r$ \tabularnewline
\hline 
01  & $r-s$  & $s$  & $-r+t+s$  & $1-t-s$\tabularnewline
\hline 
10  & $t-u$  & $u$  & $r-t+u$  & $1-r-u$\tabularnewline
\hline 
11  & $t$  & 0  & 0  & $1-t$\tabularnewline
\hline 
\end{tabular}\caption{Parametrisation of 2-input 2-output nonsignalling boxes with common
certainty of disagreement. Here, $r,s,t,u\in[0,1]$ are such that
all the entries of the box are non-negative, $r>0$, and $s-u\protect\neq r-t.$}
\label{tab:nsccd} 
\end{table}
\end{thm}

We provide an outline of the proof before turning to the proof itself.

To prove the direct implication, we first consider the case in which,
for some $n$ onwards, $\alpha_{n},\beta_{n}$ each contain only one
output, $a=0,b=0$, respectively. By the definition of $\alpha_{n+1}$,
this implies that $P(B_{n}|a=0,x=0,y=0)=1,$ and, thus, $P(01|00)=0.$
Similarly, the definition of $\beta_{n+1}$ gives $P(10|00)=0.$ Perfect
correlations in the inputs $x=1,y=1$ imply that $P(01|11)=P(10|11)=0,$
and the rest of the table is deduced in terms of parameters $r,s,t,u$
by using nonsignalling and normalisation constraints. The condition
$r>0$ ensures that $P(00|00)>0$, as per the input and output that
the agents in fact obtained. Finally, $q_{A}\neq q_{B}$ if and only
if $s-u\neq r-t,$ which concludes the proof of this case.

If, for all $n,$ one or both of $\alpha_{n},\beta_{n}$ contain(s)
both outputs, we find $q_{A}=q_{B},$ contradicting common certainty
of disagreement.

The converse implication is proved by writing $q_{A},q_{B}$ in terms
of the parameters of the box. If $\alpha_{0}=\{a=0\}$ and $\beta_{0}=\{b=0\},$
then we find $\alpha_{1}=\alpha_{0}$ and $\beta_{1}=\beta_{0},$
and common certainty of disagreement follows.

If the parameters of the box are such that $\alpha_{0}=\{a=0,a=1\}$
but $\beta_{0}=\{b=0\}$, then the definition of $\beta_{1}$ implies
$P(b=0|a=0,x=0,y=0)=1$; therefore, $(0,0,0,0)\in A_{1}$, and common
certainty of disagreement follows. One can reason symmetrically if
$\alpha_{0}=\{a=0\}$ but $\beta_{0}=\{b=0,b=1\}$. Finally, if both
$\alpha_{0}$ and $\beta_{0}$ are the full set of outcomes, then
$s-u=r-t,$ contradicting the statement of the Theorem.

In the proof of Theorem \ref{thm:2in2outcc}, we will make use of
the following Lemma: 
\begin{lem}
\label{lem:forccd}Consider a nonsignalling box of 2 inputs and 2
outputs. Then, $\alpha_{0}=\{0,1\}$ if and only if $q_{A}=P(b=1|y=1).$
Analogously, $\beta_{0}=\{0,1\}$ if and only if $q_{B}=P(a=1|x=1).$ 
\end{lem}

\begin{proof}
\noindent By hypothesis, 
\begin{align*}
q_{A} & =P(b=1|a=0,x=0,y=1)=\frac{P(01|01)}{P(a=0|x=0)}\\
 & =P(b=1|a=1,x=0,y=1)=\frac{P(11|01)}{P(a=1|x=0)}.
\end{align*}
But now, we can write 
\begin{align*}
P(b=1|y=1)=P(01|01)+P(11|01)=P(a=0|x=0)q_{A}+P(a=1|x=0)q_{A}=q_{A}.
\end{align*}
The reverse implication is trivial. The analogous statement can be
proved by interchanging the roles of Alice and Bob. 
\end{proof}
\begin{proof}[Proof of Theorem \ref{thm:2in2outcc}]
We first prove that common certainty of disagreement imposes the
claimed structure for the nonsignalling box. Therefore, we assume
common certainty of disagreement, i.e., 
\begin{equation}
(0,0,0,0)\in A_{n}\cap B_{n}\qquad\forall n\in\mathbb{N}.\label{hyp1}
\end{equation}

In particular, we also assume that Alice and Bob input $x=y=0$ and
obtain $a=b=0.$ This implies 
\begin{equation}
P(00|00)>0.
\end{equation}

We split the proof into three cases based on the contents of the sets
$A_{n},B_{n}$: 
\begin{casenv}
\item $1\notin\alpha_{n},\,1\notin\beta_{n}$ for some $n$.\footnote{This need not happen at the same stage, i.e., possibly $1\notin\alpha_{m}$,
for some $m<n$. However in this case, since the sets are nonempty
by assumption, we have $\alpha_{n}=\alpha_{m}.$ } From common certainty of disagreement (equation \eqref{hyp1}), we
have that 
\begin{align*}
P(B_{n}|a=0,x=0,y=0)=1,\qquad P(A_{n}|b=0,x=0,y=0)=1,
\end{align*}
which, together with $1\notin\alpha_{n},\,1\notin\beta_{n}$, translates
into: 
\begin{align*}
P(01|00)=0,\qquad P(10|00)=0.
\end{align*}
We also assumed that the agents in fact obtained outputs $a=0,b=0$
on inputs $x=0,y=0$, so we must have $P(00|00)>0$. The rest of the
table is determined by nonsignalling constraints in terms of parameters
$r$, $s$, $t$ and $u$. Given the box in the statement of the theorem,
$q_{A}\ne q_{B}$ if and only if $s-u\ne r-t$, which concludes the
proof of this case. 
\item $\alpha_{n}=\lbrace0,1\rbrace$, for all $n\in\mathbb{N}$ while $1\notin\beta_{m}$
for some $m$. We show that this case implies $q_{A}=q_{B}$, so it
contradicts common certainty of disagreement. Indeed, the definition
of $\alpha_{m+1}$ enforces the conditions: 
\begin{align*}
P(b=0|a=0,x=0,y=0)=1=P(b=0|a=1,x=0,y=0).
\end{align*}
This implies 
\begin{align*}
P(b=1|a=0,x=0,y=0) & =\frac{P(01|00)}{P(a=0|x=0)}=0\quad\Rightarrow\quad P(01|00)=0,\\
P(b=1|a=1,x=0,y=0) & =\frac{P(11|00)}{P(a=1|x=0)}=0\quad\Rightarrow\quad P(11|00)=0.
\end{align*}
Adding nonsignalling conditions to these last equations, we also obtain
\begin{equation}
0=P(b=1|y=0)=P(01|10)+P(11|10),
\end{equation}
and so 
\begin{equation}
P(01|10)=0=P(11|10)\label{ccd-zeroprob}
\end{equation}
and 
\begin{equation}
P(b=0|y=0)=1.\label{ccd-oneprob}
\end{equation}
This allows us to identify $q_{B}$ with $P(a=1|x=1)$, since 
\begin{align*}
q_{B} & =P(a=1|b=0,x=1,y=0)\\
 & =\frac{P(10|10)}{P(b=0|y=0)}\\
 & =P(10|10)\\
 & =P(a=1|x=1)-P(11|10)\\
 & =P(a=1|x=1),
\end{align*}
where the third and last equalities follow from equations \eqref{ccd-oneprob}
and \eqref{ccd-zeroprob} respectively. Now, taking into account Lemma
\ref{lem:forccd} and perfect correlations, we have 
\[
q_{A}=P(b=1|y=1)=P(a=1|x=1),
\]
which shows that $q_{A}=q_{B}$, as mentioned above. 
\item $\alpha_{n}=\lbrace0,1\rbrace$, $\beta_{n}=\lbrace0,1\rbrace$ for
all $n\in\mathbb{N}$. We now show that this case also implies $q_{A}=q_{B}$,
contradicting common certainty of disagreement. Using Lemma \ref{lem:forccd}
we have 
\[
q_{B}=P(a=1|x=1)\qquad\text{ as well as }\qquad q_{A}=P(b=1|y=1).
\]
Now, perfect correlations impose that $P(a=1|x=1)=P(b=1|y=1)$, that
is, $q_{A}=q_{B}$. 
\end{casenv}
Next, we prove the converse implication of the theorem. We show that
any nonsignalling box of the above form must exhibit common certainty
of disagreement. Since $s-u\ne r-t$, we have that Alice and Bob assign
different probabilities to output $a,b=1$ on input $x,y=1$: 
\begin{equation}
\begin{aligned}q_{A}=P(b=1|a=0,x=0,y=1) & =s/r,\\
q_{B}=P(a=1|b=0,x=1,y=0) & =(r-t+u)/r\,.
\end{aligned}
\label{eq:qAqB}
\end{equation}

In the case that $1\not\in\alpha_{0},\:1\not\in\beta_{0},$ we also
have that $\alpha_{1}=\alpha_{0}$ and $\beta_{1}=\beta_{0}\,,$ and
common certainty of disagreement follows, because $(0,0,0,0)$ is
in $A_{n}\cap B_{n}$ for all $n$.

If the parameters are such that 
\begin{equation}
\frac{1-t-s}{1-r}=\frac{s}{r},
\end{equation}
but 
\begin{equation}
\frac{1-r-u}{1-r}\neq\frac{r-t+u}{r},
\end{equation}
then 
\begin{equation}
P(b=1|a=1,x=0,y=1)=q_{A},
\end{equation}
as well, but 
\begin{equation}
P(a=1|b=1,x=1,y=0)\neq q_{B},
\end{equation}
and so $1\in\alpha_{0},\:1\not\in\beta_{0}.$ Since we have 
\begin{equation}
P(b=0|a=0,x=0,y=0)=1,
\end{equation}
we find $(0,0,0,0)\in A_{1}$,\footnote{Note $(1,0,0,0)\not\in A_{1},$ though this does not affect the present
proof.} and hence all $A_{n}$ still contain $(0,0,0,0),$ yielding common
certainty of disagreement.

Symmetric reasoning covers the case $1\not\in\alpha_{0},\:1\in\beta_{0}$,
and only the case where $\alpha_{0}=\{0,1\}$, $\beta_{0}=\{0,1\}$
remains. This happens when 
\begin{equation}
\begin{aligned}P(b=1|a=1,x=0,y=1) & =P(b=1|a=0,x=0,y=1),\\
P(a=1|b=0,x=1,y=0) & =P(a=1|b=1,x=1,y=0)
\end{aligned}
\end{equation}
which, in terms of the parameters, is equivalent to 
\begin{align}
\frac{1-t-s}{1-r} & =\frac{s}{r},\label{cond1}\\
\frac{1-r-u}{1-r} & =\frac{r-t+u}{r}.\label{cond2}
\end{align}
However, these two conditions are satisfied simultaneously only when
$s-u=r-t$, as we now show. From Equation \eqref{cond1} we get 
\[
s=r(1-t),
\]
while from Equation \eqref{cond2} we obtain 
\[
u=t(1-r).
\]
This means that if Equations \eqref{cond1} and \eqref{cond2} are
both satisfied, then 
\[
s-u=r(1-t)-t(1-r)=r-t,
\]
which contradicts the statement of the Theorem. 
\end{proof}

\section{Quantum agents cannot agree to disagree}

While some nonsignalling distributions can exhibit common certainty
of disagreement, we find that probability distributions arising in
quantum mechanics do satisfy the agreement theorem. This is surprising:
it is well-known that a given measurement of a quantum system (say,
that corresponding to the input \$x,y=0\$) need not offer any information
about the outcome of an incompatible measurement on the same system
(say, \$x,y=1\$). However, some consistency remains: common certainty
of disagreement is impossible, even for incompatible measurements.
\begin{thm}
\label{thm:ccd-not-quantum} No 2-input 2-output quantum box can give
rise to common certainty of disagreement. 
\end{thm}

\begin{proof}
In order to give rise to common certainty of disagreement, the probability
distribution that the state and measurements generate must be of the
form of Table \ref{tab:nsccd}. Theorem 1 in Tsirelson's seminal paper
\cite{cirelson_quantum_1980} implies that, if there is a quantum
realisation of the box, then there exist real, unit vectors 
\begin{equation}
\ket{w_{x}},\ket{v_{y}}
\end{equation}
such that the correlations 
\begin{equation}
c_{xy}:=P(a=b|xy)-P(a\neq b|xy)
\end{equation}
satisfy 
\begin{equation}
c_{xy}=\bra{w_{x}}\ket{v_{y}}
\end{equation}
for each $x,y\,.$ For the box in Theorem \ref{thm:2in2outcc}, this
means, in particular, that 
\begin{equation}
\begin{aligned}\bra{w_{0}}\ket{v_{0}} & =1\,,\\
\bra{w_{1}}\ket{v_{1}} & =1\,,
\end{aligned}
\end{equation}
and, since the vectors have unit norm, this implies that 
\begin{equation}
\begin{aligned}\ket{w_{0}} & =\ket{v_{0}}\,,\\
\ket{w_{1}} & =\ket{v_{1}}\,.
\end{aligned}
\end{equation}
Then, we are left with 
\begin{equation}
\begin{aligned}c_{01} & =\bra{w_{0}}\ket{w_{1}}\,,\\
c_{10} & =\bra{w_{1}}\ket{w_{0}}\,.
\end{aligned}
\end{equation}
Since the vectors are real, we find 
\begin{equation}
c_{01}=c_{10}\,,
\end{equation}
but this implies that 
\begin{equation}
s-u=r-t,
\end{equation}
which implies that $q_{A}=q_{B}$ and, hence, impedes disagreement. 
\end{proof}
We have seen that no 2-input 2-output quantum box can give rise to
common certainty of disagreement. We now lift the restriction on the
number of inputs and outputs and show that no quantum box can give
rise to common certainty of disagreement.

First, notice that the proof for 2 inputs and outputs did not require
\textit{common} certainty, but only first-order mutual certainty.
Indeed, by observing the definitions of the sets $\alpha_{n},\beta_{n}$,
one can see that $\alpha_{n}=\alpha_{1}$ and $\beta_{n}=\beta_{1}$
for all $n\geq1$. This means that first-order mutual certainty \textit{implies}
common certainty, and, therefore, first-order certainty suffices to
characterise the nonsignalling box that displays common certainty
of disagreement.

As the number of outputs grows, first-order mutual certainty is no
longer sufficient. However, since the number of outputs is always
finite, there exists an $N\in\mathbb{N}$ such that $\alpha_{n}=\alpha_{N}$
and $\beta_{n}=\beta_{N}$ for all $n\geq N$. Since $\alpha_{n+1}\subseteq\alpha_{n}\,\forall n$,
and similarly for $\beta$, the sets $\alpha_{N},\beta_{N}$ are the
\textit{smallest} sets of outputs for which the disagreement occurs.
Because of this, any $(a,b,x,y)$ that belongs to $A_{N}\cap B_{N}$
will also belong to $A_{n}\cap B_{n}$ for all $n$; that is, $N$th-order
mutual certainty implies common certainty. So, for any finite nonsignalling
box, one needs only $N$th-order mutual certainty to characterise
it. As the number of outputs grows unboundedly, one needs \textit{common}
certainty to hold \cite{geanakoplos_we_1982}. These observations
will be relevant to extending Theorem \ref{thm:ccd-not-quantum} beyond
two inputs and outputs. 
\begin{thm}
\label{thm:generalization} No quantum box can give rise to common
certainty of disagreement. 
\end{thm}

We show that any nonsignalling box with common certainty of disagreement
induces a 2-input 2-output nonsignalling box with the same property.
Thus, if there existed a quantum system that could generate the bigger
box, it could also generate the smaller box. Then, Theorem \ref{thm:2in2outcc}
implies that no quantum box can give rise to common certainty of disagreement.

To show the reduction of the box, we use the ideas presented in Chapter
\ref{chap:intro} about transforming probability distributions while
preserving locality (which, as we shall see, preserve normalisation
and nonsignalling too). Since the original, larger box exhibits common
certainty of disagreement \emph{at} event $(0,0,0,0)$ \emph{about}
event $(1,1,1,1)$, it is enough to consider inputs $0,1$ for each
party, and any extra available inputs can be ignored. Grouping the
outputs of the original box in two sets in order to map them to the
effective box is not as straightforward, as we must ensure that the
effective box also displays common certainty of disagreement. Recalling
the discussion preceding Theorem \ref{thm:generalization}, there
exists an $N\in\mathbb{N}$ such that $\alpha_{n}=\alpha_{N}$ and
$\beta_{n}=\beta_{N}$ for all $n\geq N$. Outputs for each agent
are then grouped according to whether or not they belong in each of
these sets respectively. Because the transformations in the probabilities
are local, the effective box is still normalised and nonsignalling.
It is then possible to check that the effective box satisfies common
certainty of disagreement if the original box did. 
\begin{proof}
We define a mapping from a distribution $\left\{ P(ab|xy)\right\} _{a\in\mathcal{A},b\in\mathcal{B},x\in\mathcal{X},y\in\mathcal{Y}}$
to an effective distribution $\left\{ \tilde{P}(\tilde{a}\tilde{b}|\tilde{x}\tilde{y})\right\} _{\tilde{a},\tilde{b},\tilde{x},\tilde{y}\in\left\{ 0,1\right\} }$
such that the following conditions hold: 
\begin{enumerate}
\item if $\left\{ P(ab|xy)\right\} $ is quantum, then so is $\left\{ \tilde{P}(\tilde{a}\tilde{b}|\tilde{x}\tilde{y})\right\} \,,$\label{enu:ptilde-qu} 
\item if $\left\{ P(ab|xy)\right\} $ satisfies common certainty of disagreement,
then so does $\left\{ \tilde{P}(\tilde{a}\tilde{b}|\tilde{x}\tilde{y})\right\} \,.$\label{enu:ptilde-ccd} 
\end{enumerate}
First, notice that the number of inputs can be reduced to 2 without
loss of generality, as common certainty of disagreement is always
defined to be \emph{at }an event (wlog, $(0,0,0,0)$) \emph{about}
another event (wlog, $(1,1,1,1)$). One can associate the inputs $x=0,y=0$
with $\tilde{x}=0,\tilde{y}=0$, respectively, and $x=1,y=1$ with
$\tilde{x}=1,\tilde{y}=1$ respectively, and ignore all other possible
inputs in $\mathcal{X},\mathcal{Y}.$ The outputs, instead, must be
grouped according to whether or not they belong in the sets $\alpha_{n},\beta_{n}$
(for input 0) and whether or not they correspond to the event obtaining,
i.e., whether or not they are equal to 1 (for input 1).

Since $P$ satisfies common certainty of disagreement, we know that
$(0,0,0,0)\in A_{n}\cap B_{n}.$ Moreover, by the definitions of the
sets $\alpha_{n},\beta_{n}$ (and since we only consider finite sets
$\mathcal{A},\mathcal{B},\mathcal{X},\mathcal{Y}$) there exists an
$N\in\mathbb{N}$ such that $\alpha_{n}=\alpha_{N}$ and $\beta_{n}=\beta_{N}$
for all $n\geq N\,.$ Take such $N,$ and define the following indicator
functions: 
\begin{equation}
\begin{aligned}\chi_{0|0}^{\alpha}(a) & =\begin{cases}
0 & a\not\in\alpha_{N}\\
1 & a\in\alpha_{N}
\end{cases}\\
\chi_{0|0}^{\beta}(b) & =\begin{cases}
0 & b\not\in\beta_{N}\\
1 & b\in\beta_{N}
\end{cases}\\
\chi_{0|1}^{\alpha}(c)=\chi_{0|1}^{\beta}(c) & =\begin{cases}
0 & c=1\\
1 & c\neq1
\end{cases}
\end{aligned}
\end{equation}
(where $c$ stands for output $a,b$ for Alice and Bob, respectively),
with 
\begin{equation}
\begin{aligned}\chi_{1|x}^{\alpha}(a) & =1-\chi_{0|x}^{\alpha}(a)\,\\
\chi_{1|y}^{\beta}(b) & =1-\chi_{0|y}^{\beta}(b)
\end{aligned}
\end{equation}
for each $a,b,x,y$. Then, the mapping from $P$ to $\tilde{P}$ is
defined as follows: 
\begin{equation}
\tilde{P}(\tilde{a}\tilde{b}|\tilde{x}\tilde{y})=\sum_{a,b}\delta_{x,\tilde{x}}\delta_{y,\tilde{y}}\chi_{\tilde{a}|x}^{\alpha}(a)\chi_{\tilde{b}|y}^{\beta}(b)P(ab|xy)\label{eq:ptoptilde}
\end{equation}
where 
\begin{equation}
\delta_{s,t}=\begin{cases}
0 & s\neq t\\
1 & s=t\,.
\end{cases}
\end{equation}

We note that the distribution $\tilde{P}$ is merely a local post-processing
of $P,$ and hence it is quantum if $P$ is. Indeed, the function
$\chi$ that defines $\tilde{P}$ only relates the inputs and outputs
of each agent individually. Therefore, condition \ref{enu:ptilde-qu}
holds, as, letting $E_{a|x},F_{b|y},\rho$ be the POVMs and state
defining $P$, we have 
\begin{equation}
\begin{aligned}\tilde{P}(\tilde{a}\tilde{b}|\tilde{x}\tilde{y}) & =\sum_{a,b}\delta_{x,\tilde{x}}\delta_{y,\tilde{y}}\chi_{\tilde{a}|x}^{\alpha}(a)\chi_{\tilde{b}|y}^{\beta}(b)\tr(E_{a|x}\otimes F_{b|y}\rho)\\
 & =\tr\left[\left(\sum_{a}\delta_{x,\tilde{x}}\chi_{\tilde{a}|x}^{\alpha}(a)E_{a|x}\right)\otimes\left(\sum_{b}\delta_{y,\tilde{y}}\chi_{\tilde{b}|y}^{\beta}(b)F_{b|y}\right)\rho\right]\\
 & =\tr\left[E_{\tilde{a}|\tilde{x}}\otimes F_{\tilde{b}|\tilde{y}}\rho\right],
\end{aligned}
\end{equation}
where 
\begin{equation}
\begin{aligned}E_{\tilde{a}|\tilde{x}} & =\sum_{a}\delta_{x,\tilde{x}}\chi_{\tilde{a}|x}^{\alpha}(a)E_{a|x}\\
F_{\tilde{b}|\tilde{y}} & =\sum_{b}\delta_{y,\tilde{y}}\chi_{\tilde{b}|y}^{\beta}(b)F_{b|y}
\end{aligned}
\end{equation}
for each $\tilde{a},\tilde{b},\tilde{x},\tilde{y},x,y.$

In particular, one can check that $\tilde{P}$ is normalised and nonsignalling
provided that $P$ is normalised and nonsignalling. Normalisation
follows straightforwardly from the definition, since for each input,
each output in $P$ gets mapped to a unique output in $\tilde{P},$
and all of the outputs in $P$ get mapped to some output in $\tilde{P}$
(i.e. the map from $P$ to $\tilde{P}$ is a surjective function).
Because the map is defined differently for each pair of inputs and
outputs, the nonsignalling conditions need to be checked for each
line. However, the computations all follow the same pattern, and we
perform only one as an example: 
\begin{equation}
\begin{aligned}\sum_{\tilde{a}}\tilde{P}(\tilde{a}0|00) & =\sum_{\substack{a\in\alpha_{N}\\
b\in\beta_{N}
}
}P(ab|00)+\sum_{\substack{a\not\in\alpha_{N}\\
b\in\beta_{N}
}
}P(ab|00)\\
 & =\sum_{\substack{a\in\mathcal{A}\\
b\in\beta_{N}
}
}P(ab|00)\\
 & =\sum_{\substack{a\in\mathcal{A}\\
b\in\beta_{N}
}
}P(ab|10)\\
 & =\sum_{\substack{a\neq1\\
b\in\beta_{N}
}
}P(ab|10)+\sum_{b\in\beta_{N}}P(1b|10)\\
 & =\sum_{\tilde{a}}\tilde{P}(\tilde{a}0|10)
\end{aligned}
\end{equation}
where we have used the nonsignalling property of $P$ in the third
line, and the rest follows from the definition of the map (\ref{eq:ptoptilde}).

To check condition \ref{enu:ptilde-ccd}, let $N$ be as in the definition
of the map (\ref{eq:ptoptilde}) and let $a\in\alpha_{N}.$ Then,
by definition of the set $\alpha_{N+1},$ we have

\begin{equation}
P(\beta_{N}|a,x=0,y=0)=1
\end{equation}
and, therefore, 
\begin{equation}
\frac{\sum_{b\in\beta_{N}}P(ab|00)}{\sum_{b\in\mathcal{B}}P(ab|00)}=1\,,
\end{equation}
which entails 
\begin{equation}
\sum_{b\not\in\beta_{N}}P(ab|00)=0\,.
\end{equation}
Summing over $a\in\alpha_{N}\,,$ we get 
\begin{equation}
\sum_{\substack{a\in\alpha_{N}\\
b\not\in\beta_{N}
}
}P(ab|00)=\tilde{P}(01|00)=0\,.
\end{equation}
Similarly, we find $\tilde{P}(10|00)=0\,.$ Since $P$ satisfies common
certainty of disagreement, its outputs on input $x=1,y=1$ must be
perfectly correlated. That is, $P(ab|11)=0$ if $a\neq b\,.$ Hence,
\begin{equation}
\tilde{P}(01|11)=\sum_{a\neq1}P(a1|11)=0
\end{equation}
and similarly for $\tilde{P}(10|11)\,.$ So far, the nonsignalling
box corresponding to $\tilde{P}$ has two zeros in the first row and
another two in the last. Using normalisation and nonsignalling conditions
to fill in the rest of the table, we find it is of the form of the
nonsignalling box in Theorem \ref{thm:2in2outcc}. There remains to
check for disagreement, i.e. that if 
\begin{equation}
q_{A}=P(b=1|a=0,x=0,y=1)\neq P(a=1|b=0,x=1,y=0)=q_{B}
\end{equation}
then 
\begin{equation}
\tilde{P}(\tilde{b}=1|\tilde{a},\tilde{x}=0,\tilde{y}=1)\neq\tilde{P}(\tilde{a}=1|\tilde{b},\tilde{x}=1,\tilde{y}=0)\,.
\end{equation}
Since $\alpha_{N}\subseteq\alpha_{0}$ and $\beta_{N}\subseteq\beta_{0},$
$P(b=1|a^{*},x=0,y=1)\neq P(a=1|b^{*},x=1,y=0)$ holds in particular
for all $a^{*}\in\alpha_{N},b^{*}\in\beta_{N}.$ This means that,
for $a^{*}\in\alpha_{N},b^{*}\in\beta_{N},$ 
\begin{equation}
\frac{P(a^{*}1|01)}{\sum_{b\in\mathcal{B}}P(a^{*}b|01)}\neq\frac{P(1b^{*}|10)}{\sum_{a\in\mathcal{A}}P(ab^{*}|10)}
\end{equation}
and so 
\begin{equation}
P(a^{*}1|01)\sum_{a\in\mathcal{A}}P(ab^{*}|10)\neq P(1b^{*}|10)\sum_{b\in\mathcal{B}}P(a^{*}b|01)\,.
\end{equation}
Then, we can sum over $\alpha_{N}$ and $\beta_{N}$ on both sides
to find 
\begin{equation}
\sum_{a^{*}\in\alpha_{N}}P(a^{*}1|01)\sum_{\substack{a\in\mathcal{A}\\
b^{*}\in\beta_{N}
}
}P(ab^{*}|10)\neq\sum_{b^{*}\in\beta_{N}}P(1b^{*}|10)\sum_{\substack{a^{*}\in\alpha_{N}\\
b\in\mathcal{B}
}
}P(a^{*}b|01)\,.
\end{equation}
But in terms of $\tilde{P},$ this corresponds to 
\begin{equation}
\tilde{P}(01|01)\sum_{\tilde{a}\in\{0,1\}}\tilde{P}(\tilde{a}0|10)\neq\tilde{P}(10|10)\sum_{\tilde{b}\in\{0,1\}}\tilde{P}(0\tilde{b}|01)
\end{equation}
which implies 
\begin{equation}
\tilde{P}(\tilde{b}=1|\tilde{a}=0,\tilde{x}=0,\tilde{y}=1)\neq\tilde{P}(\tilde{a}=1|\tilde{b}=0,\tilde{x}=1,\tilde{y}=0)
\end{equation}
and hence the disagreement occurs for the $\tilde{P}$ distribution
as well, which proves the result.

Notice that the sets $\tilde{\alpha}_{0},\tilde{\beta}_{0}$ in the
distribution $\tilde{P}$ (defined analogously to $\alpha_{0},\beta_{0}$
in the distribution $P$) will correspond to outputs $\tilde{a},\tilde{b}=0$,
respectively. This is to be expected, as the map $P\rightarrow\tilde{P}$
gives rise to a nonsignalling box of the form of the one in Theorem
\ref{thm:2in2outcc}, where the sets $\tilde{\alpha}_{0},\tilde{\beta}_{0}$
contain a single element each. (In effect, this means we are ignoring
the outputs $a^{*}\in\alpha_{0}\backslash\alpha_{N}$ and $b^{*}\in\beta_{0}\backslash\beta_{N}$,
but those outputs lead to disagreement but not to common certainty
of it, so they can be safely discarded.) 

Thus, if there existed a quantum box with common certainty of disagreement,
there would also exist a 2-input 2-output quantum box with the same
property. By Theorem \ref{thm:2in2outcc} implies that no quantum
box can give rise to common certainty of disagreement.
\end{proof}

\section{Quantum agents cannot disagree singularly}

We explore other forms of disagreement that might arise about perfectly
correlated events. Since common certainty is a strong requirement,
we remove it and, instead, suppose that the agents assign probabilities
that differ maximally. We find that this new notion of disagreement
exhibits the same behaviour as common certainty of disagreement.

In a nonsignalling box, there is \emph{singular disagreement} about
the probabilities assigned by Alice and Bob to perfectly correlated
events $F_{A}=\lbrace(1,b,1,y)\rbrace_{b\in\mathcal{B},y\in\mathcal{Y}}$
and $F_{B}=\lbrace(a,1,x,1)\rbrace_{a\in\mathcal{A},x\in\mathcal{X}}$,
respectively, at event $(0,0,0,0)$ if it holds that 
\begin{equation}
q_{A}=1,\;q_{B}=0.\label{eq:singdiscondt}
\end{equation}
This time, there is no notion of common certainty---we just require
that Alice's and Bob's assignments differ maximally.

Similarly to the previous section, we refer to the above definition
simply as ``singular disagreement.''

We restrict ourselves first to boxes of two inputs and outputs and
show that local boxes cannot exhibit singular disagreement. Then,
we characterise the nonsignalling boxes that do satisfy singular disagreement
and show they cannot be quantum. Finally, we generalise to boxes of
any number of inputs and outputs. 
\begin{thm}
\label{thm:loc-sd}There is no local 2-input 2-output box that gives
rise to singular disagreement. 
\end{thm}

\begin{proof}
Assume Alice and Bob input $x=y=0$ and obtain $a=b=0.$ This implies
\begin{equation}
P(00|00)>0.\label{eq:singdis00>0-1}
\end{equation}
Alice assigns 
\begin{equation}
P(b=1|a=0,x=0,y=1)=1,\label{eq:singdisAlice1-1}
\end{equation}
and Bob assigns 
\begin{equation}
P(a=1|b=0,x=1,y=0)=0\,.\label{eq:singdisBob0-1}
\end{equation}
Further, the outputs for input $(x,y)=(1,1)$ are perfectly correlated,
so, in particular, 
\begin{equation}
P(01|11)=0.\label{eq:singdisPerfcorr-1}
\end{equation}
. Equations \eqref{eq:singdisAlice1-1} and \eqref{eq:singdisBob0-1}
imply, respectively, 
\begin{equation}
P(00|01)=0\textnormal{ and }P(10|10)=0.\label{eq:singdisAliceBob}
\end{equation}
However, equations \eqref{eq:singdis00>0-1}, \eqref{eq:singdisPerfcorr-1}
and \eqref{eq:singdisAliceBob} make up a form of Hardy's paradox
\cite{hardy_quantum_1992}, which is known not to hold for local distributions. 
\end{proof}
We now lift the local restriction and characterise the nonsignalling
boxes in which singular disagreement occurs. 
\begin{thm}
\label{thm:singdis} A 2-input 2-output nonsignalling box gives rise
to singular disagreement if and only if it takes the form of Table
\ref{tab:nssd}. 
\begin{table}[ht]
\centering %
\begin{tabular}{|c|c|c|c|c|}
\hline 
$xy\backslash ab$  & 00  & 01  & 10  & 11\tabularnewline
\hline 
\hline 
00  & $s$  & $t$  & $1-s-u-t$  & $u$\tabularnewline
\hline 
01  & 0  & $s+t$  & $r$  & $1-s-t-r$\tabularnewline
\hline 
10  & $1-u-t$  & $u+t+r-1$  & 0  & $1-r$\tabularnewline
\hline 
11  & $r$  & 0  & 0  & $1-r$\tabularnewline
\hline 
\end{tabular}\caption{Parametrisation of 2-input 2-output nonsignalling boxes with singular
disagreement. Here, $r,\,s,\,t,\,u,\,\in[0,1]$ are such that all
the entries of the box are non-negative, $s>0$, and $s+t\protect\neq0$
and $u+t\protect\neq1$.}
\label{tab:nssd} 
\end{table}
\end{thm}

\begin{proof}
First, we show that singular disagreement implies that the nonsignalling
box must be of the above form. By construction, the inputs $x=y=1$
have perfectly correlated outputs, so that 
\begin{equation}
P(01|11)=P(10|11)=0\,.
\end{equation}
Also, singular disagreement requires

\begin{eqnarray}
 & P(b=1|a=0,x=0,y=1) & =1,\label{eq:sd-Alice1}\\
 & P(a=1|b=0,x=1,y=0) & =0.\label{eq:sd-Bob0}
\end{eqnarray}
Equation (\ref{eq:sd-Alice1}) implies that $P(00|01)=0$ and $P(01|01)\neq0$,
while Equation (\ref{eq:sd-Bob0}) implies that $P(10|10)=0$ and
$P(00|10)\neq0$. The rest of the entries follow from normalisation
and nonsignalling conditions. The condition $s>0$ ensures that $P(00|00)>0$,
as per the input and output that the agents in fact obtained. Therefore,
any two-input two-output nonsignalling box that gives rise to singular
disagreement must be of the above form.

Proving the converse is straightforward, as it suffices to check that
equations (\ref{eq:sd-Alice1}) and (\ref{eq:sd-Bob0}) are satisfied
for the parameters of the box. 
\end{proof}
However, singular disagreement cannot arise in quantum systems. This
is another way in which quantum mechanics provides some consistency
between (possibly incompatible) measurements, just like in the case
of common certainty of disagreement.
\begin{thm}
\label{thm:sd-not-quantum-2x2}No 2-input 2-output quantum box can
give rise to singular disagreement. 
\end{thm}

\begin{proof}
Due to their form, the boxes in Theorem \ref{thm:singdis} are \textit{quantum
voids} \cite{rai_geometry_2019}; i.e., they are either local or post-quantum.
This can be seen by observing that the mapping 
\begin{equation}
x\mapsto x\oplus1\,,
\end{equation}
which is a symmetry of the box, makes all four 0's lie in entries
$P(ab|xy)$ such that $a\oplus b\oplus1=xy$. As stated in Sections
III and V.B of Ref.~\cite{rai_geometry_2019}, all boxes with four
0's in entries of the above form are quantum voids.

Therefore, the box in Table \ref{tab:nssd} is either local, in which
case it does not lead to singular disagreement, or has no possible
quantum realisation, proving the claim. 
\end{proof}
Finally, the above results can be generalised to any finite box: 
\begin{thm}
\label{thm:sd-not-quantum} No quantum box can give rise to singular
disagreement. 
\end{thm}

\begin{proof}
Like in Theorem \ref{thm:generalization}, we show that any nonsignalling
box with singular disagreement induces a 2-input 2-output nonsignalling
box with the same property, and rely on Theorem \ref{thm:sd-not-quantum-2x2}
to deduce that no quantum system can give rise to singular disagreement.
The mapping from $P$ to $\tilde{P}$ for inputs $x,y=0$ is as in
Theorem \ref{thm:generalization} but substitutes $\alpha_{N},\beta_{N}$
for $\alpha_{0},\beta_{0}$ respectively.

Analogously to Theorem \ref{thm:generalization}, to prove the Theorem
for singular disagreement we define a mapping from a distribution
$\left\{ P(ab|xy)\right\} _{a\in\mathcal{A},b\in\mathcal{B},x\in\mathcal{X},y\in\mathcal{Y}}$
to an effective distribution $\left\{ \tilde{P}(\tilde{a}\tilde{b}|\tilde{x}\tilde{y})\right\} _{\tilde{a},\tilde{b},\tilde{x},\tilde{y}\in\left\{ 0,1\right\} }$
such that the following conditions hold: 
\begin{enumerate}
\item if $\left\{ P(ab|xy)\right\} $ is quantum, then so is $\left\{ \tilde{P}(\tilde{a}\tilde{b}|\tilde{x}\tilde{y})\right\} \,,$\label{enu:ptilde-singdis-qu} 
\item if $\left\{ P(ab|xy)\right\} $ satisfies singular disagreement, then
so does $\left\{ \tilde{P}(\tilde{a}\tilde{b}|\tilde{x}\tilde{y})\right\} \,.$\label{enu:ptilde-singdis} 
\end{enumerate}
Again, the number of inputs can be reduced to 2 without loss of generality\textbf{.}
To group the outputs, we notice that the sets $A_{0},B_{0}$ also
play a role in singular disagreement, as they group the outputs of
each party which lead them to assign their respective probabilities
to the event. Then, we group the outputs according to whether or not
they belong in the sets $\alpha_{0},\beta_{0}$ (for input 0) and
whether or not they correspond to the event obtaining, i.e. whether
or not they are equal to 1 (for input 1). We obtain the same mapping
(\ref{eq:ptoptilde}) as before, substituting $\alpha_{N}$ for $\alpha_{0}$
and $\beta_{N}$ for $\beta_{0}$. With this replacement, condition
\ref{enu:ptilde-singdis-qu} follows by the same argument as before.
To check condition \ref{enu:ptilde-singdis}, we know that, for all
$a^{*}\in\alpha_{0},$ 
\begin{equation}
P(b=1|a^{*},x=0,y=1)=1
\end{equation}
and so 
\begin{equation}
\begin{aligned}P(a^{*}1|01) & =\sum_{b\in\mathcal{B}}P(a^{*}b|01).\end{aligned}
\end{equation}
Summing over $a^{*}\in\alpha_{0}$ and rewriting the expression in
terms of $\tilde{P},$ we find 
\begin{equation}
\tilde{P}(01|01)=\sum_{\tilde{b}\in\{0,1\}}\tilde{P}(0\tilde{b}|01)
\end{equation}
which implies 
\begin{equation}
\tilde{P}(\tilde{b}=1|\tilde{a}=0,\tilde{x}=0,\tilde{y}=1)=1.
\end{equation}
Similarly, for all $b^{*}\in\beta_{0}$ we have 
\begin{equation}
P(a=1|b^{*},x=1,y=0)=0,
\end{equation}
hence 
\begin{equation}
P(1b^{*}|10)=0
\end{equation}
and so, by adding over $b^{*}\in\beta_{0}$ and mapping to $\tilde{P},$
we find 
\begin{equation}
\tilde{P}(10|10)=0
\end{equation}
as required.
\end{proof}
%\bibliographystyle{amsalpha}
%\bibliography{entnonloc}

%\end{document}

%%% LyX 2.3.0rc1 created this file.  For more info, see http://www.lyx.org/.
%%% Do not edit unless you really know what you are doing.
%\documentclass[oneside,english]{book}
%\usepackage[T1]{fontenc}
%\usepackage[latin9]{inputenc}
%\usepackage{geometry}
%\geometry{verbose,tmargin=1in,bmargin=1in,lmargin=1in,rmargin=1in}
%\setcounter{secnumdepth}{3}
%\setcounter{tocdepth}{3}
%\usepackage{amsmath}
%\usepackage{amsthm}
%\usepackage{setspace}
%\onehalfspacing
%
%\makeatletter
%%%%%%%%%%%%%%%%%%%%%%%%%%%%%%% User specified LaTeX commands.
%\usepackage{physics}
%\usepackage{cite}
%\usepackage{bbm}
%\usepackage{enumitem}
%\usepackage{amsfonts}
%
%
%\renewcommand{\theenumi}{(\roman{enumi})}
%\renewcommand{\labelenumi}{\theenumi}
%\providecommand{\ketbra}[1]{\ket{#1}\bra{#1}}
%\providecommand{\trace}{\textnormal{tr}}
%\providecommand{\FSP}{\mathcal{FSP}}
%
%\newcommand{\ui}{\mathrm{i}}
%\newcommand{\ue}{\mathrm{e}}
%
%
%
%\usepackage{babel}
%\providecommand{\definitionname}{Definition}
%\providecommand{\theoremname}{Theorem}
%
%\makeatother
%
%\usepackage{babel}
%\begin{document}
%\setcounter{chapter}{5}

\chapter{\label{chap:concl}Conclusions}

\section{Multipartite entanglement and nonlocality}

In a time where some quantum technologies are already within our reach,
the theoretical study of multicomponent systems is essential to develop
more applications and implement the ones that have already been proposed.
Entanglement and nonlocality are the main quantum effects behind these
applications, hence they have been the main object of study of this
thesis. First, we have ordered the set of multipartite entangled states
via a resource theory, in order to understand which multipartite states
are more useful than which others. Then, we have focused on a particular
kind of multipartite states that are arguably the easiest to implement
in practice: pair-entangled network states. We have shown that GMNL
is intrinsic to these states if the pair-entanglement is pure, and
given fundamental limitations to the obtention of GMNL and even GME
from some of these states when the pair-entanglement is mixed. Further,
we have shown several ways to overcome these limitations, including
by adding more connections to the network or taking several copies
of it (thus achieving superactivation).

In Chapter \ref{chap:maxent}, we have addressed the problem of ordering
the set of multipartite entangled states. While LOCC and its stochastic
variant give rise to inequivalent forms of entanglement and isolated
states which cannot be converted to or from any other state (hence
rendering the resource theory trivial), we have shown that enlarging
the set of free operations makes it possible to obtain non-trivial
resource theories of entanglement without inequivalent classes. However,
no resource theory of non-full-separability can have a maximally entangled
state for 3-qubit states, since this is not possible under full separability-preserving
transformations, the largest conceivable class of free operations.
While we conjecture that this no-go result extends beyond 3 qubits,
in future work it would be interesting to study whether it holds in
full generality:
\begin{itemize}
\item Can there exist a resource theory of non-full-separability with a
maximally entangled state for more than 3 parties, or local dimension
larger than 2?
\end{itemize}
On the other hand, the biseparability-preserving paradigm induces
a resource theory of GME with a maximally resourceful state. Given
this positive result, it would be interesting to analyse further features
of this theory. In particular, 
\begin{itemize}
\item Can we find an operational grounding to the conceptually satisfying
structure that biseparability-preserving operations induce?
\end{itemize}
Despite the fact that the resource theory of pure bipartite entanglement
yields only a partial order for single-copy transformations \cite{nielsen_conditions_1999},
asymptotic transformations give rise to a total order in terms of
the entropy of entanglement by measuring the cost and distillation
rates with respect to the maximally entangled state \cite{bennett_concentrating_1996}.
Thus, the existence of this state acts as a gold standard that leads
to a unique measure of entanglement in the asymptotic setting. It
is known that this is not possible for pure multipartite states. An
asymptotically reversible theory in this case cannot exist with respect
to a single reference state \cite{bennett_exact_2000}, a non-surprising
result perhaps given the lack of existence of a unique maximally multipartite
entangled state under the LOCC paradigm. This has led to the search
of a minimal reversible entanglement generating set (MREGS), which
would at least enable to define a collection of reversible asymptotic
rates with respect to the states in this set. However, progress in
this problem has been scarce and it is believed that that the cardinality
of the MREGS might be infinite. Our result that the single-copy resource
theory of entanglement under BSP operations has a unique maximally
GME state invites one to think that an asymptotically reversible theory
of pure-state GME could be possible in this paradigm using this state
as the reference state.
\begin{itemize}
\item Do BSP operations lead to an asymptotically reversible theory of pure-state
GME in which the maximally GME state acts a gold standard to measure
the cost and distillation rates? If so, what would be the corresponding
unique measure of GME for pure states?
\end{itemize}
Regarding mixed states, the work of Refs. \cite{brandao_entanglement_2008,brandao_generalization_2010}
shows that the resource theory of multipartite entanglement under
FSP is not reversible, which includes the case of BSP for two parties.
Nevertheless, Refs. \cite{brandao_entanglement_2008,brandao_generalization_2010}
show that such a theory is possible by extending the set of FSP operations
to asymptotically FSP operations. Furthermore, \cite{brandao_reversible_2015}
shows that this result remains true for any resource theory fulfilling
some general postulates under asymptotically resource non generating
operations . However, this does not extend to our case because GME
dos not meet Postulate 1 therein (the set of biseparable states is
not closed under tensor products).
\begin{itemize}
\item Is an asymptotically reversible theory of GME for general (mixed)
states possible under BSP or asymptotically BSP operations?
\end{itemize}
In the realm of quantum networks, in Chapters \ref{chap:gmnl} and
\ref{chap:mixed} we have studied pair-entangled network states, i.e.,
multipartite states where each party shares a bipartite entangled
state with one or more of the others. Strikingly, we have shown that
GMNL is intrinsic to these networks if the bipartite entangled states
are pure: GMNL can be obtained by distributing arbitrary pure bipartite
entanglement in any connected topology. This paves the way towards
feasible generation of GMNL from any network. In fact, our results
imply that, given a set of nodes, distributing pure-state entanglement
in the form of a tree is sufficient to observe GMNL.

Further, we have shown that a tensor product of finitely many GME
states is always GMNL. However, our construction is not necessarily
optimal in the number of copies, therefore we ask:
\begin{itemize}
\item What is the smallest number of copies of a pure GME state needed to
obtain GMNL?
\end{itemize}
And, in particular, the multipartite analogue of Gisin's theorem remains
open:
\begin{itemize}
\item Do all single-copy pure GME states give rise to GMNL?
\end{itemize}
On a different note, the assumption that the distributions $P_{M},\,P_{\overline{M}}$
are nonsignalling in the GMNL definition is physically natural. Still,
removing it raises the stakes to achieve nonlocality. Therefore,
\begin{itemize}
\item Is it possible to establish analogous results to those in Chapter
\ref{chap:gmnl} with the stronger definition of GMNL where the distributions
$P_{M},\,P_{\overline{M}}$ may be signalling? 
\end{itemize}
Very recently, Ref. \cite{navascues_genuine_2020} proposed the concept
of ``genuine network entanglement'', a stricter notion than GME
which rules out states which are a tensor product of non-GME states.
One might hope that states that are GME but not genuine network entangled
might be detected device independently by not passing GMNL tests.
However, our results show this will not work. Any distribution of
pure bipartite states, even with arbitrarily weak entanglement, always
displays GMNL as long as all parties are connected. This further motivates:
\begin{itemize}
\item Can an analogous concept of genuine network nonlocality be found,
that may detect genuine network entanglement?
\end{itemize}
In practical applications, the entanglement shared by the nodes of
a network would unavoidably degrade to mixed-state form. By continuity,
the GMNL in the pure pair-entangled networks considered in Chapter
\ref{chap:gmnl} must be robust to some noise. However, as we showed
in Chapter \ref{chap:mixed}, topology plays a key role in the entanglement
and nonlocality properties of general mixed pair-entangled networks.
In particular, tree networks are not sufficient to establish GME between
the nodes, even for arbitrarily low noise, if the networks are large
enough. In sharp contrast, a completely connected network exhibits
GME for any number of parties for all visibilities above a threshold.
While distributing bipartite entanglement in the edges of a network
is experimentally very feasible, adding edges to the network undoubtedly
comes at a cost. A scheme in which a resourceful central lab prepares
entangled states and sends them to the remaining less powerful parties,
as in Ref. \cite{vardoyan_stochastic_2019}, is doomed to failure
in any realistic scenario in which entanglement preparation and distribution
is bound to a certain degree of noise. Such protocols can only give
rise to genuine multipartite effects for a bounded number of parties.
In fact, our study also shows that this does not work if all parties
are able to distribute entanglement but with moderate capacities so
as to lead to tree networks. Our results require that all parties
are technologically fully capable to entangle themselves with all
others. For this reason, it would be crucial for applications to establish
a middle ground between our results. Understanding whether a square
lattice can give rise to GME might be a good starting point, and,
in more generality, we ask:
\begin{itemize}
\item What is the network with the lowest connectivity that leads to GME
for non-zero noise in the shared states, when the network is large?
\end{itemize}
Conversely, we have provided new network states that are GME but not
GMNL, and our constructions can be used to establish new such examples.
Still, we have found that the main factor compromising the GMNL of
a network is the non-steerability of the states in one or more of
the nodes. Locality is a weaker condition than non-steerability, and
the possibility of having local (possibly steerable) states forming
a GMNL network remains open:
\begin{itemize}
\item Can a network of local states give rise to GMNL?
\end{itemize}

While relatively low noise on the edges can already compromise GMNL
in the network, we have shown that taking many copies can restore
the nonlocality. We have provided an example of superactivation of
GMNL in networks, which to our knowledge constitutes a completely
new result. Further, the ideas presented here go beyond this specific
example, and can be used to construct more networks exhibiting this
phenomenon.

The understanding of pair-entangled networks, in particular for applications,
would be significantly advanced by answering:
\begin{itemize}
\item Which pair-entangled network topologies and noise tolerances can lead
to GMNL?
\end{itemize}
Finally, our results show that GME is robust in the fully connected
network as the number of parties grows. Extending this result to GMNL
remains an open question:
\begin{itemize}
\item Is the fully connected network robust not only for GME, but also to
GMNL, as the number of parties grows?
\end{itemize}

\section{A physical principle from observers' agreement}

In addition to the results on multipartite entanglement and nonlocality,
in this thesis we have also questioned whether the quantum description
of Nature is the best possible, or the only one possible. In order
to constrain the set of theories that are physically `reasonable',
we have provided a principle that should be satisfied by all physical
theories: the impossibility of disagreement. In Chapter \ref{chap:agreement}
we have defined two notions of disagreement, common certainty of disagreement
and singular disagreement, inspired by notions from epistemics. We
have shown that nonsignalling boxes can be disagreeing in each of
these senses, while quantum and local boxes cannot.

Additionally, both notions of disagreement induce an immediate test
for new theories---namely, the tables in Theorems \ref{thm:2in2outcc}
and \ref{thm:singdis}. These tests are very general, in the sense
that they are based only on the capability of a theory to realise
undesirable correlations between non-communicating parties. Also,
both principles have their roots in epistemics, common certainty of
disagreement being closer to Aumann's original idea, singular disagreement
having a simpler description.

These two definitions are compatible, and it is indeed possible to
find examples displaying both kinds of disagreement at once. Strikingly,
a prime example of this is given by the Popescu-Rohrlich box \cite{popescu_quantum_1994},
proving that it is not only an extremal resource as an extreme point
of the polytope of nonsignalling distributions, but also as a disagreeing
distribution in the strongest possible sense.

On a speculative note, it would be very interesting to explore the
application of the notions introduced here to practical tasks in which
consensus between parties plays a role, such as the coordination of
the action of distributed agents or the verification of distributed
computations. The impossibility of disagreement could be useful in
order for distant agents to coordinate while not having access to
each other's full information. Ref. \cite{fagin_reasoning_2004} proposes
some specific connections along these lines in the classical case,
and it would be interesting to find such connections in the quantum
realm:
\begin{itemize}
\item Can the impossibility of quantum disagreement be used to perform some
practical information-processing task?
\end{itemize}
As hinted above, our results suggest that agreement can be used to
design experiments to test the behaviour of Nature. In experimental
settings, noise is unavoidable. Adding white noise to the boxes in
Tables \ref{tab:nsccd} and \ref{tab:nssd} (both of which are quantum
voids) would mean the zeros in the boxes are now small but finite
parameters. Ref. \cite{rai_geometry_2019} claimed that there is strong
numerical evidence of the robustness of quantum voids to this type
of noise, which would imply the same kind of robustness in our results.
Still, analytical confirmation of this phenomenon would be desirable:
\begin{itemize}
\item If a box with common certainty of disagreement or singular disagreement
is mixed with white noise, is it still impossible for it to be quantum?
\end{itemize}
Alternatively, another future direction for continuing this work concerns
defining approximate notions for disagreement. Several approaches
are possible, and they are compatible:
\begin{itemize}
\item Is the impossibility of quantum disagreement preserved if the events
which were perfectly correlated are now \emph{approximately} perfectly
correlated?
\item Is the impossibility of quantum common certainty of disagreement preserved
if the certainty is approximate? I.e., if the agents assign a probability
bounded away from 1 to each other's outcomes?
\item Is the impossibility of quantum singular disagreement preserved if
the difference between the agents' estimations of probabilities is
bounded away from 1?
\end{itemize}
The robustness of our results to different kinds of noise would make
it possible to test our principle experimentally. Obtaining disagreeing
correlations in an experiment would be groundbreaking for science,
as it would imply both that disagreement is not a physical principle
\emph{and }that Nature is not quantum. While this seems to be a very
challenging question to tackle, future work should move towards answering:
\begin{itemize}
\item Can disagreeing correlations be found in Nature?
\end{itemize}
A complementary approach is the study of disagreement in theories
generalising quantum theory. For instance, almost quantum correlations
\cite{navascues_almost_2015} is a set of correlations strictly larger
than those achievable by measuring quantum states but that were designed
to satisfy all physical principles previously proposed in the literature.
While their usefulness has been since questioned \cite{sainz_almost_2018},
almost quantum correlations are well characterised in terms of nonsignalling
boxes, so a natural question is whether they display common certainty
of disagreement or singular disagreement. First, a straightforward
adjustment of the proof of Theorem \ref{thm:ccd-not-quantum} shows
that common certainty of disagreement is \textit{not} present in almost
quantum correlations. As for singular disagreement, due to the simple
characterisation of almost quantum correlations in terms of a semidefinite
program \cite{navascues_convergent_2008}, it is possible numerically
to search for almost quantum boxes displaying singular disagreement.
Using this, we have found numerical evidence that singular disagreement
is also \emph{not} present in almost quantum correlations. Hence,
our principles give more support to the claim that quantum theory
need not be the best description of Nature: new theories giving rise
to almost quantum correlations can be physically reasonable. In any
case, it would be desirable to obtain an analytical proof that almost
quantum correlations cannot display singular disagreement:
\begin{itemize}
\item Can we show analytically that almost quantum correlations cannot give
rise to singular disagreement?
\end{itemize}
More importantly, the quest for physical principles external to quantum
theory has given rise to several proposals, and continues to be a
fruitful line of research. Understanding their compatibility, i.e.,
whether or not some are implied by, or equivalent to, others, is crucial
if we are to use such principles to constrain the allowed correlations
in Nature. However, the variety of ways in which such principles are
phrased makes this a challenging task. In what concerns our work,
an ambitious open question is:
\begin{itemize}
\item Does disagreement imply, or is it implied by, any of the other physical
principles proposed so far?
\end{itemize}
Finally, this work is a very modest step towards characterising quantum
theory in terms of constraints external to it. Possibly the main open
question that this Chapter leaves is to complete this task:
\begin{itemize}
\item Can quantum theory be characterised in terms of external principles?
\end{itemize}
%\bibliographystyle{amsalpha}
%\bibliography{entnonloc}
%
%\end{document}

%% OBSERVACION
% Para el correcto funcionamiento de caracteres especiales: acento, ñ
% el archivo .tex debe estar codificado con UTF-8.
% Para crear un capitulo nuevo se sugiere duplicar (copiar y pegar con distinto nombre)
% un capitulo existente

% Conclusiones
%\input{conclusiones/conclusiones}

%% ************************ Bibliografia ********************
%% Crea la bibliografia utilizando bibtex
%% Carga la informacion del archivo bibliografia/referencias.bib
%\input{clases/bibliografia}
\newcommand{\etalchar}[1]{$^{#1}$}
\providecommand{\bysame}{\leavevmode\hbox to3em{\hrulefill}\thinspace}
\providecommand{\MR}{\relax\ifhmode\unskip\space\fi MR }
\providecommand{\MRhref}[2]{%
  \href{http://www.ams.org/mathscinet-getitem?mr=#1}{#2}
}
\providecommand{\href}[2]{#2}

%% ************************* Apendices ***********************
%% Todo que aparezca a partir del comando \appendix es numerado
%% con letras. Incluir apendices relevantes en archivos separados
%\appendix 

%\input{apendice1/apendice1}
\end{document}